\documentclass[onecolumn]{IEEEtran}

\usepackage[utf8]{inputenc}
\usepackage{amsmath}
\usepackage{amsthm}
\usepackage{mathtools}
\usepackage{amsfonts}
\usepackage{bbm}
\usepackage{hyperref}
\usepackage{url}
\usepackage{stackengine}
\usepackage{tikz}
\usepackage{pgfplots}
\pgfplotsset{compat=1.13}
\usepackage[nonumberlist]{glossaries}
\glsdisablehyper
\usepackage[noadjust]{cite}
\usepackage{aligned-overset}
\usepackage{todonotes}
\usepackage{graphicx}
\usepackage{xcolor}
\usepackage{caption}
\usepackage{subcaption}
\usepackage{siunitx}
\graphicspath{ {./Figures/} }

\newcommand\blankfootnote[1]{%
  \begin{NoHyper}
  \let\svthefootnote\thefootnote%
  \let\thefootnote\relax\footnotetext{#1}%
  \let\thefootnote\svthefootnote%
  \end{NoHyper}
}

\newcommand{\inputAlphabet}{\mathcal{X}}
\newcommand{\inputAlphabetElement}{x}
\newcommand{\inputAlgebra}{\Sigma}

\newcommand{\hilbertSpace}{\mathcal{H}}
\newcommand{\eveHilbertSpace}{\hilbertSpace_\eve}
\newcommand{\bobHilbertSpace}{\hilbertSpace_\bob}
\newcommand{\densityOperators}[1]{\mathcal{S}({#1})}
\newcommand{\classicalQuantumChannel}{D}
\newcommand{\bobcqChannel}{\classicalQuantumChannel_\bob}
\newcommand{\eveChannel}{\classicalQuantumChannel_\eve}
\newcommand{\inputDistribution}{P}
\newcommand{\outputDistribution}{U}
\newcommand{\naturals}{\mathbb{N}}
\newcommand{\reals}{\mathbb{R}}
\newcommand{\blocklength}{n}
\newcommand{\blockindex}{i}
\newcommand{\codebookRate}{R}
\newcommand{\traceClass}[1]{\mathcal{T}({#1})}
\newcommand{\compactOperators}[1]{\mathcal{K}({#1})}
\newcommand{\traceNorm}[1]{\left\lVert {#1} \right\rVert_\mathrm{tr}}
\newcommand{\traceNormBiggLeft}{\Bigg\lVert}
\newcommand{\traceNormBiggRight}{\Bigg\rVert_\mathrm{tr}}
\newcommand{\operatorNorm}[1]{\left\lVert {#1} \right\rVert_\mathrm{op}}
\newcommand{\generalIndexMax}{\ell}
\newcommand{\generalIndex}{i}
\newcommand{\generalIndexTwo}{j}
\newcommand{\singularValue}{s}
\newcommand{\generalConstant}{c}
\newcommand{\generalPolynomial}{p}
\newcommand{\generalPolynomialCoefficient}{a}
\newcommand{\generalIntervalLowerBound}{a}
\newcommand{\generalIntervalUpperBound}{b}
\newcommand{\generalFunction}{f}
\newcommand{\generalFunctionTwo}{g}
\newcommand{\generalReal}{t}
\newcommand{\generalRealTwo}{s}
\newcommand{\generalQuantumState}{\rho}
\newcommand{\generalQuantumStateTwo}{\sigma}
\newcommand{\hilbertSpaceElement}{h}
\newcommand{\innerProduct}[2]{\left\langle {#1}, {#2} \right\rangle}
\newcommand{\eigenvalue}{\lambda}
\newcommand{\eigenvalueIndex}{\ell}
\newcommand{\eigenvector}{e}
\newcommand{\bra}[1]{\left\langle {#1} \right\rvert}
\newcommand{\ket}[1]{\left\lvert {#1} \right\rangle}
\newcommand{\rankOneOperator}[1]{\ket{#1}\bra{#1}}
\newcommand{\maxnorm}[1]{\left\lVert {#1} \right\rVert_\infty}
\newcommand{\boundedOperators}[1]{\mathcal{B}({#1})}
\newcommand{\complexNumbers}{\mathbb{C}}
\newcommand{\generalOperator}{A}
\newcommand{\generalOperatorTwo}{B}
\newcommand{\generalOperatorRV}{A}
\newcommand{\dualSpace}[1]{#1'}
\newcommand{\trace}{\mathrm{tr}}
\newcommand{\generalFunctional}{\varphi}
\newcommand{\onbVector}{e}
\newcommand{\adjoint}[1]{#1^*}
\newcommand{\indicatorFunction}[1]{\mathbbm{1}_{#1}}
\newcommand{\spectrum}[1]{\sigma\left({#1}\right)}
\newcommand{\probabilitySpace}{\Omega}
\newcommand{\probabilitySpaceAlgebra}{\mathcal{F}}
\newcommand{\Expectation}{\mathbb{E}}
\newcommand{\Probability}{\mathbb{P}}
\newcommand{\generalPMeasure}{\mu}
\newcommand{\probabilitySpaceElement}{\omega}
\newcommand{\classicalOutputIndex}{y}
\newcommand{\classicalOutputSubset}{\mathcal{Y}}
\newcommand{\inputRV}{X}
\newcommand{\outputRV}{Y}
\newcommand{\jointEntropy}{H_P}
\newcommand{\outputEntropy}{H_U}
\newcommand{\jointEntropySubscript}[1]{H_{P_{#1}}}
\newcommand{\outputEntropySubscript}[1]{H_{U_{#1}}}
\newcommand{\quantumInformation}{\holevoInformation{P}{D}}
\newcommand{\holevoInformation}[2]{\chi(#1;#2)}
\newcommand{\quantumRelativeEntropy}[2]{H(#1;#2)}
\newcommand{\vonNeumannEntropy}[1]{H\left(#1\right)}
\newcommand{\legitInformation}{I(P,W)}
\newcommand{\legitRenyDiv}[1]{I_{#1}(P,W)}
\newcommand{\codebook}{\mathcal{C}}
\newcommand{\codewordIndex}{m}
\newcommand{\codebookSize}{M}
\newcommand{\finalconst}{\gamma}
\newcommand{\proofconst}{\beta}
\newcommand{\renyiorder}{\alpha}
\newcommand{\operatorValuedMap}{T}
\newcommand{\rademacherRV}{\mathcal{E}}
\newcommand{\typicalityParameter}{\varepsilon}
\newcommand{\jointTypicalitySet}{\mathcal{P}}
\newcommand{\outputTypicalitySet}{\mathcal{U}}
\newcommand{\jointTypicalityOperator}{\Psi}
\newcommand{\outputTypicalityOperator}{\Theta}
\newcommand{\typicalityOperatorProduct}{\Gamma}
\newcommand{\identityOperator}{\mathbf{1}}
\newcommand{\conditionalRenyiEntropy}[1]{H_{P,{#1}}}
\newcommand{\outputRenyiEntropy}[1]{H_{U,{#1}}}
\newcommand{\conditionalRenyiEntropySubscript}[2]{H_{P_{#2},{#1}}}
\newcommand{\outputRenyiEntropySubscript}[2]{H_{U_{#2},{#1}}}
\newcommand{\generalpmf}{p}
\newcommand{\cardinality}[1]{\left\lvert #1 \right \rvert}
\newcommand{\absolute}[1]{\left \lvert #1 \right \rvert}
\newcommand{\legitOutputAlphabet}{\hat{\mathcal{X}}}
\newcommand{\legitOutputAlphabetElement}{\hat{x}}
\newcommand{\legitOutputRV}{\hat{X}}
\newcommand{\legitOutputAlgebra}{\hat{\Sigma}}
\newcommand{\bobChannel}{W}
\newcommand{\alice}{\mathfrak{A}}
\newcommand{\bob}{\mathfrak{B}}
\newcommand{\eve}{\mathfrak{E}}
\newcommand{\differenceBounds}{v}
\newcommand{\encoder}{\mathrm{Enc}}
\newcommand{\decoder}{\mathrm{Dec}}
\newcommand{\decodingError}{\varepsilon}
\newcommand{\securityNumber}{\delta}
\newcommand{\messageRV}{\mathfrak{M}}
\newcommand{\randomnessRate}{\tilde{R}}
\newcommand{\combinedRate}{\hat{R}}
\newcommand{\numCodebooks}{L}
\newcommand{\indexCodebooks}{\ell}
\newcommand{\RNDerivative}[2]{\frac{d{#1}}{d{#2}}}
\newcommand{\legitOutputDistribution}{\hat{P}}
\newcommand{\messageSpacePartition}{\Pi}

\newcommand{\messageSpacePartitionElement}{\pi}
\newcommand{\eavesdropperSuccessProbability}{\mathrm{Succ}}
\newcommand{\eavesdropperGuessProbability}{\mathrm{Guess}}
\newcommand{\messageProbabilityDistribution}{P_\messageRV}
\newcommand{\eavesdropperAdvantageWeak}{\mathrm{Adv}_\mathrm{weak}}
\newcommand{\eavesdropperAdvantageStrong}{\mathrm{Adv}_\mathrm{str}}
\newcommand{\eavesdropperAdvantageMutualInformation}{\mathrm{Adv}_\mathrm{inf}}
\newcommand{\eavesdropperAdvantageDistinguishing}{\mathrm{Adv}_\mathrm{dist}}
\newcommand{\eavesdropperAdvantageSemantic}{\mathrm{Adv}_\mathrm{sem}}
\newcommand{\povm}{F}
\newcommand{\evepmf}{q}

\newcommand{\legitInformationDensity}[3]{i_{#1}({#2};{#3})}
\newcommand{\generalProbabilitySpace}{\mathcal{X}}
\newcommand{\generalProbabilitySpaceElement}{x}
\newcommand{\generalRV}{X}
\newcommand{\hilbertSpaceDimension}{d}

\newcommand{\decoderpovm}{Y}
\newcommand{\typicalityOperatorProductDecoder}{\Phi}
\newcommand{\hnconstant}{c}
\newcommand{\pseudoinverse}[1]{#1^{-1}}
\newcommand{\generalNatural}{k}
\newcommand{\image}{\mathrm{im}}
\newcommand{\kernel}{\mathrm{ker}}
\newcommand{\costConstraint}{C}
\newcommand{\costFunction}{c}
\newcommand{\badCodewordsSet}{\mathbb{B}}

\newcommand{\generalFiniteSet}{V}
\newcommand{\generalFiniteSetElement}{v}
\newcommand{\inputAlphabetSubset}{B}
\newcommand{\LTwoOnReals}{L^2(\reals)}

\newcommand{\photonNumberNatural}{k}
\newcommand{\coherentStateVector}{\zeta}

\newcommand{\quantumNoisePower}{N}
\newcommand{\gaussiancqUnitary}[1]{U_{#1}}
\newcommand{\inputEnergy}{E}
\newcommand{\gordonFunction}{g}

\newcommand{\eveChannelOutput}[1]{\classicalQuantumChannel_{\eve, #1}}
\newcommand{\gibbsObservable}{G}

\newcommand{\gibbsParameter}{\beta}
\newcommand{\gibbsEnergy}{E}
\newcommand{\binaryEntropy}{h}
\newcommand{\gibbsEigenvalue}{g}
\newcommand{\gibbsEigenvector}{e}
\newcommand{\gibbsEntropy}[1]{H_{#1}}
\newcommand{\expConstant}{\alpha}
\newcommand{\multConstant}{c}
\newcommand{\indexModes}{j}
\newcommand{\numModes}{s}
\newcommand{\frequency}{\omega}
\newcommand{\outputStateEve}{\rho_{\eve}^{\blocklength}}
\newcommand{\evePostprocChannel}{\mathcal{N}}
\newcommand{\transmittivity}{\eta}
\newcommand{\receiverNoise}{\quantumNoisePower}
\newcommand{\beamSplitter}{B}
\newcommand{\additiveNoiseChannel}{\mathcal{N}}
\newcommand{\distance}{L_\mathrm{dist}}
\newcommand{\diameter}{L_\mathrm{diam}}
\newcommand{\beamsplitterOperation}{\mathcal{V}}
\newcommand{\env}{\textup{env}}
\newcommand{\identityMap}{\textrm{id}}

\DeclareMathOperator*{\argmin}{arg\,min}
\newcommand{\beamSplitterAngle}{\theta}
\newcommand{\distributedAccTo}{\sim}
\newcommand{\atypicalTermsFunc}{\mathcal{R}}
\newcommand{\resolvabilityBoundFunc}{\mathcal{W}_\mathrm{res}}
\newcommand{\costBoundFunc}{\mathcal{W}_\mathrm{cost}}
\newcommand{\codingBoundFunc}{\mathcal{W}_\mathrm{coding}}
\newcommand{\classicalCodingBoundFunc}{\mathcal{W}_\mathrm{class}}

\newtheorem{theorem}{Theorem}

\newtheorem{lemma}{Lemma}
\newtheorem{cor}{Corollary}
\newtheorem{Prob}{Problem}
\newtheorem{remark}{Remark}
\newtheorem{definition}{Definition}

\newtheorem{innercustomthm}{Theorem}
\newenvironment{customthm}[1]
  {\renewcommand\theinnercustomthm{#1}\innercustomthm}
  {\endinnercustomthm}

\newacronym{pmf}{p.m.f.}{probability mass function}
\newacronym{iid}{i.i.d.}{identically and independently distributed}
\newacronym{cq}{cq}{classical input, quantum output}
\newacronym{qq}{qq}{quantum input, quantum output}
\newacronym{ccq}{c-cq}{classical input, one classical output, one quantum output}
\newacronym{cqq}{c-qq}{classical input, quantum output}
\newacronym{povm}{POVM}{positive operator-valued measure}
\newacronym{pls}{PLS}{physical layer security}

\title{Semantic Security with Infinite-dimensional Quantum Eavesdropping Channel}
\author{
Matthias Frey\IEEEauthorrefmark{1},
Igor Bjelaković\IEEEauthorrefmark{2}\IEEEauthorrefmark{3},
Janis Nötzel\IEEEauthorrefmark{4}, and
Sławomir Stańczak\IEEEauthorrefmark{2}\IEEEauthorrefmark{3}\\
\vspace{0.3cm}
\IEEEauthorrefmark{1}Department of Electrical and Electronic Engineering, The University of Melbourne, Australia\\
\IEEEauthorrefmark{2}Technische Universität Berlin, Germany\\
\IEEEauthorrefmark{3}Fraunhofer Heinrich Hertz Institute, Berlin, Germany\\
\IEEEauthorrefmark{4}Technical University of Munich, Germany
}

\allowdisplaybreaks

\begin{document}

\maketitle

\blankfootnote{MF, IB, and SS acknowledge the financial support by the Federal Ministry of Education and Research of Germany in the program of ``Souverän. Digital. Vernetzt.'' Joint project 6G-RIC, project identification numbers: 16KISK020K and 16KISK030.
This work was also supported by the Federal Ministry of Education and Research of Germany (BMBF) under grant 16KIS1686, as well as by the German Research Foundation (DFG) under grant STA 864/15-1.

JN acknowledges the financial support by the DFG Emmy-Noether program under grant number NO 1129/2-1 and by the Federal Ministry of Education and Research of Germany in
the program of ``Souver\"an. Digital. Vernetzt.''. Joint project 6G-life, project identification number: 16KISK002, and via projects 16KISQ077 and 16KISQ039.

This work was presented in part at the 2022 IEEE Information Theory Workshop (ITW).
}

\begin{abstract}
We propose a new proof method for direct coding theorems for wiretap channels where the eavesdropper has access to a quantum version of the transmitted signal on an infinite-dimensional Hilbert space and the legitimate parties communicate through a classical channel or a \gls{cq} channel. The transmitter input can be subject to an additive cost constraint, which specializes to the case of an average energy constraint. This method yields errors that decay exponentially with increasing block lengths. Moreover, it provides a guarantee of a quantum version of semantic security, which is an established concept in classical cryptography and physical layer security. Therefore, it complements existing works which either do not prove the exponential error decay or use weaker notions of security. The main part of this proof method is a direct coding result on channel resolvability which states that there is only a doubly exponentially small probability that a standard random codebook does not solve the channel resolvability problem for the \gls{cq} channel. Semantic security has strong operational implications meaning essentially that the eavesdropper cannot use its quantum observation to gather any meaningful information about the transmitted signal. We also discuss the connections between semantic security and various other established notions of secrecy.
\end{abstract}

\section{Introduction}
\label{sec:intro}
Developments in the area of quantum computing in the last decades have put a spotlight on how vulnerable many state-of-the-art security techniques for communication networks are against attacks based on execution of quantum algorithms that exploit the laws  of quantum physics \cite{Shor97}. Another aspect of the rapid development in experimental quantum physics, however, has received significantly less attention in communication engineering: Instead of simply performing quantum processing steps on an intercepted classical signal, an attacker can use quantum measurement devices to exploit the quantum nature of the signals themselves. For instance, radio waves as well as visible light that is used in optical fiber communications both consist of photons. This fact was exploited in \cite{Ferrigno2008} to introduce the so-called photonic side channels.

Any communication network that involves a vast number of interconnected network elements and systems communicating with each other, by its very nature exposes a large attack surface to potential adversaries. This is even exacerbated if many of the communication paths are wireless, as is the case in cellular networks such as 6G. Examples for possible attacks that can be carried out against various parts of such networks are algorithm implementation attacks, jamming attacks, side-channel attacks, and attacks on the physical layer of the communication system. This multi-faceted nature of the threat necessitates a diverse range of countermeasures. It has therefore been a longstanding expectation that established defenses based on cryptography will need to be complemented with defenses based on \gls{pls} which is a promising approach to protect against lower-layer attacks. It can be used as an additional layer of security to either increase the overall system security or reduce the complexity of cryptographic algorithms and protocols running at higher layers of the protocol stack. We point out that complexity may be a key aspect in massive wireless networks of the future such as 6G, where many low-cost, resource and computationally constrained devices will be deployed, making it difficult to use advanced cryptographic techniques.

In cryptography, sequences of bits are protected against attacks. The main threat posed by the recent progress in the implementation of quantum computers and known quantum algorithms \cite{Shor97, grover1996fast} which directly affect the security of certain cryptographic schemes, therefore, is that an attacker may have the ability to process cipher texts with quantum computers. This has triggered significant research and development efforts in the field of post-quantum cryptography \cite{Buchmann2016}. The goal in this field is to develop new cryptographic algorithms and protocols that are resistant to attacks by quantum computers. The threat for defense mechanisms based on \gls{pls}, on the other hand, is of a different nature: Since there is no assumption regarding the computational capability of attackers needed to guarantee security, such techniques are inherently safe against attacks with quantum computers that process classically represented signals. But since \gls{pls} seeks to protect the communication signals themselves against attacks, they are vulnerable to violations of system assumptions regarding what type of signal the attacker can intercept. Therefore, security is not guaranteed if an attack is carried out with quantum hardware such as optical quantum detectors and similar quantum measurement equipment. For example, the photon emission from integrated circuits has been exploited in \cite{Ferrigno2008, Juliane2014roleofphotons} to read out the secret keys used. The underlying physical process represents a quantum side channel since the properties of the emitted photons strongly depend on the operation that the involved device performs and the data being processed. This shortcoming of (classical) \gls{pls} techniques can be addressed by establishing results for \gls{pls} that take the quantum nature of wireless communication signals and side channels in transmitters and receivers into account. For example, the work \cite{vazquez-castro} defined the notion of an exclusion region with the goal of guaranteeing information-theoretic security as long as no wiretapper can interact with signals inside the exclusion region.

Another important aspect is the availability of bounds for secure communication in the finite block length regime. These are important for many practical purposes, such as the construction of secrecy maps for indoor and outdoor wireless networks in \cite{utkovski2019, utkovski2021, schulz2021} that we expect to be crucial for the integration of \gls{pls} into mobile communication networks.

In this work, we prove direct coding results for wiretap channels that take both of these aspects into consideration. The channel models considered have a quantum output at the eavesdropper's channel terminal, and the derived bounds can be numerically evaluated at finite block lengths. Besides this, we also discuss operational implications of the resulting security guarantees and compare them to other notions of security that are commonly used in the literature.

\subsection{Prior work}
\label{sec:literature}
The results we provide in this paper are rooted in and draw from several branches of research like classical \gls{pls} and its connection to cryptography, channel resolvability, and quantum Shannon theory. For this reason, we give a short overview of existing literature in these fields which is closely related to this work.

\paragraph{Classical \gls{pls}}
\label{par:pls}
An important branch of research with long history in \gls{pls} addresses fundamental bounds to confidentiality of communication, which is traditionally based on the communication model of the wiretap channel \cite{WynerWiretap, csiszarkoerner1978, CsiszarSecrecy, HayashiResolvability, BlochStrongSecrecy, watanabe2014optimal}.

It is important to emphasize that the measure of confidentiality itself has undergone a tremendous evolution. While the results of \cite{WynerWiretap, csiszarkoerner1978} use equivocation as the underlying measure, \cite{CsiszarSecrecy, HayashiResolvability, BlochStrongSecrecy, watanabe2014optimal} rely on strong secrecy or closely related measures as the security metric. The disadvantage of these security metrics is that they allow no or very limited operational interpretation, i.e., they do not provide a way to quantify leakage of information for specific types of eavesdropping attacks. A confidentiality measure that has been known in the cryptography community for a long time and allows a clear operational interpretation is semantic security \cite{GoldwasserProbabilistic}. It has been adopted in the \gls{pls} community as a suitable measure of secrecy and used for wiretap channel coding problems in \cite{bellare2012semantic, CuffSoftCovering, GoldfeldSemantic, GoldfeldSemanticConference}.

\paragraph{Achievability for wiretap channels with quantum outputs}
\label{par:q-achievability}
A quantum version of the wiretap channel was first analyzed in~\cite{schumacher1998quantum} for a one-shot scenario. \cite{DevetakPrivateCapacity,cai2004quantum} derive non-asymptotic results for general finite-dimensional wiretap channels with classical input and quantum outputs as well as in the case of quantum input and quantum outputs under the strong secrecy criterion, but due to the proof methods used, semantic security is implicitly established as well. In~\cite{hayashi2015quantum}, error exponents and equivocation rates are established for the case that there is only randomness of limited quality available at the transmitter. \cite{guha2008bosonicwiretap} extends the result from \cite{DevetakPrivateCapacity} to pure-loss bosonic channels under the strong secrecy criterion. The series of works~\cite{wildePublicPrivateTradeoffs2012,wilde2012quantumtradeoff,qi2017capacityamplifierchannels,wilde2018privateenrgyconstrained} explores a trade-off region between semantically secure communication, public communication and secret key generation which can be specialized to just semantically secure communication in a straightforward way. The channel models considered are the pure-loss bosonic channel, the thermal-noise bosonic channel, the quantum amplifier channel and general infinite-dimensional quantum channels. All achievability results are of asymptotic nature. In~\cite{dingguha2018noisyfeedback}, a feedback scenario is considered under weak secrecy. One-shot achievability bounds for the finite-dimensional broadcast channel with private and confidential messages are given in~\cite{renes2011noisy,wilde2017position,salek2020single}. In contrast to the prior works discussed, the present work derives achievability results for a general infinite-dimensional \gls{cq} channels, providing bounds that can be numerically evaluated for finite block lengths.

\paragraph{Converse results for wiretap channels with quantum outputs}
\label{par:q-converse}
Multi-letter converse results for finite-dimensional channels appear in~\cite{DevetakPrivateCapacity,cai2004quantum}. One-shot converse results for finite-dimensional broadcast channels with private and confidential messages are proposed in~\cite{renes2011noisy,wilde2017position}. In~\cite{guha2008bosonicwiretap,wilde2012quantumtradeoff}, single-letter converse results for pure-loss bosonic channels are given, but they rely on an unproven conjecture, the entropy photon number inequality. Converse results for general quantum channels and certain bosonic channels that do not rely on unproven conjectures are given in~\cite{wilde2018privateenrgyconstrained,qi2017capacityamplifierchannels,hayashi2015quantum}. Single-letter versions are available only in the case of degradable channels, otherwise the converse results are multi-letter.

\paragraph{Semantic security for channels with quantum outputs}
\label{par:q-semantic}
As mentioned above, many works have either given results that imply semantic security or established them implicitly in their proofs. The works~\cite{hayashi2015quantum,boche2022semantic}, however, have considered the question of security notion explicitly and established an interesting connection between strong secrecy and semantic security. \cite{hayashi2015quantum} contains a construction that can transform any wiretap code which achieves strong secrecy over a \gls{cqq} wiretap channel into a code that ensures semantic security and has (asymptotically) the same rate. As was observed in~\cite{boche2022semantic}, this holds also in the infinite-dimensional case. For the finite-dimensional case,~\cite{boche2022semantic} gives an alternative construction for this transformation; the advantage here is that the construction is explicit so that it can be expected that if it is used on a practically feasible code that achieves strong secrecy, the transformation will result in a semantically secure wiretap code that retains the practical feasibility. The prior works study semantic security either implicitly or explicitly under one of multiple possible definitions. In this work, we connect these definitions with each other and with strong and weak secrecy by showing, similarly as was done before for classical channels~\cite{bellare2012cryptographic}, what implications hold between them. In doing so, we also pay special attention to the subtleties that arise in the infinite-dimensional case and make explicit which of these implications continue to hold.

\paragraph{Resolvability for \gls{cq} channels}
\label{par:q-resolvability}
To the best of our knowledge, the first works that contain resolvability results for \gls{cq} channels are~\cite{DevetakPrivateCapacity,cai2004quantum}. Further results appeared in~\cite{hayashi2006quantum} and a much more in-depth treatment with generalizations to the case of imperfect randomness at the transmitter can be found in~\cite{hayashi2015quantum}. All of these works are specific to the finite-dimensional case, while in this paper, we propose results that are valid for general infinite-dimensional \gls{cq} channels.

\subsection{Contribution and Outline}
\label{sec:contributions}
The contribution of this paper can be summarized as follows:
\begin{itemize}
    \item We establish a semantic security based direct coding theorem for wiretap channels with infinite-dimensional quantum output observed by the eavesdropper (cf. Fig.~\ref{fig:wiretap}). This result applies both to the case where the legitimate parties communicate over a (possibly) continuous-alphabet classical channel, and to the case in which they also use a channel with classical input and infinite-dimensional quantum output.
    \item We develop a new proof method for wiretap channels with the eavesdropper having access to a quantum version of the input signal. This method also establishes a resolvability result for \gls{cq} channels with infinite-dimensional output. Our method is based on symmetrization arguments used to prove non-asymptotic versions of uniform laws of large numbers based on Rademacher complexity. One advantage of this new method, besides its amenability to infinite-dimensional channels, is that it naturally yields nonasymptotic results that are valid for any block length and not only asymptotically for block lengths tending to infinity.
    \item We state versions of our results that are amenable to numerical evaluations for finite block lengths. While it may not be possible from these results to obtain nontrivial bounds for very small block lengths, it is possible to numerically determine minimum block lengths for which our theorems guarantee certain performance metrics.
    \item We establish our results under additive cost constraints. These types of constraints specialize in particular to the case of average input energy constraints.
    \item We illustrate the finite-block length nature of our methods with numerical evaluations of the obtained error bounds in the special case where both the legitimate communication parties and the eavesdropper use Gaussian \gls{cq} channels.
    \item We discuss various established secrecy metrics for quantum communication systems and show how they are related to each other.
\end{itemize}
In Section~\ref{sec:main-result}, we introduce the wiretap channel models that are considered in this paper, and state the main results. In Section~\ref{sec:security}, we give the definitions of semantic security and other established secrecy metrics for quantum communication systems, and we discuss how they are related. In particular, we show how the results stated in Section~\ref{sec:main-result} imply a semantic security guarantee. To establish the main results, we need a theorem on \gls{cq} channel resolvability and a coding theorem for the \gls{cq} channel, which we state and briefly discuss in Section~\ref{sec:res-code}. Section~\ref{sec:proofs} contains the proofs of the theorems in Sections~\ref{sec:main-result} and~\ref{sec:res-code}. Finally, in Section~\ref{sec:gaussian-cq}, we specialize our results to Gaussian \gls{cq} channels. For the example of a set of system parameters that is plausible for a real-world optical communication system, we numerically evaluate and plot the bounds derived in this paper. A number of technical lemmas are relegated to the appendix.

\section{Problem Statement and Main Result}
In this section, we introduce the wiretap channel model and state our main results.
\label{sec:main-result}
\begin{figure}
    \centering
    \begin{subfigure}[t]{0.45\textwidth}
        \includegraphics[width=\textwidth]{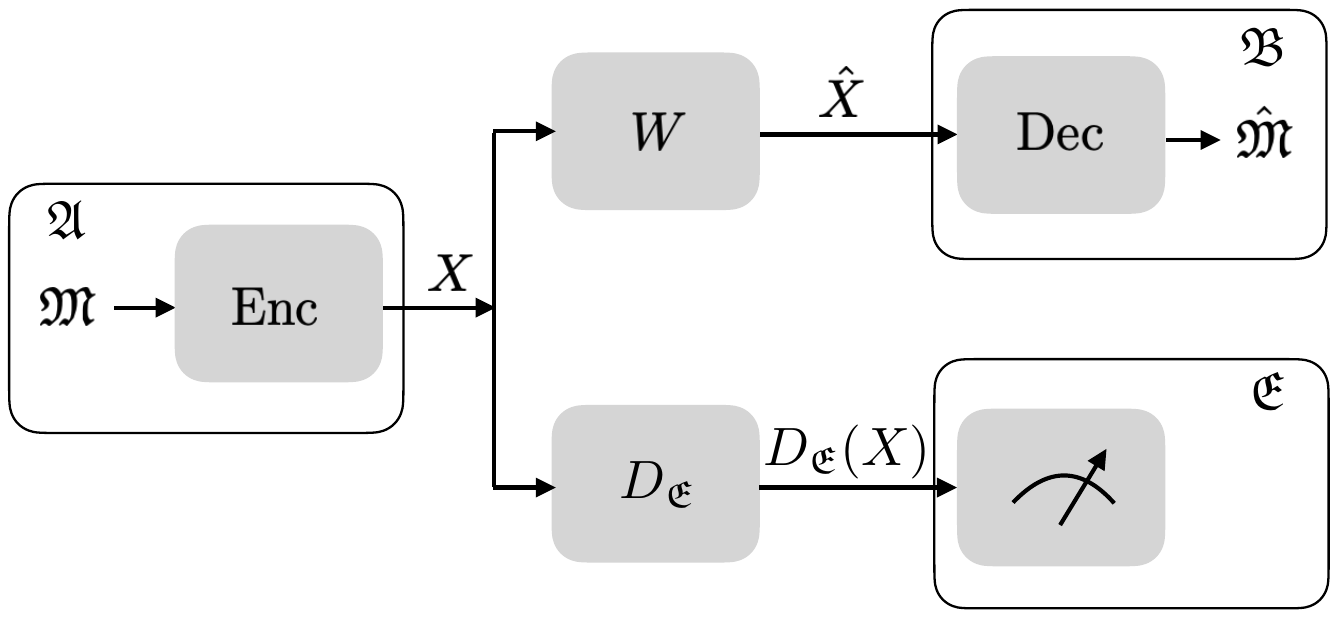}
        \caption{\acrshort{ccq} wiretap channel consisting of a classical channel $ \bobChannel$ with continuous input and output alphabet and a channel $ \eveChannel$ with classical continuous input alphabet and quantum output on separable infinite-dimensional Hilbert space.}
        \label{fig:ccq}
    \end{subfigure}
    \hfill
    \begin{subfigure}[t]{0.45\textwidth}
        \includegraphics[width=\textwidth]{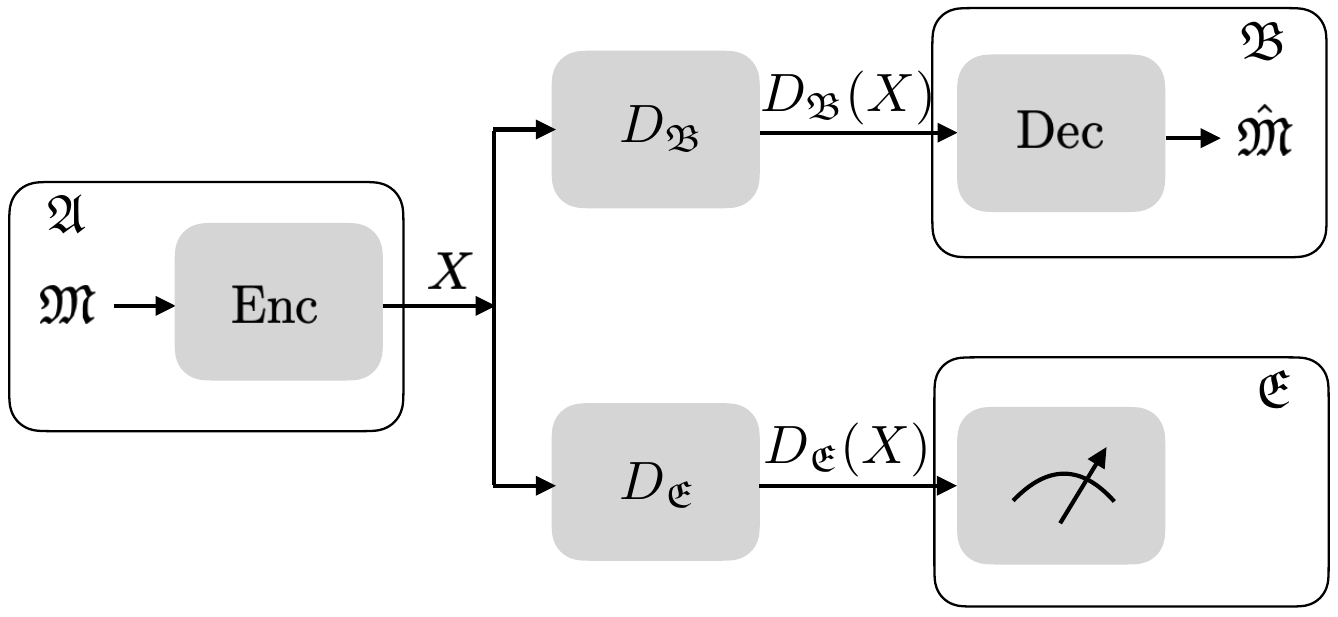}
        \caption{\acrshort{cqq} wiretap channel consisting of the pair of channels $(\bobcqChannel, \eveChannel)$ both with classical continuous input alphabet and quantum output on separable infinite-dimensional Hilbert space.}
        \label{fig:cq}
    \end{subfigure}

    \caption{Quantum wiretap channel models.}
    \label{fig:wiretap}
\end{figure}
\subsection{Notation and Conventions}\label{sec:notation}
Before we can make the formal problem statement, it is necessary to introduce some preliminaries. Let $\hilbertSpace$ be a separable  Hilbert space over $\mathbb{C} $ and let
\[
\boundedOperators{\hilbertSpace}:=\left\{\generalOperator: \hilbertSpace\to \hilbertSpace: \generalOperator \textrm{ is linear and } \sup_{\hilbertSpaceElement\in \hilbertSpace: || \hilbertSpaceElement||\le 1} || \generalOperator \hilbertSpaceElement||< \infty \right\}
\]
be the set of bounded linear operators on $\hilbertSpace$ where $ || \cdot ||$ denotes the norm on the Hilbert space $\hilbertSpace$ induced by the inner product. The set $\traceClass{\hilbertSpace}$ of trace class operators \cite{reed1980functional} is defined by
\[
\traceClass{\hilbertSpace}:=\left\{\generalOperator\in \boundedOperators{\hilbertSpace}: \trace (\adjoint{\generalOperator} \generalOperator)^{\frac{1}{2}}< \infty \right\}.
\]
Endowed with the norm $\traceNorm{\cdot}: \traceClass{\hilbertSpace}\to \reals_{+}$ defined by
\[
\traceNorm{A}:=\trace (\adjoint{\generalOperator} \generalOperator)^{\frac{1}{2}},
\]
for $\generalOperator\in \traceClass{\hilbertSpace}$, the pair
$(\traceClass{\hilbertSpace}, \traceNorm{\cdot})$ is a Banach space (cf. \cite{reed1980functional}) and is called the trace class on $\hilbertSpace$.

The set of states or density operators $\densityOperators{\hilbertSpace}\subset \traceClass{\hilbertSpace} $ is given by
\begin{equation}\label{eq:def-state}
\densityOperators{\hilbertSpace}
:=
\left\{
  \generalQuantumState \in \traceClass{\hilbertSpace}:~
  \generalQuantumState \geq 0,~
  \trace \generalQuantumState = 1
\right\},
\end{equation}
where $\geq$ denotes the positive semi-definite partial ordering of operators, and we follow the tradition of using lowercase Greek letters for states.
We will also occasionally use the operator norm $\operatorNorm{\cdot}$ on $\boundedOperators{\hilbertSpace}$ which is given by
\[
\operatorNorm{A}:= \sup_{\hilbertSpaceElement\in \hilbertSpace: || \hilbertSpaceElement||\le 1} || \generalOperator \hilbertSpaceElement||
\]
for $\generalOperator\in \boundedOperators{\hilbertSpace}$, and for the reader's convenience we summarize the specific properties of the operator and trace norms that we use in Lemma~\ref{lemma:norm-basics} in Appendix~\ref{appendix:funcana}. Note that the pair $(\boundedOperators{\hilbertSpace}, \operatorNorm{\cdot})$ is a Banach space as well (cf. \cite{reed1980functional}).

Given any finite set $\generalFiniteSet$, a $\generalFiniteSet$-valued \emph{\gls{povm}} \cite{heinosaari2012mathematical, Busch-Quantum-Measurement} is a sequence $(\povm_\generalFiniteSetElement)_{\generalFiniteSetElement \in \generalFiniteSet}$ of operators on $\hilbertSpace$ such that for all $\generalFiniteSetElement$, $\povm_\generalFiniteSetElement \geq 0$, and
\begin{equation}\label{eq:def-povm}
\sum_{\generalFiniteSetElement \in \generalFiniteSet}
  \povm_\generalFiniteSetElement
=
\identityOperator.
\end{equation}
The \gls{povm} $(\povm_\generalFiniteSetElement)_{\generalFiniteSetElement \in \generalFiniteSet}$ is a mathematical description of a measurement with the possible outcomes $ \generalFiniteSetElement \in \generalFiniteSet$. For a system in the state $\generalQuantumState \in \densityOperators{\hilbertSpace}$, the probability of the outcome $\generalFiniteSetElement \in \generalFiniteSet$ when performing the measurement represented by the \gls{povm} $(\povm_\generalFiniteSetElement)_{\generalFiniteSetElement \in \generalFiniteSet}$ is given by the Born rule as
\begin{equation*}
p(\generalFiniteSetElement):=\trace\left(\generalQuantumState \povm_{\generalFiniteSetElement}\right).
\end{equation*}
Since $\generalQuantumState\ge 0$ and $\povm_{\generalFiniteSetElement}\ge 0$, we have $p(\generalFiniteSetElement)\ge 0$.  The defining relations (\ref{eq:def-state}), (\ref{eq:def-povm}), and the linearity of the trace show that
\[
\sum_{\generalFiniteSetElement\in \generalFiniteSet}p(\generalFiniteSetElement)=1,
\]
so that, indeed, $p$ is a \gls{pmf} on $\generalFiniteSet$.

Throughout the paper, we use the following abbreviations. For a Hilbert space $\hilbertSpace$, $\generalOperator\in \boundedOperators{\hilbertSpace}$, and $\blocklength\in \naturals$, $\hilbertSpace^{\otimes \blocklength}:=\hilbertSpace\otimes \ldots \otimes \hilbertSpace$ stands for $\blocklength$-fold tensor product of $\hilbertSpace$ and $\generalOperator^{\otimes \blocklength}:= \generalOperator\otimes\ldots \otimes \generalOperator$ denotes the $\blocklength$-fold tensor power of $\generalOperator$. For a set $\inputAlphabet$ and $\blocklength \in \naturals$ we set $\inputAlphabet^{\blocklength}:=\inputAlphabet \times \ldots \times \inputAlphabet$ for $\blocklength$-fold Cartesian product of the set $\inputAlphabet$. Moreover, we use the abbreviation $\inputAlphabetElement^{\blocklength}:=(x_1, \ldots , x_{\blocklength})\in \inputAlphabet^{\blocklength}$.

Throughout this work, expectations and integrals of operator-valued random variables and functions will play an important role. Since we do not assume that the input alphabet $\inputAlphabet$ is necessarily finite or discrete, we assume that $\inputAlphabet$ is equipped with a $\sigma$-algebra $\inputAlgebra$ which represents a collection of measurable sets. This means in particular that the probability space underlying these expectations is possibly infinite as well.  The expectations are therefore formally defined as Bochner integrals on the Banach space $(\traceClass{\hilbertSpace}, \traceNorm{\cdot})$. The exact definitions and many important facts about Bochner integration can, e.g., be found in~\cite[Section V.5]{yosida1980functional}. Therefore, many technical questions of measurability and integrability arise in our proofs. We summarize the preliminaries on measurability that we use in this paper in Lemma~\ref{lemma:spectral-decompositions-measurable} and the necessary preliminaries on Bochner integration in Lemma~\ref{lemma:bochner-integral-basics} in Appendix~\ref{appendix:funcana}. When referring to the notions of continuity (measurability) of functions, it is important to specify with respect to which topology ($\sigma$-algebra) the function in question is continuous (measurable). We therefore adopt the following conventions: When the domain or range of a function is given as $\inputAlphabet$, measurability is with respect to $\inputAlgebra$. When the domain or range is given as $\boundedOperators{\hilbertSpace}$, the continuity (measurability) of the function is with respect to the topology (Borel $\sigma$-algebra) induced by the operator norm. When the domain or range is given as $\traceClass{\hilbertSpace}$ or $\densityOperators{\hilbertSpace}$, the continuity (measurability) of the function is with respect to the topology (Borel $\sigma$-algebra) induced by the trace norm. When the domain or range is a subset of $\reals$ or $\complexNumbers$, we consider the topology (Borel $\sigma$-algebra) induced by the Euclidean norm.

\subsection{System Model and Main Result}\label{sec:system-model}
\glsreset{ccq}
\glsreset{cqq}
In this work, we study wiretap channels with \gls{ccq} and wiretap channels with \gls{cqq}.
Formally, a \gls{ccq} wiretap channel is a pair $(\bobChannel, \eveChannel)$, where $\bobChannel$ is a stochastic kernel (i.e., a classical channel), $\eveChannel$ is a \gls{cq} channel, and $\bobChannel$ and $\eveChannel$ share the same input alphabet. The idea behind this definition is that the eavesdropper observes a quantum system instead of a classical output while the legitimate receiver observes a classical output. A \gls{cqq} wiretap channel is a pair $(\bobcqChannel,\eveChannel)$ where both $\bobcqChannel$ and $\eveChannel$ are \gls{cq} channels.

In the following, we describe the system model (which is also depicted in Fig.~\ref{fig:wiretap}) in more detail. The transmitter $\alice$ transmits a channel input $\inputRV$ which is a random variable ranging over a measurable space $(\inputAlphabet,\inputAlgebra)$, the input alphabet. The eavesdropper $\eve$ observes the output of a \gls{cq} channel which is described by a measurable map $\eveChannel: \inputAlphabet \rightarrow \densityOperators{\eveHilbertSpace}$ where $\eveHilbertSpace$ is a separable Hilbert space. The measurability of $\eveChannel$ is with respect to $\inputAlgebra$ and the Borel $\sigma$-algebra induced by the trace norm. In the case of the \gls{cqq} wiretap channel (Fig.~\ref{fig:cq}), the legitimate receiver $\bob$ also observes the output of a \gls{cq} channel, described by a measurable map  $\bobcqChannel: \inputAlphabet \rightarrow \densityOperators{\bobHilbertSpace}$, where $\bobHilbertSpace$ is a separable Hilbert space. In the case of the \gls{ccq} wiretap channel (Fig.~\ref{fig:ccq}), the legitimate receiver observes the output $\legitOutputRV$ which is a random variable ranging over another measurable space $(\legitOutputAlphabet,\legitOutputAlgebra)$ and its relationship with $\inputRV$ is described by a stochastic kernel $\bobChannel: \legitOutputAlgebra \times \inputAlphabet\to [0,1]$, the legitimate receiver's (classical) channel.

For $\inputAlphabetElement^\blocklength :=(\inputAlphabetElement_1, \ldots ,\inputAlphabetElement_{\blocklength})\in \inputAlphabet^\blocklength$ and $B_1,\ldots, B_n \in \legitOutputAlgebra$, we set
\begin{equation*}
    \bobChannel^\blocklength(\inputAlphabetSubset_1\times \ldots \times \inputAlphabetSubset_n,\inputAlphabetElement^\blocklength ):= \prod_{\generalIndex=1}^{\blocklength}\bobChannel (\inputAlphabetSubset_\generalIndex, \inputAlphabetElement_\generalIndex),\
\end{equation*}
for classical channels $\bobChannel$, and 
\begin{equation*}
  \classicalQuantumChannel^{\blocklength} (\inputAlphabetElement^\blocklength) := \bigotimes_{\generalIndex=1}^{\blocklength} \classicalQuantumChannel(\inputAlphabetElement_\generalIndex)= \classicalQuantumChannel(\inputAlphabetElement_1)\otimes \dots \otimes \classicalQuantumChannel(\inputAlphabetElement_{\blocklength}),
\end{equation*}
for \gls{cq} channels $\classicalQuantumChannel$, i.e., we consider $\blocklength$-th memoryless extensions of the channels $\bobChannel$ and $\classicalQuantumChannel$.

A \emph{wiretap code} for the \gls{ccq} channel $(\bobChannel, \eveChannel)$ (for the \gls{cqq} channel $(\bobcqChannel, \eveChannel)$) with message set size $\codebookSize$ and \emph{block length} $\blocklength$ is a pair $(\encoder, \decoder)$ such that
\begin{itemize}
 \item The encoder is a stochastic kernel mapping $\encoder: \inputAlgebra^{\otimes \blocklength} \times \{1, \dots, \codebookSize\} \to [0,1]$, where $\inputAlphabet$ is the common input alphabet of $\bobChannel$ and $\eveChannel$ (of $\bobcqChannel$ and $\eveChannel$).
 \item In case of a classical channel $\bobChannel$ to the legitimate receiver, the decoder is a deterministic mapping $ \decoder: \legitOutputAlphabet^\blocklength \to \{1, \dots, \codebookSize\}$, where $\legitOutputAlphabet$ is the output alphabet of $\bobChannel$.
 \item In case of a \gls{cq} channel $\bobcqChannel$ to the legitimate receiver, the decoder is a $\{1, \dots, \codebookSize\}$-valued \gls{povm} $(\decoderpovm_\codewordIndex)_{\codewordIndex=1}^\codebookSize$ on $\bobHilbertSpace$.
\end{itemize}

The \emph{rate} of a wiretap code is defined as $\log (\codebookSize) / \blocklength$. In this paper, we always use the logarithm (as well as the exponential function denoted $\exp$) with Euler's number as a base and consequently, the rate is given in nats per channel use. For the \gls{ccq} channel and $\decodingError\ge 0$, we say that the wiretap code has \emph{average error} $\decodingError$ if
\begin{equation}
\label{eq:def-average-error}
\Expectation_\messageRV
  \Probability\big(
    \decoder (\legitOutputRV^\blocklength)
    \neq
    \messageRV
  \big)
\leq
\decodingError
\end{equation}
where $ \legitOutputRV^\blocklength$ is the random variable observed by the legitimate receiver, resulting from encoding and transmission of a uniformly distributed message $\messageRV$ through the channel $ \bobChannel^\blocklength$. For the \gls{cqq} channel and $\decodingError\ge 0$, we say that the wiretap code has \emph{average error} $\decodingError$ if
\begin{equation}
\label{eq:def-average-error-cq}
\Expectation_\messageRV
  \trace\left(
    \bobcqChannel^\blocklength \circ \encoder (\messageRV)
    (\identityOperator - \decoderpovm_\messageRV)
  \right)
\leq
\decodingError
\end{equation}
for a uniformly distributed message $\messageRV$. The composition of the stochastic kernel mapping  $\encoder: \inputAlgebra^{\otimes \blocklength} \times \{1, \dots, \codebookSize\} \to [0,1]$ and the \gls{cq} channel $\bobcqChannel^\blocklength:\inputAlphabet^\blocklength\to \densityOperators{\bobHilbertSpace^{\otimes \blocklength}} $ is given by
\begin{equation}\label{eq:def-channel-composition}
\bobcqChannel^\blocklength \circ \encoder (\codewordIndex) := \int_{\inputAlphabet^\blocklength} \bobcqChannel^\blocklength (\inputAlphabetElement^\blocklength)\encoder (d \inputAlphabetElement^\blocklength, \codewordIndex),
\end{equation}
for $\codewordIndex\in \{1, \dots, \codebookSize\}$. Note that the integral is well-defined by the measurability of $\bobcqChannel$ and that the integral exists in Bochner sense by Lemma \ref{lemma:bochner-integral-basics}-\ref{item:bochner-integral-basics-existence}. Moreover, $\bobcqChannel^\blocklength \circ \encoder (\codewordIndex)\in \densityOperators{\bobHilbertSpace^{\otimes \blocklength}}$ by Lemma~\ref{lemma:bochner-integral-basics}-\ref{item:bochner-integral-basics-trace} and Lemma~\ref{lemma:bochner-integral-basics}-\ref{item:bochner-integral-basics-bounded-functional}.
\begin{remark}
Despite their different appearance, the expressions in (\ref{eq:def-average-error}) and (\ref{eq:def-average-error-cq}) are closely related, as we will briefly explain. Introducing the indicator functions of the decoding sets $\indicatorFunction{\decoder^{-1}(\codewordIndex)}$ for $\codewordIndex\in \{1, \ldots, \codebookSize  \}$, we can write (\ref{eq:def-average-error}) as 
\begin{equation}\label{eq:remark-average-errors-1}
\Expectation_\messageRV
  \Probability\big(
    \decoder (\legitOutputRV^\blocklength)
    \neq
    \messageRV
  \big)
  = \frac{1}{\codebookSize}\sum_{\codewordIndex=1}^{\codebookSize}\sum_{\substack{\hat{\codewordIndex}=1\\ \hat{\codewordIndex}\neq \codewordIndex}}^{\codebookSize} \int \left( \int\indicatorFunction{\decoder^{-1}(\hat{\codewordIndex})}(\legitOutputAlphabetElement^{\blocklength})\bobChannel^\blocklength (d \legitOutputAlphabetElement^{\blocklength}, \inputAlphabetElement^\blocklength)\right) \encoder (d \inputAlphabetElement^\blocklength, \codewordIndex).
\end{equation}
On the other hand, using (\ref{eq:def-povm}) for the decoding \gls{povm} $(\decoderpovm_\codewordIndex)_{\codewordIndex=1}^\codebookSize$ together with the linearity of the trace and the integral, we have
\begin{equation}\label{eq:remark-average-errors-2}
\Expectation_\messageRV
  \trace\left(
    \bobcqChannel^\blocklength \circ \encoder (\messageRV)
    (\identityOperator - \decoderpovm_\messageRV)
  \right)
  = \frac{1}{\codebookSize}\sum_{\codewordIndex=1}^{\codebookSize}\sum_{\substack{\hat{\codewordIndex}=1\\ \hat{\codewordIndex}\neq \codewordIndex}}^{\codebookSize} \int \trace \left(\decoderpovm_{\hat{\codewordIndex}}\bobcqChannel^\blocklength (\inputAlphabetElement^\blocklength)\right)\encoder (d \inputAlphabetElement^\blocklength, \codewordIndex) .
  \end{equation}
  The inner integral in (\ref{eq:remark-average-errors-1}) and the term containing the trace in (\ref{eq:remark-average-errors-2}) both describe the probability that the message $\hat{\codewordIndex}$ is detected/measured given that the channel input is $\inputAlphabetElement^{\blocklength}\in \inputAlphabet^{\blocklength}$.
\end{remark}
For $\securityNumber \ge 0$, we say that the wiretap code $(\encoder, \decoder)$ has \emph{distinguishing security level} $\securityNumber$ if
\begin{equation}
\label{eq:def-distinguishing-security}
\forall \codewordIndex_1, \codewordIndex_2 \in \{1, \dots, \codebookSize\}
~~
\traceNorm{
  \eveChannel^\blocklength \circ \encoder (\codewordIndex_1)
  -
  \eveChannel^\blocklength \circ \encoder (\codewordIndex_2)
}
\leq
\securityNumber,
\end{equation}
where $\eveChannel^\blocklength \circ \encoder$ is defined in an analogous way to (\ref{eq:def-channel-composition}).

The main results of this work state that wiretap codes for certain \gls{ccq} and \gls{cqq} channels exist that simultaneously achieve low average error and low distinguishing security level. Achievable secrecy rates are characterized by two information quantities. Information density and mutual information of the classical channel $\bobChannel$ associated with an input distribution $\inputDistribution$ on $\inputAlphabet$ are defined in the usual way as
\begin{equation*}
\legitInformationDensity{\inputDistribution}{
  \inputAlphabetElement^\blocklength
}{
  \legitOutputAlphabetElement^\blocklength
}
:=
\log
\RNDerivative{\bobChannel^\blocklength(\cdot, \inputAlphabetElement^\blocklength)}
             {\legitOutputDistribution}
  (\legitOutputAlphabetElement^\blocklength)
,~~
\legitInformation
:=
\Expectation
  \legitInformationDensity{\inputDistribution}{
    \inputRV
  }{
    \legitOutputRV
  }
\end{equation*}
where $\legitOutputDistribution$ is the distribution of the output of $\bobChannel$ under the input distribution $\inputDistribution$, and $\legitOutputRV$ is a random variable distributed according to $\legitOutputDistribution$ (or $\legitOutputRV \distributedAccTo \legitOutputDistribution$ for short). For $\generalQuantumState \in \densityOperators{\hilbertSpace}$, we define the \emph{von Neumann entropy}
\[
\vonNeumannEntropy{\generalQuantumState}
:=
-
\trace(
  \generalQuantumState \log \generalQuantumState
)
\]
with the convention $\vonNeumannEntropy{\generalQuantumState} = \infty$ if $\generalQuantumState \log \generalQuantumState \notin \traceClass{\hilbertSpace}$. For an input distribution $\inputDistribution$, $\inputRV \distributedAccTo \inputDistribution$, and a \gls{cq} channel $\classicalQuantumChannel: \inputAlphabet \rightarrow \densityOperators{\hilbertSpace}$, we define 
\begin{equation*}
\classicalQuantumChannel_\inputDistribution := \Expectation \classicalQuantumChannel(\inputRV)
\end{equation*}
to be the density operator of the output of the channel under input distribution $\inputDistribution$. Note that since $\traceNorm{\classicalQuantumChannel(\inputAlphabetElement)}=1$ for all $ \inputAlphabetElement\in \inputAlphabet$, the expectation $ \Expectation \classicalQuantumChannel(\inputRV)$ exists by Lemma~\ref{lemma:bochner-integral-basics}-\ref{item:bochner-integral-basics-existence}. Moreover, von Neumann entropy $\vonNeumannEntropy{\cdot}$ is lower semi-continuous with respect to the topology induced by the trace norm $\traceNorm{\cdot}$ (cf. \cite[Theorem 11.6]{holevo2019quantum}). Consequently, the map $\inputAlphabet\ni \inputAlphabetElement \mapsto\vonNeumannEntropy{\classicalQuantumChannel(\inputAlphabetElement)}$ is measurable. Therefore, for a random variable $\inputRV \distributedAccTo \inputDistribution$, we can define the \emph{Holevo information} as
\begin{equation*}
\quantumInformation
:=
\Expectation
  \trace\big(
    \classicalQuantumChannel(\inputRV)
    \log \classicalQuantumChannel(\inputRV)
  \big)
    -
\trace\big(
  \classicalQuantumChannel_\inputDistribution
  \log \classicalQuantumChannel_\inputDistribution
\big)
=
\vonNeumannEntropy{\classicalQuantumChannel_\inputDistribution}
-
\Expectation
  \vonNeumannEntropy{\classicalQuantumChannel(\inputRV)},
\end{equation*}
where we adopt the convention that $\quantumInformation = \infty$ whenever $\vonNeumannEntropy{\classicalQuantumChannel_\inputDistribution} = \infty$ or if the integral
$
\Expectation
  \vonNeumannEntropy{\classicalQuantumChannel(\inputRV)}
$
does not exist. More information on definition, measurability, and structural properties of $ \quantumInformation$ can be found in \cite{Holevo-Shirokov-Continuous-Ensembles}.

\begin{definition}
\label{def:cost-constraint}
An \emph{additive cost constraint} $(\costFunction,\costConstraint)$ consists of a measurable cost function $\costFunction: \inputAlphabet \rightarrow [0,\infty)$ and a constraint $\costConstraint \in (0,\infty)$. A tuple $\inputAlphabetElement^\blocklength \in \inputAlphabet^\blocklength$ satisfies the additive cost constraint $(\costFunction,\costConstraint)$ if
\[
\costFunction(\inputAlphabetElement_1)
+
\dots
+
\costFunction(\inputAlphabetElement_\blocklength)
\leq
\blocklength\costConstraint.
\]
We say that the cost constraint is \emph{compatible} with a distribution $\inputDistribution$ on $\inputAlphabet$ if, for $\inputRV \distributedAccTo \inputDistribution$, we have
\[
\Expectation \costFunction(\inputRV) < \costConstraint,
\]
and if there is $\generalReal \in (0,\infty)$ with $\Expectation \exp(\generalReal \costFunction(\inputRV)) < \infty$.
\end{definition}

We have now introduced all necessary terminology to state the main result of this work. For better readability, we state the result for the \gls{ccq} wiretap channel (Theorem~\ref{theorem:wiretap-ccq}) and for the \gls{cqq} wiretap channel (Theorem~\ref{theorem:wiretap-cq}) separately although the theorems are very similar. The proofs are deferred to Section~\ref{sec:proofs-main-result}.

\begin{theorem}
\label{theorem:wiretap-ccq}
Let $(\bobChannel, \classicalQuantumChannel)$ be a $\gls{ccq}$ wiretap channel, let $\inputDistribution$ be a probability distribution on the input alphabet $\inputAlphabet$, and let $\inputRV \distributedAccTo \inputDistribution$ such that
\begin{itemize}
 \item there is $\renyiorder_{\min} < 1$ with $\eveChannel(\inputAlphabetElement)^{\renyiorder_{\min}} \in \traceClass{\hilbertSpace}$ for almost all $\inputAlphabetElement$, $\classicalQuantumChannel_\inputDistribution^{\renyiorder_{\min}} \in \traceClass{\hilbertSpace}$, and the Bochner integral $\Expectation \eveChannel(\inputRV)^{\renyiorder_{\min}}$ exists;
 \item there is $\generalReal > 0$ such that $\Expectation \exp(\generalReal \legitInformationDensity{\inputDistribution}{\inputRV}{\legitOutputRV}) < \infty$.
\end{itemize}
 Let $(\costFunction, \costConstraint)$ be a cost constraint compatible with $\inputDistribution$, and let $\codebookRate < \legitInformation - \holevoInformation{\inputDistribution}{\eveChannel}$. Then there are $\finalconst_1, \finalconst_2 \in (0,\infty)$ such that for sufficiently large $\blocklength$, there exists a wiretap code with the following properties:
\begin{enumerate}
 \item The code has a rate of at least $\codebookRate$.
 \item The output of the encoder satisfies the cost constraint $(\costFunction,\costConstraint)$ almost surely.
 \item The code has an average error $\decodingError = \exp(-\finalconst_1 \blocklength)$ as defined in (\ref{eq:def-average-error}).
 \item The code has distinguishing security level $\securityNumber = \exp(-\finalconst_2 \blocklength)$ as defined in (\ref{eq:def-distinguishing-security}).
\end{enumerate}
\end{theorem}
\begin{theorem}
\label{theorem:wiretap-cq}
Let $(\bobcqChannel,\eveChannel)$ be a \gls{cqq} wiretap channel, and let $\inputDistribution$ be a probability distribution on the input alphabet $\inputAlphabet$ such that for both choices of $\classicalQuantumChannel \in \{\eveChannel, \bobcqChannel\}$, there is $\renyiorder_{\min} < 1$ with $\classicalQuantumChannel(\inputAlphabetElement)^{\renyiorder_{\min}} \in \traceClass{\hilbertSpace}$ for almost all $\inputAlphabetElement$, $\classicalQuantumChannel_\inputDistribution^{\renyiorder_{\min}} \in \traceClass{\hilbertSpace}$, and the Bochner integral $\Expectation \classicalQuantumChannel(\inputRV)^{\renyiorder_{\min}}$ exists, where $\inputRV \distributedAccTo \inputDistribution$.

 Let $(\costFunction, \costConstraint)$ be a cost constraint compatible with $\inputDistribution$, and let $\codebookRate < \holevoInformation{\inputDistribution}{\bobcqChannel} - \holevoInformation{\inputDistribution}{\eveChannel}$. Then there are $\finalconst_1, \finalconst_2 \in (0,\infty)$ such that for sufficiently large $\blocklength$, there exists a wiretap code with the following properties:
\begin{enumerate}
 \item The code has a rate of at least $\codebookRate$.
 \item The output of the encoder satisfies the cost constraint $(\costFunction,\costConstraint)$ almost surely.
 \item The code has an average error of $\decodingError = \exp(-\finalconst_1 \blocklength)$ as defined in (\ref{eq:def-average-error-cq}).
 \item The code has distinguishing security level $\securityNumber = \exp(-\finalconst_2 \blocklength)$ as defined in (\ref{eq:def-distinguishing-security}).
\end{enumerate}
\end{theorem}
\begin{remark}
As mentioned in the introduction, converse results were derived in ~\cite{wilde2018privateenrgyconstrained,qi2017capacityamplifierchannels,hayashi2015quantum}. Single-letter versions are shown for degradable channels; otherwise the bounds are of the multi-letter variety. In \cite{Tikku-Non-Additivity}, it has been shown by means of many  explicit examples that the capacity of \gls{cqq} and \gls{ccq} wiretap channels is non-additive. Therefore, single-letter converses are not possible in general.    
\end{remark}

\subsection{Results Valid for Finite Block lengths}
\label{sec:finite-blocklength}
In this subsection, we state versions of Theorem~\ref{theorem:wiretap-ccq} and Theorem~\ref{theorem:wiretap-cq} that can be evaluated at given, finite block lengths $\blocklength$. While the expressions neither give closed-form expressions for the achievable average error and security level at a given $\blocklength$ nor guarantee nontrivial bounds for these performance metrics at arbitrarily short block lengths, they do make it possible to numerically evaluate tradeoffs between $\blocklength$, average error, and security level for given channels and communication rates. For instance, it is possible to numerically approximate the average error and security level for a given channel and a set communication rate and block length, or conversely, to determine a minimum block length at which a set combination of performance metrics can be achieved. We show examples for this in Section~\ref{sec:gaussian-cq}.

Before we can state these somewhat more involved versions of our main results, we need to make some more definitions. Given the classical channel $\bobChannel$ and input distribution $\inputDistribution$, we define for $\renyiorder \in [0,1) \cup (1,\infty)$
\[
\legitRenyDiv{\renyiorder}
:=
\frac{1}{1-\renyiorder}
\log
\Expectation_{\inputDistribution \legitOutputDistribution}\left(
  \RNDerivative{\bobChannel(\cdot, \inputRV)}{\legitOutputDistribution}(\legitOutputRV)
\right)^\renyiorder
\]
to denote the Rényi divergence between the joint input-output distribution and the product of their marginals. As before, we use $\inputRV,\legitOutputRV$ to denote the channel input and output, and we use $\legitOutputDistribution$ to denote the marginal output distribution of $\bobChannel$ under the input distribution $\inputDistribution$. It is well-known~\cite[Theorem 5]{vanerven2014renyi} that this definition can be continuously extended to $\renyiorder \in [0,\infty]$ and with this extension, $\legitRenyDiv{1}=\legitInformation$.

For every \gls{cq} channel $\classicalQuantumChannel$ and $\inputAlphabetElement \in \inputAlphabet$, we fix a spectral decomposition
\begin{equation*}
\classicalQuantumChannel(\inputAlphabetElement)
=
\sum_{\classicalOutputIndex \in \naturals}
  \inputDistribution_{\classicalQuantumChannel}( \classicalOutputIndex | \inputAlphabetElement)
  \rankOneOperator{\eigenvector_{\classicalOutputIndex | \inputAlphabetElement}},
\end{equation*}
where $ \{\eigenvector_{\classicalOutputIndex | \inputAlphabetElement}\}_{\classicalOutputIndex \in \naturals } $ is an orthonormal basis of the output Hilbert space of $\classicalQuantumChannel$ for every $\inputAlphabetElement\in \inputAlphabet$. Whenever the channel $\classicalQuantumChannel$ is clear from context, we will drop the corresponding subscript. By Lemma~\ref{lemma:spectral-decompositions-measurable}-\ref{item:spectral-decompositions-measurable-eigenvalues}, for every $\classicalOutputSubset \subseteq \naturals$, the map
$
\inputAlphabetElement
\mapsto
\sum_{\classicalOutputIndex \in \classicalOutputSubset}
  \inputDistribution_{\classicalQuantumChannel}(\classicalOutputIndex | \inputAlphabetElement)
$
is a pointwise limit of measurable maps and therefore measurable. Moreover, for all $\inputAlphabetElement$, we have
$
\sum_{\classicalOutputIndex \in \naturals}
  \inputDistribution_{\classicalQuantumChannel}(\classicalOutputIndex | \inputAlphabetElement)
=
\trace \classicalQuantumChannel(\inputAlphabetElement)
=
1
$.
Therefore, the eigenvalues $\inputDistribution(\classicalOutputIndex | \inputAlphabetElement)$ induce a stochastic kernel from $\inputAlphabet$ to $\naturals$, and together with a probability distribution $\inputDistribution$ on $\inputAlphabet$, we obtain a joint probability distribution on $\inputAlphabet \times \naturals$. In a slight abuse of notation, we will use the symbol $\inputDistribution$ to denote this joint probability distribution as well as the marginal on $\naturals$ and all corresponding \glspl{pmf}. Let $(\probabilitySpace,\probabilitySpaceAlgebra, \generalPMeasure)$ be a probability space and $\inputRV: \probabilitySpace \rightarrow \inputAlphabet$ be a random variable distributed according to $\inputDistribution$. We write the density operator of the channel output in terms of its spectral decomposition
\[
\classicalQuantumChannel_\inputDistribution
:=
\Expectation_\inputRV \classicalQuantumChannel(\inputRV)
=
\sum_{\classicalOutputIndex \in \naturals}
  \outputDistribution_{\classicalQuantumChannel}(\classicalOutputIndex)
  \rankOneOperator{\eigenvector_\classicalOutputIndex}.
\]
Clearly, $\outputDistribution$ is a \gls{pmf} on $\naturals$ which we identify with the probability measure it induces. We define the entropies
\begin{align}
\label{eq:joint-entropy}
\jointEntropySubscript{\classicalQuantumChannel}
&:=
-
\sum_{\inputAlphabetElement \in \inputAlphabet}
\sum_{\classicalOutputIndex \in \naturals}
  \inputDistribution_{\classicalQuantumChannel}(\inputAlphabetElement,\classicalOutputIndex)
  \log \inputDistribution_{\classicalQuantumChannel}(\classicalOutputIndex | \inputAlphabetElement)
=
-
\Expectation_\inputRV
  \sum_{\classicalOutputIndex \in \naturals}
    \inputDistribution_{\classicalQuantumChannel}(\classicalOutputIndex | \inputRV)
    \log \inputDistribution_{\classicalQuantumChannel}(\classicalOutputIndex | \inputRV)
=
-
\Expectation_\inputRV
  \trace(
    \classicalQuantumChannel(\inputRV)
    \log \classicalQuantumChannel(\inputRV)
  )
=
\Expectation_\inputRV
\vonNeumannEntropy{\classicalQuantumChannel(\inputRV)}
,
\\
\label{eq:output-entropy}
\outputEntropySubscript{\classicalQuantumChannel}
&:=
-
\sum_{\classicalOutputIndex \in \naturals}
  \outputDistribution_{\classicalQuantumChannel}(\classicalOutputIndex)
  \log \outputDistribution_{\classicalQuantumChannel}(\classicalOutputIndex)
=
-
\trace(
  \classicalQuantumChannel_\inputDistribution
  \log \classicalQuantumChannel_\inputDistribution
)
=
\vonNeumannEntropy{\classicalQuantumChannel_\inputDistribution}.
\end{align}
If $\quantumInformation < \infty$, then $\jointEntropy$ and $\outputEntropy$ are both finite, and we have
\begin{equation}
\label{eq:quantum-information}
\quantumInformation
=
-\jointEntropy + \outputEntropy.
\end{equation}

Furthermore, we define for $\renyiorder \in [\renyiorder_{\min},1) \cup (1,\infty)$
\begin{align}
\label{eq:conditional-renyi-entropy}
\conditionalRenyiEntropySubscript{\renyiorder}{\classicalQuantumChannel}
&:=
\frac{1}{1-\renyiorder}
\log
\Expectation_{\inputDistribution_{\classicalQuantumChannel}}\left(
  \inputDistribution_{\classicalQuantumChannel}(\outputRV | \inputRV)^{\renyiorder-1}
\right)
=
\frac{1}{1-\renyiorder}
\log
\Expectation_{\inputDistribution_{\classicalQuantumChannel}}
\sum\limits_{\classicalOutputIndex \in \naturals}
  \inputDistribution_{\classicalQuantumChannel}(\classicalOutputIndex | \inputRV)^{\renyiorder}
=
\frac{1}{1-\renyiorder}
\log\Expectation_{\inputDistribution_{\classicalQuantumChannel}}\trace\left(
  \classicalQuantumChannel(\inputRV)^{\renyiorder}
\right)
\\
\label{eq:renyi-entropy}
\outputRenyiEntropySubscript{\renyiorder}{\classicalQuantumChannel}
&:=
\frac{1}{1-\renyiorder}
\log
\Expectation_{\outputDistribution_{\classicalQuantumChannel}}\left(
  \outputDistribution_{\classicalQuantumChannel}(\outputRV)^{\renyiorder-1}
\right)
=
\frac{1}{1-\renyiorder}
\log
\sum\limits_{\classicalOutputIndex \in \naturals}
  \outputDistribution_{\classicalQuantumChannel}(\classicalOutputIndex)^{\renyiorder}
=
\frac{1}{1-\renyiorder}
\log\trace\left(
  \classicalQuantumChannel_\inputDistribution^{\renyiorder}
\right).
\end{align}

We will use properties of $\conditionalRenyiEntropySubscript{\renyiorder}{\classicalQuantumChannel}$ and $\outputRenyiEntropySubscript{\renyiorder}{\classicalQuantumChannel}$ which are stated and proved in Lemmas~\ref{lemma:renyi-entropy-basics-preliminary} and~\ref{lemma:renyi-entropy-basics} in Appendix~\ref{appendix:funcana}. Lemma~\ref{lemma:renyi-entropy-basics} allows us to expand the definitions of $\outputRenyiEntropySubscript{\renyiorder}{\classicalQuantumChannel}$ and $\conditionalRenyiEntropySubscript{\renyiorder}{\classicalQuantumChannel}$ to the domain $[\renyiorder_{\min},\infty)$ and obtain continuous functions in $\renyiorder$ with $\outputRenyiEntropySubscript{1}{\classicalQuantumChannel} = \outputEntropySubscript{\classicalQuantumChannel}$ and $\conditionalRenyiEntropySubscript{1}{\classicalQuantumChannel} = \jointEntropySubscript{\classicalQuantumChannel}$.

We define the following functions:
\begin{align}
\nonumber
\classicalCodingBoundFunc^{\bobChannel}(\codebookRate,\blocklength)
&:=
\inf_{\substack{\typicalityParameter \in (0,\\ \legitInformation-\codebookRate)}}\Bigg(
  \exp\left(
    -
    \blocklength
    \sup_{\renyiorder \in (1,\infty)}\left(
      (\renyiorder-1)
      \left(
        \legitInformation
        +
        \typicalityParameter
        -
        \legitRenyDiv{\renyiorder}
      \right)
    \right)
  \right)
  +
  \exp\left(
    -
    \blocklength
    (\legitInformation+\typicalityParameter-\codebookRate)
  \right)
\Bigg)
\\
\label{eq:atypicalTermsFunc-start}
\atypicalTermsFunc_1^\classicalQuantumChannel(\typicalityParameter, \blocklength)
&:=
\exp\left(
  -\blocklength
    \sup_{\renyiorder \in (1,\infty)}
      (\renyiorder-1)(\conditionalRenyiEntropySubscript{\renyiorder}{\classicalQuantumChannel} + \typicalityParameter - \jointEntropySubscript{\classicalQuantumChannel})
\right)
\\
\label{eq:atypicalTermsFunc-2}
\atypicalTermsFunc_2^\classicalQuantumChannel(\typicalityParameter, \blocklength)
&:=
\exp\left(
  -\blocklength
   \sup_{\renyiorder \in [\renyiorder_{\min},1)}
     (1-\renyiorder)(\jointEntropySubscript{\classicalQuantumChannel} + \typicalityParameter - \conditionalRenyiEntropySubscript{\renyiorder}{\classicalQuantumChannel})
\right)
\\
\label{eq:atypicalTermsFunc-3}
\atypicalTermsFunc_3^\classicalQuantumChannel(\typicalityParameter, \blocklength)
&:=
\exp\left(
  -
  \blocklength
  \sup_{\renyiorder \in (1,\infty)}
    (\renyiorder-1)
    (
      \outputRenyiEntropySubscript{\renyiorder}{\classicalQuantumChannel}
      +
      \typicalityParameter
      -
      \outputEntropySubscript{\classicalQuantumChannel}
    )
\right)
\\
\label{eq:atypicalTermsFunc-end}
\atypicalTermsFunc_4^\classicalQuantumChannel(\typicalityParameter, \blocklength)
&:=
\exp\left(
  -
  \blocklength
  \sup_{\renyiorder \in [\renyiorder_{\min},1)}
    (1-\renyiorder)
    (
      \outputEntropySubscript{\classicalQuantumChannel}
      +
      \typicalityParameter
      -
      \outputRenyiEntropySubscript{\renyiorder}{\classicalQuantumChannel}
    )
\right)
\\
\label{eq:codingBoundFunc}
\codingBoundFunc^\classicalQuantumChannel(\codebookRate,\blocklength)
&:=\hspace{-7pt}
\inf_{\substack{\typicalityParameter \in (0,\\(\holevoInformation{\inputDistribution}{\classicalQuantumChannel}-\codebookRate)/2)}}\Big(
  2
  \atypicalTermsFunc_1^\classicalQuantumChannel(\typicalityParameter,\blocklength)
  +
  2
  \atypicalTermsFunc_2^\classicalQuantumChannel(\typicalityParameter,\blocklength)
  +
  4
  \atypicalTermsFunc_3^\classicalQuantumChannel(\typicalityParameter,\blocklength)
  +
  4
  \atypicalTermsFunc_4^\classicalQuantumChannel(\typicalityParameter,\blocklength)
  +
  4
  \exp\big(
    -
    \blocklength
    (\quantumInformation - \codebookRate - 2\typicalityParameter)
  \big)
\Big)
\\
\label{eq:resolvabilityBoundFunc}
\resolvabilityBoundFunc^\classicalQuantumChannel(\codebookRate,\blocklength)
&:= \hspace{-7pt}
\inf_{\substack{\typicalityParameter \in (0,\\(\codebookRate-\holevoInformation{\inputDistribution}{\classicalQuantumChannel})/4)}}\left(
  2\atypicalTermsFunc_1^\classicalQuantumChannel(\typicalityParameter, \blocklength)
  +
  2\atypicalTermsFunc_2^\classicalQuantumChannel(\typicalityParameter, \blocklength)
  +
  2\atypicalTermsFunc_3^\classicalQuantumChannel(\typicalityParameter, \blocklength)
  +
  2\atypicalTermsFunc_4^\classicalQuantumChannel(\typicalityParameter, \blocklength)
  +
  6\exp\left(
    -
    \frac{1}{2}
    \blocklength
    (\codebookRate - \quantumInformation - 4\typicalityParameter)
  \right)
\right)
\\
\label{eq:costBoundFunc}
\costBoundFunc^{(\costFunction,\costConstraint)}(
  \proofconst,
  \codebookRate,
  \blocklength
)
&:=
\exp\Big(
  -2\exp(\blocklength(\codebookRate-2\proofconst))\big(1-\exp(-\blocklength(\proofconst_1-\proofconst))\big)^2
\Big),
\end{align}
where $(\costFunction,\costConstraint)$ is an additive cost constraint and
\begin{equation}
\label{eq:mgf-parameter-choice}
\proofconst_1 := \sup_{\hat{\generalReal} \in (0,\infty)} \left(
  - \log \generalpmf(\hat{\generalReal})
\right),
\end{equation}
with $\generalpmf(\generalReal) := \Expectation_\inputDistribution \exp(\generalReal (\costFunction(\inputRV)-\costConstraint))$, the moment generating function of $\costFunction(\inputRV)-\costConstraint$.

With this, we are ready to state finite-block length versions of Theorem~\ref{theorem:wiretap-ccq} and Theorem~\ref{theorem:wiretap-cq}.

\begin{customthm}{\ref{theorem:wiretap-ccq}'}
\label{theorem:wiretap-ccq-all-blocklengths}
Let $(\bobChannel, \classicalQuantumChannel)$ be a $\gls{ccq}$ wiretap channel, let $\inputDistribution$ be a probability distribution on the input alphabet $\inputAlphabet$, and let $\inputRV \distributedAccTo \inputDistribution$ such that
\begin{itemize}
 \item there is $\renyiorder_{\min} < 1$ with $\eveChannel(\inputAlphabetElement)^{\renyiorder_{\min}} \in \traceClass{\hilbertSpace}$ for almost all $\inputAlphabetElement$, $\classicalQuantumChannel_\inputDistribution^{\renyiorder_{\min}} \in \traceClass{\hilbertSpace}$, and the Bochner integral $\Expectation \eveChannel(\inputRV)^{\renyiorder_{\min}}$ exists;
 \item there is $\generalReal > 0$ such that $\Expectation \exp(\generalReal \legitInformationDensity{\inputDistribution}{\inputRV}{\legitOutputRV}) < \infty$.
\end{itemize}
Let $(\costFunction, \costConstraint)$ be a cost constraint compatible with $\inputDistribution$, let $\codebookSize, \numCodebooks \in \naturals$, and let $\codebookRate,\randomnessRate,\combinedRate \in (0,\infty)$ with the properties
\begin{align}
\label{eq:wiretap-codebook-size}
\codebookSize &\geq \exp(\blocklength\randomnessRate)
\\
\label{eq:wiretap-codebook-size-num}
\codebookSize \numCodebooks &\leq \exp(\blocklength\combinedRate)
\\
\label{eq:wiretap-codebook-num}
\numCodebooks &\geq \exp(\blocklength\codebookRate).
\end{align}

Further assume that
\begin{align*}
\randomnessRate &> \holevoInformation{\inputDistribution}{\classicalQuantumChannel} \\
\combinedRate &< \legitInformation.
\end{align*}

Let $\proofconst_2 \in (0, \min(\proofconst_1,(\codebookRate+\numCodebooks)/2)), \proofconst_3 \in (0,\infty), \proofconst_4 \in (0, \min(\proofconst_1,\randomnessRate/2)), \proofconst_5 \in (0,
\randomnessRate/2)$, and $\blocklength \in \naturals$, where $\proofconst_1$ is defined in \eqref{eq:mgf-parameter-choice}, be chosen such that
\begin{multline}
\label{eq:wiretap-ccq-prob-one}
\classicalCodingBoundFunc^\bobChannel(\combinedRate,\blocklength)\exp(\blocklength\proofconst_3)
+
\costBoundFunc^{(\costFunction,\costConstraint)}(\proofconst_2,\codebookRate+\randomnessRate,\blocklength)
\\+
\costBoundFunc^{(\costFunction,\costConstraint)}(
  \proofconst_4,
  \randomnessRate,
  \blocklength
)
\exp(\blocklength(\combinedRate-\randomnessRate))
+
\exp\left(
  -\frac{1}{2}\exp(
    \blocklength
    (\randomnessRate-2\proofconst_5)
  )
  +
  \blocklength(
    \combinedRate-\randomnessRate
  )
\right)
<
1.
\end{multline}

Then there exists a wiretap code with the following properties:
\begin{enumerate}
 \item\label{item:wiretap-ccq-all-blocklengths-codebook-size} The code has size $\numCodebooks$.
 \item\label{item:wiretap-ccq-all-blocklengths-cost} The output of the encoder satisfies the cost constraint $(\costFunction,\costConstraint)$ almost surely.
 \item\label{item:wiretap-ccq-all-blocklengths-error} The code has an average error
 \[
 \decodingError
 =
 \exp(-\blocklength\proofconst_2)
 +
 \exp(-\blocklength\proofconst_3)
 \]
 as defined in (\ref{eq:def-average-error}).
 \item\label{item:wiretap-ccq-all-blocklengths-security} The code has distinguishing security level
  \[
  \securityNumber
  =
  2\resolvabilityBoundFunc^\classicalQuantumChannel(\randomnessRate,\blocklength)
  +
  4\exp(-\proofconst_4\blocklength)
  +
  2\exp(-\proofconst_5\blocklength)
  \]
 as defined in (\ref{eq:def-distinguishing-security}).
\end{enumerate}
\end{customthm}

\begin{customthm}{\ref{theorem:wiretap-cq}'}
\label{theorem:wiretap-cq-all-blocklengths}
Let $(\bobcqChannel,\eveChannel)$ be a \gls{cqq} wiretap channel, and let $\inputDistribution$ be a probability distribution on the input alphabet $\inputAlphabet$ such that for both choices of $\classicalQuantumChannel \in \{\eveChannel, \bobcqChannel\}$, there is $\renyiorder_{\min} < 1$ with $\classicalQuantumChannel(\inputAlphabetElement)^{\renyiorder_{\min}} \in \traceClass{\hilbertSpace}$ for almost all $\inputAlphabetElement$, $\classicalQuantumChannel_\inputDistribution^{\renyiorder_{\min}} \in \traceClass{\hilbertSpace}$, and the Bochner integral $\Expectation \classicalQuantumChannel(\inputRV)^{\renyiorder_{\min}}$ exists, where $\inputRV \distributedAccTo \inputDistribution$.

Let $(\costFunction, \costConstraint)$ be a cost constraint compatible with $\inputDistribution$, let $\codebookSize,\numCodebooks \in \naturals$ and $\codebookRate,\randomnessRate,\combinedRate \in \reals$ satisfying \eqref{eq:wiretap-codebook-size}, \eqref{eq:wiretap-codebook-size-num}, and \eqref{eq:wiretap-codebook-num}.

Further assume that
\begin{align*}
\randomnessRate &> \holevoInformation{\inputDistribution}{\eveChannel} \\
\combinedRate &< \holevoInformation{\inputDistribution}{\bobChannel}.
\end{align*}

Let $\proofconst_2 \in (0, \min(\proofconst_1,(\codebookRate+\numCodebooks)/2)), \proofconst_3 \in (0,\infty), \proofconst_4 \in (0, \min(\proofconst_1,\randomnessRate/2)), \proofconst_5 \in (0,\codebookRate/2)$, and $\blocklength \in \naturals$, where $\proofconst_1$ is defined in \eqref{eq:mgf-parameter-choice}, be chosen such that
\begin{multline}
\label{eq:wiretap-cq-prob-one}
\codingBoundFunc^{\bobcqChannel}(\combinedRate,\blocklength)\exp(\blocklength\proofconst_2)
+
\costBoundFunc^{(\costFunction,\costConstraint)}\left(\proofconst_3,\codebookRate+\randomnessRate,\blocklength\right)
\\+
\costBoundFunc^{(\costFunction,\costConstraint)}(
  \proofconst_4,
  \randomnessRate,
  \blocklength
)
\exp(\blocklength(\combinedRate-\randomnessRate))
+
\exp\left(
  -\frac{1}{2}\exp(
    \blocklength
    (\randomnessRate-2\proofconst_5)
  )
  +
  \blocklength(
    \combinedRate-\randomnessRate
  )
\right)
<
1.
\end{multline}

Then there exists a wiretap code with the following properties:
\begin{enumerate}
 \item The code has size $\numCodebooks$.
 \item The output of the encoder satisfies the cost constraint $(\costFunction,\costConstraint)$ almost surely.
 \item The code has an average error
 \[
 \decodingError
 =
 \exp(-\blocklength\proofconst_2)
 +
 2\exp(-\blocklength\proofconst_3)
 \]
 as defined in (\ref{eq:def-average-error-cq}).
 \item The code has distinguishing security level
  \[
  \securityNumber
  =
  2\resolvabilityBoundFunc^{\eveChannel}(\blocklength)
  +
  4\exp(-\proofconst_4\blocklength)
  +
  2\exp(-\proofconst_5\blocklength)
  \]
 as defined in (\ref{eq:def-distinguishing-security}).
\end{enumerate}
\end{customthm}

\section{Security Notions for Wiretap Channels with Quantum Outputs}
\label{sec:security}

The distinguishing security criterion provided by Theorems~\ref{theorem:wiretap-ccq} and~\ref{theorem:wiretap-cq} implies security guarantees against eavesdropping attacks that extend beyond an attacker's ability to reconstruct the entire transmitted message. In this section, we discuss various notions of security and their operational implications.

\subsection{Semantic Security}
\label{sec:semantic}
In this subsection, we give a definition of semantic security that is analogous to the classical archetype in~\cite{bellare2012semantic}. Formally, semantic security means that a large class of possible objectives (which includes the objective of reconstructing the entire transmitted message, but also many others) cannot be reached by the eavesdropper. To the best of our knowledge, this definition of semantic security has not appeared before in the literature for channels with quantum outputs.

\begin{Prob}
\label{prob:eavesdropper}
\emph{(Eavesdropper's objective.)}
Each possible eavesdropper's objective that we consider in this paper is defined by a partition $\messageSpacePartition$ of the message space $\{1, \dots, \codebookSize\}$. The eavesdropper's objective is then to output some $\messageSpacePartitionElement \in \messageSpacePartition$ such that the originally transmitted message is contained in $\messageSpacePartitionElement$.
\end{Prob}\penalty-200

The partition of $\{1, \dots, \codebookSize\}$ that corresponds to the reconstruction of the entire message consists of all singleton sets, i.e.,
\[
\messageSpacePartition = \big\{ \{1\}, \dots, \{\codebookSize\} \big\}.
\]

But also other possible objectives can be defined by a message space partition. For instance, the task of reconstructing the first bit of the transmitted message would be represented by
\[
\messageSpacePartition = \big\{
  \{
    \codewordIndex:~ \textrm{The binary representation of $\codewordIndex$ starts with $0$}
  \},
  \{
    \codewordIndex:~ \textrm{The binary representation of $\codewordIndex$ starts with $1$}
  \}
\big\}.
\]

In this section, we assume that the eavesdropper has prior knowledge of the probability distribution $\messageProbabilityDistribution$ from which the transmitted message $\messageRV$ is drawn. While this may not be realistic in some practical scenarios, it is certainly possible, for example when the eavesdropper knows that a protocol is executed in which only certain messages can be transmitted at certain time instances. In any case, this is not a restrictive assumption since it represents an additional advantage that the eavesdropper has compared to a case where no such prior knowledge is available.

We assume in this work that, after message $\messageRV$ has been transmitted, the eavesdropper can perform any $\messageSpacePartition$-valued \gls{povm} on the output $\outputStateEve(\messageRV)$ of the channel
\begin{equation}\label{eq:eve-q-output-state}
  \outputStateEve :=  \eveChannel^\blocklength \circ \encoder.
\end{equation}
By the Born rule, a $\messageSpacePartition$-valued \gls{povm} $(\povm_\messageSpacePartitionElement)_{\messageSpacePartitionElement \in \messageSpacePartition}$ and $\outputStateEve(\messageRV)$ induce a probability distribution on $\messageSpacePartition$, described by the \gls{pmf}
\begin{equation*}
\generalpmf(\messageSpacePartitionElement)
:=
\trace(\povm_\messageSpacePartitionElement \outputStateEve(\messageRV)).
\end{equation*}
$\generalpmf$ describes the probability distribution of the eavesdropper's conclusion of what the correct partition element $\messageSpacePartitionElement$ is under a fixed communication scheme, distribution of the message, and reconstruction strategy of the eavesdropper.

We define the eavesdropper's maximum success probability when a quantum measurement is performed on a quantum state $\outputStateEve(\messageRV)$ as
\begin{equation}
\label{eq:def-eavesdropper-success}
\eavesdropperSuccessProbability(\messageSpacePartition, \messageProbabilityDistribution, \outputStateEve)
:=
\sup\left\{
    \Expectation_\messageRV \sum_{\messageSpacePartitionElement \in \messageSpacePartition: \messageRV \in \messageSpacePartitionElement} \trace(\povm_\messageSpacePartitionElement \outputStateEve(\messageRV))
    :~
    (\povm_\messageSpacePartitionElement)_{\messageSpacePartitionElement \in \messageSpacePartition} \text{ is a \gls{povm}}
\right\}.
\end{equation}
Note that since $\messageSpacePartition$ is a partition, the sum consists of only one summand for each fixed realization of $\messageRV$.

It is important to emphasize at this point that it is not possible in general to guarantee a low eavesdropper's success probability. For instance, $\messageRV$ could be drawn from a distribution where the binary representation of $\messageRV$ almost surely starts with $0$. What we can control, however, is the advantage that the eavesdropper gains by observing $\outputStateEve(\messageRV)$ compared to the situation where it cannot observe any channel output. We denote the maximum achievable success probability in solving Problem~\ref{prob:eavesdropper} for partition $\messageSpacePartition$ and message distribution $\messageProbabilityDistribution$ without observing any quantum state (i.e., by pure guessing) by 
\[
\eavesdropperGuessProbability(\messageSpacePartition, \messageProbabilityDistribution)
:=
\sup\left\{
  \Probability(\messageRV \in \messageSpacePartitionElement):~
  \messageSpacePartitionElement \text{ is a random variable ranging over $\messageSpacePartition$}
\right\}.
\]

This leads us to the following definition of semantic security.
\begin{definition}
\label{def:semantic-security}
The \emph{eavesdropper's semantic security advantage} associated with the eavesdropper's output state $\outputStateEve(\messageRV)$ is defined as
\begin{multline}
\label{eq:def-eavesdropper-advantage}
\eavesdropperAdvantageSemantic(\outputStateEve)
:=
\sup\Big\{
  \eavesdropperSuccessProbability(\messageSpacePartition, \messageProbabilityDistribution, \outputStateEve)
  -
  \eavesdropperGuessProbability(\messageSpacePartition, \messageProbabilityDistribution)
  :~
  \\
  \messageSpacePartition \text{ is a partition of $\{1, \dots, \codebookSize\}$},~
  \messageProbabilityDistribution \text{ is a probability distribution on $\{1, \dots, \codebookSize\}$}.
\Big\}
\end{multline}
We say that \emph{semantic security level $\securityNumber$} is satisfied for the output state $\outputStateEve$ if
\[
\eavesdropperAdvantageSemantic(\outputStateEve) \leq \securityNumber.
\]
\end{definition}

This is a straightforward quantum analog of the term \emph{semantic security} which is already an established performance metric both in classical cryptography and information-theoretic secrecy~\cite{bellare2012semantic}.

\subsection{Other Security Notions}
There are alternative security definitions that are easier to analyze in proofs than semantic security in the sense of Definition~\ref{def:semantic-security}. In this section, we introduce several of these alternative established notions of security. We point out that the terminology is not the same everywhere in the literature, neither for classical nor for quantum wiretap channels. Because both distinguishing security as defined in Definition~\ref{def:alternative-security}-\ref{item:distinguishing-security} and mutual information security as defined in Definition~\ref{def:alternative-security}-\ref{item:mutual-information-security} below imply semantic security in the sense of Definition~\ref{def:semantic-security} above, some papers use these alternative definitions directly for semantic security. In this section, we define these alternative notions and show some implications between them. This is very similar to the treatment of the classical case~\cite{bellare2012semantic,bellare2012cryptographic}.

\begin{definition}
\label{def:alternative-security}
Given a \gls{cq} channel $\{1, \dots ,\codebookSize \}\ni \codewordIndex \mapsto \outputStateEve(\codewordIndex)$ (cf. (\ref{eq:eve-q-output-state})) and $\securityNumber\in\reals$, $\securityNumber\ge 0$, we say that:
\begin{enumerate}
 \item \label{item:weak-secrecy} \emph{Weak secrecy level $\securityNumber$} is satisfied if
 \[
   \eavesdropperAdvantageWeak(\outputStateEve)
   :=
   \frac{1}{n}
   \holevoInformation{\messageProbabilityDistribution}{\outputStateEve}
   \leq
   \securityNumber,
 \]
 where $\messageProbabilityDistribution$ is the uniform distribution over all messages $\{1, \dots, \codebookSize\}$.
 \item \label{item:strong-secrecy} \emph{Strong secrecy level $\securityNumber$} is satisfied if
 \[
   \eavesdropperAdvantageStrong(\outputStateEve)
   :=
   \holevoInformation{\messageProbabilityDistribution}{\outputStateEve}
   \leq
   \securityNumber,
 \]
 where $\messageProbabilityDistribution$ is the uniform distribution over all messages $\{1, \dots, \codebookSize\}$.
 \item \label{item:mutual-information-security} \emph{Mutual information security level $\securityNumber$} is satisfied if
 \[
   \eavesdropperAdvantageMutualInformation(\outputStateEve)
   :=
   \max_{\messageProbabilityDistribution}
     \holevoInformation{\messageProbabilityDistribution}{\outputStateEve}
   \leq
   \securityNumber,
 \]
 where $\messageProbabilityDistribution$ ranges over all probability distributions of messages $\{1, \dots, \codebookSize\}$.
 \item \label{item:distinguishing-security} \emph{Distinguishing security level $\securityNumber$} is satisfied if
 \[
  \eavesdropperAdvantageDistinguishing(\outputStateEve)
   :=
   \max\Big\{
     \traceNorm{\outputStateEve(\codewordIndex_1) - \outputStateEve(\codewordIndex_2)}
     ~:~
     \codewordIndex_1, \codewordIndex_2 \in \{1, \dots, \codebookSize\}
   \Big\}
   \leq
   \securityNumber.
 \]
\end{enumerate}
The terms on the left-hand side of these inequalities are called the weak secrecy (strong secrecy, mutual information security, distinguishing security) \emph{advantage} of the eavesdropper.

For a sequence of encoding schemes with growing block length $\blocklength$, we say that the sequence is \emph{semantically secure (weakly secret, strongly secret, mutual information secure, distinction secure)} if the corresponding advantage vanishes as $\blocklength$ tends to infinity.
\end{definition}

\begin{remark}
We note that post-processing steps carried out by the eavesdropper cannot increase any of the advantages given in Definition~\ref{def:semantic-security} and Definition~\ref{def:alternative-security}.
To see this, let $\evePostprocChannel: \traceClass{\eveHilbertSpace}\to \traceClass{\eveHilbertSpace'}$ be a quantum channel, i.e.,  positive, trace-preserving map such that the dual map $\evePostprocChannel': \boundedOperators{\eveHilbertSpace'}\to \boundedOperators{\eveHilbertSpace}$, given by
\begin{equation}\label{eq:cp-duality}
\trace (\generalOperator \evePostprocChannel(\generalOperatorTwo))=\trace (\evePostprocChannel'(\generalOperator) \generalOperatorTwo)
\end{equation}
for all $\generalOperator\in \boundedOperators{\eveHilbertSpace'}$ and all $\generalOperatorTwo\in\traceClass{\eveHilbertSpace} $, is completely positive \cite[Section 7.1]{Busch-Quantum-Measurement}. The map $\evePostprocChannel$ represents a possible post-processing that can be performed by the eavesdropper. The monotonicity of quantum relative entropy under the action of quantum channels \cite{lindblad1975completely} implies that
\begin{equation*}
 \eavesdropperAdvantageWeak(\evePostprocChannel \circ \outputStateEve)\le \eavesdropperAdvantageWeak(\outputStateEve),   
\end{equation*}
\begin{equation*}
 \eavesdropperAdvantageStrong(\evePostprocChannel \circ \outputStateEve) \le \eavesdropperAdvantageStrong(\outputStateEve),  
\end{equation*}
and
\begin{equation*}
\eavesdropperAdvantageMutualInformation(\evePostprocChannel \circ \outputStateEve)\le  \eavesdropperAdvantageMutualInformation(\outputStateEve).   
\end{equation*}
Moreover, since for all states $\generalQuantumState, \generalQuantumStateTwo\in\densityOperators{\hilbertSpace}$ we have  \cite[Proposition 4.37]{heinosaari2012mathematical}
\begin{equation*}
    \traceNorm{\evePostprocChannel(\generalQuantumState) - \evePostprocChannel(\generalQuantumStateTwo)}
    \le \traceNorm{\generalQuantumState-\generalQuantumStateTwo},
\end{equation*}
we can conclude that
\begin{equation*}
 \eavesdropperAdvantageDistinguishing(\evePostprocChannel \circ \outputStateEve)
 \le 
  \eavesdropperAdvantageDistinguishing(\outputStateEve),
\end{equation*}
holds as well.
Finally, we have
\begin{equation*}
\eavesdropperAdvantageSemantic(\evePostprocChannel \circ \outputStateEve)
\le 
\eavesdropperAdvantageSemantic(\outputStateEve),
\end{equation*}
which can be deduced as follows. Since for any given partition $\messageSpacePartition$ and probability distribution $\messageProbabilityDistribution$, the quantity 
\[
 \eavesdropperGuessProbability(\messageSpacePartition, \messageProbabilityDistribution)
\]
does not depend on the eavesdropper's output state, it is sufficient to prove
\begin{equation*}
 \eavesdropperSuccessProbability(\messageSpacePartition, \messageProbabilityDistribution, \evePostprocChannel \circ \outputStateEve)  \le 
  \eavesdropperSuccessProbability(\messageSpacePartition, \messageProbabilityDistribution, \outputStateEve),
\end{equation*}
for all partitions $\messageSpacePartition$ and probability distributions $\messageProbabilityDistribution$.
Showing this is the aim of the following argument:
\begin{align*}
\eavesdropperSuccessProbability(\messageSpacePartition, \messageProbabilityDistribution, \evePostprocChannel \circ \outputStateEve)
 &=
\sup\left\{
    \Expectation_\messageRV \sum_{\messageSpacePartitionElement \in \messageSpacePartition: \messageRV \in \messageSpacePartitionElement} \trace(\povm_\messageSpacePartitionElement \evePostprocChannel \circ \outputStateEve(\messageRV))
    :~
    (\povm_\messageSpacePartitionElement)_{\messageSpacePartitionElement \in \messageSpacePartition} \text{ is a \gls{povm} on } \eveHilbertSpace'
\right\} \\
\overset{(a)}&{=} 
\sup\left\{
    \Expectation_\messageRV \sum_{\messageSpacePartitionElement \in \messageSpacePartition: \messageRV \in \messageSpacePartitionElement} \trace(\evePostprocChannel'(\povm_\messageSpacePartitionElement ) \outputStateEve(\messageRV))
    :~
    (\povm_\messageSpacePartitionElement)_{\messageSpacePartitionElement \in \messageSpacePartition} \text{ is a \gls{povm} on } \eveHilbertSpace'
\right\} \displaybreak[0] \\
&=
\sup\left\{
    \Expectation_\messageRV \sum_{\messageSpacePartitionElement \in \messageSpacePartition: \messageRV \in \messageSpacePartitionElement} \trace(\povm'_\messageSpacePartitionElement \outputStateEve(\messageRV))
    :~
    \povm'_{\messageSpacePartitionElement}=\evePostprocChannel'(\povm_{\messageSpacePartitionElement}) \text{ where }(\povm_\messageSpacePartitionElement)_{\messageSpacePartitionElement \in \messageSpacePartition} \text{ is a \gls{povm} on } \eveHilbertSpace'
\right\} \displaybreak[0] \\
\overset{(b)}&{\le}
\sup\left\{
    \Expectation_\messageRV \sum_{\messageSpacePartitionElement \in \messageSpacePartition: \messageRV \in \messageSpacePartitionElement} \trace(\povm'_\messageSpacePartitionElement \outputStateEve(\messageRV))
    :~
    (\povm'_\messageSpacePartitionElement)_{\messageSpacePartitionElement \in \messageSpacePartition} \text{ is a \gls{povm} on } \eveHilbertSpace
\right\}\\
&=
\eavesdropperSuccessProbability(\messageSpacePartition, \messageProbabilityDistribution, \outputStateEve),
\end{align*}
where step (a) is by definition (\ref{eq:cp-duality}) of the dual map $\evePostprocChannel'$. The inequality (b) follows from the fact that $\evePostprocChannel'$ is a unital (i.e., $\evePostprocChannel'(\identityOperator_{\eveHilbertSpace'})=\identityOperator_{\eveHilbertSpace}$) and positive map. Therefore, it maps the \gls{povm}s on $\eveHilbertSpace'$ to a \emph{subset} of the \gls{povm}s on the Hilbert space $\eveHilbertSpace$, which leads to a larger supremum in the inequality (b).\\
Consequently, in proofs of security guarantees, it is sufficient to focus on the case $\outputStateEve = \classicalQuantumChannel^\blocklength \circ \encoder$, as we do in Theorems~\ref{theorem:wiretap-ccq} and~\ref{theorem:wiretap-cq}, since the guarantees automatically extend to all post-processing attempts represented by application of quantum channels at the eavesdropper.
\end{remark}

In the remainder of this section, we follow the arguments in \cite{bellare2012cryptographic} for the classical case and prove that similar implications also hold in the quantum case between the different security notions of Definitions~\ref{def:semantic-security} and \ref{def:alternative-security}. For the implication from distinguishing security to mutual information security, we use a slightly more general approach that applies not only to finite-dimensional systems but also to some practically relevant infinite-dimensional systems.

\begin{lemma}
\label{lemma:semantic-security}
$\eavesdropperAdvantageSemantic(\outputStateEve) \leq \frac{1}{2} \eavesdropperAdvantageDistinguishing(\outputStateEve)$.
\end{lemma}

\begin{proof}
Suppose the eavesdropper's output state is $\outputStateEve(\messageRV)$, and let $(\povm_{\messageSpacePartitionElement})_{\messageSpacePartitionElement \in \messageSpacePartition}$ be a \gls{povm}.

We define a \gls{pmf}
\begin{equation}
\label{eq:semantic-security-eve-guess}
\evepmf:
\messageSpacePartition \rightarrow [0,1],~
\messageSpacePartitionElement
\mapsto
\Expectation_\messageRV
  \trace(\povm_\messageSpacePartitionElement \outputStateEve(\messageRV)).
\end{equation}
Clearly, $\evepmf$ does not depend on an observed channel output, so we have
\begin{equation}
\label{eq:semantic-security-comparison-advantage}
\eavesdropperGuessProbability(\messageSpacePartition, \messageProbabilityDistribution)
\geq
\Expectation_\messageRV
\sum_{\messageSpacePartitionElement \in \messageSpacePartition: \messageRV \in \messageSpacePartitionElement} \evepmf(\messageSpacePartitionElement).
\end{equation}
Therefore,
\begin{align*}
\Expectation_\messageRV \sum_{\messageSpacePartitionElement \in \messageSpacePartition: \messageRV \in \messageSpacePartitionElement} \trace(\povm_\messageSpacePartitionElement \outputStateEve(\messageRV))
-
\eavesdropperGuessProbability(\messageSpacePartition, \messageProbabilityDistribution)
\overset{(\ref{eq:semantic-security-comparison-advantage})}&{\leq}
\Expectation_\messageRV
  \sum_{\messageSpacePartitionElement \in \messageSpacePartition: \messageRV \in \messageSpacePartitionElement}\left(
    \trace(\povm_\messageSpacePartitionElement \outputStateEve(\messageRV))
    -
    \evepmf(\messageSpacePartitionElement)
  \right)
\\
\overset{(\ref{eq:semantic-security-eve-guess})}&{=}
\Expectation_\messageRV
  \sum_{\messageSpacePartitionElement \in \messageSpacePartition: \messageRV \in \messageSpacePartitionElement}\left(
    \trace\big(
      \povm_{\messageSpacePartitionElement}\outputStateEve(\messageRV)
    \big)
    -
    \Expectation_\messageRV
      \trace\big(
        \povm_{\messageSpacePartitionElement}\outputStateEve(\messageRV)
      \big)
  \right)
\displaybreak[0] \\
&=
\Expectation_\messageRV
  \trace\left(
    \sum_{\messageSpacePartitionElement \in \messageSpacePartition: \messageRV \in \messageSpacePartitionElement}\left(
        \povm_{\messageSpacePartitionElement}\outputStateEve(\messageRV)
        -
        \povm_{\messageSpacePartitionElement}\Expectation_\messageRV\outputStateEve(\messageRV)
    \right)
  \right)
\displaybreak[0] \\
&=
\Expectation_\messageRV
  \trace\left(
    \left(
      \sum_{\messageSpacePartitionElement \in \messageSpacePartition: \messageRV \in \messageSpacePartitionElement}
        \povm_{\messageSpacePartitionElement}
    \right)(
      \outputStateEve(\messageRV) - \Expectation_\messageRV\outputStateEve(\messageRV)
    )
  \right)
\displaybreak[0] \\
\overset{(a)}&{\leq}
\frac{1}{2}
\Expectation_\messageRV
  \traceNorm{
    \outputStateEve(\messageRV) - \Expectation_\messageRV\outputStateEve(\messageRV)
  }
\displaybreak[0] \\
\overset{(b)}&{\leq}
\frac{1}{2}
\Expectation_{\messageRV,\messageRV'}
  \traceNorm{
    \outputStateEve(\messageRV')
    -
    \outputStateEve(\messageRV)
  }
\\
&\leq
\frac{1}{2}
\eavesdropperAdvantageDistinguishing(\outputStateEve).
\end{align*}
(a) is by Lemma~\ref{lemma:norm-basics}-\ref{item:norm-basics-zero-trace-operators}, and (b) is by replacing $\messageRV$ with an independent copy $\messageRV'$ and then applying Lemma~\ref{lemma:bochner-integral-basics}-\ref{item:bochner-integral-basics-trace-norm-jensen}. Since this derivation holds for every choice of $\povm$, $\messageProbabilityDistribution$ and $\messageSpacePartition$, the lemma immediately follows.
\end{proof}

\begin{lemma}
$\eavesdropperAdvantageDistinguishing(\outputStateEve) \leq 4 \eavesdropperAdvantageSemantic(\outputStateEve)$.
\end{lemma}
\begin{proof}
Fix arbitrary $\codewordIndex_1, \codewordIndex_2 \in \{1, \dots, \codebookSize\}$. We will show that
\[
\traceNorm{\outputStateEve(\codewordIndex_1) - \outputStateEve(\codewordIndex_2)}
\leq
4\eavesdropperAdvantageSemantic(\outputStateEve).
\]
Since this is immediately clear for $\codewordIndex_1 = \codewordIndex_2$, we may assume $\codewordIndex_1 \neq \codewordIndex_2$ in the following. We fix the message probability distribution
\[
\messageProbabilityDistribution(\codewordIndex)
:=
\begin{cases}
\frac{1}{2}, &\codewordIndex \in \{\codewordIndex_1, \codewordIndex_2\} \\
0, &\text{otherwise,}
\end{cases}
\]
and the partition $\messageSpacePartition := \{\{\codewordIndex_1\}, \{1, \dots, \codebookSize\} \setminus \{\codewordIndex_1\}\}$. Any guess over this partition will be correct with probability $1/2$. So for any operator $0 \leq \povm \leq \identityOperator$, we can calculate
\[
\eavesdropperAdvantageSemantic(\outputStateEve)
\overset{(\ref{eq:def-eavesdropper-advantage})}{\geq}
\eavesdropperSuccessProbability(\messageSpacePartition,\messageProbabilityDistribution,\outputStateEve)
-
\frac{1}{2}
\overset{(\ref{eq:def-eavesdropper-success})}{\geq}
\frac{1}{2} \trace(\povm \outputStateEve(\codewordIndex_1))
+
\frac{1}{2} \trace((\identityOperator-\povm) \outputStateEve(\codewordIndex_2))
-
\frac{1}{2}
=
\frac{1}{2} \trace(\povm (\outputStateEve(\codewordIndex_1)-\outputStateEve(\codewordIndex_2))).
\]
We now fix $\povm$ as the operator which maximizes the latter term according to Lemma~\ref{lemma:norm-basics}-\ref{item:norm-basics-zero-trace-operators}, and obtain
\[
\eavesdropperAdvantageSemantic(\generalOperator)
\geq
\frac{1}{4}
\traceNorm{\outputStateEve(\codewordIndex_1)-\outputStateEve(\codewordIndex_2)},
\]
concluding the proof.
\end{proof}

\begin{lemma}
$\eavesdropperAdvantageDistinguishing(\outputStateEve) \leq 2\sqrt{2\eavesdropperAdvantageMutualInformation(\outputStateEve)}$.
\end{lemma}
\begin{proof}
Fix arbitrary $\codewordIndex_1, \codewordIndex_2 \in \{1, \dots, \codebookSize\}$. We will show that
\begin{equation}
\label{eq:information-security-to-distinguishing-objective}
\traceNorm{\outputStateEve(\codewordIndex_1) - \outputStateEve(\codewordIndex_2)}
\leq
2\sqrt{2\eavesdropperAdvantageMutualInformation(\outputStateEve)}.
\end{equation}
Since this is immediately clear for $\codewordIndex_1 = \codewordIndex_2$, we may assume $\codewordIndex_1 \neq \codewordIndex_2$ in the following. We define quantum states
\begin{align*}
\generalQuantumState
&:=
\frac{1}{2} \rankOneOperator{\codewordIndex_1} \otimes \outputStateEve(\codewordIndex_1)
+
\frac{1}{2} \rankOneOperator{\codewordIndex_2} \otimes \outputStateEve(\codewordIndex_2)
\\
\generalQuantumStateTwo
&:=
\left(
  \frac{1}{2} \rankOneOperator{\codewordIndex_1}
  +
  \frac{1}{2} \rankOneOperator{\codewordIndex_2}
\right)
\otimes
\left(
  \frac{1}{2} \outputStateEve(\codewordIndex_1)
  +
  \frac{1}{2} \outputStateEve(\codewordIndex_2)
\right),
\end{align*}
where $\ket{\codewordIndex_1}, \ket{\codewordIndex_2}$ are an orthonormal basis of some $2$-dimensional Hilbert space. We further fix a distribution $\messageProbabilityDistribution$ on the set $\{1, \dots, \codebookSize\}$ of messages via
\[
\messageProbabilityDistribution(\codewordIndex)
:=
\begin{cases}
\frac{1}{2}, &\codewordIndex \in \{\codewordIndex_1, \codewordIndex_2\} \\
0, &\text{otherwise.}
\end{cases}
\]
Then, we can calculate
\begin{align*}
\eavesdropperAdvantageMutualInformation(\outputStateEve)
&\geq
\holevoInformation{\messageProbabilityDistribution}{\outputStateEve}
\\
\overset{(a)}&{=}
\quantumRelativeEntropy{\generalQuantumState}{\generalQuantumStateTwo}
\\
\overset{(b)}&{\geq}
\frac{1}{2}
\traceNorm{\generalQuantumState-\generalQuantumStateTwo}^2
\\
&=
\frac{1}{2}
\traceNorm{
  \frac{1}{4}
  \left(
    \rankOneOperator{\codewordIndex_1}
    -
    \rankOneOperator{\codewordIndex_2}
  \right)
  \otimes
  \left(
    \outputStateEve(\codewordIndex_1)
    -
    \outputStateEve(\codewordIndex_2)
  \right)
}^2
\\
&=
\frac{1}{2} \cdot \frac{1}{16}
\traceNorm{
  \rankOneOperator{\codewordIndex_1}
  -
  \rankOneOperator{\codewordIndex_2}
}^2
\traceNorm{
  \outputStateEve(\codewordIndex_1)
  -
  \outputStateEve(\codewordIndex_2)
}^2
\\
&=
\frac{1}{8}
\traceNorm{
  \outputStateEve(\codewordIndex_1)
  -
  \outputStateEve(\codewordIndex_2)
}^2
\end{align*}
where step (a) uses the definition of \emph{relative quantum entropy}
\[
\quantumRelativeEntropy{\generalQuantumState}{\generalQuantumStateTwo}
:=
\trace{\generalQuantumState \log \generalQuantumState}
-
\trace{\generalQuantumState \log \generalQuantumStateTwo},
\]
where by convention, $\quantumRelativeEntropy{\generalQuantumState}{\generalQuantumStateTwo} = \infty$ whenever the support of $\generalQuantumState$ is not contained in the support of $\generalQuantumStateTwo$ or one of the traces is infinite~\cite[Definition 7.1]{holevo2019quantum}. Step (b) is a quantum version of the Pinsker inequality, namely, the stronger version of~\cite[Proposition 7.3]{holevo2019quantum} that is mentioned in~\cite[Section 7.8]{holevo2019quantum} and originally proven for a more general scenario in~\cite[Theorem 3.1]{hiai1981sufficiency}. This yields (\ref{eq:information-security-to-distinguishing-objective}), concluding the proof of the lemma.
\end{proof}

For the bound of $\eavesdropperAdvantageMutualInformation(\outputStateEve)$ in terms of $\eavesdropperAdvantageDistinguishing(\outputStateEve)$, we introduce some technical terms first. A linear operator $\gibbsObservable$ on $\hilbertSpace$ is called a \emph{Gibbs observable} if it satisfies the following properties:
\begin{itemize}
    \item It is of the form
    \begin{equation}
    \label{eq:gibbs-eigenvalues}
        \gibbsObservable \ket{\hilbertSpaceElement}
        :=
        \sum_{\generalIndex \in \naturals}
          \gibbsEigenvalue_\generalIndex
          \innerProduct{\gibbsEigenvector_\generalIndex}{\hilbertSpaceElement}\ket{\gibbsEigenvector_\generalIndex},
    \end{equation}
    defined on the dense domain
    \begin{equation}
    \label{eq:gibbs-domain}
       \mathcal{D}(\gibbsObservable):= \left\{
          \hilbertSpaceElement \in \hilbertSpace
          :
          \sum_{\generalIndex \in \naturals}
            \absolute{\gibbsEigenvalue_\generalIndex}^2
            \absolute{\innerProduct{\gibbsEigenvector_\generalIndex}{\hilbertSpaceElement}}^2
          <
          \infty
        \right\},
    \end{equation}
    where $(\gibbsEigenvalue_\generalIndex)_{\generalIndex \in \naturals}$ is an unbounded sequence of nonnegative real numbers that includes $0$ and $\ket{\gibbsEigenvector_\generalIndex}_{\generalIndex \in \naturals}$ form an orthonormal basis in $\hilbertSpace$ (cf.~\cite[Definition 11.3]{holevo2019quantum} and~\cite[Section 4]{winter2016tight}). Note that  $\gibbsObservable$ is an unbounded self-adjoint operator on the domain $ \mathcal{D}(\gibbsObservable) $ (which is a simple consequence of \cite[Theorem VIII.3 (c)]{reed1980functional}).
    \item For every $\gibbsParameter \in (0,\infty)$, it holds that
    \begin{equation}\label{eq:gibbs-regularity}
\exp(-\gibbsParameter\gibbsObservable) \in \traceClass{\hilbertSpace},
    \end{equation}
      with the definition
    \[
        \exp(-\gibbsParameter\gibbsObservable)
        \ket{\hilbertSpaceElement}
        :=
        \sum_{\generalIndex \in \naturals}
          \exp(-\gibbsParameter\gibbsEigenvalue_\generalIndex)
          \innerProduct{\gibbsEigenvector_\generalIndex}{\hilbertSpaceElement} \ket{\gibbsEigenvector_\generalIndex},
    \]
    $\hilbertSpaceElement \in \hilbertSpace$ (cf.~\cite[Section 4]{winter2016tight}). Note that for (\ref{eq:gibbs-regularity}) to hold, it is necessary that the spectrum of $\gibbsObservable$ given by the sequence  $(\gibbsEigenvalue_\generalIndex)_{\generalIndex \in \naturals}$  is unbounded and every eigenvalue has finite multiplicity.
\end{itemize}
For a Gibbs observable $\gibbsObservable$ and $\generalQuantumState \in \densityOperators{\hilbertSpace}$, we adopt the convention (cf.~\cite[(11.6)]{holevo2019quantum})
\begin{equation}
\label{eq:gibbs-trace}
    \trace(\generalQuantumState \gibbsObservable)
    :=
    \sum_{\generalIndex \in \naturals}
      \gibbsEigenvalue_\generalIndex
      \innerProduct{\gibbsEigenvector_\generalIndex}{\generalQuantumState\gibbsEigenvector_\generalIndex}
    \in
    [0,\infty],
\end{equation}
and we define (see~\cite[Section 4]{winter2016tight}) for $\gibbsEnergy \in [0,\infty)$
\begin{equation}
\label{eq:gibbs-max-entropy-function}
    \gibbsEntropy{\gibbsObservable}(\gibbsEnergy)
    :=
    \sup\left\{
      \vonNeumannEntropy{\generalQuantumState}
      :
      \generalQuantumState \in \densityOperators{\hilbertSpace},~
      \trace(\generalQuantumState\gibbsObservable) \leq \gibbsEnergy
    \right\}
    \in
    [0,\infty).
\end{equation}
In~\cite[Proposition 1]{shirokov2006entropy}, it is proven that this supremum is finite and is in fact realized by a state of maximum entropy.

\begin{lemma}
\label{lemma:distinguishing-to-information}
Suppose that $\hat \gibbsObservable$ is a Gibbs observable on $\eveHilbertSpace^{\otimes \blocklength}$ and $\gibbsEnergy \in [0,\infty)$ such that for every $\codewordIndex \in \{1, \dots, \codebookSize\}$, we have $\trace(\outputStateEve(\codewordIndex) \hat \gibbsObservable) \leq \gibbsEnergy$. Assume further that $\eavesdropperAdvantageDistinguishing(\outputStateEve) \leq 1$. Then, we have
\[
\eavesdropperAdvantageMutualInformation(\outputStateEve)
\leq
\eavesdropperAdvantageDistinguishing(\outputStateEve)
\gibbsEntropy{\hat \gibbsObservable}\left(
  \frac{2\gibbsEnergy}
       {\eavesdropperAdvantageDistinguishing(\outputStateEve)}
\right)
+
\binaryEntropy\left(
  \frac{\eavesdropperAdvantageDistinguishing(\outputStateEve)}{2}
\right),
\]
where
$\binaryEntropy: [0,1] \rightarrow \reals, \generalReal \mapsto -\generalReal \log \generalReal - (1-\generalReal)\log(1-\generalReal)$ is the binary entropy.
\end{lemma}
\begin{proof}
We fix an arbitrary probability distribution $\messageProbabilityDistribution$ on $\{1, \dots, \codebookSize\}$ and define
\begin{equation}
\label{eq:security-implication-minentropy}
\tilde{\codewordIndex}
:=
\argmin_{\codewordIndex \in \{1, \dots, \codebookSize\}}
  \vonNeumannEntropy{\outputStateEve(\codewordIndex)}.
\end{equation}
Then, we have
\begin{align*}
\holevoInformation{\messageProbabilityDistribution}{\outputStateEve}
&=
\vonNeumannEntropy{\Expectation_\messageRV \outputStateEve(\codewordIndex)}
-
\Expectation_\messageRV
  \vonNeumannEntropy{\outputStateEve(\messageRV)}
\\
\overset{(\ref{eq:security-implication-minentropy})}&{\leq}
\vonNeumannEntropy{\Expectation_\messageRV \outputStateEve(\codewordIndex)}
-
\vonNeumannEntropy{\outputStateEve(\tilde{\codewordIndex})}
\\
\overset{(a)}&{\leq}
\traceNorm{
  \Expectation_\messageRV \outputStateEve(\codewordIndex)
  -
  \outputStateEve(\tilde{\codewordIndex})
}
\gibbsEntropy{\hat \gibbsObservable}\left(
  \frac{2\gibbsEnergy}
  {
    \traceNorm{
      \Expectation_\messageRV \outputStateEve(\codewordIndex)
      -
      \outputStateEve(\tilde{\codewordIndex})
    }
  }
\right)
+
\binaryEntropy\left(
  \frac{
    \traceNorm{
      \Expectation_\messageRV \outputStateEve(\codewordIndex)
      -
      \outputStateEve(\tilde{\codewordIndex})
    }
  }{
    2
  }
\right)
\\
\overset{(b)}&{\leq}
\eavesdropperAdvantageDistinguishing(\outputStateEve)
\gibbsEntropy{\hat \gibbsObservable}\left(
  \frac{
    2\gibbsEnergy
  }{
    \eavesdropperAdvantageDistinguishing(\outputStateEve)
  }
\right)
+
\binaryEntropy\left(
  \frac{
    \eavesdropperAdvantageDistinguishing(\outputStateEve)
  }{
    2
  }
\right),
\end{align*}
where (a) is an application of~\cite[Lemma 15]{winter2016tight}. For (b), we first observe that due to Lemma~\ref{lemma:bochner-integral-basics}-\ref{item:bochner-integral-basics-trace-norm-jensen} and Definition~\ref{def:alternative-security}-\ref{item:distinguishing-security}, we have
\[
\traceNorm{
  \Expectation_\messageRV \outputStateEve(\messageRV)
  -
  \outputStateEve(\tilde{\codewordIndex})
}
\leq
\Expectation_\messageRV
  \traceNorm{
    \outputStateEve(\messageRV)
    -
    \outputStateEve(\tilde{\codewordIndex})
  }
\leq
\eavesdropperAdvantageDistinguishing(\outputStateEve)
\leq
1.
\]
The inequality in the first summand then follows from~\cite[Corollary 12]{winter2016tight}, and in the second summand from the well-known fact that binary entropy is nondecreasing on $[0,1/2]$.
\end{proof}
It is not clear from Lemma~\ref{lemma:distinguishing-to-information} how $\eavesdropperAdvantageMutualInformation(\outputStateEve)$ behaves for $\blocklength \rightarrow \infty$ since both $\outputStateEve$ and the Gibbs observable $\hat \gibbsObservable$ depend on $\blocklength$. In the following Lemmas~\ref{lemma:gibbs-optimizer-tensor-product} and~\ref{lemma:distinguishing-to-information-tensorized}, we derive a bound for $ \eavesdropperAdvantageMutualInformation(\outputStateEve) $ that depends on $\blocklength$ only explicitly and via $\eavesdropperAdvantageDistinguishing(\outputStateEve)$. To this end, we need a way to lift a Gibbs observable $\gibbsObservable$ on a Hilbert space $\hilbertSpace$ to a Gibbs observable $\gibbsObservable^\blocklength$ on $\hilbertSpace^{\otimes \blocklength}$, and a relation between $\gibbsEntropy{\gibbsObservable}$ and $\gibbsEntropy{\gibbsObservable^\blocklength}$. Let $(\gibbsEigenvalue_\generalIndex)_{\generalIndex \in \naturals}$ be the eigenvalues and $\ket{\gibbsEigenvector_\generalIndex}_{\generalIndex \in \naturals}$ be the eigenvectors associated with $\gibbsObservable$, i.e., $\gibbsObservable$ can be written as in (\ref{eq:gibbs-eigenvalues}) and (\ref{eq:gibbs-domain}). Then $(\ket{\gibbsEigenvector_{\generalIndex_{1}}} \otimes \dots \otimes \ket{\gibbsEigenvector_{\generalIndex_{\blocklength}}})_{\generalIndex_{1}, \ldots ,\generalIndex_{\blocklength} \in \naturals}  $ is an orthonormal basis of $ \hilbertSpace^{\otimes \blocklength}  $
and the operator 
\begin{equation}\label{eq:n-block-gibbs}
  \gibbsObservable^\blocklength
  :=
  \gibbsObservable \otimes \identityOperator \otimes \dots \otimes \identityOperator
  +
  \dots
  +
  \identityOperator \otimes \dots \otimes \identityOperator \otimes \gibbsObservable
\end{equation}
is well-defined on the set of (finite) linear combinations of $(\ket{\gibbsEigenvector_{\generalIndex_{1}}} \otimes \dots \otimes \ket{\gibbsEigenvector_{\generalIndex_{\blocklength}}})_{\generalIndex_{1}, \ldots ,\generalIndex_{\blocklength} \in \naturals}  $. Moreover, for any $ \generalIndex_{1}, \ldots ,\generalIndex_{\blocklength} \in \naturals$, we have
\begin{equation*}
 \gibbsObservable^{\blocklength}  \ket{\gibbsEigenvector_{\generalIndex_{1}}} \otimes \dots \otimes \ket{\gibbsEigenvector_{\generalIndex_{\blocklength}}}
 = (\gibbsEigenvalue_{\generalIndex_{1}}+\dots + \gibbsEigenvalue_{\generalIndex_{\blocklength}})\ket{\gibbsEigenvector_{\generalIndex_{1}}} \otimes \dots \otimes \ket{\gibbsEigenvector_{\generalIndex_{\blocklength}}}.
\end{equation*}
Then on the dense domain
\begin{equation}\label{eq:n-block-gibbs-domain}
\mathcal{D}(\gibbsObservable^{\blocklength}):= \left\{
          \hilbertSpaceElement \in \hilbertSpace^{\otimes \blocklength}
          :
          \sum_{\generalIndex_1,\ldots , \generalIndex_{\blocklength} \in \naturals}
            \absolute{ \gibbsEigenvalue_{\generalIndex_{1}}+\dots + \gibbsEigenvalue_{\generalIndex_{\blocklength}}      }^2
            \absolute{\innerProduct{  {\gibbsEigenvector_{\generalIndex_{1}}} \otimes \dots \otimes {\gibbsEigenvector_{\generalIndex_{\blocklength}}}      }{\hilbertSpaceElement}}^2
          <
          \infty
        \right\},
\end{equation}
we can write
\begin{equation}\label{eq:n-block-gibbs-alt}
   \gibbsObservable^{\blocklength} \ket{\hilbertSpaceElement}= \sum_{\generalIndex_1,\ldots , \generalIndex_{\blocklength} \in \naturals}(\gibbsEigenvalue_{\generalIndex_{1}}+\dots + \gibbsEigenvalue_{\generalIndex_{\blocklength}})\innerProduct{  {\gibbsEigenvector_{\generalIndex_{1}}} \otimes \dots \otimes {\gibbsEigenvector_{\generalIndex_{\blocklength}}}      }{\hilbertSpaceElement}
   \ket{\gibbsEigenvector_{\generalIndex_{1}}} \otimes \dots \otimes \ket{\gibbsEigenvector_{\generalIndex_{\blocklength}}}.
\end{equation}
It is easily shown using \cite[Theorem VIII.3 (c)]{reed1980functional} that (\ref{eq:n-block-gibbs-domain}) and (\ref{eq:n-block-gibbs-alt}) define a self-adjoint operator on $\hilbertSpace^{\otimes \blocklength}$.
Additionally, for any $\gibbsParameter\in (0,\infty)$, we have 
\begin{align*}
 \trace\big(\exp(-\gibbsParameter \gibbsObservable^\blocklength)\big)
 &=
 \sum_{\generalIndex_1,\ldots , \generalIndex_{\blocklength} \in \naturals} \innerProduct{ \gibbsEigenvector_{\generalIndex_{1}} \otimes \dots \otimes \gibbsEigenvector_{\generalIndex_{\blocklength}} }{\exp(-\gibbsParameter \gibbsObservable^{\blocklength}) \gibbsEigenvector_{\generalIndex_{1}} \otimes \dots \otimes \gibbsEigenvector_{\generalIndex_{\blocklength}}    } 
 \\
 &=
 \sum_{\generalIndex_1,\ldots , \generalIndex_{\blocklength} \in \naturals}\exp(-\gibbsParameter( \gibbsEigenvalue_{\generalIndex_{1}}+\dots + \gibbsEigenvalue_{\generalIndex_{\blocklength}}   )) 
 \displaybreak[0] \\
\overset{(a)}&{=}
\left( \sum_{\generalIndex \in \naturals}\exp(-\gibbsParameter \gibbsEigenvalue_{\generalIndex})\right)^{\blocklength}
\\
&<
\infty,
\end{align*}
where (a) is due to the theorem of Fubini-Tonelli applied to the $\blocklength$-fold product of the counting measure on $\naturals$.
Therefore,
\begin{equation}\label{eq:n-block-gibbs-regularity}
\exp(-\gibbsParameter\gibbsObservable^{\blocklength}) \in \traceClass{\hilbertSpace^{\otimes \blocklength}},    
\end{equation}
for all $\gibbsParameter\in (0,\infty) $. Consequently, (\ref{eq:n-block-gibbs-domain}), (\ref{eq:n-block-gibbs-alt}), and (\ref{eq:n-block-gibbs-regularity}) show that $\gibbsObservable^{\blocklength}$ is a Gibbs observable on $\hilbertSpace^{\otimes\blocklength}$.

\begin{lemma}
\label{lemma:gibbs-optimizer-tensor-product}
Let $\gibbsObservable$ be a Gibbs observable on $\hilbertSpace$, and let $\gibbsObservable^{\blocklength}$ be the corresponding Gibbs observable on $\hilbertSpace^{\otimes \blocklength}$ given by (\ref{eq:n-block-gibbs}) -- (\ref{eq:n-block-gibbs-alt}). Let $\gibbsEnergy \in [0,\infty)$. Then
\[
  \gibbsEntropy{\gibbsObservable^\blocklength}(\gibbsEnergy)
  =
  \blocklength \gibbsEntropy{\gibbsObservable}\left(
    \frac{\gibbsEnergy}{\blocklength}
  \right).
\]
\end{lemma}
\begin{proof}
Let $\generalQuantumState \in \densityOperators{\hilbertSpace^{\otimes\blocklength}}$ be the quantum state which attains the supremum in the definition of $\gibbsEntropy{\gibbsObservable^\blocklength}(\gibbsEnergy)$ (cf.~\cite[Proposition 1]{shirokov2006entropy}) analog to (\ref{eq:gibbs-max-entropy-function}), i.e., $\trace(\generalQuantumState\gibbsObservable^\blocklength) \leq \gibbsEnergy$ and $\vonNeumannEntropy{\generalQuantumState} = \gibbsEntropy{\gibbsObservable^\blocklength}(\gibbsEnergy)$. Further, use $\generalQuantumState_1, \dots, \generalQuantumState_\blocklength$ to denote the marginal states of $\generalQuantumState$ on the factors of $\hilbertSpace^{\otimes\blocklength}$. We have
\begin{equation}
\label{eq:gibbs-extension-sum}
  \trace(\generalQuantumState\gibbsObservable^\blocklength)
  =
  \trace(\generalQuantumState (\gibbsObservable \otimes \identityOperator \otimes \dots \otimes \identityOperator))
  +
  \dots
  +
  \trace(\generalQuantumState (\identityOperator \otimes \dots \otimes \identityOperator \otimes \gibbsObservable))
  =
  \trace(\generalQuantumState_1 \gibbsObservable)
  +
  \dots
  +
  \trace(\generalQuantumState_\blocklength \gibbsObservable).
\end{equation}
By sub-additivity of von Neumann entropy, we have 
\begin{equation}\label{eq:ineq-1}
  \vonNeumannEntropy{\generalQuantumState}
  \leq
  \vonNeumannEntropy{\generalQuantumState_1}
  +
  \dots
  +
  \vonNeumannEntropy{\generalQuantumState_\blocklength}
  \le
  \gibbsEntropy{\gibbsObservable}(\gibbsEnergy_1)
  +
  \dots
  +
  \gibbsEntropy{\gibbsObservable}(\gibbsEnergy_\blocklength),
  \end{equation}
where we defined $\gibbsEnergy_\blockindex := \trace(\generalQuantumState_\blockindex \gibbsObservable)$ for $\blockindex \in \{1, \dots, \blocklength\}$. On the other hand, for states $\generalQuantumStateTwo_\blockindex  $, $\blockindex \in \{1, \dots, \blocklength\}$, with
\[
\vonNeumannEntropy{\generalQuantumStateTwo_\blockindex}= \gibbsEntropy{\gibbsObservable}(\gibbsEnergy_\blockindex)
\]
we define
\[
\tilde{\generalQuantumStateTwo}:=\generalQuantumStateTwo_1 \otimes \dots \otimes \generalQuantumStateTwo_\blocklength.
\]
Then
\[
\trace(\tilde{\generalQuantumStateTwo}\gibbsObservable^{\blocklength})\le \gibbsEnergy
\]
by a calculation similar to (\ref{eq:gibbs-extension-sum}), and by additivity of von Neumann entropy on product states we obtain
\begin{equation}\label{eq:ineq-2}
 \gibbsEntropy{\gibbsObservable}(\gibbsEnergy_1)
  +
  \dots
  +
  \gibbsEntropy{\gibbsObservable}(\gibbsEnergy_\blocklength)
  =
  \vonNeumannEntropy{\generalQuantumStateTwo_1}
  +
  \dots
  +
  \vonNeumannEntropy{\generalQuantumStateTwo_\blocklength}
  =
\vonNeumannEntropy{\tilde{\generalQuantumStateTwo}}\le \gibbsEntropy{\gibbsObservable^{\blocklength}}(\gibbsEnergy)= \vonNeumannEntropy{\generalQuantumState}.
\end{equation}
From (\ref{eq:ineq-1}) and (\ref{eq:ineq-2}) we arrive at
\begin{equation}
\label{eq:gibbs-energy-intermediate-equality}
  \gibbsEntropy{\gibbsObservable^\blocklength}(\gibbsEnergy)
  =
  \vonNeumannEntropy{\generalQuantumState}
  =
  \vonNeumannEntropy{\generalQuantumState_1}
  +
  \dots
  +
  \vonNeumannEntropy{\generalQuantumState_\blocklength}
  =
  \gibbsEntropy{\gibbsObservable}(\gibbsEnergy_1)
  +
  \dots
  +
  \gibbsEntropy{\gibbsObservable}(\gibbsEnergy_\blocklength).
\end{equation}
From (\ref{eq:gibbs-extension-sum}), it is clear that $\gibbsEnergy \geq \gibbsEnergy_1 + \dots + \gibbsEnergy_\blocklength$. Let $\gibbsEnergy_1', \dots, \gibbsEnergy_\blocklength' \in [0,\infty)$ be such that $\gibbsEnergy_1' + \dots + \gibbsEnergy_\blocklength' \leq \gibbsEnergy$, but otherwise arbitrary. By choosing states $\generalQuantumState_1', \dots, \generalQuantumState_\blocklength'$ with $\trace(\generalQuantumState_\blockindex'\gibbsObservable) \leq \gibbsEnergy_\blockindex'$ and $\vonNeumannEntropy{\generalQuantumState_\blockindex'} = \gibbsEntropy{\gibbsObservable}(\gibbsEnergy_\blockindex')$, we can argue
\[
  \gibbsEntropy{\gibbsObservable}(\gibbsEnergy_1')
  +
  \dots
  +
  \gibbsEntropy{\gibbsObservable}(\gibbsEnergy_\blocklength')
  =
  \vonNeumannEntropy{\generalQuantumState_1'}
  +
  \dots
  +
  \vonNeumannEntropy{\generalQuantumState_\blocklength'}
  =
  \vonNeumannEntropy{\generalQuantumState_1' \otimes \dots \otimes \generalQuantumState_\blocklength'}
  \overset{(a)}{\leq}
  \gibbsEntropy{\gibbsObservable^\blocklength}(\gibbsEnergy_1' + \dots + \gibbsEnergy_\blocklength')
  \overset{(b)}{\leq}
  \gibbsEntropy{\gibbsObservable^\blocklength}(\gibbsEnergy),
\]
where (a) follows from a calculation similar to (\ref{eq:gibbs-extension-sum}) and the definition (\ref{eq:gibbs-max-entropy-function}), and (b) is by monotonicity of $\gibbsEntropy{\gibbsObservable_\blocklength}$. By (\ref{eq:gibbs-energy-intermediate-equality}), these inequalities hold with equality in case $\gibbsEnergy_1 = \gibbsEnergy_1', \dots, \gibbsEnergy_\blocklength = \gibbsEnergy_\blocklength'$, so we have
\[
\gibbsEntropy{\gibbsObservable}(\gibbsEnergy_1) + \dots + \gibbsEntropy{\gibbsObservable}(\gibbsEnergy_\blocklength)
=
\max\Big\{
  \gibbsEntropy{\gibbsObservable}(\gibbsEnergy_1') + \dots + \gibbsEntropy{\gibbsObservable}(\gibbsEnergy_\blocklength')
  ~:~
  \gibbsEnergy_1', \dots, \gibbsEnergy_\blocklength' \in [0,\infty)
  ,~~
  \gibbsEnergy_1' + \dots + \gibbsEnergy_\blocklength' \leq \gibbsEnergy
\Big\}.
\]
Moreover, we have for any such choice of $\gibbsEnergy_1', \dots, \gibbsEnergy_\blocklength'$
\begin{equation}
\label{eq:gibbs-concavity-consequence}
  \blocklength
  \gibbsEntropy{\gibbsObservable}\left(
    \frac{\gibbsEnergy}{\blocklength}
  \right)
  \overset{(a)}{\geq}
  \blocklength
  \gibbsEntropy{\gibbsObservable}\left(
    \sum_{\blockindex=1}^\blocklength
      \frac{1}{\blocklength}
      \gibbsEnergy_\blockindex'
  \right)
  \overset{(b)}{\geq}
  \blocklength
  \sum_{\blockindex=1}^\blocklength
    \frac{1}{\blocklength}
    \gibbsEntropy{\gibbsObservable}(\gibbsEnergy_\blockindex')
  =
  \gibbsEntropy{\gibbsObservable}(\gibbsEnergy_1')
  +
  \dots
  +
  \gibbsEntropy{\gibbsObservable}(\gibbsEnergy_\blocklength'),
\end{equation}
where (a) is again by monotonicity and (b) follows from the fact that $\gibbsEntropy{\gibbsObservable}$ is a concave function~\cite[Proposition 1-iii)]{shirokov2006entropy}. Due to the strict monotonicity (it is shown in~\cite[Proposition 1-ii)]{shirokov2006entropy} that the derivative is strictly positive) and strict concavity of $\gibbsEntropy{\gibbsObservable}$~\cite[Proposition 1-iii)]{shirokov2006entropy}, these inequalities clearly both hold with equality iff $\gibbsEnergy_1' = \dots = \gibbsEnergy_\blocklength' = \gibbsEnergy/\blocklength$, and since we know that $\gibbsEnergy_1, \dots, \gibbsEnergy_\blocklength$ maximize the right-hand side of (\ref{eq:gibbs-concavity-consequence}), we get $\gibbsEnergy_1 = \dots = \gibbsEnergy_\blocklength = \gibbsEnergy/\blocklength$. Substituting this in (\ref{eq:gibbs-energy-intermediate-equality}) proves the lemma.
\end{proof}
\begin{lemma}
\label{lemma:distinguishing-to-information-tensorized}
Let $\outputStateEve$ be as defined in (\ref{eq:eve-q-output-state}). Assume that for all $\inputAlphabetElement \in \inputAlphabet$, we have $\trace(\eveChannel(\inputAlphabetElement) \gibbsObservable) \leq \gibbsEnergy$, where $\gibbsObservable$ is a Gibbs observable on $\eveHilbertSpace$. Suppose further that $\eavesdropperAdvantageDistinguishing(\outputStateEve) \leq 1$. Then
\[
\eavesdropperAdvantageMutualInformation(\outputStateEve)
\leq
\blocklength
\eavesdropperAdvantageDistinguishing(\outputStateEve)
\gibbsEntropy{\gibbsObservable}\left(
  \frac{2\gibbsEnergy}
       {\eavesdropperAdvantageDistinguishing(\outputStateEve)}
\right)
+
\binaryEntropy\left(
  \frac{\eavesdropperAdvantageDistinguishing(\outputStateEve)}{2}
\right).
\]
\end{lemma}
\begin{proof}
Let $ \gibbsObservable^\blocklength$ be given in (\ref{eq:n-block-gibbs}) -- (\ref{eq:n-block-gibbs-alt}).  We have, for every $\inputAlphabetElement^\blocklength \in \inputAlphabet^\blocklength$,
\begin{equation}
\label{eq:gibbs-product-energy}
\trace(
  \eveChannel^\blocklength(\inputAlphabetElement^\blocklength)
  \gibbsObservable^\blocklength
)
=
\trace(\eveChannel(\inputAlphabetElement_1) \gibbsObservable)
+
\dots
+
\trace(\eveChannel(\inputAlphabetElement_\blocklength) \gibbsObservable)
\leq
\blocklength\gibbsEnergy.
\end{equation}
Note that for each $\generalIndex \in \{1,\ldots ,\blocklength \}$, the function $ \trace(\eveChannel(\inputAlphabetElement_\generalIndex) \gibbsObservable) $ is measurable due to representation as a series of measurable functions given in (\ref{eq:gibbs-trace}). Therefore, $ \trace(
  \eveChannel^\blocklength(\inputAlphabetElement^\blocklength)
  \gibbsObservable^\blocklength
) $ is a measurable and bounded function.
Thus, using the representation (\ref{eq:n-block-gibbs-alt}) of $\gibbsObservable^{\blocklength}$, we obtain, for every $\codewordIndex$,
\begin{align*}
  \trace(\outputStateEve(\codewordIndex)\gibbsObservable^\blocklength)
  &=
  \trace\left(
    \int \eveChannel^\blocklength (\inputAlphabetElement^\blocklength)\encoder (d \inputAlphabetElement^\blocklength, \codewordIndex)
    \gibbsObservable^\blocklength
  \right)
  \\
  \overset{(\ref{eq:gibbs-trace})}&{=}
  \sum_{\generalIndex_1,\ldots , \generalIndex_{\blocklength} \in \naturals}
  (\gibbsEigenvalue_{\generalIndex_{1}}+\dots + \gibbsEigenvalue_{\generalIndex_{\blocklength}})
    \innerProduct{\int \eveChannel^\blocklength (\inputAlphabetElement^\blocklength)\encoder (d \inputAlphabetElement^\blocklength, \codewordIndex)\gibbsEigenvector_{\generalIndex_{1}}\otimes \dots \otimes  \gibbsEigenvector_{\generalIndex_{\blocklength}} }{ \gibbsEigenvector_{\generalIndex_{1}}\otimes \dots \otimes  \gibbsEigenvector_{\generalIndex_{\blocklength}} }
  \\
  \overset{(a)}&{=}
 \sum_{\generalIndex_1,\ldots , \generalIndex_{\blocklength} \in \naturals}
  (\gibbsEigenvalue_{\generalIndex_{1}}+\dots + \gibbsEigenvalue_{\generalIndex_{\blocklength}}) 
    \int
    \innerProduct{\eveChannel^\blocklength (\inputAlphabetElement^\blocklength) \gibbsEigenvector_{\generalIndex_{1}}\otimes \dots \otimes  \gibbsEigenvector_{\generalIndex_{\blocklength}} }{ \gibbsEigenvector_{\generalIndex_{1}}\otimes \dots \otimes  \gibbsEigenvector_{\generalIndex_{\blocklength}}  }
    \encoder (d \inputAlphabetElement^\blocklength, \codewordIndex)
  \\
  \overset{(b)}&{=}
  \int
   \sum_{\generalIndex_1,\ldots , \generalIndex_{\blocklength} \in \naturals}
  (\gibbsEigenvalue_{\generalIndex_{1}}+\dots + \gibbsEigenvalue_{\generalIndex_{\blocklength}})
    \innerProduct{\eveChannel^\blocklength (\inputAlphabetElement^\blocklength)  \gibbsEigenvector_{\generalIndex_{1}}\otimes \dots \otimes  \gibbsEigenvector_{\generalIndex_{\blocklength}} }{ \gibbsEigenvector_{\generalIndex_{1}}\otimes \dots \otimes  \gibbsEigenvector_{\generalIndex_{\blocklength}} }
  \encoder (d \inputAlphabetElement^\blocklength, \codewordIndex)
  \\
  \overset{(\ref{eq:gibbs-trace})}&{=}
    \int \trace(\eveChannel^\blocklength (\inputAlphabetElement^\blocklength)\gibbsObservable^\blocklength)\encoder (d \inputAlphabetElement^\blocklength, \codewordIndex)
  \\
  \overset{(\ref{eq:gibbs-product-energy})}&{\leq}
  \blocklength \gibbsEnergy,
\end{align*}
where (a) is by Lemma~\ref{lemma:bochner-integral-basics}-\ref{item:bochner-integral-basics-bounded-functional} and (b) is by the monotone convergence theorem. Therefore, we can apply Lemmas~\ref{lemma:distinguishing-to-information} and~\ref{lemma:gibbs-optimizer-tensor-product} and obtain
\begin{align*}
\eavesdropperAdvantageMutualInformation(\outputStateEve)
&\leq
\eavesdropperAdvantageDistinguishing(\outputStateEve)
\gibbsEntropy{\gibbsObservable^\blocklength}\left(
  \frac{2\blocklength\gibbsEnergy}
       {\eavesdropperAdvantageDistinguishing(\outputStateEve)}
\right)
+
\binaryEntropy\left(
  \frac{\eavesdropperAdvantageDistinguishing(\outputStateEve)}{2}
\right)
\\
&=
\blocklength
\eavesdropperAdvantageDistinguishing(\outputStateEve)
\gibbsEntropy{\gibbsObservable}\left(
  \frac{2\gibbsEnergy}
       {\eavesdropperAdvantageDistinguishing(\outputStateEve)}
\right)
+
\binaryEntropy\left(
  \frac{\eavesdropperAdvantageDistinguishing(\outputStateEve)}{2}
\right),
\end{align*}
concluding the proof.
\end{proof}
It is known from~\cite[Proposition 1-ii)]{shirokov2006entropy} that
$
\eavesdropperAdvantageDistinguishing(\outputStateEve)
\gibbsEntropy{\gibbsObservable}\left(
  2\gibbsEnergy/\eavesdropperAdvantageDistinguishing(\outputStateEve)
\right)
$
tends to $0$ as $\eavesdropperAdvantageDistinguishing(\outputStateEve)$ tends to $0$, and it is also clear that the binary entropy term vanishes in this case. However, in general, the presence of the extra factor $\blocklength$ in the upper bound of Lemma~\ref{lemma:distinguishing-to-information-tensorized} means that an additional argument is necessary to show that distinguishing security implies mutual information security. We do not have such an argument for the general case, but we show in the following lemma that the implication holds in important special cases. The first of these cases is that the eavesdropper's output is finite-dimensional. The other cases are infinite dimensional and assume the existence of Gibbs observables for the eavesdropper's output that are closely related to the Hamiltonian of harmonic oscillator of bounded energy. Therefore, they specialize to bosonic systems with either one or finitely many modes. In this case, the Hilbert space under consideration is $\eveHilbertSpace=L^2(\reals)$ and the Gibbs observable is given by
\begin{equation}\label{eq:def-number-operator}
  \gibbsObservable:=\adjoint{a}a,  
\end{equation}
defined, for the moment, on the domain
\begin{equation*}
    \mathcal{D}(\gibbsObservable)=\mathcal{S}(\reals)
\end{equation*}
where $\adjoint{a}$ and $a$ denote the creation resp. annihilation operators (cf. \cite[Chapter 12.1]{holevo2019quantum}) defined on $\mathcal{S}(\reals)$ and where $\mathcal{S}(\reals)$ is the Schwartz space of rapidly decreasing functions which is dense in $L^2(\reals)$. It is well known that $\gibbsObservable:=\adjoint{a} a  $ is essentially self-adjoint on $\mathcal{D}(\gibbsObservable)$ (cf. \cite[Section 2.2.7, Example 4]{moretti2019fundamental} from which the claim follows by a slight modification of the argument). The operator $\gibbsObservable$ is, up to additive scalar multiple of $\identityOperator$, the Hamiltonian of the quantum harmonic oscillator \cite[Section 12.1.2]{holevo2019quantum}. Moreover, the operator $\gibbsObservable $ has discrete spectrum with spectral decomposition (cf. \cite[Section 2.2.7, Example 4]{moretti2019fundamental}, \cite[Chapter 12.1]{holevo2019quantum})
\begin{equation}\label{eq:spec-decomp-number-operator}
  \gibbsObservable=\sum_{\generalIndex=0}^{\infty} \generalIndex\ket{\gibbsEigenvector_{\generalIndex}}\bra{\gibbsEigenvector_{\generalIndex}},
\end{equation}
where $\ket{\gibbsEigenvector_{\generalIndex}}_{\generalIndex\in \naturals_{0}}$ is an orthonormal basis of $ \eveHilbertSpace=L^2(\reals) $ consisting of so-called number state vectors. The final domain of the operator $\gibbsObservable$ is then given by
\begin{equation*}
\bar{\mathcal{D}}(\gibbsObservable)=\left\{\ket{\hilbertSpaceElement}\in L^2(\reals): \sum_{\generalIndex=0}^{\infty} \generalIndex^2 \absolute{\innerProduct{\gibbsEigenvector_{\generalIndex}}{\hilbertSpaceElement}}^2<\infty \right\}.
\end{equation*}
The operator $\gibbsObservable$ is self-adjoint on $\bar{\mathcal{D}}(\gibbsObservable) $.

For the case of $\numModes$ modes, we similarly define the operator
\begin{equation}\label{eq:def-s-mode-number-operator}
  \gibbsObservable^{(\numModes)} := \sum_{\indexModes=1}^\numModes \identityOperator^{\otimes \indexModes-1} \otimes \frequency_\indexModes \gibbsObservable_\indexModes \otimes \identityOperator^{\otimes \numModes - \indexModes},
  \end{equation}
on $ \mathcal{D}(\gibbsObservable^{\numModes}):=\mathcal{S}(\reals)^{\otimes \numModes} \subseteq (L^2(\reals))^{\otimes \numModes}$, where $\frequency_1,\ldots , \frequency_\numModes \in (0,\infty)$ denote the frequencies of the modes and the operators $\gibbsObservable_1, \ldots , \gibbsObservable_\numModes$ are given in (\ref{eq:def-number-operator}) with the respective creation and annihilation operators of the modes. Using the spectral decomposition (\ref{eq:spec-decomp-number-operator}) for each of the modes, we obtain the spectral decomposition
\begin{equation*}
\gibbsObservable^{(\numModes)}
=
\sum_{\generalIndex_1, \dots, \generalIndex_\numModes=0}^\infty\left(
  \sum_{\indexModes=1}^\numModes
    \frequency_\indexModes
    \generalIndex_\indexModes
  \right)
    \rankOneOperator{
      \gibbsEigenvector_{1,\generalIndex_1}
      \otimes
      \dots
      \otimes
      \gibbsEigenvector_{\numModes,\generalIndex_\numModes}
    },
\end{equation*}
where
$(
  \gibbsEigenvector_{1,\generalIndex_1}
  \otimes
  \dots
  \otimes
  \gibbsEigenvector_{\numModes,\generalIndex_\numModes}
)$
is the orthonormal basis of $ (L^2(\reals))^{\otimes \numModes} $ composed of bases consisting of number state vectors of individual modes (\ref{eq:spec-decomp-number-operator}). Defining
\begin{equation*}
\bar{\mathcal{D}}\left(\gibbsObservable^{(\numModes)}\right):= \left\{ \ket{\hilbertSpaceElement}\in (L^2(\reals))^{\otimes \numModes} : \sum_{\generalIndex_1, \dots, \generalIndex_\numModes=0}^\infty \left(
  \sum_{\indexModes=1}^\numModes
    \frequency_\indexModes
    \generalIndex_\indexModes
  \right)^2 \absolute{ \innerProduct{ \gibbsEigenvector_{1,\generalIndex_1}
      \otimes
      \dots
      \otimes
      \gibbsEigenvector_{\numModes,\generalIndex_\numModes} }{\hilbertSpaceElement}  }^2 <\infty
\right\},
\end{equation*}
we note that the operator $\gibbsObservable^{(\numModes)} $ is self-adjoint on the dense domain $ \bar{\mathcal{D}}(\gibbsObservable^{(\numModes)}) $ \cite[Theorem VIII.3 (c)]{reed1980functional}.

In order to show that $\gibbsObservable$ and $\gibbsObservable^{(\numModes)} $ are Gibbs observables, it remains to verify that for all $\gibbsParameter\in (0,\infty)$, we have 
\begin{equation*}
  \exp(-\gibbsParameter\gibbsObservable) \in \traceClass{L^2(\reals)}\quad \textrm{ and } \quad \exp(-\gibbsParameter\gibbsObservable^{(\numModes)})\in \traceClass{(L^2(\reals))^{\otimes \numModes}},
\end{equation*}
which will be proven in the implications 3) and 4) of Lemma \ref{lemma:distinguishing-to-information-specialized}. An example of a \gls{cqq} wiretap channel for which $\gibbsObservable$ is a Gibbs observable is given in Section \ref{sec:gaussian-cq-basics}, and a general account and further examples of energy-constrained quantum channels can be found in \cite[Chapter 12]{holevo2019quantum}.
According to the following lemma, superlinear convergence of the distinguishing security level to $0$ is sufficient to guarantee mutual information security in the finite-dimensional case. In the case that the eavesdropper's channel has energy constraints described by Gibbs observables of the form (\ref{eq:def-number-operator}), (\ref{eq:def-s-mode-number-operator}), superquadratic convergence is a sufficient condition. Therefore, mutual information security in these cases follows in particular from the exponential bounds on distinguishing security level that we obtain in this work.
\begin{lemma}
\label{lemma:distinguishing-to-information-specialized}
$(\outputStateEve)_{\blocklength \in \naturals}$ is mutual information secure if any of the following assumption holds:
\begin{enumerate}
    \item \label{item:distinguishing-to-information-specialized-finite}$\dim \eveHilbertSpace = \hilbertSpaceDimension < \infty$ and $\blocklength \eavesdropperAdvantageDistinguishing(\outputStateEve) \rightarrow 0$ as $\blocklength \rightarrow \infty$.
    \item \label{item:distinguishing-to-information-specialized-gibbs} There exists a Gibbs observable $\gibbsObservable$ on $\eveHilbertSpace$ and $\multConstant_1,\gibbsEnergy \in [0,\infty)$ such that for all $\gibbsParameter \in (0,\infty)$, we have $\log \trace \exp(-\gibbsParameter\gibbsObservable) \leq \multConstant_1\gibbsParameter^{-1}$, and for all $\inputAlphabetElement \in \inputAlphabet$, we have $\trace(\eveChannel(\inputAlphabetElement)\gibbsObservable) \leq \gibbsEnergy$. Moreover, there exist $\multConstant_2 \in (0, \infty), \expConstant \in (2,\infty)$ with $\eavesdropperAdvantageDistinguishing(\outputStateEve) \leq \multConstant_2 \blocklength^{-\expConstant}$.
    \item \label{item:distinguishing-to-information-specialized-harmonic} There is a Gibbs observable $\gibbsObservable := \sum_{\generalIndex=0}^\infty \generalIndex\rankOneOperator{\gibbsEigenvector_\generalIndex}$ on $\eveHilbertSpace$, where $\ket{\gibbsEigenvector_\generalIndex}_{\generalIndex=0}^\infty$ is an orthonormal basis of $\eveHilbertSpace$. Moreover, there is $\gibbsEnergy \in (0,\infty)$ such that for all $\inputAlphabetElement \in \inputAlphabet$, we have $\trace(\eveChannel(\inputAlphabetElement)\gibbsObservable) \leq \gibbsEnergy$, and there are $\multConstant \in (0, \infty), \expConstant \in (2,\infty)$ with $\eavesdropperAdvantageDistinguishing(\outputStateEve) \leq \multConstant \blocklength^{-\expConstant}$.
    \item \label{item:distinguishing-to-information-specialized-harmonic-multimode} We have $\eveHilbertSpace = \hilbertSpace^{\otimes \numModes}$ and a Gibbs observable $\gibbsObservable := \sum_{\indexModes=1}^\numModes \identityOperator^{\otimes \indexModes-1} \otimes \frequency_\indexModes \gibbsObservable_\indexModes \otimes \identityOperator^{\otimes \numModes - \indexModes}$, where each $\gibbsObservable_\indexModes$ is a Gibbs observable on $\hilbertSpace$ of the same form as in \ref{item:distinguishing-to-information-specialized-harmonic}) and each $\frequency_\indexModes \in (0,\infty)$. We assume further that there is $\gibbsEnergy \in (0,\infty)$ such that for all $\inputAlphabetElement \in \inputAlphabet$, we have $\trace(\eveChannel(\inputAlphabetElement)\gibbsObservable) \leq \gibbsEnergy$, and there are $\multConstant \in (0, \infty), \expConstant \in (2,\infty)$ with $\eavesdropperAdvantageDistinguishing(\outputStateEve) \leq \multConstant \blocklength^{-\expConstant}$.
\end{enumerate}
\end{lemma}
\begin{proof}
We first prove the lemma for case \ref{item:distinguishing-to-information-specialized-finite}). In this case, the proof proceeds in parallel to Lemma~\ref{lemma:distinguishing-to-information}. We fix an arbitrary probability distribution $\messageProbabilityDistribution$ on $\{1, \dots, \codebookSize\}$ and define $\tilde{\codewordIndex}$ as in (\ref{eq:security-implication-minentropy}). With this, we obtain, for sufficiently large $\blocklength$,
\begin{align*}
\holevoInformation{\messageProbabilityDistribution}{\outputStateEve}
&=
\vonNeumannEntropy{\Expectation_\messageRV \outputStateEve(\codewordIndex)}
-
\Expectation_\messageRV
  \vonNeumannEntropy{\outputStateEve(\messageRV)}
\\
\overset{(\ref{eq:security-implication-minentropy})}&{\leq}
\vonNeumannEntropy{\Expectation_\messageRV \outputStateEve(\codewordIndex)}
-
\vonNeumannEntropy{\outputStateEve(\tilde{\codewordIndex})}
\displaybreak[0] \\
\overset{(a)}&{\leq}
\traceNorm{
  \Expectation_\messageRV \outputStateEve(\codewordIndex)
  -
  \outputStateEve(\tilde{\codewordIndex})
}
\blocklength \log \hilbertSpaceDimension
+
\binaryEntropy\left(
  \frac{
    \traceNorm{
      \Expectation_\messageRV \outputStateEve(\codewordIndex)
      -
      \outputStateEve(\tilde{\codewordIndex})
    }
  }{
    2
  }
\right)
\\
\overset{(b)}&{\leq}
\eavesdropperAdvantageDistinguishing(\outputStateEve)
\blocklength \log \hilbertSpaceDimension
+
\binaryEntropy\left(
  \frac{
    \eavesdropperAdvantageDistinguishing(\outputStateEve)
  }{
    2
  }
\right),
\end{align*}
where (a) is due to the universal bound of~\cite[Lemma 1]{winter2016tight} (note that $\dim \eveHilbertSpace^{\otimes \blocklength} = \hilbertSpaceDimension^\blocklength$) and (b) holds for $\blocklength$ large enough so that $\eavesdropperAdvantageDistinguishing(\outputStateEve) \leq 1$. Clearly, this bound vanishes as $\blocklength \eavesdropperAdvantageDistinguishing(\outputStateEve)$ vanishes.

Next, we prove the lemma for case \ref{item:distinguishing-to-information-specialized-gibbs}). We have
\begin{align*}
\blocklength
\eavesdropperAdvantageDistinguishing(\outputStateEve)
\gibbsEntropy{\gibbsObservable}\left(
  \frac{2\gibbsEnergy}
       {\eavesdropperAdvantageDistinguishing(\outputStateEve)}
\right)
\overset{(a)}&{=}
\blocklength
\eavesdropperAdvantageDistinguishing(\outputStateEve)
\inf_{\gibbsParameter \in (0,\infty)}\left(
  \gibbsParameter
  \frac{2\gibbsEnergy}{\eavesdropperAdvantageDistinguishing(\outputStateEve)}
  +
  \log \trace \exp(-\gibbsParameter\gibbsObservable)
\right)
\\
\overset{(b)}&{\leq}
\blocklength
\eavesdropperAdvantageDistinguishing(\outputStateEve)
\inf_{\gibbsParameter \in (0,\infty)}\left(
  \gibbsParameter
  \frac{2\gibbsEnergy}{\eavesdropperAdvantageDistinguishing(\outputStateEve)}
  +
  \frac{\multConstant_1}{\gibbsParameter}
\right)
\\
\overset{(c)}&{\leq}
2\gibbsEnergy
\blocklength^{1-\expConstant'}
+
\multConstant_1
\eavesdropperAdvantageDistinguishing(\outputStateEve)
\blocklength^{1+\expConstant'}
\\
\overset{(d)}&{\leq}
2\gibbsEnergy
\blocklength^{1-\expConstant'}
+
\multConstant_1 \multConstant_2
\blocklength^{1+\expConstant'-\expConstant},
\end{align*}
where (a) uses the variational representation of the function $\gibbsEntropy{\gibbsObservable}$~\cite[Proposition 1-3)]{shirokov2006entropy}, (b) follows from the assumption $\log \trace \exp(-\gibbsParameter\gibbsObservable) \leq \multConstant_1\gibbsParameter^{-1}$, (c) follows by upper bounding the infimum with the choices $\gibbsParameter:=\blocklength^{-\expConstant'}$ and $\expConstant' \in (1,\expConstant-1)$, and in (d) we have substituted the upper bound for $\eavesdropperAdvantageDistinguishing(\outputStateEve)$ from the lemma statement. The obtained upper bound vanishes, and thus mutual information security follows from Lemma~\ref{lemma:distinguishing-to-information-tensorized} and the fact that binary entropy vanishes at $0$.

The remaining two cases are special cases of \ref{item:distinguishing-to-information-specialized-gibbs}). For \ref{item:distinguishing-to-information-specialized-harmonic}), we note that with $\gibbsObservable$ as defined in this case, we have for $\gibbsParameter \in (0,\infty)$
\begin{align}
\nonumber
\log\trace\exp(-\gibbsParameter\gibbsObservable)
&=
\log
\sum_{\generalIndex = 0}^\infty
  \exp(-\gibbsParameter\generalIndex)
\\
\nonumber
\overset{(a)}&{=}
-
\log(1-\exp(-\gibbsParameter))
\\
\nonumber
\overset{(b)}&{\leq}
-
\left(
  1
  -
  \frac{1}{1-\exp(-\gibbsParameter)}
\right)
\\
\nonumber
&=
\frac{1}{\exp(\gibbsParameter)-1}
\\
\label{eq:distinguishing-to-information-specialized-harmonic-trace}
\overset{(c)}&{\leq}
\frac{1}{\gibbsParameter},
\end{align}
where (a) is the known convergence behavior of the geometric series, (b) is due to the inequality $\forall \generalReal \in (0,\infty)~\log(\generalReal) \geq 1-1/\generalReal$, and (c) is due to the inequality $\forall \generalReal \in \reals~ \exp(\generalReal) \geq \generalReal + 1$. Thus, the assumption of case \ref{item:distinguishing-to-information-specialized-gibbs}) is satisfied.

For case \ref{item:distinguishing-to-information-specialized-harmonic-multimode}), we fix suitable orthonormal bases $\ket{\gibbsEigenvector_{1,\generalIndex}}_{\generalIndex=0}^\infty, \dots, \ket{\gibbsEigenvector_{\numModes,\generalIndex}}_{\generalIndex=0}^\infty$ of $\hilbertSpace$ such that
\begin{align*}
\gibbsObservable
&=
\sum_{\indexModes=1}^\numModes
  \sum_{\generalIndex_1, \dots, \generalIndex_\numModes=0}^\infty
    \frequency_\indexModes
    \generalIndex_\indexModes
    \rankOneOperator{
      \gibbsEigenvector_{1,\generalIndex_1}
      \otimes
      \dots
      \otimes
      \gibbsEigenvector_{\numModes,\generalIndex_\numModes}
    }
\\
&=
\sum_{\generalIndex_1, \dots, \generalIndex_\numModes=0}^\infty\left(
  \sum_{\indexModes=1}^\numModes
    \frequency_\indexModes
    \generalIndex_\indexModes
  \right)
    \rankOneOperator{
      \gibbsEigenvector_{1,\generalIndex_1}
      \otimes
      \dots
      \otimes
      \gibbsEigenvector_{\numModes,\generalIndex_\numModes}
    }.
\end{align*}
It suffices to show $\log \trace \exp(-\gibbsParameter\gibbsObservable) = \gibbsParameter^{-1}(\frequency_1^{-1} + \dots + \frequency_\numModes^{-1})$, which we do by induction on $\numModes$. For $\numModes=1$, this reduces to (\ref{eq:distinguishing-to-information-specialized-harmonic-trace}). For $\numModes > 1$, we calculate
\begin{align*}
\log \trace \exp(-\gibbsParameter\gibbsObservable)
&=
\log
\sum_{\generalIndex_1, \dots, \generalIndex_\numModes=0}^\infty
  \exp\left(
    -\gibbsParameter
    \sum_{\indexModes=1}^\numModes
      \frequency_\indexModes
      \generalIndex_\indexModes
  \right)
\\
&=
\log
\sum_{\generalIndex_1, \dots, \generalIndex_\numModes=0}^\infty
    \exp\left(
      -\gibbsParameter
      \sum_{\indexModes=1}^{\numModes-1}
        \frequency_\indexModes
        \generalIndex_\indexModes
    \right)
    \exp(
      -\gibbsParameter
      \frequency_\numModes
      \generalIndex_\numModes
    )
\displaybreak[0] \\
&=
\log
\sum_{\generalIndex_1, \dots, \generalIndex_{\numModes-1}=0}^\infty
  \exp\left(
    -\gibbsParameter
    \sum_{\indexModes=1}^{\numModes-1}
      \frequency_\indexModes
      \generalIndex_\indexModes
  \right)
+
\log
  \sum_{\generalIndex_\numModes=0}^\infty
    \exp(
      -\gibbsParameter
      \frequency_\numModes
      \generalIndex_\numModes
    )
\displaybreak[0] \\
\overset{(\ref{eq:distinguishing-to-information-specialized-harmonic-trace})}&{\leq}
\frac{1}{\frequency_\numModes\gibbsParameter}
+
\log
\sum_{\generalIndex_1, \dots, \generalIndex_{\numModes-1}=0}^\infty
  \exp\left(
    -\gibbsParameter
    \sum_{\indexModes=1}^{\numModes-1}
      \frequency_\indexModes
      \generalIndex_\indexModes
  \right)
\\
\overset{(a)}&{\leq}
\gibbsParameter^{-1}(\frequency_1^{-1} + \dots + \frequency_\numModes^{-1}),
\end{align*}
where the induction hypothesis is used in step (a).
\end{proof}

Observing the trivial inequalities
\[
\eavesdropperAdvantageWeak(\outputStateEve)
\leq
\eavesdropperAdvantageStrong(\outputStateEve)
\leq
\eavesdropperAdvantageMutualInformation(\outputStateEve),
\]
the preceding lemmas imply the following equivalences and implications for the asymptotic notions about sequences of codes in parallel to the classical case~\cite[Figure 2]{bellare2012cryptographic}
\[
\text{semantically secure}
\Leftrightarrow
\text{distinction secure}
\Leftrightarrow
\text{mutual information secure}
\Rightarrow
\text{strongly secure}
\Rightarrow
\text{weakly secure}.
\]
It should be noted that Lemma~\ref{lemma:distinguishing-to-information-specialized} shows
\[
\text{distinction secure} \Rightarrow \text{mutual information secure}
\]
only under additional assumptions which capture, however, some practically important cases. But even in the general case, both mutual information secure and distinction secure sequences of codes are also semantically secure according to Definition~\ref{def:semantic-security}.

\section{\texorpdfstring{\gls{cq}}{cq} Channel Resolvability and Coding}
\label{sec:res-code}
\begin{figure}
    \centering
    \begin{subfigure}[t]{0.45\textwidth}
        \centering
        \includegraphics[width=.9\textwidth, trim=0 12 0 12 mm]{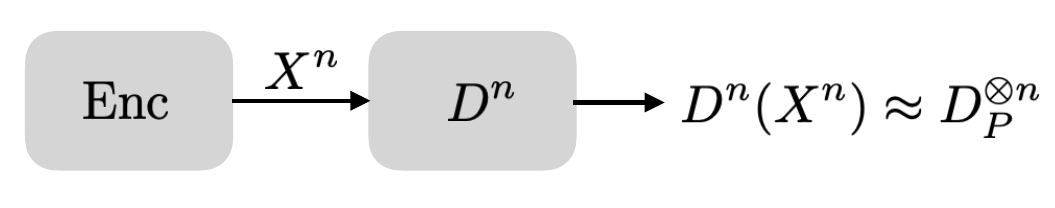}
        \caption{Generic \gls{cq} channel model for channel resolvability. }
        \label{fig:cq-resolv}
    \end{subfigure}
    \hfill
    \begin{subfigure}[t]{0.45\textwidth}
        \centering
        \includegraphics[width=\textwidth]{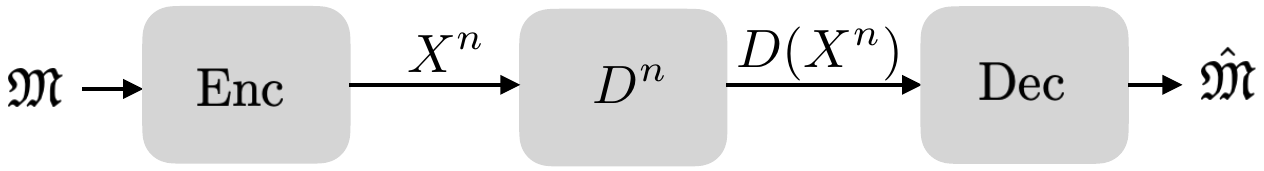}
        \caption{Generic \gls{cq} channel model for transmission of messages.}
        \label{fig:cq-coding}
    \end{subfigure}
    \caption{Point-to-point \gls{cq} channel models considered in Section~\ref{sec:res-code}.}
    \label{fig:resolvability}
\end{figure}

Important prerequisites in proving achievability of rates for the wiretap channel will be results of achievability for the channel resolvability problem, depicted in Fig.~\ref{fig:cq-resolv}, and the channel coding problem, depicted in Fig.~\ref{fig:cq-coding}.

Here, we just consider a point-to-point \gls{cq} channel described by a measurable map $\classicalQuantumChannel: \inputAlphabet \rightarrow \densityOperators{\hilbertSpace}$ with some separable Hilbert space $\hilbertSpace$. We focus on two different problems: In \emph{channel resolvability}, the question will be which transmit strategies (i.e., rules for generating $\inputRV^\blocklength$) the transmitter can employ to approximate the output state $\classicalQuantumChannel_\inputDistribution^{\otimes \blocklength}$ by $\classicalQuantumChannel^\blocklength(\inputRV^\blocklength)$ at the receiver. In \emph{channel coding}, the task is to encode a given message for transmission through the channel in such a way that the receiver can with high probability decode the message. In the proof of Theorems~\ref{theorem:wiretap-ccq} and~\ref{theorem:wiretap-cq}, the security criterion can then be shown by applying the resolvability result to the channel $\classicalQuantumChannel = \eveChannel$, and for the case of a quantum output at the legitimate receiver, the average decoding error criterion can be shown by applying the channel coding result to the channel $\classicalQuantumChannel = \bobcqChannel$ (for the case of classical output at $\bob$, we cite a classical channel coding result which is used instead).

In the solution to both of these problems, we use the same standard random codebook construction: The random codebook $\codebook$ of block length $\blocklength$ and size $\codebookSize$ is of the form $\codebook := (\codebook(\codewordIndex))_{\codewordIndex=1}^\codebookSize$ where the codewords $\codebook(\codewordIndex)$ are vectors with entries from $\inputAlphabet$ and length $\blocklength$ where all entries across all codewords are i.i.d. and follow the distribution $\inputDistribution$.

Both theorems stated in this section use the technical assumption that there is some $\renyiorder_{\min} \in (0,1)$ such that:
\begin{equation}
\label{eq:technical-assumptions}
\begin{aligned}
&\classicalQuantumChannel_\inputDistribution^{\renyiorder_{\min}} \in \traceClass{\hilbertSpace},
\\
&\text{for $\inputDistribution$-almost all }
\inputAlphabetElement:~
\classicalQuantumChannel(\inputAlphabetElement)^{\renyiorder_{\min}} \in \traceClass{\hilbertSpace},
\\
&\text{the Bochner integral }
\Expectation (\classicalQuantumChannel(\inputRV)^{\renyiorder_{\min}}) \in \traceClass{\hilbertSpace}
\text{ exists}.
\end{aligned}
\end{equation}
It is an immediate consequence of Lemmas~\ref{lemma:renyi-entropy-basics-preliminary} and~\ref{lemma:renyi-entropy-basics} that $\holevoInformation{\inputDistribution}{\classicalQuantumChannel} < \infty$ whenever (\ref{eq:technical-assumptions}) holds, and we will implicitly use this fact in the following.

For channel resolvability, each fixed realization of $\codebook$ defines an induced channel output density
\begin{equation}
\label{eq:codebook-channel-output-definition}
\classicalQuantumChannel_\codebook
:=
\frac{1}{\codebookSize}
\sum\limits_{\codewordIndex=1}^\codebookSize
  \classicalQuantumChannel^\blocklength(\codebook(\codewordIndex))
\end{equation}
which results if the transmitter chooses a codeword for transmission through the channel uniformly at random. With these definitions, we now have the necessary terminology to state our resolvability  and coding results. The proofs are deferred to Sections~\ref{sec:resolvability-proof} and~\ref{sec:coding-proof}, and in Sections~\ref{sec:proof-concentration} and ~\ref{sec:proofs-cost-constraint}, we analyze the concentration behavior of the errors and extend both results to the case of cost-constrained channel inputs.

\begin{theorem}
\label{theorem:q-resolvability}
Let $\codebookRate > \quantumInformation$, and suppose $\codebookSize \geq \exp(\blocklength\codebookRate)$. Moreover, assume that (\ref{eq:technical-assumptions}) holds. Then, for all $\blocklength \in \naturals$, we have
\begin{equation}
\label{eq:q-resolvability-all-blocklengths}
\Expectation_\codebook
  \traceNorm{
    \classicalQuantumChannel_\codebook
    -
    \classicalQuantumChannel_\inputDistribution^{\otimes \blocklength}
  }
\leq
\resolvabilityBoundFunc(\codebookRate,\blocklength),
\end{equation}
with $\resolvabilityBoundFunc$ defined in \eqref{eq:resolvabilityBoundFunc}. Furthermore, $\resolvabilityBoundFunc(\codebookRate,\blocklength)$ tends to $0$ exponentially as $\blocklength$ tends to $\infty$. That is, for sufficiently large $\blocklength \in \naturals$, we have $\finalconst \in (0,\infty)$ with
\begin{equation}
\label{eq:q-resolvability-large-blocklengths}
\resolvabilityBoundFunc(\codebookRate,\blocklength)
\leq
\exp(-\blocklength\finalconst).
\end{equation}
\end{theorem}

For channel coding, given a fixed codebook $\codebook$, the transmitter chooses the codeword $\codebook(\messageRV)$ to encode a message $\messageRV$.

\begin{theorem}
\label{theorem:decoding}
Let $\codebookRate < \quantumInformation$, and suppose $\codebookSize \leq \exp(\blocklength\codebookRate)$. Moreover, assume that (\ref{eq:technical-assumptions}) holds. Then, for each $\codebook$, there is a decoding \gls{povm} $(\decoderpovm_\codewordIndex)_{\codewordIndex=1}^\codebookSize$ such that every $\decoderpovm_\codewordIndex$ is measurable as a function of $\codebook$, and for all $\blocklength \in \naturals$, we have
\begin{equation}
\label{eq:decoding-error-all-blocklengths}
\Expectation_\codebook\left(
  \frac{1}{\codebookSize}
  \sum_{\codewordIndex=1}^\codebookSize
    \trace\left(
      \classicalQuantumChannel(\codebook(\codewordIndex))
      (\identityOperator - \decoderpovm_\codewordIndex)
    \right)
\right)
\leq
\codingBoundFunc(\codebookRate,\blocklength),
\end{equation}
with $\codingBoundFunc$ defined in \eqref{eq:codingBoundFunc}. Furthermore, $\codingBoundFunc(\codebookRate,\blocklength)$ tends to $0$ exponentially. That is, there is $\finalconst \in (0,\infty)$ such that for sufficiently large $\blocklength \in \naturals$, we have
\begin{equation}
\label{eq:decoding-error-large-blocklengths}
\codingBoundFunc(\codebookRate,\blocklength)
\leq
\exp(-\blocklength\finalconst).
\end{equation}
\end{theorem}

\section{Proofs}
\label{sec:proofs}
In this section, we prove the theorems which are stated in the preceding sections. In Section~\ref{sec:typicality-preliminaries}, we introduce notions of typicality and state two lemmas around typicality that will be needed both for the resolvability and coding results. We then proceed to proving Theorem~\ref{theorem:q-resolvability} on resolvability in Section~\ref{sec:resolvability-proof} and Theorem~\ref{theorem:decoding} on coding in Section~\ref{sec:coding-proof}. In Section~\ref{sec:proof-concentration}, we show that the error terms are very tightly concentrated around their expectations. Extensions that incorporate an additive input cost constraint can be found in Section~\ref{sec:proofs-cost-constraint}. Finally, everything is put together in Section~\ref{sec:proofs-main-result} where Theorems~\ref{theorem:wiretap-ccq-all-blocklengths}, \ref{theorem:wiretap-ccq}, \ref{theorem:wiretap-cq-all-blocklengths}, and \ref{theorem:wiretap-cq} on coding for the wiretap channel are proved. The proofs of technical lemmas are relegated to Appendix~\ref{appendix:qit}.

\subsection{Prerequisites on Typicality}
\label{sec:typicality-preliminaries}
In this section, we introduce typicality notions and related technical lemmas that are used in the proofs of Theorems~\ref{theorem:q-resolvability} and~\ref{theorem:decoding}.

For $\blocklength$-fold uses of the channel $\classicalQuantumChannel$, we note that via the definitions
\begin{align*}
\eigenvector_{\classicalOutputIndex^\blocklength | \inputAlphabetElement^\blocklength}
&:=
\bigotimes\limits_{\blockindex=1}^\blocklength
  \eigenvector_{\classicalOutputIndex_\blockindex | \inputAlphabetElement_\blockindex}
\\
\eigenvector_{\classicalOutputIndex^\blocklength}
&:=
\bigotimes\limits_{\blockindex=1}^\blocklength
  \eigenvector_{\classicalOutputIndex_\blockindex}
\\
\inputDistribution(\inputAlphabetElement^\blocklength)
&:=
\prod\limits_{\blockindex=1}^\blocklength
  \inputDistribution(\inputAlphabetElement_\blockindex)
\\
\inputDistribution(\classicalOutputIndex^\blocklength | \inputAlphabetElement^\blocklength)
&:=
\prod\limits_{\blockindex=1}^\blocklength
  \inputDistribution(\classicalOutputIndex_\blockindex | \inputAlphabetElement_\blockindex)
\\
\outputDistribution(\classicalOutputIndex^\blocklength)
&:=
\prod\limits_{\blockindex=1}^\blocklength
  \outputDistribution(\classicalOutputIndex_\blockindex),
\end{align*}
we identify $\inputDistribution$ and $\outputDistribution$ with corresponding distributions on $\inputAlphabet^\blocklength \times \naturals^\blocklength$ (and the resulting marginals) and $\naturals^\blocklength$, and we have
\begin{align*}
\classicalQuantumChannel^\blocklength(\inputAlphabetElement^\blocklength)
&=
\sum_{\classicalOutputIndex^\blocklength \in \naturals^\blocklength}
  \inputDistribution(\classicalOutputIndex^\blocklength | \inputAlphabetElement^\blocklength)
  \rankOneOperator{\eigenvector_{\classicalOutputIndex^\blocklength |  \inputAlphabetElement^\blocklength}},
\\
\classicalQuantumChannel_\inputDistribution^{\otimes \blocklength}
&=
\sum_{\classicalOutputIndex^\blocklength \in \naturals^\blocklength}
  \outputDistribution(\classicalOutputIndex^\blocklength)
  \rankOneOperator{\eigenvector_{\classicalOutputIndex^\blocklength}}.
\end{align*}

For any $\typicalityParameter \in (0, \infty)$, we define
\begin{align*}
\jointTypicalitySet_{\typicalityParameter,\blocklength}
&:=
\left\{
  (
    \inputAlphabetElement^\blocklength,
    \classicalOutputIndex^\blocklength
  )
  \in
  \inputAlphabet^\blocklength
  \times
  \naturals^\blocklength
  :~
  \blocklength(\jointEntropy - \typicalityParameter)
  <
  -\log
  \inputDistribution(
    \classicalOutputIndex^\blocklength
    ~|~
    \inputAlphabetElement^\blocklength
  )
  <
  \blocklength(\jointEntropy + \typicalityParameter)
\right\}
\\
\outputTypicalitySet_{\typicalityParameter,\blocklength}
&:=
\left\{
  \classicalOutputIndex^\blocklength
  \in
  \naturals^\blocklength
  :~
  \blocklength(\outputEntropy - \typicalityParameter)
  <
  -\log
  \outputDistribution(
    \classicalOutputIndex^\blocklength
  )
  <
  \blocklength(\outputEntropy + \typicalityParameter)
\right\}
\end{align*}
and, based on these definitions,
\begin{align}
\label{eq:spectral-decomposition-joint-typicality-operator}
\jointTypicalityOperator_{\typicalityParameter,\blocklength}
  (\inputAlphabetElement^\blocklength)
&:=
\sum\limits_{\classicalOutputIndex^\blocklength \in \naturals^\blocklength}
  \indicatorFunction{\jointTypicalitySet_{\typicalityParameter,\blocklength}}
    (
      \inputAlphabetElement^\blocklength,
      \classicalOutputIndex^\blocklength
    )
  \rankOneOperator{\eigenvector_{
    \classicalOutputIndex^\blocklength
    |
    \inputAlphabetElement^\blocklength
  }}
\displaybreak[0] \\
\nonumber
\outputTypicalityOperator_{\typicalityParameter,\blocklength}
&:=
\sum\limits_{\classicalOutputIndex^\blocklength \in \naturals^\blocklength}
  \indicatorFunction{\outputTypicalitySet_{\typicalityParameter,\blocklength}}
    (\classicalOutputIndex^\blocklength)
  \rankOneOperator{\eigenvector_{\classicalOutputIndex^\blocklength}}.
\end{align}
While $\outputTypicalityOperator_{\typicalityParameter,\blocklength}$ is a fixed operator, $\jointTypicalityOperator_{\typicalityParameter,\blocklength}$ is an operator-valued function and its measurability is not immediately clear from the definition. Therefore, we note that it can also be written as
\[
\jointTypicalityOperator_{\typicalityParameter,\blocklength}
  (\inputAlphabetElement^\blocklength)
=
\indicatorFunction{(
  \exp(-\blocklength(\jointEntropy + \typicalityParameter)),
  \exp(-\blocklength(\jointEntropy - \typicalityParameter))
)}
\big(
  \classicalQuantumChannel^\blocklength(\inputAlphabetElement^\blocklength)
\big)
\]
and is therefore a measurable map $\inputAlphabet^\blocklength \rightarrow \traceClass{\hilbertSpace^{\otimes \blocklength}}$ by Lemma~\ref{lemma:spectral-decompositions-measurable}-\ref{item:spectral-decompositions-measurable-indicators}. We also note the following basic properties of these projections which are straightforward to verify from their definitions:
\begin{equation}
\label{eq:typicality-operators-basic-properties}
\outputTypicalityOperator_{\typicalityParameter,\blocklength}^2
=
\outputTypicalityOperator_{\typicalityParameter,\blocklength}
,
~~
\jointTypicalityOperator_{\typicalityParameter,\blocklength}
  (\inputAlphabetElement^\blocklength)^2
=
\jointTypicalityOperator_{\typicalityParameter,\blocklength}
  (\inputAlphabetElement^\blocklength)
,
~~
\classicalQuantumChannel^\blocklength(\inputAlphabetElement^\blocklength)
\jointTypicalityOperator_{\typicalityParameter,\blocklength}
  (\inputAlphabetElement^\blocklength)
=
\jointTypicalityOperator_{\typicalityParameter,\blocklength}
  (\inputAlphabetElement^\blocklength)
\classicalQuantumChannel^\blocklength(\inputAlphabetElement^\blocklength).
\end{equation}

Finally, define 
\begin{align}
\label{eq:typicality-operator-product}
\typicalityOperatorProduct_{\typicalityParameter,\blocklength}:~
\inputAlphabet^\blocklength \rightarrow \traceClass{\hilbertSpace^{\otimes \blocklength}},~
\inputAlphabetElement^\blocklength
&\mapsto
\outputTypicalityOperator_{\typicalityParameter,\blocklength}
\jointTypicalityOperator_{\typicalityParameter,\blocklength}
  (\inputAlphabetElement^\blocklength)
\\
\label{eq:typicality-operator-product-decoder}
\typicalityOperatorProductDecoder_{\typicalityParameter,\blocklength}:~
\inputAlphabet^\blocklength \rightarrow \traceClass{\hilbertSpace^{\otimes \blocklength}},~
\inputAlphabetElement^\blocklength
&\mapsto
\outputTypicalityOperator_{\typicalityParameter,\blocklength}
\jointTypicalityOperator_{\typicalityParameter,\blocklength}(\inputAlphabetElement^\blocklength)
\outputTypicalityOperator_{\typicalityParameter,\blocklength}.
\end{align}
which clearly are also measurable.

These notions of typicality will allow us to split the error terms that appear in the proofs of Theorems~\ref{theorem:q-resolvability} and~\ref{theorem:decoding} into a typical and an atypical part. These terms can be bounded separately with the help of the next two lemmas. These are fairly well-known facts in quantum information theory. For this reason, we omit the proofs here, but include them in Appendix~\ref{appendix:qit} for sake of self-containedness of the paper.

\begin{lemma}
\label{lemma:q-resolvability-typical-terms}
\emph{(Bound for typical terms).}
We have
\begin{enumerate}
  \item \label{item:q-resolvability-typical-terms-joint}
  $
    \forall \inputAlphabetElement^\blocklength \in \inputAlphabet^\blocklength~~
    \jointTypicalityOperator_{\typicalityParameter,\blocklength}
      (\inputAlphabetElement^\blocklength)
    \classicalQuantumChannel^\blocklength
      (\inputAlphabetElement^\blocklength)
    \jointTypicalityOperator_{\typicalityParameter,\blocklength}
      (\inputAlphabetElement^\blocklength)
    \leq
    \exp\big(
      -
      \blocklength
      (\jointEntropy - \typicalityParameter)
    \big)
    \jointTypicalityOperator_{\typicalityParameter,\blocklength}
      (\inputAlphabetElement^\blocklength)
  $
  \item \label{item:q-resolvability-typical-terms-joint-trace}
  $
    \forall \inputAlphabetElement^\blocklength \in \inputAlphabet^\blocklength~~
    \trace \jointTypicalityOperator_{\typicalityParameter,\blocklength} (\inputAlphabetElement^\blocklength)
    <
    \exp\big(
        \blocklength
        (\jointEntropy + \typicalityParameter)
      \big)
  $
  \item \label{item:q-resolvability-typical-terms-output}
  $
    \outputTypicalityOperator_{\typicalityParameter,\blocklength}
    \classicalQuantumChannel_\inputDistribution^{\otimes \blocklength}
    \outputTypicalityOperator_{\typicalityParameter,\blocklength}
    \leq
    \exp\big(
      -
      \blocklength
      (\outputEntropy - \typicalityParameter)
    \big)
    \outputTypicalityOperator_{\typicalityParameter,\blocklength}
  $
  \item \label{item:q-resolvability-typical-terms-output-trace}
  $
  \trace \outputTypicalityOperator_{\typicalityParameter,\blocklength}
  <
  \exp\big(
      \blocklength
      (\outputEntropy + \typicalityParameter)
    \big)
  $.
\end{enumerate}
\end{lemma}

\begin{lemma}
\label{lemma:q-resolvability-atypical-terms}
\emph{(Bound for atypical terms).}
For all $\blocklength \in \naturals, \renyiorder_1, \renyiorder_3 \in (1,\infty), \renyiorder_2, \renyiorder_4 \in [\renyiorder_{\min},1)$, we have
\begin{align}
\label{eq:q-resolvability-atypical-terms}
\Expectation\trace\left(
  \classicalQuantumChannel^\blocklength(\inputRV^\blocklength)
  \adjoint{
    \typicalityOperatorProduct_{\typicalityParameter,\blocklength}
      (\inputRV^\blocklength)
  }
\right)
&\geq
1
-
\atypicalTermsFunc_1(\typicalityParameter, \blocklength)
-
\atypicalTermsFunc_2(\typicalityParameter, \blocklength)
-
\atypicalTermsFunc_3(\typicalityParameter, \blocklength)
-
\atypicalTermsFunc_4(\typicalityParameter, \blocklength)
\\
\label{eq:q-resolvability-atypical-terms-decoder}
\Expectation\trace\left(
  \classicalQuantumChannel^\blocklength(\inputRV^\blocklength)
  \typicalityOperatorProductDecoder_{\typicalityParameter,\blocklength}
    (\inputRV^\blocklength)
\right)
&\geq
1
-
\atypicalTermsFunc_1(\typicalityParameter, \blocklength)
-
\atypicalTermsFunc_2(\typicalityParameter, \blocklength)
-
2\atypicalTermsFunc_3(\typicalityParameter, \blocklength)
-
2\atypicalTermsFunc_4(\typicalityParameter, \blocklength),
\end{align}
where $\typicalityOperatorProduct_{\typicalityParameter,\blocklength}$ is defined in (\ref{eq:typicality-operator-product}), $\typicalityOperatorProductDecoder_{\typicalityParameter,\blocklength}$ is defined in (\ref{eq:typicality-operator-product-decoder}), $\atypicalTermsFunc_1$, $\atypicalTermsFunc_2$, $\atypicalTermsFunc_3$, and $\atypicalTermsFunc_4$ are defined in \eqref{eq:atypicalTermsFunc-start} to \eqref{eq:atypicalTermsFunc-end}. Furthermore, the lower bounds \eqref{eq:q-resolvability-atypical-terms} and \eqref{eq:q-resolvability-atypical-terms-decoder} tend to $1$ exponentially as $\blocklength$ tends to $\infty$. That is, there is $\finalconst_1 \in (0,\infty)$ such that for sufficiently large $\blocklength$, we have
\begin{equation}
\label{eq:atypicalTermsFunc-bounds}
\atypicalTermsFunc_1(\typicalityParameter, \blocklength),
\atypicalTermsFunc_2(\typicalityParameter, \blocklength),
\atypicalTermsFunc_3(\typicalityParameter, \blocklength),
\atypicalTermsFunc_4(\typicalityParameter, \blocklength),
<
\exp(-\finalconst_1 \blocklength).
\end{equation}
\end{lemma}

\subsection{Proof of Theorem~\ref{theorem:q-resolvability} for \texorpdfstring{\gls{cq}}{cq} Channel Resolvability}
\label{sec:resolvability-proof}
We start with a preliminary lemma which is used in the proof of Theorem~\ref{theorem:q-resolvability} and encapsulates a symmetrization argument similar to the one used in \cite[Theorem 4.10]{wainwright2019high}.

\begin{lemma}
\label{lemma:symmetrization}
Let $(\inputAlphabet, \inputAlgebra)$ be a measurable space and $\operatorValuedMap: \inputAlphabet \rightarrow \traceClass{\hilbertSpace}$ a measurable map. Let $\inputDistribution$ be a probability measure on $\inputAlphabet$, and $\inputRV := (\inputRV_1, \dots, \inputRV_\generalIndexMax)$ be a tuple of $\inputAlphabet$-valued random variables i.i.d. according to $\inputDistribution$. Let
\[
\operatorValuedMap_\inputRV
:=
\frac{1}{\generalIndexMax}
\sum\limits_{\generalIndex=1}^\generalIndexMax
  \operatorValuedMap(\inputRV_\generalIndex).
\]
Then, we have
\begin{align}
\label{eq:symmetrization-first-alternative}
\Expectation\traceNorm{
  \operatorValuedMap_\inputRV
  -
  \Expectation
    \operatorValuedMap_\inputRV
}
&\leq
2
\trace\sqrt{
  \frac{1}{\generalIndexMax}
  \Expectation\big(
    \adjoint{\operatorValuedMap(\inputRV_1)}
    \operatorValuedMap(\inputRV_1)
  \big),
}
\\
\label{eq:symmetrization-second-alternative}
\Expectation\traceNorm{
  \operatorValuedMap_\inputRV
  -
  \Expectation
    \operatorValuedMap_\inputRV
}
&\leq
2
\trace\sqrt{
  \frac{1}{\generalIndexMax}
  \Expectation\big(
    \operatorValuedMap(\inputRV_1)
    \adjoint{\operatorValuedMap(\inputRV_1)}
  \big)
}.
\end{align}
\end{lemma}
\begin{proof}
Let $\hat{\inputRV_1}, \dots, \hat{\inputRV_\generalIndexMax}$ be i.i.d. independent copies of $\inputRV_1, \dots, \inputRV_\generalIndexMax$, and let $\rademacherRV_1, \dots, \rademacherRV_\generalIndexMax$ be i.i.d. uniformly on $\{-1,1\}$. The following derivations are adapted from the proof of~\cite[Theorem 4.10]{wainwright2019high}.
\begin{align*}
\Expectation\traceNorm{
  \operatorValuedMap_\inputRV
  -
  \Expectation
    \operatorValuedMap_\inputRV
}
&=
\Expectation_\inputRV
  \traceNorm{
    \Expectation_{\hat{\inputRV}}
      \sum\limits_{\generalIndex=1}^\generalIndexMax
      \frac{1}{\generalIndexMax}
      \left(
        \operatorValuedMap(\inputRV_\generalIndex)
        -
        \operatorValuedMap(\hat{\inputRV}_\generalIndex)
      \right)
  }
\\
\overset{(a)}&{\leq}
\Expectation_{\inputRV,\hat{\inputRV}}
  \traceNorm{
    \sum\limits_{\generalIndex=1}^\generalIndexMax
    \frac{1}{\generalIndexMax}
    \left(
      \operatorValuedMap(\inputRV_\generalIndex)
      -
      \operatorValuedMap(\hat{\inputRV}_\generalIndex)
    \right)
  }
\displaybreak[0] \\
\overset{(b)}&{=}
\Expectation_{\inputRV,\hat{\inputRV},\rademacherRV}
  \traceNorm{
    \sum\limits_{\generalIndex=1}^\generalIndexMax
    \frac{1}{\generalIndexMax}
    \rademacherRV_\generalIndex
    \left(
      \operatorValuedMap(\inputRV_\generalIndex)
      -
      \operatorValuedMap(\hat{\inputRV}_\generalIndex)
    \right)
  }
\displaybreak[0] \\
&\leq
\Expectation_{\inputRV,\rademacherRV}
  \traceNorm{
    \sum\limits_{\generalIndex=1}^\generalIndexMax
    \frac{1}{\generalIndexMax}
    \rademacherRV_\generalIndex
    \operatorValuedMap(\inputRV_\generalIndex)
  }
+
\Expectation_{\hat{\inputRV},\rademacherRV}
  \traceNorm{
    \sum\limits_{\generalIndex=1}^\generalIndexMax
    \frac{1}{\generalIndexMax}
    \rademacherRV_\generalIndex
    \operatorValuedMap(\hat{\inputRV}_\generalIndex)
  }
\displaybreak[0] \\
&=
2
\Expectation_{\inputRV,\rademacherRV}
  \traceNorm{
    \sum\limits_{\generalIndex=1}^\generalIndexMax
    \frac{1}{\generalIndexMax}
    \rademacherRV_\generalIndex
    \operatorValuedMap(\inputRV_\generalIndex)
  }
\displaybreak[0] \\
&=
2
\Expectation_{\inputRV,\rademacherRV}
  \trace
  \sqrt{
    \adjoint{
      \left(
        \sum\limits_{\generalIndex=1}^\generalIndexMax
        \frac{1}{\generalIndexMax}
        \rademacherRV_\generalIndex
        \operatorValuedMap(\inputRV_\generalIndex)
      \right)
    }
    \left(
      \sum\limits_{\generalIndex=1}^\generalIndexMax
      \frac{1}{\generalIndexMax}
      \rademacherRV_\generalIndex
      \operatorValuedMap(\inputRV_\generalIndex)
    \right)
  }
\displaybreak[0] \\
\overset{(c)}&{\leq}
2
\trace
\sqrt{
  \Expectation_{\inputRV,\rademacherRV}
  \left(
    \adjoint{
      \left(
        \sum\limits_{\generalIndex=1}^\generalIndexMax
        \frac{1}{\generalIndexMax}
        \rademacherRV_\generalIndex
        \operatorValuedMap(\inputRV_\generalIndex)
      \right)
    }
    \left(
      \sum\limits_{\generalIndex=1}^\generalIndexMax
      \frac{1}{\generalIndexMax}
      \rademacherRV_\generalIndex
      \operatorValuedMap(\inputRV_\generalIndex)
    \right)
  \right)
}
\displaybreak[0] \\
&=
2
\trace
\sqrt{
  \frac{1}{\generalIndexMax^2}
  \Expectation_\inputRV
    \sum\limits_{\generalIndex_1,\generalIndex_2=1}^\generalIndexMax
      \Expectation_\rademacherRV
        (
          \rademacherRV_{\generalIndex_1}
          \rademacherRV_{\generalIndex_2}
        )
        \adjoint{
          \operatorValuedMap(\inputRV_{\generalIndex_1})
        }
        \operatorValuedMap(\inputRV_{\generalIndex_2})
}
\displaybreak[0] \\
&\overset{(d)}{=}
2
\trace
\sqrt{
  \frac{1}{\generalIndexMax^2}
  \Expectation_\inputRV
    \sum\limits_{\generalIndex=1}^\generalIndexMax
        \adjoint{
          \operatorValuedMap(\inputRV_{\generalIndex})
        }
        \operatorValuedMap(\inputRV_{\generalIndex})
}
\\
&=
2
\trace
\sqrt{
  \frac{1}{\generalIndexMax}
  \Expectation_{\inputRV_1}
    \adjoint{
      \operatorValuedMap(\inputRV_{1})
    }
    \operatorValuedMap(\inputRV_{1})
}.
\end{align*}
Step (a) follows by Lemma~\ref{lemma:bochner-integral-basics}-\ref{item:bochner-integral-basics-trace-norm-jensen}. For step (b), we observe that the equality holds conditioned on any realization of $\rademacherRV_1, \dots, \rademacherRV_\generalIndexMax$ since $\inputRV_\generalIndex$ and $\hat{\inputRV}_\generalIndex$ are identically distributed and therefore can be swapped if $\rademacherRV_\generalIndex = -1$. Inequality (c) is by Lemma~\ref{lemma:bochner-integral-basics}-\ref{item:bochner-integral-basics-square-root-jensen}. Finally, equality (d) holds because
$
\Expectation_\rademacherRV
  (
    \rademacherRV_{\generalIndex_1}
    \rademacherRV_{\generalIndex_2}
  )
$
equals $1$ if $\generalIndex_1=\generalIndex_2$ and $0$ otherwise. This concludes the proof of (\ref{eq:symmetrization-first-alternative}). (\ref{eq:symmetrization-second-alternative}) follows similarly using $\traceNorm{\generalOperator} = \traceNorm{\adjoint{\generalOperator}}$ before expanding the trace norm.
\end{proof}
\begin{remark}
\label{remark:rademacher}
\cite[Theorem 4.10]{wainwright2019high} (from the proof of which the first part of the calculation above is adapted) bounds the absolute deviation of an empirical average from its expectation in terms of the Rademacher complexity of a suitably defined class of functions. Indeed, via the trace norm duality stated in Lemma~\ref{lemma:norm-basics}-\ref{item:norm-basics-trace-norm-duality}, the term
\[
\Expectation_{\inputRV,\rademacherRV}
  \traceNorm{
    \sum\limits_{\generalIndex=1}^\generalIndexMax
    \frac{1}{\generalIndexMax}
    \rademacherRV_\generalIndex
    \operatorValuedMap(\inputRV_\generalIndex)
  }
\]
which appears in the calculation can be argued to be equal to the Rademacher complexity of the function class
\[
\{
  \inputAlphabet
  \ni
  \inputAlphabetElement
  \mapsto
  \trace\left(
    \generalOperator
    \operatorValuedMap(\inputAlphabetElement)
  \right)
  \in
  \complexNumbers
  ~:~
  \generalOperator \in \boundedOperators{\hilbertSpace},~
  \operatorNorm{\generalOperator} \leq 1
\}.
\]
\end{remark}

\begin{proof}[Proof of Theorem~\ref{theorem:q-resolvability}]
In this proof, we use the definitions of Section~\ref{sec:typicality-preliminaries}. We multiply with identities and use the triangle inequality to bound
\begin{align}
\nonumber
&\hphantom{{}={}}
\Expectation_\codebook
  \traceNorm{
    \classicalQuantumChannel_\codebook
    -
    \classicalQuantumChannel_\inputDistribution^{\otimes \blocklength}
  }
\\
\nonumber
&=
\Expectation_\codebook
  \traceNorm{
    \frac{1}{\codebookSize}
    \sum\limits_{\codewordIndex=1}^\codebookSize\Big(
      \classicalQuantumChannel^\blocklength(\codebook(\codewordIndex))
      -
      \Expectation_\codebook
        \classicalQuantumChannel^\blocklength(\codebook(\codewordIndex))
    \Big)
  }
\displaybreak[0] \\
\nonumber
&=
\Expectation_\codebook
  \traceNorm{
    \frac{1}{\codebookSize}
    \sum\limits_{\codewordIndex=1}^\codebookSize\bigg(
      \big(
        \typicalityOperatorProduct_{\typicalityParameter, \blocklength}
          (\codebook(\codewordIndex))
        +
        \identityOperator
        -
        \typicalityOperatorProduct_{\typicalityParameter, \blocklength}
          (\codebook(\codewordIndex))
      \big)
      \classicalQuantumChannel^\blocklength(\codebook(\codewordIndex))
      -
      \Expectation_\codebook\Big(
        \big(
          \typicalityOperatorProduct_{\typicalityParameter, \blocklength}
            (\codebook(\codewordIndex))
          +
          \identityOperator
          -
          \typicalityOperatorProduct_{\typicalityParameter, \blocklength}
            (\codebook(\codewordIndex))
        \big)
        \classicalQuantumChannel^\blocklength(\codebook(\codewordIndex))
      \Big)
    \bigg)
  }
\displaybreak[0] \\
\nonumber
&\leq
\begin{aligned}[t]
  &\Expectation_\codebook
    \traceNorm{
      \frac{1}{\codebookSize}
      \sum\limits_{\codewordIndex=1}^\codebookSize\Big(
        \typicalityOperatorProduct_{\typicalityParameter, \blocklength}
          (\codebook(\codewordIndex))
        \classicalQuantumChannel^\blocklength(\codebook(\codewordIndex))
        -
        \Expectation_\codebook\big(
          \typicalityOperatorProduct_{\typicalityParameter, \blocklength}
            (\codebook(\codewordIndex))
          \classicalQuantumChannel^\blocklength(\codebook(\codewordIndex))
        \big)
    }
\\
  &+
  \Expectation_\codebook
    \traceNorm{
      \frac{1}{\codebookSize}
      \sum\limits_{\codewordIndex=1}^\codebookSize\bigg(
        \big(
          \identityOperator
          -
          \typicalityOperatorProduct_{\typicalityParameter, \blocklength}
            (\codebook(\codewordIndex))
        \big)
        \classicalQuantumChannel^\blocklength(\codebook(\codewordIndex))
        -
        \Expectation_\codebook\Big(
          \big(
            \identityOperator
            -
            \typicalityOperatorProduct_{\typicalityParameter, \blocklength}
              (\codebook(\codewordIndex))
          \big)
          \classicalQuantumChannel^\blocklength(\codebook(\codewordIndex))
        \Big)
      \bigg)
    }
\end{aligned}
\displaybreak[0] \\
\nonumber
&=
\begin{aligned}[t]
  &\Expectation_\codebook
    \traceNorm{
      \frac{1}{\codebookSize}
      \sum\limits_{\codewordIndex=1}^\codebookSize\Big(
        \typicalityOperatorProduct_{\typicalityParameter, \blocklength}
          (\codebook(\codewordIndex))
        \classicalQuantumChannel^\blocklength(\codebook(\codewordIndex))
        -
        \Expectation_\codebook\big(
          \typicalityOperatorProduct_{\typicalityParameter, \blocklength}
            (\codebook(\codewordIndex))
          \classicalQuantumChannel^\blocklength(\codebook(\codewordIndex))
        \big)
    }
\\
  &+
  \begin{multlined}[t]
    \Expectation_\codebook
      \traceNormBiggLeft
        \frac{1}{\codebookSize}
        \sum\limits_{\codewordIndex=1}^\codebookSize\bigg(
          \big(
            \identityOperator
            -
            \typicalityOperatorProduct_{\typicalityParameter, \blocklength}
              (\codebook(\codewordIndex))
          \big)
          \classicalQuantumChannel^\blocklength(\codebook(\codewordIndex))
          \adjoint{\big(
            \typicalityOperatorProduct_{\typicalityParameter, \blocklength}
              (\codebook(\codewordIndex))
            +
            \identityOperator
            -
            \typicalityOperatorProduct_{\typicalityParameter, \blocklength}
              (\codebook(\codewordIndex))
          \big)}
\\ \hspace{4cm}
          -
          \Expectation_\codebook\Big(
            \big(
              \identityOperator
              -
              \typicalityOperatorProduct_{\typicalityParameter, \blocklength}
                (\codebook(\codewordIndex))
            \big)
            \classicalQuantumChannel^\blocklength(\codebook(\codewordIndex))
            \adjoint{\big(
              \typicalityOperatorProduct_{\typicalityParameter, \blocklength}
                (\codebook(\codewordIndex))
              +
              \identityOperator
              -
              \typicalityOperatorProduct_{\typicalityParameter, \blocklength}
                (\codebook(\codewordIndex))
            \big)}
          \Big)
        \bigg)
       \traceNormBiggRight
  \end{multlined}
\end{aligned}
\displaybreak[0] \\
\label{eq:q-resolvability-split}
&\begin{aligned}[c]
  {}\leq{}
  &\Expectation_\codebook
    \traceNorm{
      \frac{1}{\codebookSize}
      \sum\limits_{\codewordIndex=1}^\codebookSize\Big(
        \typicalityOperatorProduct_{\typicalityParameter, \blocklength}
          (\codebook(\codewordIndex))
        \classicalQuantumChannel^\blocklength(\codebook(\codewordIndex))
        -
        \Expectation_\codebook\big(
          \typicalityOperatorProduct_{\typicalityParameter, \blocklength}
            (\codebook(\codewordIndex))
          \classicalQuantumChannel^\blocklength(\codebook(\codewordIndex))
        \big)
    }
\\
  &+
  \Expectation_\codebook
    \traceNorm{
      \frac{1}{\codebookSize}
      \sum\limits_{\codewordIndex=1}^\codebookSize\Big(
        \classicalQuantumChannel^\blocklength(\codebook(\codewordIndex))
        \adjoint{\typicalityOperatorProduct_{\typicalityParameter, \blocklength}(\codebook(\codewordIndex))}
        -
        \Expectation_\codebook\big(
          \classicalQuantumChannel^\blocklength(\codebook(\codewordIndex))
          \adjoint{\typicalityOperatorProduct_{\typicalityParameter, \blocklength}(\codebook(\codewordIndex))}
        \big)
      \Big)
    }
\\
  &+
  \Expectation_\codebook
    \traceNorm{
      \frac{1}{\codebookSize}
      \sum\limits_{\codewordIndex=1}^\codebookSize\bigg(
        \typicalityOperatorProduct_{\typicalityParameter, \blocklength}
          (\codebook(\codewordIndex))
        \classicalQuantumChannel^\blocklength(\codebook(\codewordIndex))
        \adjoint{\typicalityOperatorProduct_{\typicalityParameter, \blocklength}(\codebook(\codewordIndex))}
        -
        \Expectation_\codebook\Big(
          \typicalityOperatorProduct_{\typicalityParameter, \blocklength}
            (\codebook(\codewordIndex))
          \classicalQuantumChannel^\blocklength(\codebook(\codewordIndex))
          \adjoint{\typicalityOperatorProduct_{\typicalityParameter, \blocklength}(\codebook(\codewordIndex))}
        \Big)
      \bigg)
    }
\\
  &+
  \begin{multlined}[t]
  \Expectation_\codebook
    \traceNormBiggLeft
      \frac{1}{\codebookSize}
      \sum\limits_{\codewordIndex=1}^\codebookSize\bigg(
        \big(
          \identityOperator
          -
          \typicalityOperatorProduct_{\typicalityParameter, \blocklength}
            (\codebook(\codewordIndex))
        \big)
        \classicalQuantumChannel^\blocklength(\codebook(\codewordIndex))
        \adjoint{\big(
          \identityOperator
          -
          \typicalityOperatorProduct_{\typicalityParameter, \blocklength}
            (\codebook(\codewordIndex))
        \big)}
        \hspace{6cm}
        \\
        -
        \Expectation_\codebook\Big(
          \big(
            \identityOperator
            -
            \typicalityOperatorProduct_{\typicalityParameter, \blocklength}
              (\codebook(\codewordIndex))
          \big)
          \classicalQuantumChannel^\blocklength(\codebook(\codewordIndex))
          \adjoint{\big(
            \identityOperator
            -
            \typicalityOperatorProduct_{\typicalityParameter, \blocklength}
              (\codebook(\codewordIndex))
          \big)}
        \Big)
      \bigg)
    \traceNormBiggRight
  \end{multlined}
\end{aligned}
\end{align}
We next bound the summands in (\ref{eq:q-resolvability-split}) separately. For the last summand, we calculate
\begin{align*}
&\hphantom{{}={}}
\begin{multlined}[t]
\Expectation_\codebook
  \traceNormBiggLeft
    \frac{1}{\codebookSize}
    \sum\limits_{\codewordIndex=1}^\codebookSize\bigg(
      \big(
        \identityOperator
        -
        \typicalityOperatorProduct_{\typicalityParameter, \blocklength}
          (\codebook(\codewordIndex))
      \big)
      \classicalQuantumChannel^\blocklength(\codebook(\codewordIndex))
      \adjoint{\big(
        \identityOperator
        -
        \typicalityOperatorProduct_{\typicalityParameter, \blocklength}
          (\codebook(\codewordIndex))
      \big)}
      \hspace{4cm} \\
      -
      \Expectation_\codebook\Big(
        \big(
          \identityOperator
          -
          \typicalityOperatorProduct_{\typicalityParameter, \blocklength}
            (\codebook(\codewordIndex))
        \big)
        \classicalQuantumChannel^\blocklength(\codebook(\codewordIndex))
        \adjoint{\big(
          \identityOperator
          -
          \typicalityOperatorProduct_{\typicalityParameter, \blocklength}
            (\codebook(\codewordIndex))
        \big)}
      \Big)
    \bigg)
  \traceNormBiggRight
\end{multlined}
\\
&\overset{(a)}{\leq}
\begin{aligned}[t]
  &\Expectation_\codebook
    \traceNorm{
      \frac{1}{\codebookSize}
      \sum\limits_{\codewordIndex=1}^\codebookSize\Big(
        \big(
          \identityOperator
          -
          \typicalityOperatorProduct_{\typicalityParameter, \blocklength}
            (\codebook(\codewordIndex))
        \big)
        \classicalQuantumChannel^\blocklength(\codebook(\codewordIndex))
        \adjoint{\big(
          \identityOperator
          -
          \typicalityOperatorProduct_{\typicalityParameter, \blocklength}
            (\codebook(\codewordIndex))
        \big)}
      \Big)
    }
  \\
  &+
  \traceNorm{
    \frac{1}{\codebookSize}
    \Expectation_\codebook
    \sum\limits_{\codewordIndex=1}^\codebookSize\Big(
      \big(
        \identityOperator
        -
        \typicalityOperatorProduct_{\typicalityParameter, \blocklength}
          (\codebook(\codewordIndex))
      \big)
      \classicalQuantumChannel^\blocklength(\codebook(\codewordIndex))
      \adjoint{\big(
        \identityOperator
        -
        \typicalityOperatorProduct_{\typicalityParameter, \blocklength}
          (\codebook(\codewordIndex))
      \big)}
    \Big)
  }
\end{aligned}
\\
&\overset{(b)}{=}
2\frac{1}{\codebookSize}
\sum\limits_{\codewordIndex=1}^\codebookSize
  \Expectation_\codebook\trace{\Big(
    \big(
      \identityOperator
      -
      \typicalityOperatorProduct_{\typicalityParameter, \blocklength}
        (\codebook(\codewordIndex))
    \big)
    \classicalQuantumChannel^\blocklength(\codebook(\codewordIndex))
    \adjoint{\big(
      \identityOperator
      -
      \typicalityOperatorProduct_{\typicalityParameter, \blocklength}
        (\codebook(\codewordIndex))
    \big)}
  \Big)}
\\
&=
2
\Expectation_{\inputRV^\blocklength}\trace{\Big(
  \big(
    \identityOperator
    -
    \typicalityOperatorProduct_{\typicalityParameter, \blocklength}
      (\inputRV^\blocklength)
  \big)
  \classicalQuantumChannel^\blocklength(\inputRV^\blocklength)
  \adjoint{\big(
    \identityOperator
    -
    \typicalityOperatorProduct_{\typicalityParameter, \blocklength}
      (\inputRV^\blocklength)
  \big)}
\Big)}
\\
&=
2\Expectation_{\inputRV^\blocklength}\left(
  1
  -
  \trace{\Big(
    \outputTypicalityOperator_{\typicalityParameter,\blocklength}
    \jointTypicalityOperator_{\typicalityParameter,\blocklength}
      (\inputRV^\blocklength)
    \classicalQuantumChannel^\blocklength(\inputRV^\blocklength)
  \Big)}
  -
  \trace{\Big(
    \classicalQuantumChannel^\blocklength(\inputRV^\blocklength)
    \adjoint{\typicalityOperatorProduct_{\typicalityParameter, \blocklength}(\inputRV^\blocklength)}
  \Big)}
  +
  \trace{\Big(
    \outputTypicalityOperator_{\typicalityParameter,\blocklength}
    \jointTypicalityOperator_{\typicalityParameter,\blocklength}
      (\inputRV^\blocklength)
    \classicalQuantumChannel^\blocklength(\inputRV^\blocklength)
    \jointTypicalityOperator_{\typicalityParameter,\blocklength}
      (\inputRV^\blocklength)
    \outputTypicalityOperator_{\typicalityParameter,\blocklength}
  \Big)}
\right)
\\
&\overset{(c)}{=}
2\left(
  1
  -
  \Expectation_{\inputRV^\blocklength}\Big(
    \trace{\big(
      \classicalQuantumChannel^\blocklength(\inputRV^\blocklength)
      \adjoint{\typicalityOperatorProduct_{\typicalityParameter, \blocklength}(\inputRV^\blocklength)}
    \big)}
  \Big)
\right)
\\
&\overset{(d)}{\leq}
2\atypicalTermsFunc_1(\typicalityParameter, \blocklength)
+
2\atypicalTermsFunc_2(\typicalityParameter, \blocklength)
+
2\atypicalTermsFunc_3(\typicalityParameter, \blocklength)
+
2\atypicalTermsFunc_4(\typicalityParameter, \blocklength).
\end{align*}
(a) is due to the triangle inequality, (b) uses Lemma~\ref{lemma:bochner-integral-basics}-\ref{item:bochner-integral-basics-bounded-functional} and the fact that trace and trace norm are equal for positive semi-definite operators, (c) is by cyclic permutations inside the trace and using (\ref{eq:typicality-operators-basic-properties}). Finally, (d) is due to Lemma~\ref{lemma:q-resolvability-atypical-terms}.

For the first summand in (\ref{eq:q-resolvability-split}), we apply (\ref{eq:symmetrization-second-alternative}) of Lemma~\ref{lemma:symmetrization} with
$
\operatorValuedMap:
\inputAlphabet^\blocklength
\rightarrow
\traceClass{\hilbertSpace},
\inputAlphabetElement^\blocklength
\mapsto
\typicalityOperatorProduct_{\typicalityParameter, \blocklength}
  (\inputAlphabetElement^\blocklength)
\classicalQuantumChannel^\blocklength(\inputAlphabetElement^\blocklength)
$ (which is measurable by Lemma~\ref{lemma:operator-compositions-continuous}),
$\generalIndexMax = \codebookSize$ and $\inputRV_\codewordIndex = \codebook(\codewordIndex)$. This yields
\begin{align*}
&\hphantom{{}={}}
\Expectation_\codebook
  \traceNorm{
    \frac{1}{\codebookSize}
    \sum\limits_{\codewordIndex=1}^\codebookSize\Big(
      \typicalityOperatorProduct_{\typicalityParameter, \blocklength}
        (\codebook(\codewordIndex))
      \classicalQuantumChannel^\blocklength(\codebook(\codewordIndex))
      -
      \Expectation_\codebook\big(
        \typicalityOperatorProduct_{\typicalityParameter, \blocklength}
          (\codebook(\codewordIndex))
        \classicalQuantumChannel^\blocklength(\codebook(\codewordIndex))
      \big)
    \Big)
  }
\\
&\leq
2\trace{
  \sqrt{
    \frac{1}{\codebookSize}
    \Expectation_{\inputRV^\blocklength}\big(
      \typicalityOperatorProduct_{\typicalityParameter, \blocklength}
        (\inputRV^\blocklength)
      \classicalQuantumChannel^\blocklength(\inputRV^\blocklength)^2
      \adjoint{
        \typicalityOperatorProduct_{\typicalityParameter, \blocklength}
          (\inputRV^\blocklength)
      }
    \big)
  }
}
\\
\overset{(\ref{eq:typicality-operators-basic-properties})}&{=}
2\trace{
  \sqrt{
    \frac{1}{\codebookSize}
    \Expectation_{\inputRV^\blocklength}\big(
      \outputTypicalityOperator_{\typicalityParameter, \blocklength}
      \classicalQuantumChannel^\blocklength
        (\inputRV^\blocklength)^{\frac{1}{2}}
      \jointTypicalityOperator_{\typicalityParameter, \blocklength}
        (\inputRV^\blocklength)
      \classicalQuantumChannel^\blocklength(\inputRV^\blocklength)
      \jointTypicalityOperator_{\typicalityParameter, \blocklength}
        (\inputRV^\blocklength)
      \classicalQuantumChannel^\blocklength
        (\inputRV^\blocklength)^{\frac{1}{2}}
      \outputTypicalityOperator_{\typicalityParameter, \blocklength}
    \big)
  }
}.
\end{align*}
For the second summand in (\ref{eq:q-resolvability-split}), we use (\ref{eq:symmetrization-first-alternative}) of Lemma~\ref{lemma:symmetrization} with
$
\operatorValuedMap:
\inputAlphabet^\blocklength
\rightarrow
\traceClass{\hilbertSpace},
\inputAlphabetElement^\blocklength
\mapsto
\classicalQuantumChannel^\blocklength(\inputAlphabetElement^\blocklength)
\adjoint{\typicalityOperatorProduct_{\typicalityParameter, \blocklength}
  (\inputAlphabetElement^\blocklength)}
$ (which is again measurable by Lemma~\ref{lemma:operator-compositions-continuous})
and obtain the exact same upper bound.

For the third summand in (\ref{eq:q-resolvability-split}), we use Lemma~\ref{lemma:symmetrization} one more time with
$
\operatorValuedMap:
\inputAlphabet^\blocklength
\rightarrow
\traceClass{\hilbertSpace},
\inputAlphabetElement^\blocklength
\mapsto
\typicalityOperatorProduct_{\typicalityParameter, \blocklength}
  (\inputAlphabetElement^\blocklength)
\classicalQuantumChannel^\blocklength(\inputAlphabetElement^\blocklength)
\adjoint{\typicalityOperatorProduct_{\typicalityParameter, \blocklength}
  (\inputAlphabetElement^\blocklength)}
$
(measurability is again by Lemma~\ref{lemma:operator-compositions-continuous}; this time it does not matter which alternative we use because $\operatorValuedMap$ has self-adjoint values) and obtain
\begin{align*}
&\hphantom{{}={}}
\Expectation_\codebook
  \traceNorm{
    \frac{1}{\codebookSize}
    \sum\limits_{\codewordIndex=1}^\codebookSize\Big(
      \typicalityOperatorProduct_{\typicalityParameter, \blocklength}
        (\codebook(\codewordIndex))
      \classicalQuantumChannel^\blocklength(\codebook(\codewordIndex))
      \adjoint{\typicalityOperatorProduct_{\typicalityParameter, \blocklength}(\codebook(\codewordIndex))}
      -
      \Expectation_\codebook\big(
        \typicalityOperatorProduct_{\typicalityParameter, \blocklength}
          (\codebook(\codewordIndex))
        \classicalQuantumChannel^\blocklength(\codebook(\codewordIndex))
        \adjoint{\typicalityOperatorProduct_{\typicalityParameter, \blocklength}(\codebook(\codewordIndex))}
      \big)
    \Big)
  }
\\
&\leq
2\trace\sqrt{
  \frac{1}{\codebookSize}
  \Expectation_{\inputRV^\blocklength}\left(
  \typicalityOperatorProduct_{\typicalityParameter, \blocklength}
    (\inputRV^\blocklength)
  \classicalQuantumChannel^\blocklength(\inputRV^\blocklength)
  \adjoint{
    \typicalityOperatorProduct_{\typicalityParameter, \blocklength}
      (\inputRV^\blocklength)
  }
  \typicalityOperatorProduct_{\typicalityParameter, \blocklength}
    (\inputRV^\blocklength)
  \classicalQuantumChannel^\blocklength(\inputRV^\blocklength)
  \adjoint{
    \typicalityOperatorProduct_{\typicalityParameter, \blocklength}
      (\inputRV^\blocklength)
  }
  \right)
}
\\
\overset{(\ref{eq:typicality-operators-basic-properties})}&{=}
2\trace{
  \sqrt{
    \frac{1}{\codebookSize}
    \Expectation_{\inputRV^\blocklength}\big(
      \outputTypicalityOperator_{\typicalityParameter, \blocklength}
      \classicalQuantumChannel^\blocklength
        (\inputRV^\blocklength)^{\frac{1}{2}}
      \jointTypicalityOperator_{\typicalityParameter, \blocklength}
        (\inputRV^\blocklength)
      \classicalQuantumChannel^\blocklength
        (\inputRV^\blocklength)^{\frac{1}{2}}
      \jointTypicalityOperator_{\typicalityParameter, \blocklength}
        (\inputRV^\blocklength)
      \outputTypicalityOperator_{\typicalityParameter, \blocklength}
      \jointTypicalityOperator_{\typicalityParameter, \blocklength}
        (\inputRV^\blocklength)
      \classicalQuantumChannel^\blocklength
        (\inputRV^\blocklength)^{\frac{1}{2}}
      \jointTypicalityOperator_{\typicalityParameter, \blocklength}
        (\inputRV^\blocklength)
      \classicalQuantumChannel^\blocklength
        (\inputRV^\blocklength)^{\frac{1}{2}}
      \outputTypicalityOperator_{\typicalityParameter, \blocklength}
    \big)
  }
}
\\
\overset{(a)}&{\leq}
2\trace{
  \sqrt{
    \frac{1}{\codebookSize}
    \Expectation_{\inputRV^\blocklength}\big(
      \outputTypicalityOperator_{\typicalityParameter, \blocklength}
      \classicalQuantumChannel^\blocklength
        (\inputRV^\blocklength)^{\frac{1}{2}}
      \jointTypicalityOperator_{\typicalityParameter, \blocklength}
        (\inputRV^\blocklength)
      \classicalQuantumChannel^\blocklength(\inputRV^\blocklength)
      \jointTypicalityOperator_{\typicalityParameter, \blocklength}
        (\inputRV^\blocklength)
      \classicalQuantumChannel^\blocklength
        (\inputRV^\blocklength)^{\frac{1}{2}}
      \outputTypicalityOperator_{\typicalityParameter, \blocklength}
    \big)
  }
}.
\end{align*}
(a) is by $\outputTypicalityOperator_{\typicalityParameter,\blocklength} \leq \identityOperator$ and once more applying (\ref{eq:typicality-operators-basic-properties}). We have obtained matching upper bounds in all three cases and can conclude our calculation with a successive application of the sub-items of Lemma~\ref{lemma:q-resolvability-typical-terms} in conjunction with the operator monotonicity of the square root~\cite{pedersen1972some}. The number of the subitem of Lemma~\ref{lemma:q-resolvability-typical-terms} is indicated above the inequality sign where it is applied.
\begin{align*}
&\hphantom{{}={}}
\trace{
  \sqrt{
    \frac{1}{\codebookSize}
    \Expectation_{\inputRV^\blocklength}\big(
      \outputTypicalityOperator_{\typicalityParameter, \blocklength}
      \classicalQuantumChannel^\blocklength
        (\inputRV^\blocklength)^{\frac{1}{2}}
      \jointTypicalityOperator_{\typicalityParameter, \blocklength}
        (\inputRV^\blocklength)
      \classicalQuantumChannel^\blocklength(\inputRV^\blocklength)
      \jointTypicalityOperator_{\typicalityParameter, \blocklength}
        (\inputRV^\blocklength)
      \classicalQuantumChannel^\blocklength
        (\inputRV^\blocklength)^{\frac{1}{2}}
      \outputTypicalityOperator_{\typicalityParameter, \blocklength}
    \big)
  }
}
\\
&\overset{\ref{item:q-resolvability-typical-terms-joint})}{\leq}
\trace{
  \sqrt{
    \frac{1}{\codebookSize}
    \exp\big(
      -
      \blocklength
      (\jointEntropy - \typicalityParameter)
    \big)
    \Expectation_{\inputRV^\blocklength}\big(
      \outputTypicalityOperator_{\typicalityParameter, \blocklength}
      \classicalQuantumChannel^\blocklength
        (\inputRV^\blocklength)
      \outputTypicalityOperator_{\typicalityParameter, \blocklength}
    \big)
  }
}
\displaybreak[0]
\\
&=
\trace{
  \sqrt{
    \frac{1}{\codebookSize}
    \exp\big(
      -
      \blocklength
      (\jointEntropy - \typicalityParameter)
    \big)
      \outputTypicalityOperator_{\typicalityParameter, \blocklength}
      \classicalQuantumChannel_\inputDistribution^{\otimes \blocklength}
      \outputTypicalityOperator_{\typicalityParameter, \blocklength}
  }
}
\displaybreak[0]
\\
&\overset{\ref{item:q-resolvability-typical-terms-output})}{\leq}
\trace{
  \sqrt{
    \frac{1}{\codebookSize}
    \exp\big(
      -
      \blocklength
      (\jointEntropy + \outputEntropy - 2\typicalityParameter)
    \big)
      \outputTypicalityOperator_{\typicalityParameter, \blocklength}
  }
}
\displaybreak[0]
\\
&\overset{(a)}{\leq}
\exp\left(
  -
  \frac{1}{2}
  \blocklength
  (\jointEntropy + \outputEntropy + \codebookRate - 2\typicalityParameter)
\right)
\trace{
  \outputTypicalityOperator_{\typicalityParameter, \blocklength}
}
\displaybreak[0]
\\
&\overset{\ref{item:q-resolvability-typical-terms-output-trace})}{\leq}
\exp\left(
  -
  \frac{1}{2}
  \blocklength
  (\jointEntropy - \outputEntropy + \codebookRate - 4\typicalityParameter)
\right)
\\
\overset{eq. (\ref{eq:quantum-information})}&{=}
\exp\left(
  -
  \frac{1}{2}
  \blocklength
  (\codebookRate - \quantumInformation - 4\typicalityParameter)
\right).
\end{align*}
Step (a) is due to the assumption $\codebookSize \geq \exp(\blocklength\codebookRate)$ and (\ref{eq:typicality-operators-basic-properties}). \eqref{eq:q-resolvability-all-blocklengths} now follows from \eqref{eq:q-resolvability-split} and the upper bounds for the four summands on its right hand side which we have calculated. To prove \eqref{eq:q-resolvability-large-blocklengths}, we choose $\typicalityParameter \in (0,(\codebookRate-\quantumInformation)/4)$ and invoke Lemma~\ref{lemma:q-resolvability-atypical-terms} to fix $\finalconst_1$ which satisfies \eqref{eq:atypicalTermsFunc-bounds} for this choice of $\typicalityParameter$. The infimum in \eqref{eq:resolvabilityBoundFunc} is clearly upper bounded by the realization for our choice of $\typicalityParameter$, so we have argued that \eqref{eq:q-resolvability-large-blocklengths} holds for any choice
\[
\finalconst
\in
\left(
  0,
  \min\left(
    \finalconst_1,
    \frac{1}{2}
    (\codebookRate - \quantumInformation - 4\typicalityParameter)
  \right)
\right).
\qedhere
\]
\end{proof}

\subsection{Proof of Theorem~\ref{theorem:decoding} for \texorpdfstring{\gls{cq}}{cq} Channel Coding}
\label{sec:coding-proof}
In this section, we follow the methodology in~\cite{hayashi2003general}, making adaptations as needed to derive the exponential error bound as stated in Theorem~\ref{theorem:decoding}. An essential ingredient will be the following lemma from~\cite{hayashi2003general}. In the statement of the lemma, we need the notion of Moore-Penrose pseudoinverse which assigns to any $\generalOperator\in\boundedOperators{\hilbertSpace}$ an (unbounded) operator $\pseudoinverse{A}$ acting on  $\hilbertSpace$ \cite[Definition 2.2]{engl2000regularization}. Moreover, for $\generalOperator\in\boundedOperators{\hilbertSpace}$, we use the notation $\image (A):=\{\hilbertSpaceElement\in \hilbertSpace: \exists \tilde{\hilbertSpaceElement} \in \hilbertSpace \textrm{ such that }\hilbertSpaceElement= \generalOperator\tilde{\hilbertSpaceElement} \}$. The following lemma is a well-established fact in quantum information theory, but some care needs to be taken in the infinite-dimensional case to ensure the Moore-Penrose pseudoinverse which appears is well-behaved. Therefore, we include a proof of this part of the lemma in Appendix~\ref{appendix:qit}.
\begin{lemma}
\label{lemma:hninequality}
\emph{(Hayashi-Nagaoka~\cite[Lemma 2]{hayashi2003general})}~
Let $\hnconstant \in (0,\infty)$, and let $\generalOperator, \generalOperatorTwo \in \boundedOperators{\hilbertSpace}$ with $0 \leq \generalOperator \leq \identityOperator$ and $0 \leq \generalOperatorTwo$ such that $\image(\generalOperator+\generalOperatorTwo)$ is closed. Then the following statements hold true:
\begin{enumerate}
    \item\label{item:hninequality-bounded} The Moore-Penrose pseudoinverse $\pseudoinverse{\sqrt{\generalOperator+\generalOperatorTwo}}$ is a bounded linear operator, i.e., $\pseudoinverse{\sqrt{\generalOperator+\generalOperatorTwo}} \in \boundedOperators{\hilbertSpace}$.
    \item\label{item:hninequality-main} For any real number $\hnconstant>0$, we have
\[
\identityOperator
-
\sqrt{\generalOperator + \generalOperatorTwo}^{-1}
\generalOperator
\sqrt{\generalOperator + \generalOperatorTwo}^{-1}
\leq
(1+\hnconstant)
(\identityOperator - \generalOperator)
+
(2+\hnconstant+\hnconstant^{-1})
\generalOperatorTwo.
\]
\end{enumerate}
\end{lemma}
For the proof of Theorem~\ref{theorem:decoding}, we use the definitions from Section~\ref{sec:typicality-preliminaries}. For decoding, we choose the \gls{povm} $(\decoderpovm_\codewordIndex)_{\codewordIndex=1}^\codebookSize$ defined as
\begin{equation*}
\decoderpovm_\codewordIndex
:=
\sqrt{
  \sum_{\hat{\codewordIndex}=1}^\codebookSize
    \typicalityOperatorProductDecoder_{\typicalityParameter,\blocklength}
      (\codebook(\hat{\codewordIndex}))
}^{-1}
\typicalityOperatorProductDecoder_{\typicalityParameter,\blocklength}
      (\codebook(\codewordIndex))
\sqrt{
  \sum_{\hat{\codewordIndex}=1}^\codebookSize
    \typicalityOperatorProductDecoder_{\typicalityParameter,\blocklength}
      (\codebook(\hat{\codewordIndex}))
}^{-1},
\end{equation*}
where we use $\typicalityOperatorProductDecoder_{\typicalityParameter,\blocklength}$ defined in (\ref{eq:typicality-operator-product-decoder}). We note that
\begin{equation*}
\dim \image\left(\sum_{\hat{\codewordIndex}=1}^\codebookSize
    \typicalityOperatorProductDecoder_{\typicalityParameter,\blocklength}
      (\codebook(\hat{\codewordIndex}))\right)< \infty, 
\end{equation*}
by Lemma \ref{lemma:q-resolvability-typical-terms}-\ref{item:q-resolvability-typical-terms-output-trace} which implies that
\begin{equation*}
\image\left(\sum_{\hat{\codewordIndex}=1}^\codebookSize
    \typicalityOperatorProductDecoder_{\typicalityParameter,\blocklength}
      (\codebook(\hat{\codewordIndex}))\right)
\end{equation*}
is closed. Consequently, by the first statement of Lemma \ref{lemma:hninequality} we have
\begin{equation*}
\sqrt{
  \sum_{\hat{\codewordIndex}=1}^\codebookSize
    \typicalityOperatorProductDecoder_{\typicalityParameter,\blocklength}
      (\codebook(\hat{\codewordIndex}))
}^{-1}
\in \boundedOperators{\hilbertSpace}.
\end{equation*}

The first part of the statement of Theorem~\ref{theorem:decoding}, namely that the $\decoderpovm_\codewordIndex$ are measurable functions of $\codebook$, is proven in the next two lemmas. This measurability is also essential for the proof of the remainder of Theorem~\ref{theorem:decoding}.

\begin{lemma}
\label{lemma:pseudoinverse-measurable}
Let $\hilbertSpace$ be a finite-dimensional complex Hilbert space. Then the function
\[
\pseudoinverse{\cdot}:~
\boundedOperators{\hilbertSpace}
\rightarrow
\boundedOperators{\hilbertSpace}
,~~
\generalOperator
\mapsto
\pseudoinverse{\generalOperator}
\]
which maps every operator to its Moore-Penrose pseudoinverse is measurable.
\end{lemma}
\begin{proof}
We can represent the Moore-Penrose pseudoinverse as a limit (see \cite[Chapter 3, Ex. 25]{ben2003generalized})
\[
\pseudoinverse{\generalOperator}
=
\lim_{\generalNatural \rightarrow \infty}
  \left(
    \adjoint{\generalOperator} \generalOperator
    +
    \frac{1}{\generalNatural} \identityOperator
  \right)^{-1}
  \adjoint{\generalOperator}.
\]
Matrix inversion is continuous (see~\cite[Chapter 6, eq. (127)]{ben2003generalized}), and so are addition and multiplication. Therefore, $\pseudoinverse{\cdot}$ is represented as a pointwise limit of continuous (and therefore measurable) functions, hence it is measurable.
\end{proof}

\begin{lemma}
For every $\codewordIndex \in \{1, \dots, \codebookSize\}$,
$
\inputAlphabet^{\blocklength\codebookSize}
\rightarrow
\traceClass{\hilbertSpace}
,~~
\codebook \mapsto \decoderpovm_\codewordIndex
$
is measurable.
\end{lemma}
\begin{proof}
Clearly,
\[
\sum_{\hat{\codewordIndex}=1}^\codebookSize
  \typicalityOperatorProductDecoder_{\typicalityParameter,\blocklength}
    (\codebook(\hat{\codewordIndex}))
\]
is a measurable function of $\codebook$, so by Lemma~\ref{lemma:continuous-functions-operator-measurable}, its square root is also measurable. Denote the restriction of
\[
\sqrt{
  \sum_{\hat{\codewordIndex}=1}^\codebookSize
    \typicalityOperatorProductDecoder_{\typicalityParameter,\blocklength}
      (\codebook(\hat{\codewordIndex}))
}
\]
to $\image (\outputTypicalityOperator_{\typicalityParameter, \blocklength})$ by $\generalOperator$. It can be seen in (\ref{eq:typicality-operator-product-decoder}) that $\generalOperator: \image (\outputTypicalityOperator_{\typicalityParameter, \blocklength}) \rightarrow \image (\outputTypicalityOperator_{\typicalityParameter, \blocklength})$, and that we can write
\begin{equation}
\label{eq:inverse-representation}
\sqrt{
  \sum_{\hat{\codewordIndex}=1}^\codebookSize
    \typicalityOperatorProductDecoder_{\typicalityParameter,\blocklength}
      (\codebook(\hat{\codewordIndex}))
}^{-1}
(\hilbertSpaceElement)
=
\begin{cases}
\pseudoinverse{\generalOperator} \hilbertSpaceElement, &\hilbertSpaceElement \in \image (\outputTypicalityOperator_{\typicalityParameter, \blocklength}) \\
0, &\text{otherwise.}
\end{cases}
\end{equation}
By Lemma~\ref{lemma:q-resolvability-typical-terms}-\ref{item:q-resolvability-typical-terms-output-trace}, $\image (\outputTypicalityOperator_{\typicalityParameter, \blocklength})$ is finite-dimensional, so we may apply Lemma~\ref{lemma:pseudoinverse-measurable} to argue that the operator represented in (\ref{eq:inverse-representation}) is a measurable function of $\codebook$. Since all norms are equivalent on finite-dimensional spaces, this measurability also applies with respect to the trace norm. The measurability of $\decoderpovm_\codewordIndex$ is then a straightforward consequence of Lemma~\ref{lemma:operator-compositions-continuous}.
\end{proof}

\begin{proof}[Proof of Theorem~\ref{theorem:decoding}]
Clearly, $\decoderpovm_\codewordIndex \geq 0$ and $\decoderpovm_1 + \dots + \decoderpovm_\codebookSize = \identityOperator$, so $(\decoderpovm_\codewordIndex)_{\codewordIndex=1}^\codebookSize$ is a $\{1, \dots, \codebookSize\}$-valued \gls{povm}. We have for the decoding error
\begin{align}
\nonumber
&\hphantom{{}={}}
\frac{1}{\codebookSize}
\sum_{\codewordIndex=1}^\codebookSize
  \trace\left(
    \classicalQuantumChannel(\codebook(\codewordIndex))
    (\identityOperator - \decoderpovm_\codewordIndex)    
  \right)
\\
\nonumber
&=
\frac{1}{\codebookSize}
\sum_{\codewordIndex=1}^\codebookSize
  \trace\left(
    \sqrt{\classicalQuantumChannel(\codebook(\codewordIndex))}
    \left(
      \identityOperator
      -
      \sqrt{
        \sum_{\hat{\codewordIndex}=1}^\codebookSize
          \typicalityOperatorProductDecoder_{\typicalityParameter,\blocklength}
            (\codebook(\hat{\codewordIndex}))
      }^{-1}
      \typicalityOperatorProductDecoder_{\typicalityParameter,\blocklength}
            (\codebook(\codewordIndex))
      \sqrt{
        \sum_{\hat{\codewordIndex}=1}^\codebookSize
          \typicalityOperatorProductDecoder_{\typicalityParameter,\blocklength}
            (\codebook(\hat{\codewordIndex}))
      }^{-1}
    \right)
    \sqrt{\classicalQuantumChannel(\codebook(\codewordIndex))}
  \right)
\\
\label{eq:decoding-typical-split}
\overset{(a)}&{\leq}
\frac{2}{\codebookSize}
\sum_{\codewordIndex=1}^\codebookSize
  \trace\Big(
    \classicalQuantumChannel(\codebook(\codewordIndex))
    \big(\identityOperator-\typicalityOperatorProductDecoder_{\typicalityParameter,\blocklength}(\codebook(\codewordIndex))\big)
  \Big)
+
\frac{4}{\codebookSize}
\sum_{\codewordIndex=1}^\codebookSize
  \trace\left(
    \classicalQuantumChannel(\codebook(\codewordIndex))
    \sum_{\substack{\hat{\codewordIndex}=1\\ \hat{\codewordIndex} \neq \codewordIndex}}^\codebookSize
      \typicalityOperatorProductDecoder_{\typicalityParameter,\blocklength}
        (\codebook(\hat{\codewordIndex}))
  \right),
\end{align}
where (a) is an application of Lemma~\ref{lemma:hninequality} with
\[
\generalOperator
:=
\typicalityOperatorProductDecoder_{\typicalityParameter,\blocklength}
  (\codebook(\codewordIndex))
,~~
\generalOperatorTwo
:=
\sum_{\substack{\hat{\codewordIndex}=1\\ \hat{\codewordIndex} \neq \codewordIndex}}^\codebookSize
  \typicalityOperatorProductDecoder_{\typicalityParameter,\blocklength}
    (\codebook(\hat{\codewordIndex}))
,~~
\hnconstant := 1,
\]
where $\image(\generalOperator+\generalOperatorTwo)$ is closed because it is finite-dimensional by Lemma~\ref{lemma:q-resolvability-typical-terms}-\ref{item:q-resolvability-typical-terms-output-trace}.

As a prerequisite to bounding the expectation of the second summand in (\ref{eq:decoding-typical-split}), we calculate
\begin{align*}
\trace\left(
  \classicalQuantumChannel_\inputDistribution^{\otimes \blocklength}
  \typicalityOperatorProductDecoder_{\typicalityParameter,\blocklength}
    (\codebook(\codewordIndex))
\right)
\overset{(\ref{eq:typicality-operator-product-decoder})}&{=}
\trace\left(
  \classicalQuantumChannel_\inputDistribution^{\otimes \blocklength}
  \outputTypicalityOperator_{\typicalityParameter,\blocklength}
  \jointTypicalityOperator_{\typicalityParameter,\blocklength}(\codebook(\codewordIndex))
  \outputTypicalityOperator_{\typicalityParameter,\blocklength}
\right)
\\
\overset{(a)}&{=}
\trace\left(
  \jointTypicalityOperator_{\typicalityParameter,\blocklength}(\codebook(\codewordIndex))
  \outputTypicalityOperator_{\typicalityParameter,\blocklength}
  \classicalQuantumChannel_\inputDistribution^{\otimes \blocklength}
  \outputTypicalityOperator_{\typicalityParameter,\blocklength}
  \jointTypicalityOperator_{\typicalityParameter,\blocklength}(\codebook(\codewordIndex))
\right)
\\
\overset{(b)}&{\leq}
\exp\big(
  -
  \blocklength
  (\outputEntropy - \typicalityParameter)
\big)
\trace\left(
  \jointTypicalityOperator_{\typicalityParameter,\blocklength}(\codebook(\codewordIndex))
  \outputTypicalityOperator_{\typicalityParameter,\blocklength}
  \jointTypicalityOperator_{\typicalityParameter,\blocklength}(\codebook(\codewordIndex))
\right)
\\
\overset{(c)}&{\leq}
\exp\big(
  -
  \blocklength
  (\outputEntropy - \typicalityParameter)
\big)
\trace\left(
  \jointTypicalityOperator_{\typicalityParameter,\blocklength}(\codebook(\codewordIndex))
\right)
\\
\overset{(d)}&{\leq}
\exp\big(
  -
  \blocklength
  (\outputEntropy - \jointEntropy - 2\typicalityParameter)
\big)
\\
\overset{(\ref{eq:quantum-information})}&{=}
\exp\big(
  -
  \blocklength
  (\quantumInformation - 2\typicalityParameter)
\big),
\end{align*}
where (a) is by (\ref{eq:typicality-operators-basic-properties}) and the cyclic property of the trace, (b) is by Lemma~\ref{lemma:q-resolvability-typical-terms}-\ref{item:q-resolvability-typical-terms-output}, (c) is due to $\outputTypicalityOperator_{\typicalityParameter,\blocklength} \leq \identityOperator$ and (\ref{eq:typicality-operators-basic-properties}), and (d) is by Lemma~\ref{lemma:q-resolvability-typical-terms}-\ref{item:q-resolvability-typical-terms-joint-trace}. We can use this in conjunction with the independence of codewords and bound
\begin{align}
\nonumber
\Expectation_\codebook
  \trace\left(
    \classicalQuantumChannel(\codebook(\codewordIndex))
    \sum_{\substack{\hat{\codewordIndex}=1\\ \hat{\codewordIndex} \neq \codewordIndex}}^\codebookSize
      \typicalityOperatorProductDecoder_{\typicalityParameter,\blocklength}
        (\codebook(\hat{\codewordIndex}))
  \right)
&=
\Expectation_\codebook
  \trace\left(
    \classicalQuantumChannel_{\inputDistribution}^{\otimes \blocklength}
    \sum_{\substack{\hat{\codewordIndex}=1\\ \hat{\codewordIndex} \neq \codewordIndex}}^\codebookSize
      \typicalityOperatorProductDecoder_{\typicalityParameter,\blocklength}
        (\codebook(\hat{\codewordIndex}))
  \right)
\\
\nonumber
&\leq
\codebookSize
\exp\big(
  -
  \blocklength
  (\quantumInformation - 2\typicalityParameter)
\big)
\\
\label{eq:decoding-typical-expectation}
&\leq
\exp\big(
  -
  \blocklength
  (\quantumInformation - \codebookRate - 2\typicalityParameter)
\big).
\end{align}

Next, we apply $\Expectation_\codebook$ in (\ref{eq:decoding-typical-split}). We use Lemma~\ref{lemma:q-resolvability-atypical-terms} in the first summand and (\ref{eq:decoding-typical-expectation}) in the second summand which yields \eqref{eq:decoding-error-all-blocklengths}. Next, we note that the infimum in \eqref{eq:codingBoundFunc} is upper bounded by the realization for any fixed $\typicalityParameter$. So we pick any $\typicalityParameter \in (0,(\quantumInformation - \codebookRate)/2)$, then invoke Lemma~\ref{lemma:q-resolvability-atypical-terms} and fix some $\finalconst_1$ which satisfies \eqref{eq:atypicalTermsFunc-bounds}. With the choice $\finalconst \in (0, \min(\finalconst_1, \quantumInformation - \codebookRate - 2\typicalityParameter))$, this proves \eqref{eq:decoding-error-large-blocklengths}.
\end{proof}

\subsection{Concentration of Error}
\label{sec:proof-concentration}
Theorems~\ref{theorem:q-resolvability} and~\ref{theorem:decoding} are formulated in terms of expectation, but both the decoding error and the trace distance from the ideal output distribution are concentrated around their mean, as can be seen in the following corollaries.

\begin{cor}
\label{cor:q-resolvability}
Under the assumptions of Theorem~\ref{theorem:q-resolvability} and for every $\finalconst \in (0,\codebookRate/2)$, we have
\begin{equation}
\label{eq:q-resolvability-cor-statement-all-blocklengths}
\Probability_\codebook\left(
  \traceNorm{
    \classicalQuantumChannel_\codebook
    -
    \classicalQuantumChannel_\inputDistribution^{\otimes \blocklength}
  }
  \geq
  \resolvabilityBoundFunc(\codebookRate,\blocklength)
  +
  \exp(-\finalconst\blocklength)
\right)
\leq
\exp\left(
  -\frac{1}{2}\exp(
    \blocklength
    (\codebookRate-2\finalconst)
  )
\right).
\end{equation}

Furthermore, \eqref{eq:q-resolvability-cor-statement-all-blocklengths} implies that
$
\traceNorm{
  \classicalQuantumChannel_\codebook
  -
  \classicalQuantumChannel_\inputDistribution^{\otimes \blocklength}
}
$
tends to $0$ exponentially with a doubly exponentially small error probability as $\blocklength \rightarrow \infty$. That is, there are $\finalconst_1,\finalconst_2 \in (0,\infty)$ such that \eqref{eq:q-resolvability-cor-statement-all-blocklengths} implies for all sufficiently large $\blocklength$ that
\begin{equation}
\label{eq:q-resolvability-cor-statement-large-blocklengths}
\Probability_\codebook\left(
  \traceNorm{
    \classicalQuantumChannel_\codebook
    -
    \classicalQuantumChannel_\inputDistribution^{\otimes \blocklength}
  }
  \geq
  \exp(-\finalconst_1\blocklength)
\right)
\leq
\exp(
  -\exp(
    \finalconst_2
    \blocklength
  )
).
\end{equation}
\end{cor}

\begin{cor}
\label{cor:decoding}
Under the assumptions of Theorem~\ref{theorem:decoding} and for any choice of $\finalconst_1 \in (0,\infty)$, we have
\begin{equation}
\label{eq:decoding-cor-all-blocklengths}
\Probability_\codebook\left(
  \frac{1}{\codebookSize}
  \sum_{\codewordIndex=1}^\codebookSize
    \trace\left(
      \classicalQuantumChannel^\blocklength(\codewordIndex) (\identityOperator-\decoderpovm_\codewordIndex)
    \right)
  \geq
  \exp(-\blocklength\finalconst_1)
\right)
\leq
\exp(\blocklength\finalconst_1)
\codingBoundFunc(\codebookRate,\blocklength).
\end{equation}
Furthermore, \eqref{eq:decoding-cor-all-blocklengths} implies exponentially small decoding error with exponentially small error. That is, there is a suitable choice for $\finalconst_1$ and some $\finalconst_2 \in (0,\infty)$ such that for sufficiently large $\blocklength$,
\begin{equation}
\label{eq:decoding-cor-large-blocklengths}
\Probability_\codebook\left(
  \frac{1}{\codebookSize}
  \sum_{\codewordIndex=1}^\codebookSize
    \trace\left(
      \classicalQuantumChannel^\blocklength(\codewordIndex) (\identityOperator-\decoderpovm_\codewordIndex)
    \right)
  \geq
  \exp(-\blocklength\finalconst_1)
\right)
\leq
\exp(-\blocklength\finalconst_2).
\end{equation}
\end{cor}

The proof of Corollary~\ref{cor:q-resolvability} is essentially an application of the bounded differences inequality~\cite{mcdiarmid1989method}. For the reader's convenience, we reproduce the result here in the form that we will be using.

\begin{theorem}
\label{theorem:bounded-differences}
\emph{(Bounded differences inequality as stated in~\cite[Theorem 6.2]{boucheron2013concentration}.)}
Let $\generalProbabilitySpace$ be a measurable space, and let $f: \generalProbabilitySpace^\generalIndexMax \rightarrow \reals$ be measurable. Assume that there are nonnegative constants $\generalConstant_1, \dots, \generalConstant_\generalIndexMax$ with the property
\begin{equation}
\label{eq:bounded-differences}
\forall \generalIndex \in \{1, \dots, \generalIndexMax\}~
\sup\Big\{
  \absolute{
    \generalFunction(\generalProbabilitySpaceElement_1, \dots, \generalProbabilitySpaceElement_\generalIndexMax)
    -
    \generalFunction(
      \generalProbabilitySpaceElement_1, \dots, \generalProbabilitySpaceElement_{\generalIndex - 1},
      \generalProbabilitySpaceElement_\generalIndex',
      \generalProbabilitySpaceElement_{\generalIndex+1}, \dots, \generalProbabilitySpaceElement_\generalIndexMax
    )
  }
  :~
  \generalProbabilitySpaceElement_1, \dots, \generalProbabilitySpaceElement_\generalIndexMax, \generalProbabilitySpaceElement_\generalIndex'
  \in
  \generalProbabilitySpace
\Big\}
\leq
\generalConstant_\generalIndex,
\end{equation}
and denote
\[
\differenceBounds := \frac{1}{4} \sum_{\generalIndex=1}^\generalIndexMax \generalConstant_\generalIndex^2.
\]
Let $\generalRV_1, \dots, \generalRV_\generalIndexMax$ be independent random variables such that $\Expectation \generalFunction(\generalRV_1, \dots, \generalRV_\generalIndexMax)$ exists. Then, for any $\generalReal \in (0, \infty)$,
\[
\Probability\left(
  \generalFunction(\generalRV_1, \dots, \generalRV_\generalIndexMax)
  -
  \Expectation \generalFunction(\generalRV_1, \dots, \generalRV_\generalIndexMax)
  >
  \generalReal
\right)
\leq
\exp\left(
  -\frac{\generalReal^2}{2\differenceBounds}
\right).
\]
\end{theorem}
Condition (\ref{eq:bounded-differences}) is called the \emph{bounded differences property} which gives Theorem~\ref{theorem:bounded-differences} its name.
\begin{proof}[Proof of Corollary~\ref{cor:q-resolvability}]
We apply Theorem~\ref{theorem:bounded-differences} to the function
\[
\generalFunction
:
(\codebook(1), \dots, \codebook(\codebookSize))
\mapsto
\traceNorm{
  \classicalQuantumChannel_\codebook
  -
  \classicalQuantumChannel_\inputDistribution^{\otimes \blocklength}
},
\]
so we first need to compute $\generalConstant_{\hat{\codewordIndex}}$ with
\[
\sup\left\{
  \vphantom{\bigg(}
  \absolute{
    \vphantom{\Big(}
      \traceNorm{
        \classicalQuantumChannel_\codebook
        -
        \classicalQuantumChannel_\inputDistribution^{\otimes \blocklength}
      }
      -
      \traceNorm{
        \classicalQuantumChannel_{\codebook'}
        -
        \classicalQuantumChannel_\inputDistribution^{\otimes \blocklength}
      }
  }
  :~
  \forall
  \codewordIndex
  \in
  \{
    1,
    \dots,
    \hat{\codewordIndex}-1,
    \hat{\codewordIndex}+1,
    \dots,
    \codebookSize
  \}
  :~
  \codebook(\codewordIndex)
  =
  \codebook'(\codewordIndex)
\right\}
\leq
\generalConstant_{\hat{\codewordIndex}}.
\]
To this end, we calculate
\begin{align*}
\absolute{
  \vphantom{\Big(}
    \traceNorm{
      \classicalQuantumChannel_\codebook
      -
      \classicalQuantumChannel_\inputDistribution^{\otimes \blocklength}
    }
    -
    \traceNorm{
      \classicalQuantumChannel_{\codebook'}
      -
      \classicalQuantumChannel_\inputDistribution^{\otimes \blocklength}
    }
}
\overset{(a)}&{\leq}
\traceNorm{
  \classicalQuantumChannel_\codebook
  -
  \classicalQuantumChannel_\inputDistribution^{\otimes \blocklength}
  -
  \classicalQuantumChannel_{\codebook'}
  +
  \classicalQuantumChannel_\inputDistribution^{\otimes \blocklength}
}
\\
\overset{(\ref{eq:codebook-channel-output-definition})}&{=}
\traceNorm{
  \frac{1}{\codebookSize}
  \sum_{\codewordIndex=1}^\codebookSize\Big(
    \classicalQuantumChannel^\blocklength(\codebook(\codewordIndex))
    -
    \classicalQuantumChannel^\blocklength(\codebook'(\codewordIndex))
  \Big)
}
\\
\overset{(b)}&{=}
\frac{1}{\codebookSize}
\traceNorm{
  \classicalQuantumChannel^\blocklength(\codebook(\hat{\codewordIndex}))
  -
  \classicalQuantumChannel^\blocklength(\codebook'(\hat{\codewordIndex}))
}
\\
\overset{(c)}&{\leq}
\frac{2}{\codebookSize},
\end{align*}
where (a) is due to the triangle inequality, (b) is because $\codebook(\codewordIndex) = \codebook'(\codewordIndex)$ whenever $\codewordIndex \neq \hat{\codewordIndex}$, and (c) is an application of the triangle inequality and the fact $\classicalQuantumChannel$ is a map to $\densityOperators{\hilbertSpace}$.

This means that we can choose $\generalConstant_{\hat{\codewordIndex }}:= 2/\codebookSize$ for all $\hat{\codewordIndex}$. Consequently, we obtain
\[
\differenceBounds
=
\frac{1}{4}
\sum_{\hat{\codewordIndex}=1}^\codebookSize
  \frac{2^2}{\codebookSize^2}
=
\codebookSize^{-1}.
\]

Hence,
\begin{align*}
\Probability_\codebook\left(
  \traceNorm{
    \classicalQuantumChannel_\codebook
    -
    \classicalQuantumChannel_\inputDistribution^{\otimes \blocklength}
  }
  \geq
  \resolvabilityBoundFunc(\codebookRate,\blocklength)
  +
  \exp(-\finalconst\blocklength)
\right)
\overset{(a)}&{\leq}
\Probability_\codebook\left(
  \traceNorm{
    \classicalQuantumChannel_\codebook
    -
    \classicalQuantumChannel_\inputDistribution^{\otimes \blocklength}
  }
  -
  \Expectation_\codebook
    \traceNorm{
      \classicalQuantumChannel_\codebook
      -
      \classicalQuantumChannel_\inputDistribution^{\otimes \blocklength}
    }
  \geq
  \exp(-\finalconst\blocklength)
\right)
\\
\overset{(b)}&{\leq}
\exp\left(
  -\frac{1}{2}\exp(-2\finalconst\blocklength)\codebookSize
\right),
\end{align*}
where (a) is the application of Theorem~\ref{theorem:q-resolvability} and (b) the application of Theorem~\ref{theorem:bounded-differences}. \eqref{eq:q-resolvability-cor-statement-all-blocklengths} then follows from $\codebookSize \geq \exp(\blocklength\codebookRate)$. In order to prove \eqref{eq:q-resolvability-cor-statement-large-blocklengths}, we invoke \eqref{eq:q-resolvability-large-blocklengths} of Theorem~\ref{theorem:q-resolvability} to find $\proofconst \in (0,\infty)$ such that $\resolvabilityBoundFunc(\codebookRate,\blocklength) \leq \exp(-\proofconst\blocklength)$ for large enough $\blocklength$. Then, fixing some $\finalconst_1 \in (0,\min(\proofconst,\finalconst))$, we have for sufficiently large $\blocklength$,
\begin{align*}
\Probability_\codebook\left(
  \traceNorm{
    \classicalQuantumChannel_\codebook
    -
    \classicalQuantumChannel_\inputDistribution^{\otimes \blocklength}
  }
  \geq
  \exp(-\finalconst_1\blocklength)
\right)
&\leq
\Probability_\codebook\left(
  \traceNorm{
    \classicalQuantumChannel_\codebook
    -
    \classicalQuantumChannel_\inputDistribution^{\otimes \blocklength}
  }
  \geq
  \exp(-\proofconst\blocklength)
  +
  \exp(-\finalconst\blocklength)
\right)
\\
&\leq
\Probability_\codebook\left(
  \traceNorm{
    \classicalQuantumChannel_\codebook
    -
    \classicalQuantumChannel_\inputDistribution^{\otimes \blocklength}
  }
  \geq
  \resolvabilityBoundFunc(\codebookRate,\blocklength)
  +
  \exp(-\finalconst\blocklength)
\right)
\\
\overset{\eqref{eq:q-resolvability-cor-statement-all-blocklengths}}&{\leq}
\exp\left(
  -\frac{1}{2}\exp(
    \blocklength
    (\codebookRate-2\finalconst)
  )
\right).
\end{align*}
This proves \eqref{eq:q-resolvability-cor-statement-large-blocklengths} with any choice of $\finalconst_2 \in (0,\codebookRate-2\finalconst)$.
\end{proof}

\begin{proof}[Proof of Corollary~\ref{cor:decoding}]
\eqref{eq:decoding-cor-all-blocklengths} is an immediate consequence of Theorem~\ref{theorem:decoding} and Markov's inequality. In order to argue \eqref{eq:decoding-cor-large-blocklengths}, we first invoke Theorem~\ref{theorem:decoding} to obtain $\finalconst \in (0,\infty)$ such that the \eqref{eq:decoding-error-large-blocklengths} holds. \eqref{eq:decoding-cor-large-blocklengths} then follows from \eqref{eq:decoding-cor-all-blocklengths} with the choices $\finalconst_1 \in (0,\finalconst)$ and $\finalconst_2 := \finalconst - \finalconst_1$.
\end{proof}

\subsection{Cost constraint}
\label{sec:proofs-cost-constraint}
In this section, we extend Corollaries~\ref{cor:q-resolvability} and~\ref{cor:decoding} to a cost-constrained version $\codebook_{\costFunction,\costConstraint}$ of the codebook $\codebook$, where $(\costFunction,\costConstraint)$ is an additive cost constraint compatible with the input distribution $\inputDistribution$ chosen for the generation of $\codebook$.

As long as there exists at least one $\inputAlphabetElement^\blocklength \in \inputAlphabet^\blocklength$ which satisfies the cost constraint (which is always the case if there is an input distribution compatible with the cost constraint), we can (given any codebook $\codebook$) define the \emph{cost-constrained codebook} $\codebook_{\costFunction,\costConstraint}$ via
\[
\codebook_{\costFunction,\costConstraint}(\codewordIndex)
=
\begin{cases}
\codebook(\codewordIndex),
&\codebook(\codewordIndex) \text{ satisfies the cost constraint } (\costFunction,\costConstraint),
\\
\inputAlphabetElement^\blocklength, &\text{otherwise.}
\end{cases}
\]

We define a set of \emph{bad codeword indices}
\[
\badCodewordsSet
:=
\left\{
  \codewordIndex \in \{1, \dots, \codebookSize\}
  :~~
  \codebook(\codewordIndex)
  \neq
  \codebook_{\costFunction,\costConstraint}(\codewordIndex)
\right\}.
\]

It is possible with methods similar to the ones used in~\cite[Section 3.3]{elgamal2011network} to bound the probability that the codebook contains a large number of bad code words as in the following lemma. For completeness, we give a full proof in Appendix~\ref{appendix:qit}.

\begin{lemma}
\label{lemma:bad-codewords}
Let $\codebook$ be a random codebook generated from a channel input distribution which is compatible with the additive cost constraint ${(\costFunction,\costConstraint)}$. Let $\proofconst_1$ be as defined in \eqref{eq:mgf-parameter-choice}. Then, $\proofconst_1 > 0$ and for all $\proofconst \in (0,\proofconst_1)$, we have
\begin{equation}
\label{eq:bad-codewords-all-blocklengths}
\Probability_\codebook\left(
  \cardinality{\badCodewordsSet}
  \geq
  \codebookSize
  \exp(-\blocklength\proofconst)
\right)
\leq
\costBoundFunc^{(\costFunction,\costConstraint)}\left(
  \proofconst,
  \frac{\log \codebookSize}{\blocklength},
  \blocklength
\right),
\end{equation}
with $\costBoundFunc^{(\costFunction,\costConstraint)}$ defined in \eqref{eq:costBoundFunc}. Furthermore, for all $\codebookRate \in (0,\infty)$, $\proofconst \in (0,\min(\codebookRate/2,\proofconst_1)$, $\costBoundFunc^{(\costFunction,\costConstraint)}(\proofconst,\codebookRate,\blocklength)$ tends to $0$ doubly exponentially fast as $\blocklength$ tends to $\infty$. That is, for every $\codebookRate$ and $\proofconst \in (0,\min(\codebookRate/2,\proofconst_1)$, $\finalconst \in (0,\codebookRate-2\proofconst)$, we have, for sufficiently large $\blocklength$,
\begin{equation}
\label{eq:bad-codewords-large-blocklengths}
\costBoundFunc^{(\costFunction,\costConstraint)}\left(
  \proofconst,
  \codebookRate,
  \blocklength
\right)
\leq
\exp(-\exp(\finalconst\blocklength)).
\end{equation}
\end{lemma}

\begin{cor}
\label{cor:q-resolvability-constrained}
Make the same assumptions as in Theorem~\ref{theorem:q-resolvability}, and let $(\costFunction, \costConstraint)$ be an additive cost constraint which is compatible with the input distribution $\inputDistribution$ and induces the cost-constrained random codebook $\codebook_{\costFunction,\costConstraint}$. Then, we have for $\proofconst_1$ chosen as in \eqref{eq:mgf-parameter-choice} and all $\proofconst_2 \in (0,\min(\proofconst_1,\codebookRate/2)), \proofconst_3 \in (0,\codebookRate/2)$
\begin{equation}
\label{eq:q-resolvability-constrained-all-blocklengths}
\hphantom{{}={}}
\Probability_\codebook\bigg(
  \traceNorm{
    \classicalQuantumChannel_{\codebook_{\costFunction,\costConstraint}}
    -
    \classicalQuantumChannel_\inputDistribution^{\otimes \blocklength}
  }
  \geq
  \resolvabilityBoundFunc(\codebookRate,\blocklength)
  +
  2\exp(-\proofconst_2\blocklength)
  +
  \exp(-\proofconst_3\blocklength)
\bigg)
\leq
\costBoundFunc^{(\costFunction,\costConstraint)}(
  \proofconst_2,
  \codebookRate,
  \blocklength
)
+
\exp\left(
  -\frac{1}{2}\exp(
    \blocklength
    (\codebookRate-2\proofconst_3)
  )
\right)
\end{equation}

Furthermore, \eqref{eq:q-resolvability-constrained-all-blocklengths} implies that
$
\traceNorm{
  \classicalQuantumChannel_{\codebook_{\costFunction,\costConstraint}}
  -
  \classicalQuantumChannel_\inputDistribution^{\otimes \blocklength}
}
$
tends to $0$ exponentially with a doubly exponentially small error probability as $\blocklength \rightarrow \infty$. That is, there are $\finalconst_1,\finalconst_2 \in (0,\infty)$ such that \eqref{eq:q-resolvability-constrained-all-blocklengths} implies for all sufficiently large $\blocklength$ that
\begin{equation}
\label{eq:q-resolvability-constrained-large-blocklengths}
\Probability_\codebook\left(
  \traceNorm{
    \classicalQuantumChannel_{\codebook_{\costFunction,\costConstraint}}
    -
    \classicalQuantumChannel_\inputDistribution^{\otimes \blocklength}
  }
  \geq
  \exp(-\finalconst_1\blocklength)
\right)
\leq
\exp(
  -\exp(
    \finalconst_2
    \blocklength
  )
).
\end{equation}
\end{cor}

\begin{proof}
We will bound $\traceNorm{\classicalQuantumChannel_\codebook - \classicalQuantumChannel_{\codebook_{\costFunction,\costConstraint}}}$ in such a way that this corollary follows as an immediate consequence of Corollary~\ref{cor:q-resolvability}.

Conditioned on the event
$
\cardinality{\badCodewordsSet}
<
\codebookSize
\exp(-\blocklength\proofconst_2)
$,
we have almost surely
\begin{align}
\nonumber
\traceNorm{
  \classicalQuantumChannel_\codebook - \classicalQuantumChannel_{\codebook_{\costFunction,\costConstraint}}
}
&=
\traceNorm{
  \frac{1}{\codebookSize}
  \sum_{\codewordIndex\in\badCodewordsSet}
    \Big(
      \classicalQuantumChannel^\blocklength\big(
        \codebook(\codewordIndex)
      \big)
      -
      \classicalQuantumChannel^\blocklength\big(
        \codebook_{\costFunction,\costConstraint}(\codewordIndex)
      \big)
    \Big)
}
\\
\nonumber
\overset{(a)}&{\leq}
\frac{1}{\codebookSize}
\sum_{\codewordIndex\in\badCodewordsSet}\left(
  \traceNorm{
    \classicalQuantumChannel^\blocklength\big(
      \codebook(\codewordIndex)
    \big)
  }
  +
  \traceNorm{
    \classicalQuantumChannel^\blocklength\big(
      \codebook_{\costFunction,\costConstraint}(\codewordIndex)
    \big)
  }
\right)
\displaybreak[0] \\
\nonumber
&=
\frac{2\cardinality{\badCodewordsSet}}{\codebookSize}
\\
\label{eq:resolvability-cost-constraint-bad-codeword-conditioned}
&<
2\exp(-\blocklength\proofconst_2),
\end{align}
where step (a) is due to the triangle inequality. Consequently, we obtain
\begin{align*}
&\hphantom{{}={}}
\Probability_\codebook\bigg(
  \traceNorm{
    \classicalQuantumChannel_{\codebook_{\costFunction,\costConstraint}}
    -
    \classicalQuantumChannel_\inputDistribution^{\otimes \blocklength}
  }
  \geq
  \resolvabilityBoundFunc(\codebookRate,\blocklength)
  +
  2\exp(-\proofconst_2\blocklength)
  +
  \exp(-\proofconst_3\blocklength)
\bigg)
\\
\overset{(a)}&{\leq}
\Probability_\codebook\bigg(
  \traceNorm{
    \classicalQuantumChannel_{\codebook_{\costFunction,\costConstraint}}
    -
    \classicalQuantumChannel_{\codebook}
  }
  +
  \traceNorm{
    \classicalQuantumChannel_{\codebook}
    -
    \classicalQuantumChannel_\inputDistribution^{\otimes \blocklength}
  }
  \geq
  \resolvabilityBoundFunc(\codebookRate,\blocklength)
  +
  2\exp(-\proofconst_2\blocklength)
  +
  \exp(-\proofconst_3\blocklength)
\bigg)
\\
\overset{(b)}&{\leq}
\Probability_\codebook\left(
  \traceNorm{
    \classicalQuantumChannel_{\codebook_{\costFunction,\costConstraint}}
    -
    \classicalQuantumChannel_{\codebook}
  }
  \geq
  2\exp(-\proofconst_2\blocklength)
\right)
+
\Probability_\codebook\left(
  \traceNorm{
    \classicalQuantumChannel_{\codebook}
    -
    \classicalQuantumChannel_\inputDistribution^{\otimes \blocklength}
  }
  \geq
  \resolvabilityBoundFunc(\codebookRate,\blocklength)
  +
  \exp(-\proofconst_3\blocklength)
\right),
\end{align*}
where (a) is due to the triangle inequality and (b) is by the union bound. We conclude the proof of \eqref{eq:q-resolvability-constrained-all-blocklengths} by applying Lemma~\ref{lemma:bad-codewords} along with \eqref{eq:resolvability-cost-constraint-bad-codeword-conditioned} and $\codebookSize \geq \exp(\blocklength\codebookRate)$ in the first summand and \eqref{eq:q-resolvability-cor-statement-all-blocklengths} of Corollary~\ref{cor:q-resolvability} in the second summand.

In order to prove \eqref{eq:q-resolvability-constrained-large-blocklengths}, we invoke \eqref{eq:q-resolvability-large-blocklengths} of Theorem~\ref{theorem:q-resolvability} to find $\proofconst_4 \in (0,\infty)$ such that $\resolvabilityBoundFunc(\codebookRate,\blocklength) \leq \exp(-\proofconst_4 \blocklength)$ for large enough $\blocklength$. Then with any choice of $\finalconst_1 \in (0,\min(\proofconst_2,\proofconst_3,\proofconst_4)), \finalconst_2 \in (0,\min(\codebookRate-2\proofconst_2,\codebookRate-2\proofconst_3))$,we obtain for sufficiently large $\blocklength$
\begin{align*}
&\hphantom{{}={}}
\Probability_\codebook\bigg(
  \traceNorm{
    \classicalQuantumChannel_{\codebook_{\costFunction,\costConstraint}}
    -
    \classicalQuantumChannel_\inputDistribution^{\otimes \blocklength}
  }
  \geq
  \exp(-\blocklength\finalconst_1)
\bigg)
\\
&\leq
\Probability_\codebook\bigg(
  \traceNorm{
    \classicalQuantumChannel_{\codebook_{\costFunction,\costConstraint}}
    -
    \classicalQuantumChannel_\inputDistribution^{\otimes \blocklength}
  }
  \geq
  \exp(-\proofconst_4 \blocklength)
  +
  2\exp(-\proofconst_2\blocklength)
  +
  \exp(-\proofconst_3\blocklength)
\bigg)
\\
&\leq
\Probability_\codebook\bigg(
  \traceNorm{
    \classicalQuantumChannel_{\codebook_{\costFunction,\costConstraint}}
    -
    \classicalQuantumChannel_\inputDistribution^{\otimes \blocklength}
  }
  \geq
  \resolvabilityBoundFunc(\codebookRate,\blocklength)
  +
  2\exp(-\proofconst_2\blocklength)
  +
  \exp(-\proofconst_3\blocklength)
\bigg)
\\
\overset{\eqref{eq:q-resolvability-constrained-all-blocklengths}}&{\leq}
\costBoundFunc^{(\costFunction,\costConstraint)}(
  \proofconst_2,
  \codebookRate,
  \blocklength
)
+
\exp\left(
  -\frac{1}{2}\exp(
    \blocklength
    (\codebookRate-2\proofconst_3)
  )
\right)
\\
\overset{\eqref{eq:bad-codewords-large-blocklengths}}&{\leq}
\exp(-\exp(\blocklength\finalconst_2)).
\qedhere
\end{align*}
\end{proof}

\begin{cor}
\label{cor:decoding-constrained}
Make the same assumptions as in Theorem~\ref{theorem:decoding}, and let $(\costFunction,\costConstraint)$ be an additive cost constraint which is compatible with the input distribution $\inputDistribution$ and induces the cost-constrained random codebook $\codebook_{\costFunction,\costConstraint}$. Then, for each $\codebook$, there is a decoding \gls{povm} $(\decoderpovm_\codewordIndex)_{\codewordIndex=1}^\codebookSize$ such that every $\decoderpovm_\codewordIndex$ is measurable as a function of $\codebook$, and for all $\proofconst_2 \in (0,\infty) , \proofconst_3 \in (0,\proofconst_1)$ with $\proofconst_1$ defined in \eqref{eq:mgf-parameter-choice}, we have
\begin{equation}
\label{eq:decoding-constrained-all-blocklengths}
\Probability_\codebook\Bigg(
  \frac{1}{\codebookSize}
  \sum_{\codewordIndex=1}^\codebookSize
    \trace\left(
      \classicalQuantumChannel(\codebook_{\costFunction,\costConstraint}(\codewordIndex))
      (\identityOperator - \decoderpovm_\codewordIndex)
    \right)
  \geq
  \exp(-\blocklength\proofconst_2)
  +
  2\exp(-\blocklength\proofconst_3)
\Bigg)
\leq
\codingBoundFunc(\codebookRate,\blocklength)\exp(\blocklength\proofconst_2)
+
\costBoundFunc^{(\costFunction,\costConstraint)}\left(\proofconst_3,\frac{\log \codebookSize}{\blocklength},\blocklength\right).
\end{equation}
Furthermore, \eqref{eq:decoding-constrained-all-blocklengths} implies that if $\codebookSize \geq \exp(\blocklength\codebookRate_{\min})$ for some $\codebookRate_{\min} \in (0,\codebookRate]$ and all $\blocklength$, the decoding error tends to $0$ exponentially with an exponentially small error probability as $\blocklength \rightarrow \infty$. That is, there are $\finalconst_1, \finalconst_2 \in (0,\infty)$ such that for sufficiently large $\blocklength$,
\begin{equation}
\label{eq:decoding-constrained-large-blocklengths}
\Probability_\codebook\left(
  \frac{1}{\codebookSize}
  \sum_{\codewordIndex=1}^\codebookSize
    \trace\left(
      \classicalQuantumChannel(\codebook(\codewordIndex))
      (\identityOperator - \decoderpovm_\codewordIndex)
    \right)
  \geq
  \exp(-\blocklength\finalconst_1)
\right)
\leq
\exp(-\blocklength\finalconst_2).
\end{equation}
\end{cor}
\begin{proof}
We apply Corollary~\ref{cor:decoding} to obtain a \gls{povm} $(\decoderpovm_\codewordIndex)_{\codewordIndex=1}^\codebookSize$ which is measurable as a function of $\codebook$. Due to the union bound,
\begin{align}
\nonumber
&\hphantom{{}={}}
\Probability_\codebook\left(
  \frac{1}{\codebookSize}
  \sum_{\codewordIndex=1}^\codebookSize
    \trace\left(
      \classicalQuantumChannel(\codebook_{\costFunction,\costConstraint}(\codewordIndex))
      (\identityOperator - \decoderpovm_\codewordIndex)
    \right)
  \geq
  \exp(-\blocklength\proofconst_2)
  +
  2\exp(-\blocklength\proofconst_3)
\right)
\\
\label{eq:decoding-constrained-split}
&\leq
\begin{multlined}[t]
\Probability_\codebook\left(
  \frac{1}{\codebookSize}
  \sum_{\codewordIndex=1}^\codebookSize
    \trace\left(
      \classicalQuantumChannel(\codebook(\codewordIndex))
      (\identityOperator - \decoderpovm_\codewordIndex)
    \right)
  \geq
  \exp(-\blocklength\proofconst_2)
\right)
\\+
\Probability_\codebook\left(
  \frac{1}{\codebookSize}
  \sum_{\codewordIndex=1}^\codebookSize
    \trace\Big(
      \big(
        \classicalQuantumChannel(\codebook_{\costFunction,\costConstraint}(\codewordIndex))
        -
        \classicalQuantumChannel(\codebook(\codewordIndex))
      \big)
      \big(\identityOperator - \decoderpovm_\codewordIndex\big)
    \Big)
  \geq
  2\exp(-\blocklength\proofconst_3)
\right).
\end{multlined}
\end{align}
The first summand in \eqref{eq:decoding-constrained-split} can now be upper bounded by $\codingBoundFunc(\codebookRate,\blocklength)\exp(\blocklength\proofconst_2)$ due to Corollary~\ref{cor:decoding}. For the second summand, we note that
\begin{align*}
\trace\Big(
  \big(
    \classicalQuantumChannel(\codebook_{\costFunction,\costConstraint}(\codewordIndex))
    -
    \classicalQuantumChannel(\codebook(\codewordIndex))
  \big)
  \big(\identityOperator - \decoderpovm_\codewordIndex\big)
\Big)
\overset{(a)}&{\leq}
\traceNorm{
  \big(
    \classicalQuantumChannel(\codebook_{\costFunction,\costConstraint}(\codewordIndex))
    -
    \classicalQuantumChannel(\codebook(\codewordIndex))
  \big)
  \big(\identityOperator - \decoderpovm_\codewordIndex\big)
}
\\
\overset{(b)}&{\leq}
\traceNorm{
    \classicalQuantumChannel(\codebook_{\costFunction,\costConstraint}(\codewordIndex))
    -
    \classicalQuantumChannel(\codebook(\codewordIndex))
}
\operatorNorm{
  \identityOperator - \decoderpovm_\codewordIndex
}
\\
\overset{(c)}&{\leq}
2,
\end{align*}
where (a) is by Lemma~\ref{lemma:norm-basics}-\ref{item:norm-basics-trace-norm-duality}, (b) is by Lemma~\ref{lemma:norm-basics}-\ref{item:norm-basics-trace-norm-submultiplicative} and (c) follows because $\decoderpovm_\codewordIndex \leq \identityOperator$ and $\classicalQuantumChannel$ maps to $\densityOperators{\hilbertSpace}$. Hence, the second summand in \eqref{eq:decoding-constrained-split} is upper bounded by
\[
\Probability_\codebook\left(
  \frac{2\cardinality{\badCodewordsSet}}
       {\codebookSize}
  \geq
  2\exp(-\blocklength\proofconst_3)
\right)
\leq
\costBoundFunc^{(\costFunction,\costConstraint)}\left(\proofconst_3,\frac{\log \codebookSize}{\blocklength},\blocklength\right)
\]
where the inequality follows by Lemma~\ref{lemma:bad-codewords}. This concludes the proof of \eqref{eq:decoding-constrained-all-blocklengths}. In order to prove \eqref{eq:decoding-constrained-large-blocklengths}, we invoke \eqref{eq:decoding-error-large-blocklengths} of Theorem~\ref{theorem:decoding} to obtain $\proofconst_4 \in (0,\infty)$ with $\codingBoundFunc(\codebookRate,\blocklength) \leq \exp(-\blocklength\proofconst_4)$ for large enough $\blocklength$ and we invoke \eqref{eq:bad-codewords-large-blocklengths} of Lemma~\ref{lemma:bad-codewords} to obtain $\proofconst_5 \in (0,\infty)$ with $\costBoundFunc^{(\costBoundFunc,\costConstraint)}(\proofconst_3,\codebookRate_{\min},\blocklength) \leq \exp(-\exp(\blocklength\proofconst_5))$ for sufficiently large $\blocklength$. We are allowed to make $\proofconst_2$ small enough so that $\proofconst_2 \in (0,\proofconst_4)$. With any choice of $\finalconst_1 \in (0, \min(\proofconst_2, \proofconst_3))$ and $\finalconst_2 \in (0,\proofconst_4-\proofconst_2)$, we then have
\begin{align*}
\Probability_\codebook\left(
  \frac{1}{\codebookSize}
  \sum_{\codewordIndex=1}^\codebookSize
    \trace\left(
      \classicalQuantumChannel(\codebook(\codewordIndex))
      (\identityOperator - \decoderpovm_\codewordIndex)
    \right)
  \geq
  \exp(-\blocklength\finalconst_1)
\right)
&\leq
\Probability_\codebook\left(
  \frac{1}{\codebookSize}
  \sum_{\codewordIndex=1}^\codebookSize
    \trace\left(
      \classicalQuantumChannel(\codebook(\codewordIndex))
      (\identityOperator - \decoderpovm_\codewordIndex)
    \right)
  \geq
  \exp(-\blocklength\proofconst_2)
  +
  2\exp(-\blocklength\proofconst_3)
\right)
\\
\overset{\eqref{eq:decoding-constrained-all-blocklengths}}&{\leq}
\codingBoundFunc(\codebookRate,\blocklength)\exp(\blocklength\proofconst_2)
+
\costBoundFunc^{(\costFunction,\costConstraint)}\left(\proofconst_3,\codebookRate_{\min},\blocklength\right)
\\
&\leq
\exp(-\blocklength(\proofconst_4-\proofconst_2)
+
\exp(-\exp(\proofconst_5 \blocklength))
\\
&\leq
\exp(-\blocklength\gamma_2).
\qedhere
\end{align*}
\end{proof}

\subsection{Proof of the Main Theorems for Wiretap Coding}
\label{sec:proofs-main-result}
We now have everything needed to prove the main results of this paper.

\begin{proof}[Proof of Theorem~\ref{theorem:wiretap-ccq-all-blocklengths}]
We prove the existence of a codebook by arguing that if we draw a codebook at random, it has all the properties claimed in the theorem statement with a positive probability.
\paragraph{Codebook generation}
We generate a wiretap codebook $\codebook := (\codebook_1, \dots, \codebook_\numCodebooks)$ which is an $\numCodebooks$-tuple of i.i.d. standard random codebooks drawn according to $\inputDistribution$. We also define the associated cost-constrained wiretap codebook $\codebook_{\costFunction,\costConstraint} = (\codebook_{1, \costFunction,\costConstraint}, \dots, \codebook_{\numCodebooks, \costFunction,\costConstraint})$.

\paragraph{Encoding procedure}
Let $\messageRV \in \{1, \dots, \numCodebooks\}$. In order to generate the corresponding output of $\encoder$, we draw a random number $\codewordIndex \in \{1, \dots, \codebookSize\}$ and output $\codebook_{\messageRV, \costFunction,\costConstraint}(\codewordIndex)$. Note that (\ref{eq:wiretap-codebook-num}) ensures that our code has a rate of at least $\codebookRate$.

\paragraph{Decoding procedure}
We use a joint typicality decoder; i.e., $\decoder$ outputs $\hat{\messageRV}$ if there is $\codewordIndex \in \{1, \dots, \codebookSize\}$ such that $\codebook_{\hat{\messageRV}}(\codewordIndex)$ is jointly typical with $\legitOutputRV^\blocklength$ and $\hat{\messageRV}, \codewordIndex$ are unique with this property. If no such $\hat{\messageRV}, \codewordIndex$ exist, the decoder outputs $1$. The definition of typicality we use is in terms of information density (cf. \cite[Def. 4.34]{bender2021nonasymptotic}). Let
\[
\decodingError_\codebook:~
\{1, \dots, \numCodebooks\} \times \{1, \dots, \codebookSize\}
\rightarrow
[0,1]
,~~
(\indexCodebooks,\codewordIndex)
\mapsto
\Probability\left(
  \hat{\messageRV}
  \neq
  \indexCodebooks
  ~|~
  \inputRV^\blocklength = \codebook_\indexCodebooks(\codewordIndex)
\right).
\]

Let us first look at the average decoding error of an alternative encoder $\encoder'$ which draws a random number $\codewordIndex \in \{1, \dots, \codebookSize\}$ and outputs $\codebook_{\messageRV}(\codewordIndex)$ to transmit message $\messageRV$. The average decoding error in case $\encoder'$ is used can be expressed as
\[
\frac{1}{\numCodebooks\codebookSize}
\sum_{\indexCodebooks=1}^\numCodebooks
\sum_{\codewordIndex=1}^\codebookSize
  \decodingError_\codebook(\indexCodebooks,\codewordIndex),
\]

and so we have by~\cite[Lemma 4.37]{bender2021nonasymptotic} that
\[
\Expectation_\codebook\left(
  \frac{1}{\numCodebooks\codebookSize}
  \sum_{\indexCodebooks=1}^\numCodebooks
  \sum_{\codewordIndex=1}^\codebookSize
    \decodingError_\codebook(\indexCodebooks,\codewordIndex)
\right)
\leq
\classicalCodingBoundFunc^{\bobChannel}(\combinedRate,\blocklength).
\]
We use Markov's inequality to infer that for any $\proofconst_3 \in (0,\infty)$, we have
\[
\Probability_\codebook\left(
  \frac{1}{\numCodebooks\codebookSize}
  \sum_{\indexCodebooks=1}^\numCodebooks
  \sum_{\codewordIndex=1}^\codebookSize
    \decodingError_\codebook(\indexCodebooks,\codewordIndex)
  \geq
  \exp(-\blocklength\proofconst_3)
\right)
\leq
\classicalCodingBoundFunc^{\bobChannel}(\combinedRate,\blocklength)\exp(\blocklength\proofconst_3).
\]
Noting that
\[
\frac{1}{\numCodebooks\codebookSize}
\sum_{\indexCodebooks=1}^\numCodebooks
\sum_{\codewordIndex=1}^\codebookSize
  \decodingError_\codebook(\indexCodebooks,\codewordIndex)
-
\frac{1}{\numCodebooks\codebookSize}
\sum_{\indexCodebooks=1}^\numCodebooks
\sum_{\codewordIndex=1}^\codebookSize
  \decodingError_{\codebook_{\costFunction, \costConstraint}}(\indexCodebooks,\codewordIndex)
\leq
\frac{\cardinality{\badCodewordsSet}}{\numCodebooks\codebookSize}
\]
and using Lemma~\ref{lemma:bad-codewords}, we obtain for every $\proofconst_2 \in (0,(\codebookRate+\randomnessRate)/2)$,
\begin{align}
\nonumber
&\hphantom{{}={}}
\Probability_\codebook\left(
  \frac{1}{\numCodebooks\codebookSize}
  \sum_{\indexCodebooks=1}^\numCodebooks
  \sum_{\codewordIndex=1}^\codebookSize
    \decodingError_{\codebook_{\costFunction, \costConstraint}}(\indexCodebooks,\codewordIndex)
  \geq
  \exp(-\blocklength\proofconst_2)
  +
  \exp(-\blocklength\proofconst_3)
\right)
\\
\nonumber
&\leq
\Probability_\codebook\left(
  \frac{1}{\numCodebooks\codebookSize}
  \sum_{\indexCodebooks=1}^\numCodebooks
  \sum_{\codewordIndex=1}^\codebookSize
    \decodingError_\codebook(\indexCodebooks,\codewordIndex)
  \geq
  \exp(-\blocklength\proofconst_3)
\right)
+
\Probability_\codebook\left(
  \frac{1}{\numCodebooks\codebookSize}
  \sum_{\indexCodebooks=1}^\numCodebooks
  \sum_{\codewordIndex=1}^\codebookSize
    \decodingError_\codebook(\indexCodebooks,\codewordIndex)
  -
  \frac{1}{\numCodebooks\codebookSize}
  \sum_{\indexCodebooks=1}^\numCodebooks
  \sum_{\codewordIndex=1}^\codebookSize
    \decodingError_{\codebook_{\costFunction, \costConstraint}}(\indexCodebooks,\codewordIndex)
  \geq
  \exp(-\blocklength\proofconst_2)
\right)
\\
\label{eq:wiretap-ccq-decoding-error}
&\leq
\classicalCodingBoundFunc^{\bobChannel}(\combinedRate,\blocklength)\exp(\blocklength\proofconst_3)
+
\costBoundFunc^{(\costFunction,\costConstraint)}(\proofconst_2,\codebookRate+\randomnessRate,\blocklength).
\end{align}

\paragraph{Distinguishing security level}
We first note that by passing $\messageRV$ through $\encoder$ and $\eveChannel^\blocklength$, we obtain the density operator $\eveChannelOutput{\codebook_{\messageRV, \costFunction,\costConstraint}}$. We apply Corollary~\ref{cor:q-resolvability-constrained} to obtain, for every $\indexCodebooks \in \{1, \dots, \numCodebooks\}$,
\[
\Probability_\codebook\bigg(
  \traceNorm{
    \eveChannelOutput{\codebook_{\indexCodebooks, \costFunction,\costConstraint}}
    -
    \eveChannelOutput{\inputDistribution}^{\otimes \blocklength}
  }
  \geq
  \resolvabilityBoundFunc^{\classicalQuantumChannel}(\randomnessRate,\blocklength)
  +
  2\exp(-\proofconst_4\blocklength)
  +
  \exp(-\proofconst_5\blocklength)
\bigg)
\leq
\costBoundFunc^{(\costFunction,\costConstraint)}(
  \proofconst_4,
  \randomnessRate,
  \blocklength
)
+
\exp\left(
  -\frac{1}{2}\exp(
    \blocklength
    (\randomnessRate-2\proofconst_5)
  )
\right).
\]

Hence,
\begin{multline}
\label{eq:ccq-wiretap-tvdist}
\Probability_\codebook\bigg(
  \exists \indexCodebooks \in \{1, \dots, \numCodebooks\}~
  \traceNorm{
    \eveChannelOutput{\codebook_{\indexCodebooks, \costFunction,\costConstraint}}
    -
    \eveChannelOutput{\inputDistribution}^{\otimes \blocklength}
  }
  \geq
  \resolvabilityBoundFunc^{\classicalQuantumChannel}(\randomnessRate,\blocklength)
  +
  2\exp(-\proofconst_4\blocklength)
  +
  \exp(-\proofconst_5\blocklength)
\bigg)
\\
\leq
\costBoundFunc^{(\costFunction,\costConstraint)}(
  \proofconst_4,
  \randomnessRate,
  \blocklength
)
\exp(\blocklength(\combinedRate-\randomnessRate))
+
\exp\left(
  -\frac{1}{2}\exp(
    \blocklength
    (\randomnessRate-2\proofconst_5)
  )
  +
  \blocklength(
    \combinedRate-\randomnessRate
  )
\right)
\end{multline}

Conditioned on the event that we have for every $\indexCodebooks \in \{1, \dots, \numCodebooks\}$
\begin{equation}
\label{eq:wiretap-security-assumption}
\traceNorm{
  \eveChannelOutput{\codebook_{\indexCodebooks, \costFunction,\costConstraint}}
  -
  \eveChannelOutput{\inputDistribution}^{\otimes \blocklength}
}
<
\resolvabilityBoundFunc^{\classicalQuantumChannel}(\randomnessRate,\blocklength)
+
2\exp(-\proofconst_4\blocklength)
+
\exp(-\proofconst_5\blocklength),
\end{equation}
the triangle inequality yields, for all $\indexCodebooks_1, \indexCodebooks_2$,
\begin{align*}
\traceNorm{
  \eveChannelOutput{\codebook_{\indexCodebooks_1, \costFunction,\costConstraint}}
  -
  \eveChannelOutput{\codebook_{\indexCodebooks_2, \costFunction,\costConstraint}}
}
&\leq
\traceNorm{
  \eveChannelOutput{\codebook_{\indexCodebooks_1, \costFunction,\costConstraint}}
  -
  \eveChannelOutput{\inputDistribution}^{\otimes \blocklength}
}
+
\traceNorm{
  \eveChannelOutput{\codebook_{\indexCodebooks_2, \costFunction,\costConstraint}}
  -
  \eveChannelOutput{\inputDistribution}^{\otimes \blocklength}
}
\\
\overset{(\ref{eq:wiretap-security-assumption})}&{\leq}
2\resolvabilityBoundFunc^{\classicalQuantumChannel}(\randomnessRate,\blocklength)
+
4\exp(-\proofconst_4\blocklength)
+
2\exp(-\proofconst_5\blocklength),
\end{align*}
which means that our wiretap code has distinguishing security level
$
2\resolvabilityBoundFunc^{\classicalQuantumChannel}(\randomnessRate,\blocklength)
+
4\exp(-\proofconst_4\blocklength)
+
2\exp(-\proofconst_5\blocklength)
$.

In conclusion, if \eqref{eq:wiretap-ccq-prob-one} is satisfied, then we can combine \eqref{eq:wiretap-ccq-decoding-error} and \eqref{eq:ccq-wiretap-tvdist} with the union bound to argue that if we carry out the random codebook construction as described above, there is a nonzero probability that the resulting cost-constrained codebook simultaneously satisfies the decoding error property claimed in item \ref{item:wiretap-ccq-all-blocklengths-error} and the security property claimed in item \ref{item:wiretap-ccq-all-blocklengths-security} of the theorem statement. Furthermore, it is clear that all codebooks that could be drawn randomly in the fashion described satisfy items \ref{item:wiretap-ccq-all-blocklengths-codebook-size} and \ref{item:wiretap-ccq-all-blocklengths-cost}. This means that we have shown that at least one codebook must exist that has all the properties claimed in the statement of Theorem~\ref{theorem:wiretap-ccq-all-blocklengths}.
\end{proof}

\begin{proof}[Proof of Theorem~\ref{theorem:wiretap-ccq}]
We have $\codebookRate < \legitInformation - \holevoInformation{\inputDistribution}{\eveChannel}$. This allows us to fix $\randomnessRate \in (\holevoInformation{\inputDistribution}{\eveChannel}, \legitInformation - \codebookRate)$, and, subsequently, $\combinedRate \in (\codebookRate + \randomnessRate, \legitInformation)$. For sufficiently large $\blocklength$, this then allows us to fix $\codebookSize, \numCodebooks \in \naturals$ satisfying \eqref{eq:wiretap-codebook-size}, \eqref{eq:wiretap-codebook-size-num}, and \eqref{eq:wiretap-codebook-num}.

We fix $\proofconst_2 \in (0, \min(\proofconst_1,(\codebookRate+\numCodebooks)/2))$ and $\proofconst_4 \in (0, \min(\proofconst_1,\codebookRate/2))$, where $\proofconst_1$ is defined in \eqref{eq:mgf-parameter-choice}. Then we invoke~\cite[Lemma 4.37]{bender2021nonasymptotic}, Lemma~\ref{lemma:bad-codewords}, and Theorem~\ref{theorem:q-resolvability} to find $\proofconst_6, \proofconst_7, \proofconst_8, \proofconst_9 \in(0,\infty)$ with the properties that for sufficiently large $\blocklength$,
\begin{align*}
\classicalCodingBoundFunc^{\bobChannel}(\combinedRate,\blocklength)
&\leq
\exp(-\proofconst_6 \blocklength)
\\
\costBoundFunc^{(\costFunction,\costConstraint)}(\proofconst_2,\codebookRate+\randomnessRate,\blocklength)
&\leq
\exp(-\exp(\proofconst_7 \blocklength))
\\
\costBoundFunc^{(\costFunction,\costConstraint)}(\proofconst_4,\randomnessRate,\blocklength)
&\leq
\exp(-\exp(\proofconst_8 \blocklength))
\\
\resolvabilityBoundFunc^{\classicalQuantumChannel}(\randomnessRate,\blocklength)
&\leq
\exp(-\proofconst_9 \blocklength).
\end{align*}
Finally, we pick $\proofconst_3 \in (0,\proofconst_6)$ and $\proofconst_5 \in (0,\codebookRate/2)$. Then the left hand side in \eqref{eq:wiretap-ccq-prob-one} is upper bounded by
\[
\exp(-(\proofconst_6-\proofconst_3) \blocklength)
+
\exp(-\exp(\proofconst_7 \blocklength))
+
\exp(-\exp(\proofconst_8 \blocklength) + \blocklength(\combinedRate-\randomnessRate))
+
\exp\left(
  -\frac{1}{2}\exp(
    \blocklength
    (\codebookRate-2\proofconst_5)
  )
  +
  \blocklength(\combinedRate-\randomnessRate)
\right),
\]
which clearly tends to $0$ for $\blocklength \rightarrow \infty$. In particular, for large enough $\blocklength$, condition \eqref{eq:wiretap-ccq-prob-one} is satisfied. With the choices $\finalconst_1 \in (0,\min(\proofconst_2,\proofconst_3)$ and $\finalconst_2 \in (0, \min(\proofconst_4,\proofconst_5,\proofconst_9))$, Theorem~\ref{theorem:wiretap-ccq-all-blocklengths} gives us a wiretap code which by item~\ref{item:wiretap-ccq-all-blocklengths-error} has, for sufficiently large $\blocklength$, average error
\[
 \decodingError
 =
 \exp(-\blocklength\proofconst_2)
 +
 \exp(-\blocklength\proofconst_3)
 <
 \exp(-\blocklength\finalconst_1)
\]
and by item~\ref{item:wiretap-ccq-all-blocklengths-security} distinguishing security level
\[
  \securityNumber
  =
  2\resolvabilityBoundFunc^{\classicalQuantumChannel}(\randomnessRate,\blocklength)
  +
  4\exp(-\proofconst_4\blocklength)
  +
  2\exp(-\proofconst_5\blocklength)
  \leq
  2\exp(-\proofconst_9 \blocklength)
  +
  4\exp(-\proofconst_4\blocklength)
  +
  2\exp(-\proofconst_5\blocklength)
  <
  \exp(-\blocklength\finalconst_2).
  \qedhere
\]
\end{proof}

\setcounter{paragraph}{0}

\begin{proof}[Proof of Theorem~\ref{theorem:wiretap-cq-all-blocklengths}]
The proof is similar to that of Theorem~\ref{theorem:wiretap-ccq-all-blocklengths}, except that we invoke Corollary~\ref{cor:decoding} for the average decoding error at the legitimate receiver.

\paragraph{Codebook generation}
See proof of Theorem~\ref{theorem:wiretap-ccq}.

\paragraph{Encoding procedure}
See proof of Theorem~\ref{theorem:wiretap-ccq}.

\paragraph{Decoding procedure}
We treat the wiretap codebook $\codebook$ as one large codebook $\tilde{\codebook}$ of size $\numCodebooks\codebookSize$ defined by $\tilde{\codebook}(\indexCodebooks,\codewordIndex) := \codebook_\indexCodebooks(\codewordIndex)$. Together with (\ref{eq:wiretap-codebook-size-num}), this allows us to invoke Corollary~\ref{cor:decoding-constrained} to obtain a decoding \gls{povm} $(\decoderpovm_{\indexCodebooks,\codewordIndex})_{\indexCodebooks,\codewordIndex=1}^{\numCodebooks,\codebookSize}$ for the codebook $\tilde{\codebook}_{\costFunction,\costConstraint}$. Define a corresponding decoding \gls{povm} $(\decoderpovm_{\indexCodebooks})_{\indexCodebooks=1}^{\numCodebooks}$ for the wiretap channel by
\[
\decoderpovm_{\indexCodebooks}
:=
\sum_{\codewordIndex=1}^\codebookSize
  \decoderpovm_{\indexCodebooks,\codewordIndex}.
\]

Then
\begin{align*}
\Expectation_\messageRV
  \trace\left(
    \bobcqChannel^\blocklength \circ \encoder (\messageRV)
    (\identityOperator - \decoderpovm_\messageRV)
  \right)
&=
\frac{1}{\numCodebooks}
\sum_{\indexCodebooks=1}^{\numCodebooks}
  \trace\left(
    \bobcqChannel^\blocklength \circ \encoder (\indexCodebooks)
    \left(\identityOperator - \decoderpovm_{\indexCodebooks}\right)
  \right)
\\
&=
\frac{1}{\numCodebooks}
\sum_{\indexCodebooks=1}^{\numCodebooks}
  \trace\left(
    \left(
      \frac{1}{\codebookSize}
      \sum_{\codewordIndex=1}^\codebookSize
        \bobcqChannel^\blocklength\left(\tilde{\codebook}_{\costFunction,\costConstraint}(\indexCodebooks,\codewordIndex)\right)
    \right)
    \left(
      \identityOperator
      -
      \sum_{\codewordIndex=1}^\codebookSize \decoderpovm_{\indexCodebooks,\codewordIndex}
    \right)
  \right)
\\
&=
\frac{1}{\numCodebooks\codebookSize}
\sum_{\indexCodebooks=1}^{\numCodebooks}
\sum_{\codewordIndex=1}^\codebookSize
  \trace\left(
    \bobcqChannel^\blocklength\left(\tilde{\codebook}_{\costFunction,\costConstraint}(\indexCodebooks,\codewordIndex)\right)
    \left(
      \identityOperator
      -
      \sum_{\codewordIndex=1}^\codebookSize \decoderpovm_{\indexCodebooks,\codewordIndex}
    \right)
  \right)
\\
&\leq
\frac{1}{\numCodebooks\codebookSize}
\sum_{\indexCodebooks=1}^{\numCodebooks}
\sum_{\codewordIndex=1}^\codebookSize
  \trace\left(
    \bobcqChannel^\blocklength\left(\tilde{\codebook}_{\costFunction,\costConstraint}(\indexCodebooks,\codewordIndex)\right)
    \left(
      \identityOperator
      -
      \decoderpovm_{\indexCodebooks,\codewordIndex}
    \right)
  \right),
\end{align*}
and therefore, we have by Corollary~\ref{cor:decoding-constrained} that
\begin{align*}
&\hphantom{{}={}}
\Probability_\codebook\Big(
  \Expectation_\messageRV
    \trace\left(
      \bobcqChannel^\blocklength \circ \encoder (\messageRV)
      (\identityOperator - \decoderpovm_\messageRV)
    \right)
  \geq
  \exp(-\blocklength\proofconst_2)
  +
  2\exp(-\blocklength\proofconst_3)
\Big)
\\
&\leq
\codingBoundFunc^{\bobcqChannel}(\combinedRate,\blocklength)\exp(\blocklength\proofconst_2)
+
\costBoundFunc^{(\costFunction,\costConstraint)}\left(\proofconst_3,\frac{\log (\numCodebooks\codebookSize)}{\blocklength},\blocklength\right)
\\
&\leq
\codingBoundFunc^{\bobcqChannel}(\combinedRate,\blocklength)\exp(\blocklength\proofconst_2)
+
\costBoundFunc^{(\costFunction,\costConstraint)}\left(\proofconst_3,\codebookRate+\randomnessRate,\blocklength\right).
\end{align*}

\paragraph{Distinguishing security level}
See proof of Theorem~\ref{theorem:wiretap-ccq}.

Similarly as in the proof of Theorem~\ref{theorem:wiretap-ccq}, the existence of a wiretap code as claimed in the theorem statement is assured for sufficiently large $\blocklength$.
\end{proof}

\begin{proof}[Proof of Theorem~\ref{theorem:wiretap-cq}]
Theorem~\ref{theorem:wiretap-cq} follows from Theorem~\ref{theorem:wiretap-cq-all-blocklengths} with the same argument we used to show that Theorem~\ref{theorem:wiretap-ccq} follows from Theorem~\ref{theorem:wiretap-ccq-all-blocklengths}.
\end{proof}

\section{Specialization to the Gaussian \texorpdfstring{\gls{cq}}{cq} Wiretap Channel}
\label{sec:gaussian-cq}
In this section, we demonstrate the finite block length nature of our results on the example of the Gaussian \gls{cqq} wiretap channel. We only give a sketch of the additional steps necessary to evaluate the bounds for this channel in the rest of this section, but the full annotated Python source code that reproduces the plots in Figures~\ref{fig:plot-security}, \ref{fig:plot-decoding} and \ref{fig:plot-blocklength} is available as an electronic supplement with this paper.

\subsection{The Gaussian \texorpdfstring{\gls{cq}}{cq} Channel}
\label{sec:gaussian-cq-basics}
In the following, we introduce the Gaussian \gls{cq} channel along with some properties that we need in this section. For more details, we refer the reader to~\cite[Section 12.1]{holevo2019quantum}.

Let $\hilbertSpace := \LTwoOnReals$ be the complex Hilbert space of equivalence classes of complex-valued, square integrable functions on $\reals$. For $\photonNumberNatural=0,1,\dots$, let $\ket{\photonNumberNatural}$ be the number state vector defined in \cite[eq. (12.17)]{holevo2019quantum}. Here, it will only be important that $\ket{0}, \ket{1}, \dots$ form an orthonormal basis of $\hilbertSpace$. For every $\coherentStateVector \in \complexNumbers$, we define a \emph{coherent state vector} (see~\cite[eq. before (12.18) and eq. before (12.21)]{holevo2019quantum})
\begin{equation}
\label{eq:coherent-state-vector}
\ket{\coherentStateVector}
:=
\exp\left(
  -\frac{\absolute{\coherentStateVector}^2}{2}
\right)
\sum_{\photonNumberNatural=0}^\infty
  \frac{\coherentStateVector^\photonNumberNatural}
       {\sqrt{\photonNumberNatural !}}
  \ket{\photonNumberNatural}.
\end{equation}
The Gaussian \gls{cq} channel is then defined (cf.~\cite[eq. (12.28)]{holevo2019quantum}) as
\begin{equation*}
\classicalQuantumChannel^{(\transmittivity,\quantumNoisePower)}:~
\complexNumbers \rightarrow \densityOperators{\hilbertSpace}
,~~
\inputAlphabetElement
\mapsto
\frac{1}{\pi\quantumNoisePower}
\int_\complexNumbers
  \rankOneOperator{\coherentStateVector}
  \exp\left(
    -\frac{\absolute{\coherentStateVector-\sqrt{\transmittivity}\inputAlphabetElement}^2}
          {\quantumNoisePower}
  \right)
d\coherentStateVector.
\end{equation*}
It is parametrized by the \emph{noise power} $\quantumNoisePower>0$, and the \emph{transmittivity} $\transmittivity \in [0,1]$. With $\transmittivity=1$, this is the same definition as in ~\cite[eq. (12.28)]{holevo2019quantum}. In the following, we summarize some properties of this channel that we need in this section. They are derived in~\cite{holevo2019quantum} for the case $\transmittivity=1$, but they carry over because $\classicalQuantumChannel^{(\transmittivity,\quantumNoisePower)}(\inputAlphabetElement) = \classicalQuantumChannel^{(1,\quantumNoisePower)}(\sqrt{\transmittivity}\inputAlphabetElement)$. For every $\inputAlphabetElement \in \complexNumbers$, there is~\cite[eq. (12.30)]{holevo2019quantum} a unitary \emph{displacement operator} $\gaussiancqUnitary{\inputAlphabetElement}$ such that
\begin{equation*}
\classicalQuantumChannel^{(1,\quantumNoisePower)}(\inputAlphabetElement)
=
\gaussiancqUnitary{\inputAlphabetElement}
\classicalQuantumChannel^{(1,\quantumNoisePower)}(0)
\adjoint{\gaussiancqUnitary{\inputAlphabetElement}},
\end{equation*}
and hence
\begin{equation}
\label{eq:gaussian-cq-unitary}
\classicalQuantumChannel^{(\transmittivity,\quantumNoisePower)}(\inputAlphabetElement)
=
\classicalQuantumChannel^{(1,\quantumNoisePower)}(\sqrt{\transmittivity}\inputAlphabetElement)
=
\gaussiancqUnitary{\sqrt{\transmittivity}\inputAlphabetElement}
\classicalQuantumChannel^{(1,\quantumNoisePower)}(0)
\adjoint{\gaussiancqUnitary{\sqrt{\transmittivity}\inputAlphabetElement}}
=
\gaussiancqUnitary{\sqrt{\transmittivity}\inputAlphabetElement}
\classicalQuantumChannel^{(\transmittivity,\quantumNoisePower)}(0)
\adjoint{\gaussiancqUnitary{\sqrt{\transmittivity}\inputAlphabetElement}}.
\end{equation}
$\classicalQuantumChannel^{(\transmittivity,\quantumNoisePower)}(0)=\classicalQuantumChannel^{(1,\quantumNoisePower)}(0)$ can be written as~\cite[eq. (12.24)]{holevo2019quantum}
\begin{equation}
\label{eq:gaussian-cq-zero-decomposition}
\classicalQuantumChannel^{(\transmittivity,\quantumNoisePower)}(0)
=
\frac{1}{\quantumNoisePower+1}
\sum_{\photonNumberNatural=0}^\infty
  \left(
    \frac{\quantumNoisePower}{\quantumNoisePower+1}
  \right)^\photonNumberNatural
  \rankOneOperator{\photonNumberNatural}.
\end{equation}
This state has the von Neumann entropy~\cite[eq. (12.26)]{holevo2019quantum}
\begin{equation}
\label{eq:gaussian-cq-entropy}
\vonNeumannEntropy{\classicalQuantumChannel^{(\transmittivity,\quantumNoisePower)}(0)}
=
\gordonFunction(\quantumNoisePower),
\end{equation}
where
\begin{equation}
\label{eq:gordon}
\gordonFunction:~ [0,\infty) \rightarrow [0,\infty),~~
\generalReal \mapsto
\begin{cases}
(\generalReal+1) \log(\generalReal+1)
-
\generalReal \log \generalReal,
&\generalReal > 0,
\\
0, &\generalReal=0
\end{cases}
\end{equation}
is called the \emph{Gordon function}.

A consequence of (\ref{eq:gaussian-cq-unitary}) and (\ref{eq:gaussian-cq-zero-decomposition}) is that for every $\generalFunction: \reals \rightarrow \reals$ and every $\inputAlphabetElement \in \complexNumbers$, we have
\begin{align}
\nonumber
\trace
\generalFunction\big(
  \classicalQuantumChannel^{(\transmittivity,\quantumNoisePower)}(\inputAlphabetElement)
\big)
\overset{(\ref{eq:gaussian-cq-unitary})}&{=}
\trace
\generalFunction\big(
  \gaussiancqUnitary{\sqrt{\transmittivity}\inputAlphabetElement}
  \classicalQuantumChannel^{(\transmittivity,\quantumNoisePower)}(0)
  \adjoint{\gaussiancqUnitary{\sqrt{\transmittivity}\inputAlphabetElement}}
\big)
\\
\nonumber
\overset{(\ref{eq:gaussian-cq-zero-decomposition})}&{=}
\trace
\generalFunction\left(
  \frac{1}{\quantumNoisePower+1}
  \sum_{\photonNumberNatural=0}^\infty
    \left(
      \frac{\quantumNoisePower}{\quantumNoisePower+1}
    \right)^\photonNumberNatural
    \gaussiancqUnitary{\sqrt{\transmittivity}\inputAlphabetElement}
    \rankOneOperator{\photonNumberNatural}
    \adjoint{\gaussiancqUnitary{\sqrt{\transmittivity}\inputAlphabetElement}}
\right)
\\
\nonumber
\overset{(a)}&{=}
\trace
\generalFunction\left(
  \frac{1}{\quantumNoisePower+1}
  \sum_{\photonNumberNatural=0}^\infty
    \left(
      \frac{\quantumNoisePower}{\quantumNoisePower+1}
    \right)^\photonNumberNatural
    \rankOneOperator{\photonNumberNatural}
\right)
\\
\label{eq:gaussian-cq-trace-function}
\overset{(\ref{eq:gaussian-cq-zero-decomposition})}&{=}
\trace \generalFunction(\classicalQuantumChannel^{(\transmittivity,\quantumNoisePower)}(0)),
\end{align}
where step (a) is due to the fact that both $\ket{0}, \ket{1}, \dots$ and $\gaussiancqUnitary{\sqrt{\transmittivity}\inputAlphabetElement}\ket{0}, \gaussiancqUnitary{\sqrt{\transmittivity}\inputAlphabetElement}\ket{1}, \dots$ are orthonormal bases of $\hilbertSpace$.

As the input distribution $\inputDistribution$, we choose a complex Gaussian distribution where real and imaginary part are independent with mean $0$ and variance $\inputEnergy/2$, and $\inputEnergy$ is the average energy per channel use. In order to relate this to the case $\transmittivity=1$, we also define an auxiliary input distribution $\inputDistribution_\transmittivity$ which is complex Gaussian with mean $0$ and variance $\transmittivity\inputEnergy/2$ per complex component. With this choice, we have~\cite[eq. 12.41]{holevo2019quantum}
\begin{equation}
\label{eq:gaussian-cq-average}
\classicalQuantumChannel^{(\transmittivity,\quantumNoisePower)}_\inputDistribution
=
\classicalQuantumChannel^{(1,\quantumNoisePower)}_{\inputDistribution_\transmittivity}
=
\frac{1}{\quantumNoisePower+\transmittivity\inputEnergy+1}
\sum_{\photonNumberNatural=1}^\infty
  \left(
    \frac{\quantumNoisePower+\transmittivity\inputEnergy}
         {\quantumNoisePower+\transmittivity\inputEnergy+1}
  \right)^\photonNumberNatural
  \rankOneOperator{\photonNumberNatural}.
\end{equation}

In this section, we consider a \gls{cqq} wiretap channel where $\bobcqChannel = \classicalQuantumChannel^{(\transmittivity_\bob, \quantumNoisePower_\bob)}$ and $\eveChannel = \classicalQuantumChannel^{(\transmittivity_\eve,\quantumNoisePower_\eve)}$.

\subsection{Physical Channel Model}
\label{sec:physical-model}

\begin{figure}[t]
    \centering
    \includegraphics[width=0.7\textwidth]{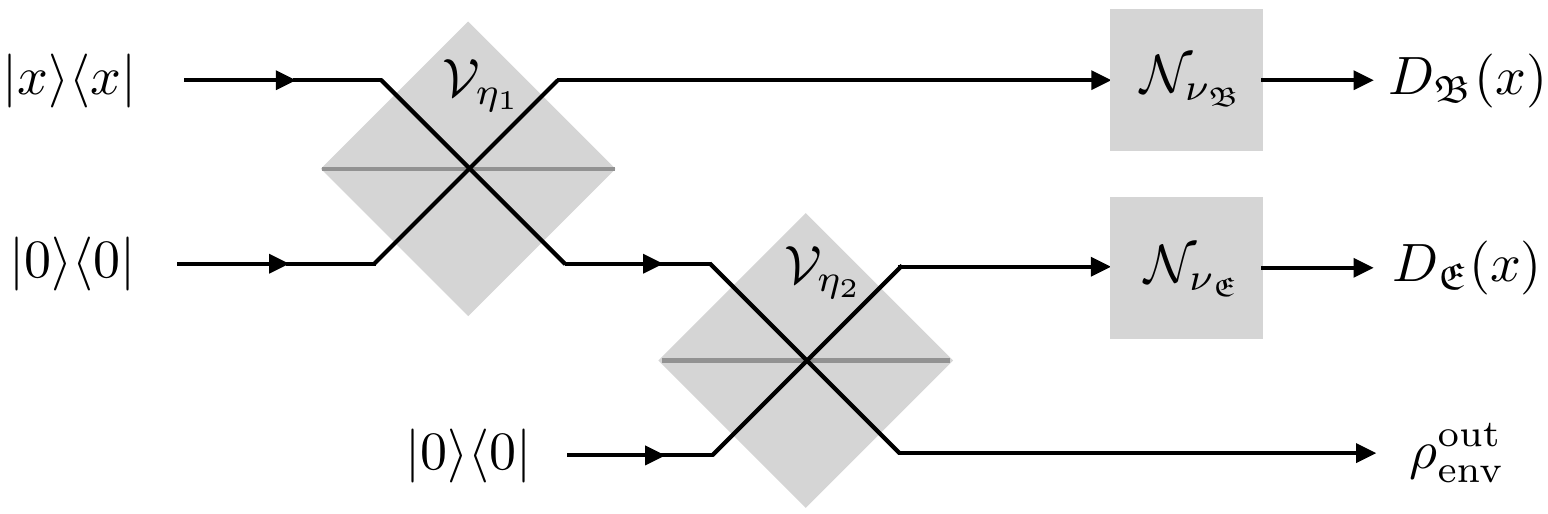}
    \caption{The cascaded channel model consisting of two beam splitters followed by additive thermal noise channels acting on the signals that are arriving at the legitimate receiver and eavesdropper respectively which is used for the numerical evaluation of the theoretical bounds. Both beam splitters are pure-loss channels which is reflected by the presence of the vacuum state $\rankOneOperator{0}$ at the second input of both devices. The third output is the state $\rho_{\textup{env}}^{\textup{out}}$ of photons detected neither by the legitimate receiver nor the eavesdropper.}
    \label{fig:beamsplitter}
\end{figure}

In the following, we describe a commonly employed channel model of physical systems~\cite{guha2008bosonicwiretap,wildeWiretap2018} that is a special case of the Gaussian \gls{cqq} wiretap channel $(\bobcqChannel,\eveChannel)$ introduced in Section~\ref{sec:gaussian-cq-basics}. It uses a quantum channel $\beamsplitterOperation_\transmittivity$ called a \emph{beam splitter} \cite[Chapter 4]{leonhardt1997measuring} with a parameter $\transmittivity\in [0,1]$ called \emph{transmittivity}.
Denoting by $a^{\ast},b^\ast$ and $a,b$ the creation, respectively annihilation operators on two copies $\hilbertSpace_\alice,\hilbertSpace'_\alice$ of the Hilbert space $\hilbertSpace=\LTwoOnReals$,
$\beamsplitterOperation_{\transmittivity}: \densityOperators{\hilbertSpace_{\alice}\otimes \hilbertSpace'_{\alice}} \to \densityOperators{\hilbertSpace_{\alice}\otimes \hilbertSpace'_{\alice}}$
is defined as
\begin{equation*}
\beamsplitterOperation_{\transmittivity}
  (\generalQuantumState)
:=
\beamSplitter_{\beamSplitterAngle}
\generalQuantumState
\adjoint{\beamSplitter_{\beamSplitterAngle}},
\end{equation*}
where $\beamSplitterAngle\in[0,\pi/2]$ is such that $\transmittivity=\cos^2(\beamSplitterAngle)$ and the unitary $\beamSplitter_{\beamSplitterAngle}: \hilbertSpace_{\alice}\otimes \hilbertSpace'_{\alice} \to\hilbertSpace_{\alice}\otimes \hilbertSpace'_{\alice}$ is given by \cite[eq. (4.16)]{leonhardt1997measuring}
\[
    \beamSplitter_\beamSplitterAngle = \exp\{\beamSplitterAngle(a^\ast \otimes b-a\otimes b^\ast)\}.
\]
For every $x,x'\in\complexNumbers$ and corresponding coherent state vectors $ \ket{x}\in\hilbertSpace_{\alice} $ and $ \ket{x'}\in \hilbertSpace'_{\alice}$, we have \cite[eq. (4.28)]{leonhardt1997measuring}
\begin{equation}
\label{eq:beam-splitter-key-property}
\beamsplitterOperation_\transmittivity\left(
  \rankOneOperator{x} \otimes \rankOneOperator{x'}
\right)
=
\rankOneOperator{\sqrt{\transmittivity}x+\sqrt{\transmittivity'}x'}
\otimes
\rankOneOperator{\sqrt{\transmittivity'}x-\sqrt{\transmittivity}x'},
\end{equation}
where $\transmittivity':=1-\transmittivity$.
Moreover, we use an \emph{additive thermal noise channel} $\additiveNoiseChannel_\quantumNoisePower(\generalQuantumState)$ which is described~\cite[eq. (3.48)]{wildeWiretap2018} by
\[
\additiveNoiseChannel_\quantumNoisePower(\generalQuantumState)
=
\frac{1}{\pi\quantumNoisePower}
\int_\complexNumbers
  \gaussiancqUnitary{\coherentStateVector}
  \generalQuantumState
  \adjoint{\gaussiancqUnitary{\coherentStateVector}}
  \exp\left(
    -
    \frac{\absolute{\coherentStateVector}^2}
         {\quantumNoisePower}
  \right)
d\coherentStateVector,
\]
where $\gaussiancqUnitary{\coherentStateVector}$ are displacement operators and $\coherentStateVector\in\complexNumbers$. For every coherent state $\rankOneOperator{\coherentStateVector}$ and $\quantumNoisePower \in (0,\infty)$, it holds that
\begin{equation}
\label{eq:additive-noise-channel-key-property}
\additiveNoiseChannel_{\quantumNoisePower}\left(
  \rankOneOperator{\coherentStateVector}
\right)
=
\classicalQuantumChannel^{(1,\quantumNoisePower)}
  (\coherentStateVector).
\end{equation}

A bosonic wiretap model \cite{guha2008bosonicwiretap,wildeWiretap2018} can be based on the assumption that every photon which is lost from the original signal can be detected by the eavesdropper. This is an extremely pessimistic model, not taking into account the photons which are detected neither by the legitimate receiver nor the eavesdropper due to channel loss. In contrast, we evaluate a channel model, depicted in Fig. \ref{fig:beamsplitter}, in which both the legitimate receiver and the eavesdropper are subject to channel loss. 
To this end, we consider Hilbert spaces $\hilbertSpace_{\alice,1}$, $\hilbertSpace'_{\alice,1}$, $\hilbertSpace_{\alice,2}$, $\hilbertSpace'_{\alice, 2}$, $\bobHilbertSpace$, $\eveHilbertSpace$, and $\hilbertSpace_{\env}$ all of which are copies of $\hilbertSpace=\LTwoOnReals$. Moreover, we choose transmittivity parameters $\transmittivity_1, \transmittivity_2 \in [0,1]$.

In this model, the transmitter's channel alphabet is given as $\inputAlphabet := \complexNumbers$. The transmission symbol $\inputAlphabetElement \in \complexNumbers$ is transformed into a coherent state $\rankOneOperator{\inputAlphabetElement}$ with $\ket{\inputAlphabetElement}\in \hilbertSpace_{\alice ,1}$ defined in (\ref{eq:coherent-state-vector}). This state $\rankOneOperator{\inputAlphabetElement}$ is passed to the first input of the beam splitter $\beamsplitterOperation_{\transmittivity_1}: \densityOperators{\hilbertSpace_{\alice ,1}\otimes \hilbertSpace'_{\alice ,1} } \to \densityOperators{\hilbertSpace_{\alice ,1}\otimes \hilbertSpace_{\alice ,2}}$, where the second input of the beam splitter receives the vacuum state $\rankOneOperator{0}\in \densityOperators{\hilbertSpace'_{\alice ,1}}$. The first output is passed through $\additiveNoiseChannel_{\receiverNoise_\bob}:\densityOperators{\hilbertSpace_{\alice ,1}}\to \densityOperators{\bobHilbertSpace}$ to model the thermal noise at the legitimate receiver. The second output is passed through another beam splitter $\beamsplitterOperation_{\transmittivity_2}: \densityOperators{\hilbertSpace_{\alice ,2}\otimes \hilbertSpace'_{\alice ,2} } \to \densityOperators{\hilbertSpace_{\alice ,2}\otimes \hilbertSpace_{\env}}$, again with the vacuum state $\rankOneOperator{0}\in \densityOperators{\hilbertSpace'_{\alice ,2}}$ at the second input. The eavesdropper receives a version of the first output of the beam splitter  $\beamsplitterOperation_{\transmittivity_2}$ which is passed through $\additiveNoiseChannel_{\receiverNoise_\eve}: \densityOperators{ \hilbertSpace_{\alice ,2} }\to \densityOperators{\eveHilbertSpace}$ to model the eavesdropper's thermal receiver noise. The remaining beam splitter output is considered an environment state which is received neither at the legitimate receiver nor at the eavesdropper. The joint quantum state at the legitimate receiver's channel output, the eavesdropper's channel output, and the environment channel output can thus be written as
\begin{align*}
&\hphantom{{}={}}
\left( \additiveNoiseChannel_{\receiverNoise_\bob}
\otimes
\additiveNoiseChannel_{\receiverNoise_\eve}
\otimes
\identityMap_{\boundedOperators{\hilbertSpace'_{\alice ,2}}}\right)
\circ
 \left(\identityMap_{\boundedOperators{\hilbertSpace_{\alice ,1}}}
\otimes
\beamsplitterOperation_{\transmittivity_2}\right) \Big(
  \beamsplitterOperation_{\transmittivity_1}\big(
    \rankOneOperator{\inputAlphabetElement} \otimes \rankOneOperator{0}
  \big)
  \otimes
  \rankOneOperator{0}
\Big)
\\
\overset{(\ref{eq:beam-splitter-key-property})}&{=}
\left(\additiveNoiseChannel_{\receiverNoise_\bob}
\otimes
\additiveNoiseChannel_{\receiverNoise_\eve}
\otimes
\identityMap_{\boundedOperators{\hilbertSpace'_{\alice ,2}}}\right)
\circ
\left(\identityMap_{\boundedOperators{\hilbertSpace_{\alice ,1}}}
\otimes
\beamsplitterOperation_{\transmittivity_2}\right)
\Big(
  \rankOneOperator{\sqrt{\transmittivity_1}\inputAlphabetElement} 
  \otimes
  \rankOneOperator{\sqrt{\transmittivity'_1}\inputAlphabetElement} 
  \otimes
  \rankOneOperator{0}
\Big)
\\
\overset{(\ref{eq:beam-splitter-key-property})}&{=}
\left(
  \additiveNoiseChannel_{\receiverNoise_\bob}
  \otimes
  \additiveNoiseChannel_{\receiverNoise_\eve}
  \otimes
  \identityMap_{\boundedOperators{\hilbertSpace'_{\alice ,2}}}
\right)
\Big(
  \rankOneOperator{\sqrt{\transmittivity_1}\inputAlphabetElement} 
  \otimes
  \rankOneOperator{\sqrt{\transmittivity_2\transmittivity'_1}\inputAlphabetElement} 
  \otimes
  \rankOneOperator{\sqrt{\transmittivity'_2\transmittivity'_1}\inputAlphabetElement}
\Big)
\\
\overset{(\ref{eq:additive-noise-channel-key-property})}&{=}
\classicalQuantumChannel^{(1,\receiverNoise_\bob)}
  \left(\sqrt{\transmittivity_1}\inputAlphabetElement\right)
\otimes
\classicalQuantumChannel^{(1,\receiverNoise_\eve)}
  \left(\sqrt{\transmittivity_2\transmittivity'_1}\inputAlphabetElement\right)
\otimes
  \rankOneOperator{\sqrt{\transmittivity'_2\transmittivity'_1}\inputAlphabetElement}
\\
&=
\classicalQuantumChannel^{(\transmittivity_1,\receiverNoise_\bob)}
  \left(\inputAlphabetElement\right)
\otimes
\classicalQuantumChannel^{(\transmittivity_2\transmittivity'_1,\receiverNoise_\eve)}
  \left(\inputAlphabetElement\right)
\otimes
  \rankOneOperator{\sqrt{\transmittivity'_2\transmittivity'_1}\inputAlphabetElement},
\end{align*}
where we have used $\transmittivity'_1 := 1-\transmittivity_1$, $\transmittivity'_2 := 1-\transmittivity_2$, and the identity maps $\identityMap_{\boundedOperators{\hilbertSpace_{\alice ,1}}}$ and $ \identityMap_{\boundedOperators{\hilbertSpace'_{\alice ,2}}}$ on the spaces $\boundedOperators{\hilbertSpace_{\alice ,1}}$ and $\boundedOperators{\hilbertSpace'_{\alice ,2}}$.

Therefore, this system is a special case of the \gls{cqq} wiretap channel $\bobcqChannel = \classicalQuantumChannel^{(\transmittivity_\bob, \quantumNoisePower_\bob)}$ and $\eveChannel = \classicalQuantumChannel^{(\transmittivity_\eve,\quantumNoisePower_\eve)}$ described in Section~\ref{sec:gaussian-cq-basics} where $\transmittivity_\bob := \transmittivity_1$ and $\transmittivity_\eve := \transmittivity_2(1 - \transmittivity_1)$.

\subsection{Numerical Evaluation of the Error Bounds in Theorems~\ref{theorem:q-resolvability} and~\ref{theorem:decoding}}
\label{sec:gaussian-cq-expectations}

In this subsection, we fix a channel $\classicalQuantumChannel^{(\transmittivity,\quantumNoisePower)}$ and show how to compute the quantities $\jointEntropy,\outputEntropy,\conditionalRenyiEntropy{\renyiorder},\outputRenyiEntropy{\renyiorder}$ and evaluate \eqref{eq:atypicalTermsFunc-start} to \eqref{eq:costBoundFunc} for such a channel under a Gaussian input distribution with energy $\inputEnergy$ as defined in Section~\ref{sec:gaussian-cq-basics}.

The properties collected in Section~\ref{sec:gaussian-cq-basics} allow us to calculate the quantities defined in (\ref{eq:joint-entropy}) and (\ref{eq:output-entropy}) as
\begin{align}
\label{eq:gaussian-joint-entropy}
\jointEntropy
\overset{(\ref{eq:joint-entropy})}&{=}
-\Expectation_\inputRV
  \trace\left(
    \classicalQuantumChannel^{(\transmittivity,\quantumNoisePower)}(\inputRV)
    \log \classicalQuantumChannel^{(\transmittivity,\quantumNoisePower)}(\inputRV)
  \right)
\overset{(\ref{eq:gaussian-cq-trace-function})}{=}
-
\trace(
  \classicalQuantumChannel^{(\transmittivity,\quantumNoisePower)}(0)
  \log \classicalQuantumChannel^{(\transmittivity,\quantumNoisePower)}(0)
)
\overset{(\ref{eq:gaussian-cq-entropy})}{=}
\gordonFunction(\quantumNoisePower)
\\
\label{eq:gaussian-output-entropy}
\outputEntropy
\overset{(\ref{eq:output-entropy})}&{=}
\vonNeumannEntropy{\classicalQuantumChannel^{(\transmittivity,\quantumNoisePower)}_\inputDistribution}
\overset{(\ref{eq:gaussian-cq-entropy}),(\ref{eq:gaussian-cq-average})}{=}
\gordonFunction(\quantumNoisePower+\transmittivity\inputEnergy).
\end{align}

Next, we use the convergence behavior of the geometric series to argue that for $\renyiorder \in (0,\infty)$,
\begin{equation}
\label{eq:gaussian-cq-power}
\trace\left(\classicalQuantumChannel^{(\transmittivity,\quantumNoisePower)}(0)^\renyiorder\right)
=
\frac{1}{(\quantumNoisePower+1)^\renyiorder}
\sum_{\photonNumberNatural=0}^\infty
  \left(
    \frac{\quantumNoisePower}{\quantumNoisePower+1}
  \right)^{\renyiorder\photonNumberNatural}
=
\frac{1}{(\quantumNoisePower+1)^\renyiorder (1-\quantumNoisePower^\renyiorder/(\quantumNoisePower+1)^\renyiorder)}
=
\frac{1}{(\quantumNoisePower+1)^\renyiorder -\quantumNoisePower^\renyiorder}
\end{equation}

With this, we can calculate the quantities defined in (\ref{eq:conditional-renyi-entropy}) and (\ref{eq:renyi-entropy}) as
\begin{align*}
\conditionalRenyiEntropy{\renyiorder}
\overset{(\ref{eq:conditional-renyi-entropy})}&{=}
\frac{1}{1-\renyiorder}
\log\Expectation_\inputRV\left(
  \trace\left(\classicalQuantumChannel^{(\transmittivity,\quantumNoisePower)}(\inputRV)^\renyiorder\right)
\right)
\overset{(\ref{eq:gaussian-cq-trace-function})}{=}
\frac{1}{1-\renyiorder}
\log \trace\left(\classicalQuantumChannel^{(\transmittivity,\quantumNoisePower)}(0)^\renyiorder\right)
\overset{(\ref{eq:gaussian-cq-power})}{=}
\frac{1}{\renyiorder-1}
\log(
  (\quantumNoisePower+1)^\renyiorder -\quantumNoisePower^\renyiorder
)
=
\gordonFunction_\renyiorder(\quantumNoisePower)
\\
\outputRenyiEntropy{\renyiorder}
\overset{(\ref{eq:renyi-entropy})}&{=}
\frac{1}{1-\renyiorder}
\log
  \trace\left(\left(\classicalQuantumChannel^{(\transmittivity,\quantumNoisePower)}_\inputDistribution\right)^\renyiorder\right)
\overset{(\ref{eq:gaussian-cq-average}),(\ref{eq:gaussian-cq-power})}{=}
\gordonFunction_\renyiorder(\quantumNoisePower+\transmittivity\inputEnergy),
\end{align*}
where we use a modified version of the Gordon function
\begin{equation}
\label{eq:gordon-renyi}
\gordonFunction_\renyiorder:~ [0,\infty) \rightarrow [0,\infty)
,~~
\generalReal
\mapsto
\frac{\log((\generalReal+1)^\renyiorder-\generalReal^\renyiorder)}
     {\renyiorder-1}.
\end{equation}
From the considerations in Section~\ref{sec:typicality-preliminaries}, we know that with the convention $\gordonFunction_1 := \gordonFunction$, the map $(0,\infty) \rightarrow (0,\infty), \renyiorder \mapsto \gordonFunction_\renyiorder(\generalReal)$ is continuous.
We can also calculate the derivative
\begin{equation}
\label{eq:gordon-renyi-derivative}
\gordonFunction_\renyiorder'(\generalReal)
:=
\frac{\partial}{\partial\renyiorder} \gordonFunction_\renyiorder(\generalReal)
=
\frac{
  \frac{
    (\generalReal+1)^\renyiorder\log(\generalReal+1)
    -
    \generalReal^\renyiorder \log \generalReal
  }{
    (\generalReal+1)^\renyiorder-\generalReal^\renyiorder}
  -
  \gordonFunction_\renyiorder(\generalReal)
}{
  \renyiorder-1
}.
\end{equation}
With this, it is possible to numerically evaluate the functions $\atypicalTermsFunc_1,\atypicalTermsFunc_2,\atypicalTermsFunc_3,\atypicalTermsFunc_4$ defined in \eqref{eq:atypicalTermsFunc-start} to \eqref{eq:atypicalTermsFunc-end} for any given values of $\typicalityParameter,\blocklength$. Consequently, the evaluation of $\codingBoundFunc$ defined in \eqref{eq:codingBoundFunc} and $\resolvabilityBoundFunc$ defined in \eqref{eq:resolvabilityBoundFunc} requires only a one-dimensional optimization over all admissible values of the parameter $\typicalityParameter$. The evaluation of the bounds of Theorem~\ref{theorem:wiretap-cq-all-blocklengths} also requires calculation of $\costBoundFunc^{(\costFunction,\costConstraint)}$ defined in \eqref{eq:costBoundFunc} and another doubly exponential term that appears in \eqref{eq:wiretap-cq-prob-one}, however, we observe that in the regime in which we evaluate the bounds for the plots in the present paper, these terms are so small that they are negligible. This means that we can avoid a costly simultaneous optimization over the parameters $\proofconst_2$, $\proofconst_3$, $\proofconst_4$, and $\proofconst_5$ that appear in the statement of Theorem~\ref{theorem:wiretap-cq-all-blocklengths}. Instead, we only consider the terms that contain $\resolvabilityBoundFunc$ or $\codingBoundFunc$ and only check that the remaining terms are indeed neglible. For details on how small we ensure these terms are, we refer the reader to Section~\ref{sec:plots}. Full details and comments on how this is ensured can be found in the source code in the electronic supplement of this paper.

\subsection{Plots}
\label{sec:plots}
\begin{figure}
\centering
\begin{tikzpicture}
\begin{axis}[
  xmin=4000,
  xmax=100000,
  ymin=1e-5,
  ymax=1,
  ymode=log,
  xlabel={$\blocklength$},
  ylabel={$\securityNumber$},
  legend style={at={(1.2,.5)},anchor=west},
]
\addplot[ color=red, mark=*,]
coordinates {
  ( 4500 , 1.0913004454738087 )
  ( 5500 , 0.5522531180298952 )
  ( 6500 , 0.26291927625797895 )
  ( 7500 , 0.11946115381984325 )
  ( 8500 , 0.05229917215779195 )
  ( 9500 , 0.022210093044067004 )
  ( 10500 , 0.009194932789333555 )
  ( 11500 , 0.0037250554561747834 )
  ( 12500 , 0.001481119989012063 )
  ( 13500 , 0.0005793642993308349 )
  ( 14500 , 0.00022338676860809262 )
  ( 15500 , 8.503600481265565e-05 )
  ( 16500 , 3.200154637666275e-05 )
  ( 17500 , 1.1919443386568799e-05 )
  ( 18500 , 4.398258107760537e-06 )
  ( 19500 , 1.6091996397435808e-06 )
};
\addplot[ color=green, mark=*,]
coordinates {
  ( 8000 , 0.9532718622704868 )
  ( 10000 , 0.42089590982923125 )
  ( 12000 , 0.17230421776511737 )
  ( 14000 , 0.06665195177074387 )
  ( 16000 , 0.02466208848492168 )
  ( 18000 , 0.008802458366903337 )
  ( 20000 , 0.00304918803334357 )
  ( 22000 , 0.0010298231435886727 )
  ( 24000 , 0.0003403160347849025 )
  ( 26000 , 0.00011034882388674897 )
  ( 28000 , 3.518916587420508e-05 )
  ( 30000 , 1.1056580287057249e-05 )
  ( 32000 , 3.4283299262543426e-06 )
};
\addplot[ color=blue, mark=*,]
coordinates {
  ( 15000 , 1.2799941894436608 )
  ( 18000 , 0.7408714455552434 )
  ( 21000 , 0.40851348842871094 )
  ( 24000 , 0.21682088228040328 )
  ( 27000 , 0.11156480211027044 )
  ( 30000 , 0.055940681433562116 )
  ( 33000 , 0.02744025321747274 )
  ( 36000 , 0.013207258939410795 )
  ( 39000 , 0.006252249399023612 )
  ( 42000 , 0.0029167257140376096 )
  ( 45000 , 0.001343008932772206 )
  ( 48000 , 0.0006111680102118049 )
  ( 51000 , 0.00027518572174133326 )
  ( 54000 , 0.00012271258183171784 )
  ( 57000 , 5.4238452180748355e-05 )
  ( 60000 , 2.377889051642277e-05 )
  ( 63000 , 1.0346952635724127e-05 )
};
\addplot[ color=magenta, mark=*,]
coordinates {
  ( 25000 , 0.9161605204280387 )
  ( 30000 , 0.4776858351562014 )
  ( 35000 , 0.23668539241920877 )
  ( 40000 , 0.11265763023404682 )
  ( 45000 , 0.05189849911263443 )
  ( 50000 , 0.023265150344905085 )
  ( 55000 , 0.010190210508185754 )
  ( 60000 , 0.004374794012186998 )
  ( 65000 , 0.001845513102099672 )
  ( 70000 , 0.0007665531632807499 )
  ( 75000 , 0.0003140203613822133 )
  ( 80000 , 0.00012704823228562923 )
  ( 85000 , 5.0826007467173273e-05 )
  ( 90000 , 2.0125492806372263e-05 )
  ( 95000 , 7.894537585824014e-06 )
  ( 100000 , 3.0701089562485317e-06 )
  ( 105000 , 1.184444023630192e-06 )
};
\addplot[ color=brown, mark=*,]
coordinates {
  ( 35000 , 1.4061090668794576 )
  ( 40000 , 0.9814845182079178 )
  ( 45000 , 0.6676254186145156 )
  ( 50000 , 0.44445435095163005 )
  ( 55000 , 0.2905113281599896 )
  ( 60000 , 0.18690515733847862 )
  ( 65000 , 0.11859280494679904 )
  ( 70000 , 0.07433047001210587 )
  ( 75000 , 0.046080418970871236 )
  ( 80000 , 0.02828668453876748 )
  ( 85000 , 0.017209369489526085 )
  ( 90000 , 0.010385026782061919 )
  ( 95000 , 0.0062202121722694514 )
  ( 100000 , 0.0037001061254294805 )
  ( 105000 , 0.0021870520504821706 )
};
\legend{
  $\codebookRate=0.3$,
  $\codebookRate=0.35$,
  $\codebookRate=0.4$,
  $\codebookRate=0.42$,
  $\codebookRate=0.44$,
}
\end{axis}
\end{tikzpicture}
\caption{Decay of semantic security level $\securityNumber$ depending on the block length of the codebook for various rates $\codebookRate$. The decoding error is fixed at $\decodingError=0.01$.}
\label{fig:plot-security}
\end{figure}

\begin{figure}
\centering
\begin{tikzpicture}
\begin{axis}[
  xmin=4000,
  xmax=100000,
  ymin=1e-5,
  ymax=1,
  ymode=log,
  xlabel={$\blocklength$},
  ylabel={$\decodingError$},
  legend style={at={(1.2,.5)},anchor=west},
]
\addplot[ color=red, mark=*,]
coordinates {
  ( 5500 , 6.282949782960966 )
  ( 6500 , 3.0352700646734614 )
  ( 7500 , 1.0190768921016489 )
  ( 8500 , 0.2565077754266116 )
  ( 9500 , 0.05107977769315823 )
  ( 10500 , 0.008373504234366141 )
  ( 11500 , 0.0011649699628291084 )
  ( 12500 , 0.0001409027313584138 )
  ( 13500 , 1.5105323340878225e-05 )
  ( 14500 , 1.458138380450706e-06 )
};
\addplot[ color=green, mark=*,]
coordinates {
  ( 8000 , 8.452128275516927 )
  ( 10000 , 5.00637741169582 )
  ( 12000 , 1.7438065975266337 )
  ( 14000 , 0.3922706010753058 )
  ( 16000 , 0.06276014218407269 )
  ( 18000 , 0.007631108428808417 )
  ( 20000 , 0.0007399611228149032 )
  ( 22000 , 5.934121282862679e-05 )
  ( 24000 , 4.04890170458351e-06 )
  ( 26000 , 2.4040046873726575e-07 )
};
\addplot[ color=blue, mark=*,]
coordinates {
  ( 18000 , 7.401666089897124 )
  ( 21000 , 4.854708672002986 )
  ( 24000 , 2.3722994223856304 )
  ( 27000 , 0.9137713577754474 )
  ( 30000 , 0.28885989658736044 )
  ( 33000 , 0.07723592808192889 )
  ( 36000 , 0.017879591067691533 )
  ( 39000 , 0.003650288398542579 )
  ( 42000 , 0.0006671561809616723 )
  ( 45000 , 0.00011051286147575391 )
  ( 48000 , 1.6763366482344626e-05 )
  ( 51000 , 2.3488924941140053e-06 )
};
\addplot[ color=magenta, mark=*,]
coordinates {
  ( 40000 , 0.9272785858255793 )
  ( 45000 , 0.2528037332788309 )
  ( 50000 , 0.05600794609427583 )
  ( 55000 , 0.010405542789226655 )
  ( 60000 , 0.0016611356123499325 )
  ( 65000 , 0.00023231046972403075 )
  ( 70000 , 2.8910923121348946e-05 )
  ( 75000 , 3.2433577532412774e-06 )
  ( 80000 , 3.315578052397605e-07 )
  ( 85000 , 3.116898087723839e-08 )
};
\addplot[ color=brown, mark=*,]
coordinates {
  ( 60000 , 1.943921425924837 )
  ( 65000 , 1.0034622218274678 )
  ( 70000 , 0.4717873149368395 )
  ( 75000 , 0.20407976414572127 )
  ( 80000 , 0.08190574995421082 )
  ( 85000 , 0.030716426063072108 )
  ( 90000 , 0.010829231185386269 )
  ( 95000 , 0.00360794578702557 )
  ( 100000 , 0.0011411060301138184 )
  ( 105000 , 0.00034396893860631517 )
};
\legend{
  $\codebookRate=0.3$,
  $\codebookRate=0.35$,
  $\codebookRate=0.4$,
  $\codebookRate=0.42$,
  $\codebookRate=0.44$,
}
\end{axis}
\end{tikzpicture}
\caption{Decay of decoding error probability $\decodingError$ depending on the block length of the codebook for various rates. The semantic security level is fixed at $\securityNumber=0.01$.}
\label{fig:plot-decoding}
\end{figure}

\begin{figure}
\centering
\begin{tikzpicture}
\begin{axis}[
  xmin=0.15,
  xmax=0.53,
  xticklabel style={
    /pgf/number format/fixed,
  },
  ymin=0,
  ymax=1e5,
  xlabel={$\codebookRate$},
  ylabel={$\blocklength$},
  legend style={at={(1.2,.5)},anchor=west},
]
\addplot[ color=lime, mark=*,]
coordinates {
  ( 0.15 , 3628.0 )
  ( 0.16478260869565217 , 3936.0 )
  ( 0.17956521739130435 , 4287.0 )
  ( 0.1943478260869565 , 4687.0 )
  ( 0.20913043478260868 , 5147.0 )
  ( 0.22391304347826085 , 5679.0 )
  ( 0.23869565217391303 , 6300.0 )
  ( 0.2534782608695652 , 7030.0 )
  ( 0.2682608695652174 , 7897.0 )
  ( 0.2830434782608695 , 8936.0 )
  ( 0.2978260869565217 , 10198.0 )
  ( 0.31260869565217386 , 11752.0 )
  ( 0.32739130434782604 , 13693.0 )
  ( 0.3421739130434782 , 16165.0 )
  ( 0.3569565217391304 , 19379.0 )
  ( 0.37173913043478257 , 23667.0 )
  ( 0.38652173913043475 , 29567.0 )
  ( 0.4013043478260869 , 38006.0 )
  ( 0.4160869565217391 , 50683.0 )
  ( 0.4308695652173913 , 71004.0 )
  ( 0.44565217391304346 , 106639.0 )
  ( 0.46043478260869564 , 178009.0 )
  ( 0.4752173913043478 , 355783.0 )
  ( 0.49 , 1038562.0 )
};
\addplot[ color=purple, mark=*,]
coordinates {
  ( 0.15 , 4415.0 )
  ( 0.16478260869565217 , 4790.0 )
  ( 0.17956521739130435 , 5216.0 )
  ( 0.1943478260869565 , 5702.0 )
  ( 0.20913043478260868 , 6262.0 )
  ( 0.22391304347826085 , 6909.0 )
  ( 0.23869565217391303 , 7664.0 )
  ( 0.2534782608695652 , 8551.0 )
  ( 0.2682608695652174 , 9605.0 )
  ( 0.2830434782608695 , 10869.0 )
  ( 0.2978260869565217 , 12404.0 )
  ( 0.31260869565217386 , 14293.0 )
  ( 0.32739130434782604 , 16655.0 )
  ( 0.3421739130434782 , 19662.0 )
  ( 0.3569565217391304 , 23572.0 )
  ( 0.37173913043478257 , 28788.0 )
  ( 0.38652173913043475 , 35968.0 )
  ( 0.4013043478260869 , 46237.0 )
  ( 0.4160869565217391 , 61664.0 )
  ( 0.4308695652173913 , 86395.0 )
  ( 0.44565217391304346 , 129768.0 )
  ( 0.46043478260869564 , 216644.0 )
  ( 0.4752173913043478 , 433063.0 )
  ( 0.49 , 1264359.0 )
};
\addplot[ color=orange, mark=*,]
coordinates {
  ( 0.15 , 5344.0 )
  ( 0.16478260869565217 , 5798.0 )
  ( 0.17956521739130435 , 6314.0 )
  ( 0.1943478260869565 , 6903.0 )
  ( 0.20913043478260868 , 7581.0 )
  ( 0.22391304347826085 , 8365.0 )
  ( 0.23869565217391303 , 9280.0 )
  ( 0.2534782608695652 , 10356.0 )
  ( 0.2682608695652174 , 11633.0 )
  ( 0.2830434782608695 , 13166.0 )
  ( 0.2978260869565217 , 15027.0 )
  ( 0.31260869565217386 , 17317.0 )
  ( 0.32739130434782604 , 20180.0 )
  ( 0.3421739130434782 , 23826.0 )
  ( 0.3569565217391304 , 28568.0 )
  ( 0.37173913043478257 , 34894.0 )
  ( 0.38652173913043475 , 43602.0 )
  ( 0.4013043478260869 , 56058.0 )
  ( 0.4160869565217391 , 74772.0 )
  ( 0.4308695652173913 , 104776.0 )
  ( 0.44565217391304346 , 157402.0 )
  ( 0.46043478260869564 , 262820.0 )
  ( 0.4752173913043478 , 525458.0 )
  ( 0.49 , 1534401.0 )
};
\addplot[ color=violet, mark=*,]
coordinates {
  ( 0.15 , 6281.0 )
  ( 0.16478260869565217 , 6815.0 )
  ( 0.17956521739130435 , 7422.0 )
  ( 0.1943478260869565 , 8115.0 )
  ( 0.20913043478260868 , 8912.0 )
  ( 0.22391304347826085 , 9834.0 )
  ( 0.23869565217391303 , 10910.0 )
  ( 0.2534782608695652 , 12176.0 )
  ( 0.2682608695652174 , 13678.0 )
  ( 0.2830434782608695 , 15481.0 )
  ( 0.2978260869565217 , 17670.0 )
  ( 0.31260869565217386 , 20364.0 )
  ( 0.32739130434782604 , 23733.0 )
  ( 0.3421739130434782 , 28023.0 )
  ( 0.3569565217391304 , 33602.0 )
  ( 0.37173913043478257 , 41047.0 )
  ( 0.38652173913043475 , 51294.0 )
  ( 0.4013043478260869 , 65953.0 )
  ( 0.4160869565217391 , 87979.0 )
  ( 0.4308695652173913 , 123295.0 )
  ( 0.44565217391304346 , 185241.0 )
  ( 0.46043478260869564 , 309338.0 )
  ( 0.4752173913043478 , 618537.0 )
  ( 0.49 , 1806441.0 )
};
\draw[black, densely dashed] ({axis cs:0.5108022919521318,0}|-{rel axis cs:0,1}) -- ({axis cs:0.5108022919521318,0}|-{rel axis cs:0,0});
\addlegendimage{line legend,black,densely dashed}
\legend{
  {$\decodingError=0.01, \securityNumber=0.01$},
  {$\decodingError=0.0001, \securityNumber=0.01$},
  {$\decodingError=0.01, \securityNumber=0.0001$},
  {$\decodingError=0.0001, \securityNumber=0.0001$},
  {upper achievable rates bound}
}
\end{axis}
\end{tikzpicture}
\caption{Block length necessary to achieve a given decoding error $\decodingError$, semantic security level $\securityNumber$ and rate $\codebookRate$.}
\label{fig:plot-blocklength}
\end{figure}

We plot various quantities of interest for a model of an optical communication system with a noiseless source as described in Section~\ref{sec:physical-model}. In order to find reasonable system parameters, we make a series of rough estimates of what values these parameters may have in a realistic communication system. When a transmitter communicates a line of sight optical signal to a receiver, the transmittivity $\transmittivity \in [0,1]$ can be roughly estimated as 
\[
    \transmittivity = \multConstant \cdot\left(\frac{\diameter}{\distance}\right)^2
\]
where $\multConstant$ is a system-dependent constant, $\diameter$ the receiver diameter and $\distance$ the distance between transmitter and receiver~\cite{prokesLoss2009}. In slightly non-optimal situations where fog disturbs the link, this transmittivity in a free-space optical system can easily be as low as $\transmittivity=10^{-5}$ already at distances $\distance=\qty{1}{\km}$ \cite{prokesLoss2009}. For our plots, we choose $\transmittivity_1 := 10^{-5}$ and $\transmittivity_2 := 6 \cdot 10^{-6}$ to obtain a scenario in which the eavesdropper is at a slightly greater distance from the sender than the legitimate receiver (or equivalently, guaranteed to be unable to place its receive antenna close enough to the center of the optical beam, cf.~\cite{vazquez-castro}), and otherwise uses equivalent equipment. For the average transmit energy, we note that the number of photons per channel use can be estimated by considering a laser with $\qty{1}{\W}$ output power at $\qty{1550}{\nm}$. This system will emit in the order of $10^{19}$ photons per second, and typically use a modulation format such that $10^{10}$ pulses are emitted per second. Thus, $\inputEnergy := 10^9$ photons per channel use are a fair estimate. We assume that the number of noise photons per channel use is around $10^{-5}$, which is at the lower end of the plausible range of values. At a wave length of $\qty{1550}{\nm}$ and a baud rate of $10^{10}$ pulses per second, this is equivalent to a background noise power of about $\qty{1.5e-14}{\W}$. This means that in our example, we choose $\receiverNoise_\alice := \receiverNoise_\bob := 10^{-5}$.

This yields the system parameters $\transmittivity_\bob = \transmittivity_1 = 10^{-5}$, $\transmittivity_\eve = \transmittivity_2(1-\transmittivity_1) \approx 6 \cdot 10^{-6}$, $\quantumNoisePower_\bob = \quantumNoisePower_\eve = 10^{-5}$. By Theorem~\ref{theorem:wiretap-cq} and (\ref{eq:gaussian-joint-entropy}), (\ref{eq:gaussian-output-entropy}), we obtain a bound
\[
\holevoInformation{\inputDistribution}{\bobcqChannel}
-
\holevoInformation{\inputDistribution}{\eveChannel}
=
\gordonFunction(\quantumNoisePower_\bob + \transmittivity_\bob \inputEnergy)
-
\gordonFunction(\quantumNoisePower_\bob)
-
\gordonFunction(\quantumNoisePower_\eve + \transmittivity_\eve \inputEnergy)
+
\gordonFunction(\quantumNoisePower_\eve)
\approx 0.5108
\]
on the achievable rates. When calculating the achievable bounds, we do not optimize over the doubly exponential terms that appear in \eqref{eq:wiretap-cq-prob-one}, but we make sure they do not exceed $\exp(-100)$ in value. This threshold is chosen arbitrarily, but we believe it is small enough to convince the reader that neglecting this term does not have more severe effects than other limitations that are inherent in numerical evaluations such as the machine precision of the computer used. For the cost constraint and the concentration of the security level and decoding error around its expectation, we also need to allow for a little bit of slack due to the nature of the results we evaluate. We allow the security level and decoding error to differ by at most $10^{-8}$ from the expected value and the individual code words to exceed the expected energy per code word by at most $10\%$. Again, the choice of these threshold values is somewhat arbitrary, but we believe the values are small enough to convince the reader that the difference between expected and actual security level and decoding error does not significantly impact the accuracy of the plots given that we do not display values of security level and decoding error smaller than $10^{-5}$ and the possible additional energy consumption of some of the code words is not so large that it would meaningfully impact the practicality of such a code. In Fig.~\ref{fig:plot-security} and~\ref{fig:plot-decoding}, it can be seen how the security level and decoding error decay with increasing block length for various rates. As expected, the decay is always at least exponential, but the slope of the lines is steeper for larger gaps to the upper bound of achievable rates. In Fig.~\ref{fig:plot-blocklength}, we plot the block lengths necessary to achieve a few selected combinations of decoding error and security level in dependence of the rate. It can be seen that the necessary block length increases sharply when the rate gets close to the upper bound.

\appendix
\subsection{Proofs of Technical Lemmas}
\label{appendix:qit}
In this appendix, we give the proofs for the technical lemmas related to quantum information theory and stochastics that are omitted in the main part of the paper.

\begin{proof}[Proof of Lemma~\ref{lemma:q-resolvability-typical-terms}]
For \ref{item:q-resolvability-typical-terms-joint}), we use the presence of the indicator in the definition of $\jointTypicalityOperator_{\typicalityParameter,\blocklength}(\inputAlphabetElement^\blocklength)$ to bound
\begin{align*}
\jointTypicalityOperator_{\typicalityParameter,\blocklength}
  (\inputAlphabetElement^\blocklength)
\classicalQuantumChannel^\blocklength
  (\inputAlphabetElement^\blocklength)
\jointTypicalityOperator_{\typicalityParameter,\blocklength}
  (\inputAlphabetElement^\blocklength)
&=
\sum_{\classicalOutputIndex^\blocklength \in \naturals^\blocklength}
  \indicatorFunction{\jointTypicalitySet_{\typicalityParameter,\blocklength}}
    (\inputAlphabetElement^\blocklength,\classicalOutputIndex^\blocklength)
  \inputDistribution(\classicalOutputIndex^\blocklength | \inputAlphabetElement^\blocklength)
  \rankOneOperator{
    \eigenvector_{\classicalOutputIndex^\blocklength | \inputAlphabetElement^\blocklength}
  }
\\
&\leq
\sum_{\classicalOutputIndex^\blocklength \in \naturals^\blocklength}
  \indicatorFunction{\jointTypicalitySet_{\typicalityParameter,\blocklength}}
    (\inputAlphabetElement^\blocklength,\classicalOutputIndex^\blocklength)
  \exp\big(-\blocklength(\jointEntropy - \typicalityParameter)\big)
  \rankOneOperator{
    \eigenvector_{\classicalOutputIndex^\blocklength | \inputAlphabetElement^\blocklength}
  }
\\
&=
\exp\big(
  -
  \blocklength
  (\jointEntropy - \typicalityParameter)
\big)
\jointTypicalityOperator_{\typicalityParameter,\blocklength}
  (\inputAlphabetElement^\blocklength).
\end{align*}

For \ref{item:q-resolvability-typical-terms-joint-trace}, we fix some $\inputAlphabetElement^\blocklength$ and note
\begin{equation}
\label{eq:q-resolvability-typical-terms-joint-cardinality}
\trace \jointTypicalityOperator_{\typicalityParameter,\blocklength}(\inputAlphabetElement^\blocklength)
=
\trace
  \sum\limits_{\classicalOutputIndex^\blocklength \in \naturals^\blocklength}
    \indicatorFunction{\jointTypicalitySet_{\typicalityParameter,\blocklength}}
      (
        \inputAlphabetElement^\blocklength,
        \classicalOutputIndex^\blocklength
      )
    \rankOneOperator{\eigenvector_{
      \classicalOutputIndex^\blocklength
      |
      \inputAlphabetElement^\blocklength
    }}
=
\cardinality{
  \{
    \classicalOutputIndex^\blocklength:~
    (\inputAlphabetElement^\blocklength, \classicalOutputIndex^\blocklength) \in \jointTypicalitySet_{\typicalityParameter, \blocklength}
  \}
}.
\end{equation}
For every $(\inputAlphabetElement^\blocklength, \classicalOutputIndex^\blocklength) \in \jointTypicalitySet_{\typicalityParameter, \blocklength}$, we have
\[
\inputDistribution(\classicalOutputIndex^\blocklength | \inputAlphabetElement^\blocklength)
>
\exp\big(
  -
  \blocklength
  (\jointEntropy + \typicalityParameter)
\big).
\]
This allows us to argue
\[
1
\geq
\sum_{\classicalOutputIndex^\blocklength:~ (\inputAlphabetElement^\blocklength, \classicalOutputIndex^\blocklength) \in \jointTypicalitySet_{\typicalityParameter, \blocklength}}
  \inputDistribution(\classicalOutputIndex^\blocklength | \inputAlphabetElement^\blocklength)
>
\sum_{\classicalOutputIndex^\blocklength:~ (\inputAlphabetElement^\blocklength, \classicalOutputIndex^\blocklength) \in \jointTypicalitySet_{\typicalityParameter, \blocklength}}
  \exp\big(
    -
    \blocklength
    (\jointEntropy + \typicalityParameter)
  \big)
=
\cardinality{
  \{
    \classicalOutputIndex^\blocklength:~
    (\inputAlphabetElement^\blocklength, \classicalOutputIndex^\blocklength) \in \jointTypicalitySet_{\typicalityParameter, \blocklength}
  \}
}
\exp\big(
  -
  \blocklength
  (\jointEntropy + \typicalityParameter)
\big),
\]
from which item \ref{item:q-resolvability-typical-terms-joint-trace}) follows by (\ref{eq:q-resolvability-typical-terms-joint-cardinality}).

For \ref{item:q-resolvability-typical-terms-output}), we have a calculation very similar to the one for \ref{item:q-resolvability-typical-terms-joint})
\begin{align*}
\outputTypicalityOperator_{\typicalityParameter,\blocklength}
\classicalQuantumChannel_\inputDistribution^{\otimes \blocklength}
\outputTypicalityOperator_{\typicalityParameter,\blocklength}
&=
\sum_{\classicalOutputIndex^\blocklength \in \naturals^\blocklength}
  \indicatorFunction{\outputTypicalitySet_{\typicalityParameter,\blocklength}}
    (\classicalOutputIndex^\blocklength)
  \outputDistribution(\classicalOutputIndex^\blocklength)
  \rankOneOperator{\eigenvector_{\classicalOutputIndex^\blocklength}}
\\
&\leq
\sum_{\classicalOutputIndex^\blocklength \in \naturals^\blocklength}
  \indicatorFunction{\outputTypicalitySet_{\typicalityParameter,\blocklength}}
    (\classicalOutputIndex^\blocklength)
  \exp\big(
    -
    \blocklength
    (\outputEntropy - \typicalityParameter)
  \big)
  \rankOneOperator{\eigenvector_{\classicalOutputIndex^\blocklength}}
\\
&=
\exp\big(
  -
  \blocklength
  (\outputEntropy - \typicalityParameter)
\big)
\outputTypicalityOperator_{\typicalityParameter,\blocklength}.
\end{align*}
For \ref{item:q-resolvability-typical-terms-output-trace}), the argument is very similar to the one for \ref{item:q-resolvability-typical-terms-joint-trace}). We note
\begin{equation}
\label{eq:q-resolvability-typical-terms-output-cardinality}
\trace \outputTypicalityOperator_{\typicalityParameter,\blocklength}
=
\trace \sum_{\classicalOutputIndex^\blocklength \in \outputTypicalitySet_{\typicalityParameter, \blocklength}}
  \rankOneOperator{\eigenvector_{\classicalOutputIndex^\blocklength}}
=
\cardinality{
  \outputTypicalitySet_{\typicalityParameter, \blocklength}
}.
\end{equation}
For every $\classicalOutputIndex^\blocklength \in \outputTypicalitySet_{\typicalityParameter, \blocklength}$, we have
\[
\outputDistribution(\classicalOutputIndex^\blocklength)
>
\exp\big(
  -
  \blocklength
  (\outputEntropy + \typicalityParameter)
\big).
\]
This allows us to argue
\[
1
\geq
\outputDistribution(
  \classicalOutputIndex^\blocklength
  \in
  \outputTypicalitySet_{\typicalityParameter,\blocklength}
)
=
\sum_{\classicalOutputIndex^\blocklength \in \outputTypicalitySet_{\typicalityParameter, \blocklength}}
  \outputDistribution(\classicalOutputIndex^\blocklength)
>
\sum_{\classicalOutputIndex^\blocklength \in \outputTypicalitySet_{\typicalityParameter, \blocklength}}
  \exp\big(
    -
    \blocklength
    (\outputEntropy + \typicalityParameter)
  \big)
=
\cardinality{
  \outputTypicalitySet_{\typicalityParameter, \blocklength}
}
\exp\big(
  -
  \blocklength
  (\outputEntropy + \typicalityParameter)
\big),
\]
from which item \ref{item:q-resolvability-typical-terms-output-trace}) follows by (\ref{eq:q-resolvability-typical-terms-output-cardinality}).
\end{proof}

\begin{proof}[Proof of Lemma~\ref{lemma:q-resolvability-atypical-terms}]
We bound the trace in (\ref{eq:q-resolvability-atypical-terms}) as
\begin{align}
\nonumber
\trace\left(
  \classicalQuantumChannel^\blocklength(\inputRV^\blocklength)
  \adjoint{
    \typicalityOperatorProduct_{\typicalityParameter,\blocklength}
      (\inputRV^\blocklength)
  }
\right)
\overset{(\ref{eq:typicality-operator-product})}&{=}
\trace\left(
  \classicalQuantumChannel^\blocklength(\inputRV^\blocklength)
  \jointTypicalityOperator_{\typicalityParameter,\blocklength}
    (\inputRV^\blocklength)
  \outputTypicalityOperator_{\typicalityParameter,\blocklength}
\right)
\\
\nonumber
&=
\trace\left(
  \classicalQuantumChannel^\blocklength(\inputRV^\blocklength)
  \jointTypicalityOperator_{\typicalityParameter,\blocklength}
    (\inputRV^\blocklength)
\right)
-
\trace\left(
  \classicalQuantumChannel^\blocklength(\inputRV^\blocklength)
  \jointTypicalityOperator_{\typicalityParameter,\blocklength}
    (\inputRV^\blocklength)
  (\identityOperator-\outputTypicalityOperator_{\typicalityParameter,\blocklength})
\right)
\\
\nonumber
\overset{(a)}&{=}
\trace\left(
  \classicalQuantumChannel^\blocklength(\inputRV^\blocklength)
  \jointTypicalityOperator_{\typicalityParameter,\blocklength}
    (\inputRV^\blocklength)
\right)
-
\trace\left(
  (\identityOperator-\outputTypicalityOperator_{\typicalityParameter,\blocklength})
  \sqrt{\classicalQuantumChannel^\blocklength(\inputRV^\blocklength)}
  \jointTypicalityOperator_{\typicalityParameter,\blocklength}
    (\inputRV^\blocklength)
  \sqrt{\classicalQuantumChannel^\blocklength(\inputRV^\blocklength)}
  (\identityOperator-\outputTypicalityOperator_{\typicalityParameter,\blocklength})
\right)
\\
\nonumber
&\overset{(b)}{\geq}
\trace\left(
  \classicalQuantumChannel^\blocklength(\inputRV^\blocklength)
  \jointTypicalityOperator_{\typicalityParameter,\blocklength}
    (\inputRV^\blocklength)
\right)
-
\trace\left(
  \classicalQuantumChannel^\blocklength(\inputRV^\blocklength)
  (\identityOperator-\outputTypicalityOperator_{\typicalityParameter,\blocklength})
\right)
\\
\label{eq:q-resolvability-atypical-terms-split}
&=
1
-
\trace\left(
  \classicalQuantumChannel^\blocklength(\inputRV^\blocklength)
  (
    \identityOperator
    -
    \jointTypicalityOperator_{\typicalityParameter,\blocklength}
      (\inputRV^\blocklength)
  )
\right)
-
\trace\left(
  \classicalQuantumChannel^\blocklength(\inputRV^\blocklength)
  (\identityOperator-\outputTypicalityOperator_{\typicalityParameter,\blocklength})
\right)
\end{align}
where (a) and (b) use (\ref{eq:typicality-operators-basic-properties}) along with the cyclic property of the trace and (b) additionally uses $\jointTypicalityOperator_{\typicalityParameter,\blocklength}(\inputRV^\blocklength) \leq \identityOperator$, which can be verified with the spectral representation (\ref{eq:spectral-decomposition-joint-typicality-operator}).

The trace in (\ref{eq:q-resolvability-atypical-terms-decoder}) can be bounded as
\begin{align}
\nonumber
\trace\left(
  \classicalQuantumChannel(\inputRV^\blocklength)
  \typicalityOperatorProductDecoder_{\typicalityParameter,\blocklength}
    (\inputRV^\blocklength)
\right)
\overset{(\ref{eq:typicality-operator-product-decoder})}&{=}
\trace\Big(
  \classicalQuantumChannel(\inputRV^\blocklength)
  \outputTypicalityOperator_{\typicalityParameter,\blocklength}
  \jointTypicalityOperator_{\typicalityParameter,\blocklength}(\inputRV^\blocklength)
  \outputTypicalityOperator_{\typicalityParameter,\blocklength}
\Big)
\\
\nonumber
\overset{(a)}&{\geq}
\trace\Big(
  \classicalQuantumChannel(\inputRV^\blocklength)
  \jointTypicalityOperator_{\typicalityParameter,\blocklength}(\inputRV^\blocklength)
\Big)
-
2
\trace\Big(
  \classicalQuantumChannel(\inputRV^\blocklength)
  (\identityOperator - \outputTypicalityOperator_{\typicalityParameter,\blocklength})
\Big)
\\
\label{eq:decoder-atypical-terms-split}
&=
1
-
\trace\Big(
  \classicalQuantumChannel(\inputRV^\blocklength)
  (\identityOperator - \jointTypicalityOperator_{\typicalityParameter,\blocklength}(\inputRV^\blocklength))
\Big)
-
2
\trace\Big(
  \classicalQuantumChannel(\inputRV^\blocklength)
  (\identityOperator - \outputTypicalityOperator_{\typicalityParameter,\blocklength})
\Big),
\end{align}
where (a) is due to~\cite[Lemma 6]{hayashi2003general}.

The same two terms appear in (\ref{eq:q-resolvability-atypical-terms-split}) and (\ref{eq:decoder-atypical-terms-split}), and we will bound their expectations separately. For the first term, we recall the definition (\ref{eq:conditional-renyi-entropy}). Since $\trace(\classicalQuantumChannel(\inputRV)^\renyiorder)$ is nonincreasing in $\renyiorder$, it is clear that (\ref{eq:technical-assumptions}) and Lemma~\ref{lemma:bochner-integral-basics}-\ref{item:bochner-integral-basics-existence} ensure that $\conditionalRenyiEntropy{\renyiorder} < \infty$ for all $\renyiorder \in [\renyiorder_{\min},\infty)$. With that, we have for every $\renyiorder_1 \in (1,\infty)$ and $\renyiorder_2 \in [\renyiorder_{\min},1)$
\begin{align}
\nonumber
&\hphantom{{}={}}
\Expectation \trace\left(
  \classicalQuantumChannel^\blocklength(\inputRV^\blocklength)
  (
    \identityOperator
    -
    \jointTypicalityOperator_{\typicalityParameter,\blocklength}
    (\inputRV^\blocklength)
  )
\right)
\\
\nonumber
&=
\Expectation
  \trace\left(
    \sum\limits_{\classicalOutputIndex^\blocklength \in \naturals^\blocklength}
      \indicatorFunction{
        (\inputAlphabet^\blocklength \times \naturals^\blocklength)
        \setminus
        \jointTypicalitySet_{\typicalityParameter,\blocklength}
      }
        (
         \inputRV^\blocklength,
          \classicalOutputIndex^\blocklength
        )
      \inputDistribution(\classicalOutputIndex^\blocklength | \inputRV^\blocklength)
      \rankOneOperator{\eigenvector_{
        \classicalOutputIndex^\blocklength
        |
        \inputRV^\blocklength
      }}
  \right)
\displaybreak[0] \\
\nonumber
&=
\inputDistribution(
  (\inputAlphabet^\blocklength \times \naturals^\blocklength)
  \setminus
  \jointTypicalitySet_{\typicalityParameter,\blocklength}
)
\displaybreak[0] \\
\nonumber
&=
\inputDistribution\big(
  -\log(\inputDistribution(\outputRV^\blocklength | \inputRV^\blocklength))
  \leq
  \blocklength(\jointEntropy - \typicalityParameter)
\big)
+
\inputDistribution\big(
  -\log(\inputDistribution(\outputRV^\blocklength | \inputRV^\blocklength))
  \geq
  \blocklength(\jointEntropy + \typicalityParameter)
\big)
\displaybreak[0] \\
\nonumber
&=
\inputDistribution\big(
  \inputDistribution(\outputRV^\blocklength | \inputRV^\blocklength)^{\renyiorder_1 - 1}
  \geq
  \exp\left(
    -\blocklength(\renyiorder_1-1)(\jointEntropy - \typicalityParameter)
  \right)
\big)
+
\inputDistribution\big(
  \inputDistribution(\outputRV^\blocklength | \inputRV^\blocklength)^{\renyiorder_2-1}
  \geq
  \exp\left(
    -\blocklength(\renyiorder_2-1)(\jointEntropy + \typicalityParameter)
  \right)
\big)
\displaybreak[0] \\
\nonumber
\overset{(a)}&{\leq}
\Expectation_P\left(
  \inputDistribution(\outputRV^\blocklength | \inputRV^\blocklength)^{\renyiorder_1 - 1}
\right)
\exp\left(
  \blocklength(\renyiorder_1-1)(\jointEntropy - \typicalityParameter)
\right)
+
\Expectation_P\left(
  \inputDistribution(\outputRV^\blocklength | \inputRV^\blocklength)^{\renyiorder_2-1}
\right)
\exp\left(
  \blocklength(\renyiorder_2-1)(\jointEntropy + \typicalityParameter)
\right)
\\
\label{eq:q-resolvability-atypical-terms-joint}
&=
\exp\left(
  -\blocklength(\renyiorder_1-1)(\conditionalRenyiEntropy{\renyiorder_1} + \typicalityParameter - \jointEntropy)
\right)
+
\exp\left(
  -\blocklength(1-\renyiorder_2)(\jointEntropy + \typicalityParameter - \conditionalRenyiEntropy{\renyiorder_2})
\right),
\end{align}
where the inequality step (a) is due to Markov's inequality.

Similarly, we recall the definition (\ref{eq:renyi-entropy}) and note that if $\classicalQuantumChannel_\inputDistribution^\renyiorder \in \traceClass{\hilbertSpace}$, then $\outputRenyiEntropy{\renyiorder} < \infty$. For $\renyiorder_3 \in (1,\infty)$ and $\renyiorder_4 \in [\renyiorder_{\min}, 1)$, we obtain
\begin{align}
\nonumber
&\hphantom{{}={}}
\Expectation \trace\left(
  \classicalQuantumChannel^\blocklength(\inputRV^\blocklength)
  (\identityOperator-\outputTypicalityOperator_{\typicalityParameter,\blocklength})
\right)
\\
\nonumber
&=
\trace\left(
  \classicalQuantumChannel_\inputDistribution^{\otimes \blocklength}
  (\identityOperator-\outputTypicalityOperator_{\typicalityParameter,\blocklength})
\right)
\displaybreak[0]
\\
\nonumber
&=
\trace\left(
  \sum_{\classicalOutputIndex^\blocklength \in \naturals}
    \indicatorFunction{
      \naturals^\blocklength
      \setminus
      \outputTypicalitySet_{\typicalityParameter,\blocklength}
    }
  \outputDistribution(\classicalOutputIndex^\blocklength)
  \rankOneOperator{\eigenvector_{\classicalOutputIndex^\blocklength}}
\right)
\displaybreak[0]
\\
\nonumber
&=
\outputDistribution(
  \naturals^\blocklength
  \setminus
  \outputTypicalitySet_{\typicalityParameter,\blocklength}
)
\displaybreak[0]
\\
\nonumber
&=
\outputDistribution\Big(
  -\log
  \outputDistribution(
    \outputRV^\blocklength
  )
  \leq
  \blocklength(\outputEntropy - \typicalityParameter)
\Big)
+
\outputDistribution\Big(
  -\log
  \outputDistribution(
    \outputRV^\blocklength
  )
  \geq
  \blocklength(\outputEntropy + \typicalityParameter)
\Big)
\displaybreak[0]
\\
\nonumber
&=
\outputDistribution\Big(
  \outputDistribution(
    \outputRV^\blocklength
  )^{\renyiorder_3-1}
  \geq
  \exp\big(
    -
    \blocklength
    (\renyiorder_3-1)
    (\outputEntropy - \typicalityParameter)
  \big)
\Big)
+
\outputDistribution\Big(
  \outputDistribution(
    \outputRV^\blocklength
  )^{\renyiorder_4-1}
  \geq
  \exp\big(
    -
    \blocklength
    (\renyiorder_4-1)
    (\outputEntropy + \typicalityParameter)
  \big)
\Big)
\displaybreak[0]
\\
\nonumber
&\leq
\Expectation_\outputDistribution\Big(
  \outputDistribution(
    \outputRV^\blocklength
  )^{\renyiorder_3-1}
\Big)
\exp\big(
  \blocklength
  (\renyiorder_3-1)
  (\outputEntropy - \typicalityParameter)
\big)
+
\Expectation_\outputDistribution\Big(
  \outputDistribution(
    \outputRV^\blocklength
  )^{\renyiorder_4-1}
\Big)
\exp\big(
  \blocklength
  (\renyiorder_4-1)
  (\outputEntropy + \typicalityParameter)
\big)
\\
\label{eq:q-resolvability-atypical-terms-output}
&=
\exp\big(
  -
  \blocklength
  (\renyiorder_3-1)
  (
    \outputRenyiEntropy{\renyiorder_3}
    +
    \typicalityParameter
    -
    \outputEntropy
  )
\big)
+
\exp\big(
  -
  \blocklength
  (1-\renyiorder_4)
  (
    \outputEntropy
    +
    \typicalityParameter
    -
    \outputRenyiEntropy{\renyiorder_4}
  )
\big).
\end{align}

Optimizing over the choices of $\renyiorder_1$, $\renyiorder_2$, $\renyiorder_3$, and $\renyiorder_4$, \eqref{eq:q-resolvability-atypical-terms} now follows from \eqref{eq:q-resolvability-atypical-terms-split}, \eqref{eq:q-resolvability-atypical-terms-joint}, and \eqref{eq:q-resolvability-atypical-terms-output}, and \eqref{eq:q-resolvability-atypical-terms-decoder} follows from \eqref{eq:decoder-atypical-terms-split}, \eqref{eq:q-resolvability-atypical-terms-joint}, and \eqref{eq:q-resolvability-atypical-terms-output}. To argue \eqref{eq:atypicalTermsFunc-bounds}, we invoke Lemma~\ref{lemma:renyi-entropy-basics} to argue that there exist $\renyiorder_1 \in (1,\infty)$ with $\conditionalRenyiEntropy{\renyiorder_1} > \jointEntropy - \typicalityParameter$, $\renyiorder_2 \in [\renyiorder_{\min},1)$ with $\conditionalRenyiEntropy{\renyiorder_2} < \jointEntropy + \typicalityParameter$, $\renyiorder_3 \in (1,\infty)$ with $\outputRenyiEntropy{\renyiorder_3} > \outputEntropy - \typicalityParameter$, and $\renyiorder_4 \in [\renyiorder_{\min},1)$ with $\outputRenyiEntropy{\renyiorder_4} < \outputEntropy + \typicalityParameter$. These properties ensure we can choose
\begin{align*}
\finalconst_1 \in \Big(
  0,
  \min\big(
    &(\renyiorder_1-1)(\conditionalRenyiEntropy{\renyiorder_1} + \typicalityParameter - \jointEntropy)
    ,
    \\
    &(1-\renyiorder_2)(\jointEntropy + \typicalityParameter - \conditionalRenyiEntropy{\renyiorder_2})
    ,
    \\
    &(\renyiorder_3-1)
    (
      \outputRenyiEntropy{\renyiorder_3}
      +
      \typicalityParameter
      -
      \outputEntropy
    )
    ,
    \\
    &(1-\renyiorder_4)
    (
      \outputEntropy
      +
      \typicalityParameter
      -
      \outputRenyiEntropy{\renyiorder_4}
    )
  \big)
\Big).
\end{align*}
This choice of $\finalconst_1$ clearly satisfies
\begin{align*}
\finalconst_1
&<
(\renyiorder_1-1)(\conditionalRenyiEntropy{\renyiorder_1} + \typicalityParameter - \jointEntropy)
<
\sup_{\renyiorder \in (1,\infty)}
  (\renyiorder-1)(\conditionalRenyiEntropy{\renyiorder} + \typicalityParameter - \jointEntropy)
\\
\finalconst_1
&<
(1-\renyiorder_2)(\jointEntropy + \typicalityParameter - \conditionalRenyiEntropy{\renyiorder_2})
<
\sup_{\renyiorder \in [\renyiorder_{\min},1)}
  (1-\renyiorder)(\jointEntropy + \typicalityParameter - \conditionalRenyiEntropy{\renyiorder_2})
\displaybreak[0] \\
\finalconst_1
&<
(\renyiorder_3-1)
(
  \outputRenyiEntropy{\renyiorder_3}
  +
  \typicalityParameter
  -
  \outputEntropy
)
<
\sup_{\renyiorder \in (1,\infty)}
  (\renyiorder-1)
  (
    \outputRenyiEntropy{\renyiorder}
    +
    \typicalityParameter
    -
    \outputEntropy
  )
\\
\finalconst_1
&<
(1-\renyiorder_4)
(
  \outputEntropy
  +
  \typicalityParameter
  -
  \outputRenyiEntropy{\renyiorder_4}
)
<
\sup_{\renyiorder \in [\renyiorder_{\min},1)}
  (1-\renyiorder)
  (
    \outputEntropy
    +
    \typicalityParameter
    -
    \outputRenyiEntropy{\renyiorder}
  ).
\end{align*}
Since $\finalconst_1$ is strictly less than any of the exponents that appear in the definitions of $\atypicalTermsFunc_1$, $\atypicalTermsFunc_2$, $\atypicalTermsFunc_3$, and $\atypicalTermsFunc_4$, these choices ensure that \eqref{eq:atypicalTermsFunc-bounds} is also satisfied as long as $\blocklength$ is large enough.
\end{proof}

\begin{proof}[Proof of Lemma~\ref{lemma:hninequality}]
\ref{item:hninequality-bounded}) In the proof of the first claim, we will use the following fact: If $T\in \boundedOperators{\hilbertSpace}$, $T\ge 0$, and $\image(T)$ is closed, then $\image(T)=\image(\sqrt{T})$. To prove this, first observe that clearly $\image(T)\subseteq \image(\sqrt{T})$ holds since $T= \sqrt{T}\cdot \sqrt{T}$. Therefore, it suffices to show $\image(\sqrt{T})\subseteq \image(T)$. The last inclusion is a consequence of the elementary relations $\kernel(T)=\kernel(\sqrt{T})$, and $\kernel(T)=\image(T)^{\perp}$, where $\kernel(T)=\{\hilbertSpaceElement\in\hilbertSpace: T\hilbertSpaceElement=0  \}$ and $M^{\perp}$ denotes the orthogonal complement of a subspace $M\subseteq\hilbertSpace$:
\begin{equation}\label{eq:images-1}
    \image(T)^{\perp}=\kernel(T)=\kernel(\sqrt{T})=\image(\sqrt{T})^{\perp}.
\end{equation}
Passing to the orthogonal complement in (\ref{eq:images-1}), using our assumption that $\image(T)$ is closed, and the fact that for every subspace $M\subseteq \hilbertSpace$ we have $M^{\perp \perp}= \overline{M}$, we obtain
\begin{equation*}
  \image(T)= \image(\sqrt{T})^{\perp \perp}=\overline{\image(\sqrt{T})}\supseteq \image(\sqrt{T}).
\end{equation*}

Since $\image(\generalOperator+\generalOperatorTwo)$ is closed by assumption, we have $\image(\generalOperator+ \generalOperatorTwo)=\image(\sqrt{\generalOperator+\generalOperatorTwo})$. Consequently, $\image(\sqrt{\generalOperator+\generalOperatorTwo})$ is closed as well and \cite[Proposition 2.4]{engl2000regularization} states that in this case the Moore-Penrose pseudoinverse $\pseudoinverse{\sqrt{\generalOperator+\generalOperatorTwo}}$ is a bounded linear operator, i.e. $\pseudoinverse{\sqrt{\generalOperator+\generalOperatorTwo}} \in \boundedOperators{\hilbertSpace}$.\\
\ref{item:hninequality-main}) The operator inequality in the lemma follows from the proof of \cite[Lemma 2]{hayashi2003general}.
\end{proof}
\begin{proof}[Proof of Lemma~\ref{lemma:bad-codewords}]
The compatibility of the cost constraint by Definition~\ref{def:cost-constraint} clearly implies that $\generalpmf(\generalReal) < \infty$ in an interval around $0$ and that $\generalpmf'(0) = \Expectation_\inputDistribution(\costFunction(\inputRV)-\costConstraint) < 0$. Since $\generalpmf(0) = 1$, this means that there is $\hat{\generalReal} > 0$ with $\generalpmf(\hat{\generalReal}) < 1$ and hence, $\proofconst_1>0$.

Furthermore, for any $\hat{\generalReal} > 0$ with $\generalpmf(\hat{\generalReal}) < \infty$, we have
\begin{align*}
\Probability_\codebook\left(
  \sum_{\blockindex=1}^\blocklength \costFunction(\codebook(\codewordIndex)_\blockindex) > \blocklength\costConstraint
\right)
&=
\Probability_\codebook\left(
  \sum_{\blockindex=1}^\blocklength
    \exp\left(
      \hat{\generalReal}
      \left(
        \costFunction(\codebook(\codewordIndex)_\blockindex)
        -
        \costConstraint
      \right)
    \right)
  >
  1
\right)
\\
\overset{(a)}&{\leq}
\prod_{\blockindex=1}^\blocklength
  \Expectation_\codebook
    \exp\left(
      \hat{\generalReal}
      \left(
        \costFunction(\codebook(\codewordIndex)_\blockindex)
        -
        \costConstraint
      \right)
    \right)
\\
&=
\generalpmf(\hat{\generalReal})^\blocklength,
\end{align*}
where step (a) uses Markov's inequality and the independence of the codeword entries. This clearly implies
\begin{equation*}
\Probability_\codebook\left(
  \sum_{\blockindex=1}^\blocklength \costFunction(\codebook(\codewordIndex)_\blockindex) > \blocklength\costConstraint
\right)
\leq
\exp(-\blocklength\proofconst_1).
\end{equation*}
with $\proofconst_1$ defined in \eqref{eq:mgf-parameter-choice}.

This means that $\cardinality{\badCodewordsSet}$ follows a binomial distribution with $\codebookSize$ trials and some success probability $\generalRealTwo \leq \exp(-\blocklength\proofconst_1)$. Hence, for every $\proofconst \in (0,\proofconst_1)$,
\begin{align*}
\Probability_\codebook\big(
  \cardinality{\badCodewordsSet} \geq \codebookSize\exp(-\blocklength\proofconst)
\big)
&=
\Probability_\codebook\big(
  \cardinality{\badCodewordsSet}
  \geq
  \codebookSize\generalRealTwo
  +
  \codebookSize(\exp(-\blocklength\proofconst) - \generalRealTwo)
\big)
\\
&=
\Probability_\codebook\big(
  \cardinality{\badCodewordsSet}
  -
  \Expectation\cardinality{\badCodewordsSet}
  \geq
  \codebookSize(\exp(-\blocklength\proofconst) - \generalRealTwo)
\big)
\\
\overset{(a)}&{\leq}
\exp\left(
  -2\frac{\codebookSize^2(\exp(-\blocklength\proofconst) - \generalRealTwo)^2}{\codebookSize}
\right)
\\
&=
\exp\left(
  -2\codebookSize(\exp(-\blocklength\proofconst) - \generalRealTwo)^2
\right)
\\
\overset{(b)}&{\leq}
\exp\Big(
  -2\codebookSize\exp(-2\blocklength\proofconst)\big(1-\exp(-\blocklength(\proofconst_1-\proofconst))\big)^2
\Big)
\\
&=
\costBoundFunc^{(\costFunction,\costConstraint)}\left(
  \proofconst,
  \frac{\log \codebookSize}{\blocklength},
  \blocklength
\right)
\end{align*}
where step (a) is due to the Chernoff-Hoeffding inequality in the form stated, e.g., in~\cite[Theorem 2.8]{boucheron2013concentration} and step (b) uses $\generalRealTwo\leq\exp(-\blocklength\proofconst_1)\leq\exp(-\blocklength\proofconst)$. This concludes the proof of \eqref{eq:bad-codewords-all-blocklengths}. Furthermore, since $\proofconst<\proofconst_1$, it is clear that $1-\exp(-\blocklength(\proofconst_1-\proofconst))$ tends to $1$ as $\blocklength$ tends to infinity. In particular, this expression is lower bounded by, say, $1/\sqrt{2}$ for sufficiently large $\blocklength$. The bound \eqref{eq:bad-codewords-large-blocklengths} is then clear.
\end{proof}

\subsection{Preliminaries on Functional Analysis and Rényi Entropy}
\label{appendix:funcana}
In this appendix, we collect technical facts on functional analysis and Rényi entropy that we need for the proofs in this paper.

\begin{lemma}
\label{lemma:norm-basics}
The norms $\operatorNorm{\cdot}$ on $\boundedOperators{\hilbertSpace}$ and $\traceNorm{\cdot}$ on $\traceClass{\hilbertSpace}$ have the following properties:
\begin{enumerate}
 \item \label{item:norm-basics-trace-norm-bound}
 $
 \forall \generalOperator \in \traceClass{\hilbertSpace}:~
   \operatorNorm{\generalOperator} \leq \traceNorm{\generalOperator}
 $
 \item \label{item:norm-basics-trace-class-ideal}
 $
 \forall
 \generalOperator_1 \in \boundedOperators{\hilbertSpace},
 \generalOperator_2 \in \traceClass{\hilbertSpace},
 \generalOperator_3 \in \boundedOperators{\hilbertSpace}
 :~
   \generalOperator_1 \generalOperator_2 \generalOperator_3
   \in
   \traceClass{\hilbertSpace}
   \wedge
   \traceNorm{\generalOperator_1 \generalOperator_2 \generalOperator_3}
   \leq
   \operatorNorm{\generalOperator_1}
   \traceNorm{\generalOperator_2}
   \operatorNorm{\generalOperator_3}
 $
 \item \label{item:norm-basics-trace-norm-submultiplicative}
 $
 \forall
 \generalOperator_1,
 \generalOperator_2 \in \traceClass{\hilbertSpace}
 :~
   \traceNorm{\generalOperator_1 \generalOperator_2}
   \leq
   \traceNorm{\generalOperator_1} \traceNorm{\generalOperator_2}
 $
 \item \label{item:norm-basics-trace-norm-duality}
 The maps $\traceClass{\hilbertSpace} \rightarrow \dualSpace{\compactOperators{\hilbertSpace}}, \generalOperator \mapsto \trace(\generalOperator \cdot)$ and $\boundedOperators{\hilbertSpace} \rightarrow \dualSpace{\traceClass{\hilbertSpace}}, \generalOperator \mapsto \trace(\generalOperator \cdot)$ are isometric isomorphisms. In particular,
 \[
 \forall \generalOperator \in \traceClass{\hilbertSpace}:~
   \traceNorm{\generalOperator}
   =
   \sup\left\{
     \absolute{\trace(\generalOperator\generalOperatorTwo)}
     :~
     \generalOperatorTwo \in \boundedOperators{\hilbertSpace}
     \wedge
     \operatorNorm{\generalOperatorTwo} \leq 1
   \right\}.
 \]
 \item \label{item:norm-basics-zero-trace-operators}
 For any $\generalQuantumState, \generalQuantumStateTwo \in \densityOperators{\hilbertSpace}$, we have
 \[
  \traceNorm{\generalQuantumState - \generalQuantumStateTwo}
   =
   2\max\left\{
     \trace(\generalOperatorTwo(\generalQuantumState - \generalQuantumStateTwo))
     :~
     0 \leq \generalOperatorTwo \leq \identityOperator
   \right\}.
 \]

\end{enumerate}
\end{lemma}
Since these are well-known facts, we provide references to textbooks in lieu of a proof: \ref{item:norm-basics-trace-norm-bound}) follows from \cite[Section III.1, Theorem 4]{schatten1970norm}, \ref{item:norm-basics-trace-class-ideal}) follows from \cite[Section III.1, Lemma 8(v)]{schatten1970norm}, \ref{item:norm-basics-trace-norm-submultiplicative}) is immediate from \ref{item:norm-basics-trace-norm-bound}) and \ref{item:norm-basics-trace-class-ideal}), \ref{item:norm-basics-trace-norm-duality} is~\cite[Theorem VI.26]{reed1980functional}, and \ref{item:norm-basics-zero-trace-operators} is stated in~\cite[Lemma 9.1.7]{wilde2013quantum} for the finite-dimensional case, however, the proof also works for infinite-dimensional spaces without modification.

\begin{lemma}
\label{lemma:operator-compositions-continuous}
Let $\generalIndexMax \in \naturals$. Then the map $\traceClass{\hilbertSpace}^\generalIndexMax \rightarrow \traceClass{\hilbertSpace}, (\generalOperator_1, \dots, \generalOperator_\generalIndexMax) \mapsto \generalOperator_1 \circ \cdots \circ \generalOperator_\generalIndexMax$ is continuous.
\end{lemma}

\begin{proof}
The proof is by induction on $\generalIndexMax$. The case $\generalIndexMax = 1$ is clear, and the case $\generalIndexMax > 2$ follows by induction hypothesis and the case $\generalIndexMax = 2$.

Hence, the only case left to prove is $\generalIndexMax = 2$. Let $\generalOperator_1, \generalOperator_2 \in \traceClass{\hilbertSpace}$, and for all $\generalIndex \in \{1,2\}$, let $(\generalOperator_\generalIndex^{(\generalIndexTwo)})_{\generalIndexTwo \in \naturals}$ be a sequence with
\[
\lim\limits_{\generalIndexTwo \rightarrow \infty}
  \traceNorm{
    \generalOperator_\generalIndex
    -
    \generalOperator_\generalIndex^{(\generalIndexTwo)}
  }
=
0.
\]
Then we use the triangle inequality and sub-multiplicativity of the trace norm in Lemma~\ref{lemma:norm-basics}-\ref{item:norm-basics-trace-norm-submultiplicative} to argue
\begin{align*}
\traceNorm{
  \generalOperator_1
  \generalOperator_2
  -
  \generalOperator_1^{(\generalIndexTwo)}
  \generalOperator_2^{(\generalIndexTwo)}
}
&=
\traceNorm{
  \generalOperator_1
  \generalOperator_2
  -
  \generalOperator_1^{(\generalIndexTwo)}
  \generalOperator_2
  +
  \generalOperator_1^{(\generalIndexTwo)}
  \generalOperator_2
  -
  \generalOperator_1^{(\generalIndexTwo)}
  \generalOperator_2^{(\generalIndexTwo)}
}
\\
&\leq
\traceNorm{
  \generalOperator_1
  -
  \generalOperator_1^{(\generalIndexTwo)}
}
\traceNorm{
  \generalOperator_2
}
+
\traceNorm{
  \generalOperator_1^{(\generalIndexTwo)}
}
\traceNorm{
  \generalOperator_2
  -
  \generalOperator_2^{(\generalIndexTwo)}
}
\end{align*}
Since the norm of a convergent sequence is upper bounded, we can apply the limit on both sides and obtain
\[
\lim\limits_{\generalIndexTwo \rightarrow \infty}
  \traceNorm{
    \generalOperator_1
    \generalOperator_2
    -
    \generalOperator_1^{(\generalIndexTwo)}
    \generalOperator_2^{(\generalIndexTwo)}
  }
=
0.
\qedhere
\]
\end{proof}

\begin{lemma}
\label{lemma:continuous-functions-operator-measurable}
Let $\generalRealTwo \in (0,\infty)$, and let $\classicalQuantumChannel: \inputAlphabet \rightarrow \{ \generalOperator \in \traceClass{\hilbertSpace}:~ 0 \leq \generalOperator \leq \generalRealTwo \identityOperator\}$ be measurable. Let $\generalFunction: [0,\generalRealTwo] \rightarrow \reals$ be continuous with $\generalFunction(0)=0$, and assume that $\generalFunction(\classicalQuantumChannel(\inputAlphabetElement)) \in \traceClass{\hilbertSpace}$ for all $\inputAlphabetElement \in \inputAlphabet$. Then $\inputAlphabet \rightarrow \traceClass{\hilbertSpace}, \inputAlphabetElement \mapsto \generalFunction(\classicalQuantumChannel(\inputAlphabetElement))$ is measurable.

A sufficient condition for $\generalFunction(\classicalQuantumChannel(\inputAlphabetElement)) \in \traceClass{\hilbertSpace}$ is that there is $\generalConstant \in (0,\generalRealTwo]$ with $\generalFunction(t) = 0$ for all $\generalReal < \generalConstant$. 
\end{lemma}
\begin{proof}
We begin with the second part of the statement and show that if there is $\generalConstant \in (0,\generalRealTwo]$ with $\generalFunction(t) = 0$ for all $\generalReal < \generalConstant$, then $\generalFunction(\classicalQuantumChannel(\inputAlphabetElement)) \in \traceClass{\hilbertSpace}$ for all $\inputAlphabetElement \in \inputAlphabet$. To this end, we apply the spectral theorem for self-adjoint trace class operators to write
\begin{equation*}
\generalFunction(\classicalQuantumChannel(\inputAlphabetElement))
=
\sum_{\eigenvalueIndex=1}^\infty
  \generalFunction(\eigenvalue_\eigenvalueIndex(\inputAlphabetElement))
  \rankOneOperator{\eigenvector_\eigenvalueIndex(\inputAlphabetElement)},
\end{equation*}
where $(\eigenvector_\eigenvalueIndex(\inputAlphabetElement))_{\eigenvalueIndex \in \naturals}$ is a sequence of corresponding orthonormal eigenvectors and $(\eigenvalue_\eigenvalueIndex(\inputAlphabetElement))_{\eigenvalueIndex \in \naturals}$ is the non-increasing sequence of eigenvalues of $\classicalQuantumChannel(\inputAlphabetElement)$ which contains all nonzero eigenvalues counted with multiplicity and converges to $0$. The Riesz-Schauder theorem (see, e.g.,~\cite[Theorem VI.15]{reed1980functional} ensures that it is possible to arrange the eigenvalues in such a way. For every $\inputAlphabetElement$, we have $\generalFunction(\eigenvalue_\eigenvalueIndex(\inputAlphabetElement)) \neq 0$ for only finitely many $\eigenvalueIndex$, and thus $\generalFunction(\classicalQuantumChannel(\inputAlphabetElement)) \in \traceClass{\hilbertSpace}$.

Let us briefly outline the proof of the first part of the lemma statement. We will first prove that if $\generalPolynomial: \reals \rightarrow \reals$ is a polynomial function with $\generalPolynomial(0) = 0$, 
\begin{equation}
\label{eq:continuous-functions-operator-measurable-polynomials}
\inputAlphabet \rightarrow \traceClass{\hilbertSpace},
\inputAlphabetElement
\mapsto
\generalPolynomial(\classicalQuantumChannel(\inputAlphabetElement))
\textrm{ is measurable.}
\end{equation}
Next, we will infer that
\begin{equation}
\label{eq:continuous-functions-operator-measurable-weak-measurability}
\forall \hilbertSpaceElement_1, \hilbertSpaceElement_2 \in \hilbertSpace
:~
\inputAlphabet \rightarrow \complexNumbers,
\inputAlphabetElement
\mapsto
\innerProduct
  {\hilbertSpaceElement_1}
  {
    \generalFunction(\classicalQuantumChannel(\inputAlphabetElement)) \hilbertSpaceElement_2
  }
\textrm{ is measurable.}
\end{equation}
Finally, we will use (\ref{eq:continuous-functions-operator-measurable-weak-measurability}) to verify the criterion \cite[Theorem E.9]{cohn2013measure} for measurability of operator valued functions.

For (\ref{eq:continuous-functions-operator-measurable-polynomials}), we note that since $\generalPolynomial(0) = 0$, we can write
\begin{equation*}
\generalPolynomial(\classicalQuantumChannel(\inputAlphabetElement))
=
\sum\limits_{\generalIndex=1}^\generalIndexMax
  \generalPolynomialCoefficient_\generalIndex
  \classicalQuantumChannel(\inputAlphabetElement)^\generalIndex.
\end{equation*}
Each summand is a composition of the measurable map $\classicalQuantumChannel^\generalIndex: \inputAlphabet \rightarrow \traceClass{\hilbertSpace}^\generalIndex, \inputAlphabetElement \mapsto (\classicalQuantumChannel(\inputAlphabetElement), \dots, \classicalQuantumChannel(\inputAlphabetElement))$, and the product of operators which is continuous by Lemma~\ref{lemma:operator-compositions-continuous}. Both the trace class and the set of measurable functions are closed under finite summations, hence we obtain (\ref{eq:continuous-functions-operator-measurable-polynomials}).

In order to prove (\ref{eq:continuous-functions-operator-measurable-weak-measurability}), we apply the Weierstraß approximation theorem to obtain a sequence of polynomials $(\generalPolynomial_\generalIndex)_{\generalIndex \in \naturals}$ with $\generalPolynomial_\generalIndex(0) = 0$ for every $\generalIndex$ and
\[
\lim\limits_{\generalIndex \rightarrow \infty}
\maxnorm{
  \generalFunction
  -
  \generalPolynomial_\generalIndex
}
=
0,
\]
where the difference is pointwise and $\generalPolynomial_\generalIndex$ is implicitly identified with its restriction to $[0,\generalRealTwo]$.
By the continuity of the continuous functional calculus, this implies, for every $\inputAlphabetElement$,
\[
\lim\limits_{\generalIndex \rightarrow \infty}
\operatorNorm{
  \generalFunction(\classicalQuantumChannel(\inputAlphabetElement))
  -
  \generalPolynomial_\generalIndex(\classicalQuantumChannel(\inputAlphabetElement))
}
=
0.
\]
The map $\boundedOperators{\hilbertSpace} \rightarrow \complexNumbers, \generalOperator \mapsto \innerProduct{\hilbertSpaceElement_1}{\generalOperator \hilbertSpaceElement_2}$ is continuous for every fixed $\hilbertSpaceElement_1,\hilbertSpaceElement_2 \in \hilbertSpace$, so it follows that
\begin{equation}
\label{eq:operator-functions-weakly-measurable-limit-representation}
\lim\limits_{\generalIndex \rightarrow \infty}
  \innerProduct{\hilbertSpaceElement_1}
               {\generalPolynomial_\generalIndex(\classicalQuantumChannel(\inputAlphabetElement)) \hilbertSpaceElement_2}
=
\innerProduct{\hilbertSpaceElement_1}
             {\generalFunction(\classicalQuantumChannel(\inputAlphabetElement)) \hilbertSpaceElement_2}.
\end{equation}
Clearly, the identity map $\traceClass{\hilbertSpace} \rightarrow \boundedOperators{\hilbertSpace}$ is also continuous and $\inputAlphabet \rightarrow \traceClass{\hilbertSpace}, \inputAlphabetElement \mapsto \generalPolynomial_\generalIndex(\classicalQuantumChannel(\inputAlphabetElement))$ is measurable by (\ref{eq:continuous-functions-operator-measurable-polynomials}). Therefore, (\ref{eq:operator-functions-weakly-measurable-limit-representation}) is a representation of the map $\inputAlphabetElement \mapsto \innerProduct{\hilbertSpaceElement_1}{ \generalFunction(\classicalQuantumChannel(\inputAlphabetElement)) \hilbertSpaceElement_2}$ as a pointwise limit of measurable functions, hence we obtain (\ref{eq:continuous-functions-operator-measurable-weak-measurability}).

Clearly, $\image \generalFunction \circ \classicalQuantumChannel \subseteq \traceClass{\hilbertSpace}$ is separable due to the separability of $\hilbertSpace$, so the remaining criterion~\cite[Theorem E.9]{cohn2013measure} requires us to verify that for every $\generalFunctional \in \dualSpace{\traceClass{\hilbertSpace}}$, the map $\inputAlphabet \rightarrow \complexNumbers, \inputAlphabetElement \mapsto \generalFunctional(\generalFunction(\classicalQuantumChannel(\inputAlphabetElement)))$ is measurable. By Lemma~\ref{lemma:norm-basics}-\ref{item:norm-basics-trace-norm-duality}, the map
\[
\boundedOperators{\hilbertSpace}
\rightarrow
\dualSpace{\traceClass{\hilbertSpace}},
\generalOperator \mapsto \trace(\generalOperator~\cdot)
\]
is an isometric isomorphism. Therefore, it remains to show that for all $\generalOperator \in \boundedOperators{\hilbertSpace}$, the map $\inputAlphabet \rightarrow \complexNumbers, \inputAlphabetElement \mapsto \trace(\generalOperator \generalFunction(\classicalQuantumChannel(\inputAlphabetElement)))$ is measurable.

By Lemma~\ref{lemma:norm-basics}-\ref{item:norm-basics-trace-class-ideal}, $\generalOperator \generalFunction(\classicalQuantumChannel(\inputAlphabetElement)) \in \traceClass{\hilbertSpace}$. Pick an orthonormal basis $(\onbVector_\generalIndex)_{\generalIndex \in \naturals}$ of $\hilbertSpace$. Then, by~\cite[Theorem VI.24]{reed1980functional},
\[
\trace(\generalOperator \generalFunction(\classicalQuantumChannel(\inputAlphabetElement)))
:=
\sum_{\generalIndex=1}^\infty
  \innerProduct{\onbVector_\generalIndex}
               {\generalOperator \generalFunction(\classicalQuantumChannel(\inputAlphabetElement))\onbVector_\generalIndex}
\]
converges absolutely and regardless of the choice of basis. For every $\generalIndexMax \in \naturals$, we have
\[
\sum\limits_{\generalIndex=1}^\generalIndexMax
  \innerProduct{\onbVector_\generalIndex}
               {\generalOperator \generalFunction(\classicalQuantumChannel(\inputAlphabetElement))\onbVector_\generalIndex}
=
\sum\limits_{\generalIndex=1}^\generalIndexMax
  \innerProduct{\adjoint{\generalOperator}\onbVector_\generalIndex}
               {\generalFunction(\classicalQuantumChannel(\inputAlphabetElement))\onbVector_\generalIndex},
\]
which is a measurable function of $\inputAlphabetElement$ by (\ref{eq:continuous-functions-operator-measurable-weak-measurability}). So we have written the map $\inputAlphabet \rightarrow \complexNumbers, \inputAlphabetElement \mapsto \trace(\generalOperator \generalFunction(\classicalQuantumChannel(\inputAlphabetElement)))$ as a pointwise limit of measurable maps.
\end{proof}

\begin{lemma}
\label{lemma:spectral-decompositions-measurable}
Let $\classicalQuantumChannel: \inputAlphabet \rightarrow \densityOperators{\hilbertSpace}$ be measurable. For every $\eigenvalueIndex \in \naturals$, define a map $\eigenvalue_\eigenvalueIndex: \inputAlphabet \rightarrow [0,1]$ in such a way that for every $\inputAlphabetElement \in \inputAlphabet$, $(\eigenvalue_\eigenvalueIndex(\inputAlphabetElement))_{\eigenvalueIndex \in \naturals}$ is the non-increasing sequence which contains all nonzero eigenvalues of $\classicalQuantumChannel(\inputAlphabetElement)$ counted with multiplicity.

Then, we have the following:
\begin{enumerate}
 \item \label{item:spectral-decompositions-measurable-eigenvalues}
 For every $\eigenvalueIndex \in \naturals$, the map $\eigenvalue_\eigenvalueIndex: \inputAlphabet \rightarrow [0,1]$ is measurable.
 \item \label{item:spectral-decompositions-measurable-indicators}
 Let $0 < \generalIntervalLowerBound < \generalIntervalUpperBound \leq 1$. Then
 $
 \inputAlphabet \rightarrow \traceClass{\hilbertSpace},
 \inputAlphabetElement
 \mapsto
 \indicatorFunction{(\generalIntervalLowerBound,\generalIntervalUpperBound)}
   (\classicalQuantumChannel(\inputAlphabetElement))
 $
 is measurable.
\end{enumerate}
\end{lemma}

\begin{proof}
For item \ref{item:spectral-decompositions-measurable-eigenvalues}), we let $\singularValue_\eigenvalueIndex: \densityOperators{\hilbertSpace} \rightarrow [0,1]$ be the map that maps each $\generalOperator$ to its $\eigenvalueIndex$-th largest eigenvalue (counted with multiplicity). We note that $\singularValue_\eigenvalueIndex \circ \classicalQuantumChannel = \eigenvalue_\eigenvalueIndex$. Since by Lemma~\ref{lemma:norm-basics}-\ref{item:norm-basics-trace-norm-bound}, the operator norm is upper bounded by the trace norm, \cite[Application (b) following Theorem III.9.1]{gohberg1980basic} implies that the map $\singularValue_\eigenvalueIndex$ is continuous. Therefore, $\eigenvalue_\eigenvalueIndex$ is measurable, since it is the composition of a continuous and a measurable map.

For item \ref{item:spectral-decompositions-measurable-indicators}), we define a sequence $(\generalFunction_\generalIndex)_{\generalIndex \in \naturals}$ of functions $[0,1] \rightarrow \reals$ as follows:
\[
\generalFunction_\generalIndex(\generalReal)
:=
\begin{cases}
  1, &\generalReal \in [\generalIntervalLowerBound + \frac{1}{\generalIndex},
                        \generalIntervalUpperBound - \frac{1}{\generalIndex}] \\
  \generalIndex(\generalReal - \generalIntervalLowerBound),
     &\generalReal \in (\generalIntervalLowerBound,
                        \generalIntervalLowerBound + \frac{1}{\generalIndex}) \\
  \generalIndex(\generalIntervalUpperBound - \generalReal),
     &\generalReal \in (\generalIntervalUpperBound - \frac{1}{\generalIndex},
                        \generalIntervalUpperBound) \\
  0, &\generalReal \notin (\generalIntervalLowerBound, \generalIntervalUpperBound).
\end{cases}
\]
It is straightforward to verify that $(\generalFunction_\generalIndex)_{\generalIndex \in \naturals}$ is a sequence of continuous functions and that for all $\generalIndex$ and all $\generalReal < \generalIntervalLowerBound$, we have $\generalFunction_\generalIndex(\generalReal) = 0$. This allows us to apply Lemma~\ref{lemma:continuous-functions-operator-measurable} and conclude that the functions $\inputAlphabet \rightarrow \traceClass{\hilbertSpace}, \inputAlphabetElement \mapsto \generalFunction_\generalIndex(\classicalQuantumChannel(\inputAlphabetElement))$ are measurable.

It remains to show that for every $\inputAlphabetElement$, $\generalFunction_\generalIndex(\classicalQuantumChannel(\inputAlphabetElement))$ converges to $\indicatorFunction{(\generalIntervalLowerBound,\generalIntervalUpperBound)}(\classicalQuantumChannel(\inputAlphabetElement))$ in trace norm. We are going to prove the stronger statement
\[
\forall \inputAlphabetElement \in \inputAlphabet~
\exists \generalIndex' \in \naturals~
\forall \generalIndex \geq \generalIndex':
  \generalFunction_\generalIndex(\classicalQuantumChannel(\inputAlphabetElement))
  =
  \indicatorFunction{(\generalIntervalLowerBound,\generalIntervalUpperBound)}(\classicalQuantumChannel(\inputAlphabetElement)).
\]
To this end, we observe that we have equality if $\generalFunction_\generalIndex$ and $\indicatorFunction{(\generalIntervalLowerBound,\generalIntervalUpperBound)}$ coincide on $\spectrum{\classicalQuantumChannel(\inputAlphabetElement)}$. It can easily be verified that these functions coincide everywhere but on $(\generalIntervalLowerBound, \generalIntervalLowerBound + \frac{1}{\generalIndex})$ and $(\generalIntervalUpperBound - \frac{1}{\generalIndex}, \generalIntervalUpperBound)$. But if $\spectrum{\classicalQuantumChannel(\inputAlphabetElement)}$ has nonempty intersection with one of these intervals for infinitely many $\generalIndex$, we have found an accumulation point of $\spectrum{\classicalQuantumChannel(\inputAlphabetElement)}$ either at $\generalIntervalLowerBound$ or $\generalIntervalUpperBound$. By the Riesz-Schauder theorem (see, e.g.,~\cite[Theorem VI.15]{reed1980functional}), $\spectrum{\classicalQuantumChannel(\inputAlphabetElement)}$ does not have nonzero accumulation points, so this contradicts $\generalIntervalLowerBound, \generalIntervalUpperBound \neq 0$.
\end{proof}

\begin{lemma}
\label{lemma:bochner-integral-basics}
Let $(\probabilitySpace, \probabilitySpaceAlgebra, \generalPMeasure)$ be a probability space, $\hilbertSpace$ a separable Hilbert space and let $\generalOperatorRV: \probabilitySpace \rightarrow  \traceClass{\hilbertSpace}$ be a random variable. Then, we have the following:
\begin{enumerate}
 \item \label{item:bochner-integral-basics-existence}
 The Bochner integral $\Expectation \generalOperatorRV \in \traceClass{\hilbertSpace}$ exists iff the Lebesgue integral $\Expectation \traceNorm{\generalOperatorRV}$ exists and is finite.
 \item \label{item:bochner-integral-basics-trace-norm-jensen}
 $
 \traceNorm{\Expectation \generalOperatorRV}
 \leq
 \Expectation \traceNorm{\generalOperatorRV}
 $
 \item \label{item:bochner-integral-basics-bounded-functional}
 Let $\generalOperatorTwo$ be a bounded linear functional on $\traceClass{\hilbertSpace}$. Then, if $\Expectation \generalOperator$ exists, $\Expectation(\generalOperatorTwo(\generalOperator))$ exists and $\generalOperatorTwo (\Expectation \generalOperator) = \Expectation(\generalOperatorTwo (\generalOperator))$,
 \item \label{item:bochner-integral-basics-trace}
 If $\Expectation \generalOperatorRV$ exists, then $\Expectation \trace \generalOperatorRV$ exists and $\trace \Expectation \generalOperatorRV = \Expectation \trace \generalOperatorRV$.
 \item \label{item:bochner-integral-basics-square-root-jensen}
 If for all $\probabilitySpaceElement \in \probabilitySpace$, we have $\generalOperatorRV(\probabilitySpaceElement) \geq 0$ and $\sqrt{\generalOperatorRV(\probabilitySpaceElement)} \in \traceClass{\hilbertSpace}$, then
  \[
  \Expectation\sqrt{\generalOperatorRV}
  \leq
  \sqrt{\Expectation{\generalOperatorRV}}.
  \]
\end{enumerate}
\end{lemma}

\begin{proof}
For \ref{item:bochner-integral-basics-existence}) and \ref{item:bochner-integral-basics-trace-norm-jensen}), we note that the separability of $\hilbertSpace$ implies that $\traceClass{\hilbertSpace}$ is separable, which means that $\generalOperatorRV$ has separable range. Therefore, the statements are proven in~\cite[Section V.5, Theorem 1 and Corollary 1]{yosida1980functional}. \ref{item:bochner-integral-basics-bounded-functional}) is a special case of~\cite[Section V.5, Corollary 2]{yosida1980functional}. \ref{item:bochner-integral-basics-trace}) is a special case of \ref{item:bochner-integral-basics-bounded-functional}), since the trace is a bounded linear operator~\cite[Section III.1, Lemmas 4 and 8(vi)]{schatten1970norm}.

For \ref{item:bochner-integral-basics-square-root-jensen}), first observe that the measurability of $\probabilitySpaceElement \mapsto \sqrt{\generalOperator(\probabilitySpaceElement)}$ follows from Lemma~\ref{lemma:continuous-functions-operator-measurable}. We apply the theorem in~\cite{to1975generalized} in the form of eq. (K') with the function $\traceClass{\hilbertSpace} \rightarrow \traceClass{\hilbertSpace}, \generalOperator \mapsto \adjoint{\generalOperator} \generalOperator$ substituted for $g$. To this end, we first have to note that the function is continuous by Lemma~\ref{lemma:operator-compositions-continuous} and that it satisfies the convexity condition in~\cite{to1975generalized}. In order to show the latter, we calculate for $\generalReal \in [0,1]$ and $\generalOperator_1, \generalOperator_2 \in \traceClass{\hilbertSpace}$:
\begin{align*}
&\hphantom{{}={}}
\adjoint{
  \big(
    \generalReal \generalOperator_1
    +
    (1-\generalReal) \generalOperator_2
  \big)
}
\big(
  \generalReal \generalOperator_1
  +
  (1-\generalReal) \generalOperator_2
\big)
\\
&=
\generalReal^2 \adjoint{\generalOperator_1} \generalOperator_1
+
\generalReal(1-\generalReal)(
  \adjoint{\generalOperator_1} \generalOperator_2
  +
  \adjoint{\generalOperator_2} \generalOperator_1
)
+
(1-\generalReal)^2 \adjoint{\generalOperator_2} \generalOperator_2
\\
&=
\generalReal \adjoint{\generalOperator_1} \generalOperator_1
+
(1-\generalReal) \adjoint{\generalOperator_2} \generalOperator_2
+
(\generalReal^2 - \generalReal) \adjoint{\generalOperator_1} \generalOperator_1
+
\big((1-\generalReal)^2-(1-\generalReal)\big) \adjoint{\generalOperator_2} \generalOperator_2
+
\generalReal(1-\generalReal)(
  \adjoint{\generalOperator_1} \generalOperator_2
  +
  \adjoint{\generalOperator_2} \generalOperator_1
)
\\
&=
\generalReal \adjoint{\generalOperator_1} \generalOperator_1
+
(1-\generalReal) \adjoint{\generalOperator_2} \generalOperator_2
-
\generalReal(1-\generalReal)
\big(
  \adjoint{\generalOperator_1} \generalOperator_1
  -
  \adjoint{\generalOperator_1} \generalOperator_2
  -
  \adjoint{\generalOperator_2} \generalOperator_1
  +
  \adjoint{\generalOperator_2} \generalOperator_2
\big)
\\
&=
\generalReal \adjoint{\generalOperator_1} \generalOperator_1
+
(1-\generalReal) \adjoint{\generalOperator_2} \generalOperator_2
-
\generalReal(1-\generalReal)
\adjoint{
  \big(
    \generalOperator_1
    -
    \generalOperator_2
  \big)
}
\big(
  \generalOperator_1
  -
  \generalOperator_2
\big)
\\
&\leq
\generalReal \adjoint{\generalOperator_1} \generalOperator_1
+
(1-\generalReal) \adjoint{\generalOperator_2} \generalOperator_2.
\end{align*}
We can now argue
\begin{equation*}
\Expectation\sqrt{\generalOperatorRV}
\overset{(a)}{=}
\sqrt{
  \adjoint{
    \left(
      \Expectation \sqrt{\generalOperatorRV}
    \right)
  }
  \Expectation \sqrt{\generalOperatorRV}
}
\overset{(b)}{\leq}
\sqrt{
  \Expectation
  \left(
    \adjoint{\sqrt{\generalOperatorRV}}
    \sqrt{\generalOperatorRV}
  \right)
}
\overset{(c)}{=}
\sqrt{\Expectation{\generalOperatorRV}}.
\end{equation*}
(a) is due to $\Expectation \sqrt{\generalOperatorRV} \geq 0$, (b) is the application of~\cite[eq. (K')]{to1975generalized} combined with the fact that the square root is operator monotone~\cite{pedersen1972some}, and (c) follows directly from $\generalOperatorRV \geq 0$.
\end{proof}

\begin{lemma}
\label{lemma:renyi-entropy-basics-preliminary}
Let $\generalpmf: \naturals \rightarrow [0,1]$ be a \gls{pmf} with
\[
\sum\limits_{\generalIndex \in \naturals}
  \generalpmf(\generalIndex)^{\renyiorder_{\min}}
  <
  \infty.
\]
Then the function
\begin{equation}
\label{eq:renyi-entropy-basics-preliminary-function}
\generalFunction_\generalpmf:
\renyiorder
\mapsto
  \sum\limits_{\generalIndex \in \naturals}
    \generalpmf(\generalIndex)^\renyiorder
\end{equation}
is continuously differentiable on $(\renyiorder_{\min},\infty)$, and its derivative is
\begin{equation}
\label{eq:renyi-entropy-basics-preliminary-derivative}
\generalFunction_\generalpmf':
\renyiorder
\mapsto
  \sum\limits_{\generalIndex \in \naturals}
    \generalpmf(\generalIndex)^\renyiorder
    \log \generalpmf(\generalIndex)
\in
(-\infty,0].
\end{equation}
Moreover, for every $\renyiorder_0 > \renyiorder_{\min}$, there exists $\generalFunctionTwo_{\max}(\renyiorder_0) \in [0, \infty)$ such that for all $\renyiorder \in [\renyiorder_0,\infty)$
\begin{equation}
\label{eq:renyi-entropy-basics-preliminary-derivative-bound}
\absolute{
  \generalFunction_\generalpmf'(\renyiorder)
}
\leq
\generalFunctionTwo_{\max}(\renyiorder_0)
\sum\limits_{\generalIndex \in \naturals}
  \generalpmf(\generalIndex)^{\renyiorder_{\min}}
<
\infty
.
\end{equation}
\end{lemma}
\begin{proof}
Fix an arbitrary $\renyiorder_0 > \renyiorder_{\min}$. Because $\renyiorder_0$ can be chosen arbitrarily close to $\renyiorder_{\min}$, it is sufficient to compute $\generalFunction'$ on $[\renyiorder_0,\infty)$. By~\cite[Theorem 7.17]{rudin1976principles}, we only have to show that
\[
\sum\limits_{\generalIndex \in \naturals}
  \frac{d}{d \renyiorder}
  \generalpmf(\generalIndex)^\renyiorder
=
\sum\limits_{\generalIndex \in \naturals}
  \generalpmf(\generalIndex)^\renyiorder
  \log \generalpmf(\generalIndex)
\]
converges uniformly for $\renyiorder \in [\renyiorder_0,\infty)$. To this end, we write
\[
-
\generalpmf(\generalIndex)^\renyiorder
\log \generalpmf(\generalIndex)
=
-
\generalpmf(\generalIndex)^{\renyiorder_{\min}}
\generalpmf(\generalIndex)^{\renyiorder_0 - \renyiorder_{\min}}
\generalpmf(\generalIndex)^{\renyiorder - \renyiorder_0}
\log \generalpmf(\generalIndex)
\leq
-
\generalpmf(\generalIndex)^{\renyiorder_{\min}}
\generalpmf(\generalIndex)^{\renyiorder_0 - \renyiorder_{\min}}
\log \generalpmf(\generalIndex)
\]
By the Weierstraß criterion, we only have to show that
\[
\sum\limits_{\generalIndex \in \naturals}
\Big(
  -
  \generalpmf(\generalIndex)^{\renyiorder_{\min}}
  \generalpmf(\generalIndex)^{\renyiorder_0 - \renyiorder_{\min}}
  \log \generalpmf(\generalIndex)
\Big)
\]
converges. To this end, we first examine the function
\[
\generalFunctionTwo: (0,1] \rightarrow [0,\infty),
\generalReal
\mapsto
-
\generalReal^{\renyiorder_0 - \renyiorder_{\min}}
\log \generalReal.
\]
Clearly, $\generalFunctionTwo(1) = 0$. Moreover, by substituting $\hat{\generalReal} := - \log \generalReal$, we get
\[
\lim\limits_{\generalReal \rightarrow 0}
  \generalFunctionTwo(\generalReal)
=
\lim\limits_{\generalReal \rightarrow 0}
\Big(
  -
  \generalReal^{\renyiorder_0 - \renyiorder_{\min}}
  \log \generalReal
\Big)
=
\lim\limits_{\hat{\generalReal} \rightarrow \infty}
  \hat{\generalReal}
  \exp\left(-\hat{\generalReal}(\renyiorder_0 - \renyiorder_{\min})\right)
=
0,
\]
so $\generalFunctionTwo$ can be extended to a continuous function on the whole interval $[0,1]$ with $\generalFunctionTwo(0) = \generalFunctionTwo(1) = 0$. This means that $\generalFunctionTwo$ takes a finite maximum value $\generalFunctionTwo_{\max}$ on $[0,1]$. So we have
\[
\sum\limits_{\generalIndex \in \naturals}
\Big(
  -
  \generalpmf(\generalIndex)^{\renyiorder_{\min}}
  \generalpmf(\generalIndex)^{\renyiorder_0 - \renyiorder_{\min}}
  \log \generalpmf(\generalIndex)
\Big)
\leq
\generalFunctionTwo_{\max}
\sum\limits_{\generalIndex \in \naturals}
  \generalpmf(\generalIndex)^{\renyiorder_{\min}}
<
\infty
\]
by the assumption of the lemma, concluding the proof.
\end{proof}

\begin{lemma}
\label{lemma:renyi-entropy-basics}
Let $\inputDistribution$ and $\classicalQuantumChannel$ be such that (\ref{eq:technical-assumptions}) holds, let $\jointEntropy$ be as defined as in \ref{eq:joint-entropy}, $\outputEntropy$ as defined in (\ref{eq:output-entropy}), $\conditionalRenyiEntropy{\renyiorder}$ as defined in (\ref{eq:conditional-renyi-entropy}), and $\outputRenyiEntropy{\renyiorder}$ as defined in (\ref{eq:renyi-entropy}).

Then $\renyiorder \mapsto \outputRenyiEntropy{\renyiorder}$ and $\renyiorder \mapsto \conditionalRenyiEntropy{\renyiorder}$ are continuously differentiable on $(\renyiorder_{\min},1) \cup (1,\infty)$. Moreover, we have
\begin{align}
\label{eq:renyi-entropy-basics-output}
\lim\limits_{\renyiorder \rightarrow 1}
  \outputRenyiEntropy{\renyiorder}
&=
\outputEntropy
\\
\label{eq:renyi-entropy-basics-joint}
\lim\limits_{\renyiorder \rightarrow 1}
  \conditionalRenyiEntropy{\renyiorder}
&=
\jointEntropy.
\end{align}
\end{lemma}
\begin{proof}
In order to apply Lemma~\ref{lemma:renyi-entropy-basics-preliminary}, we define functions $\generalFunction_\outputDistribution$ and $\generalFunction_{\inputDistribution(\cdot | \inputRV)}$ as in (\ref{eq:renyi-entropy-basics-preliminary-function}), where the \gls{pmf} $\outputDistribution$ (respectively $\inputDistribution(\cdot | \inputRV)$) is substituted for $\generalpmf$. Observe that
\[
\outputRenyiEntropy{\renyiorder}
=
\frac{1}{1-\renyiorder}
\log
  \generalFunction_\outputDistribution(\renyiorder),
\]
where $\generalFunction_\outputDistribution$ is defined by (\ref{eq:renyi-entropy-basics-preliminary-function}) with $\generalpmf:=\outputDistribution$. Therefore, the continuous differentiability follows immediately from Lemma~\ref{lemma:renyi-entropy-basics-preliminary}.

The argument for the continuous differentiability of
\[
\conditionalRenyiEntropy{\renyiorder}
=
\frac{1}{1-\renyiorder}
\log\Expectation
  \generalFunction_{\inputDistribution(\cdot | \inputRV)}
    (\renyiorder)
\]
is similar, but we need an additional argument that
\begin{equation}
\label{eq:renyi-entropy-basics-intermediate-derivative}
\frac{d}{d\renyiorder}
\Expectation
  \generalFunction_{\inputDistribution(\cdot | \inputRV)}
    (\renyiorder)
=
\Expectation
  \generalFunction_{\inputDistribution(\cdot | \inputRV)}'
    (\renyiorder),
\end{equation}
(and, implicitly, that the derivative on the left-hand side exists) where $\generalFunction_{\inputDistribution(\cdot | \inputRV)}'$ is given in (\ref{eq:renyi-entropy-basics-preliminary-derivative}). It is sufficient to show that this is the case for all $\renyiorder > \renyiorder_0 > \renyiorder_{\min}$. We argue this with the criterion~\cite[Theorem 16.8]{billingsley2012probability} for interchangeability of differentiation and expectation, which requires us to show that $\generalFunction_{\inputDistribution(\cdot | \inputAlphabetElement)}'$ is a continuous function of $\renyiorder$ for fixed $\inputAlphabetElement$ and that $\absolute{\generalFunction_{\inputDistribution(\cdot | \inputRV)}'(\renyiorder)}$ has an upper bound that is uniform in $\renyiorder$ and integrable over $\inputRV$. The former clearly follows from Lemma~\ref{lemma:renyi-entropy-basics-preliminary}, and for the latter, we note the bound (\ref{eq:renyi-entropy-basics-preliminary-derivative-bound}) in Lemma~\ref{lemma:renyi-entropy-basics-preliminary} and observe that
\[
\sum\limits_{\classicalOutputIndex \in \naturals}
  \inputDistribution(\classicalOutputIndex | \inputAlphabetElement)^{\renyiorder_{\min}}
=
\trace\left(
  \classicalQuantumChannel(\inputAlphabetElement)^{\renyiorder_{\min}}
\right),
\]
so by the interchangeability of expectation and trace, the integrability follows from the assumption that $\Expectation \classicalQuantumChannel(\inputRV)^{\renyiorder_{\min}} \in \traceClass{\hilbertSpace}$ exists.

In order to show (\ref{eq:renyi-entropy-basics-output}), we note that
\[
\lim\limits_{\renyiorder \rightarrow 1}
  \outputRenyiEntropy{\renyiorder}
=
-
\lim\limits_{\renyiorder \rightarrow 1}
  \frac{
    \log
      \generalFunction_\outputDistribution(\renyiorder)
    -
    \log
      \generalFunction_\outputDistribution(1)
}{
  \renyiorder - 1
}
=
(\log \circ \generalFunction_\outputDistribution)'
  (1).
\]
Lemma~\ref{lemma:renyi-entropy-basics-preliminary} and the chain rule yield
\[
(\log \circ \generalFunction_\outputDistribution)'
  (1)
=
\frac{\generalFunction'_\outputDistribution(1)}
     {\generalFunction_\outputDistribution(1)}
=
\outputEntropy.
\]

For (\ref{eq:renyi-entropy-basics-joint}), we have
\[
\lim\limits_{\renyiorder \rightarrow 1}
  \conditionalRenyiEntropy{\renyiorder}
=
-
\lim\limits_{\renyiorder \rightarrow 1}
  \frac{
    \log\Expectation
      \generalFunction_{\inputDistribution(\cdot | \inputRV)}
        (\renyiorder)
    -
    \log\Expectation
      \generalFunction_{\inputDistribution(\cdot | \inputRV)}
        (1)
}{
  \renyiorder - 1
}
=
\generalFunctionTwo'
  (1),
\]
where $\generalFunctionTwo'$ is the derivative of
\[
\generalFunctionTwo:
  \renyiorder
  \mapsto
  \log\Expectation
    \generalFunction_{\inputDistribution(\cdot | \inputRV)}
      (\renyiorder).
\]
We use (\ref{eq:renyi-entropy-basics-intermediate-derivative}) and the chain rule for differentiation to obtain
\[
\generalFunctionTwo'(\renyiorder)
=
\frac{
  \Expectation \generalFunction'_{\inputDistribution(\cdot | \inputRV)}(\renyiorder)
}{
  \Expectation \generalFunction_{\inputDistribution(\cdot | \inputRV)}(\renyiorder)
}
=
\frac{
  \Expectation \sum_{\classicalOutputIndex \in \naturals}
    \inputDistribution(\classicalOutputIndex | \inputRV)^\renyiorder
    \log \inputDistribution(\classicalOutputIndex | \inputRV)
}{
  \Expectation \sum_{\classicalOutputIndex \in \naturals}
    \inputDistribution(\classicalOutputIndex | \inputRV)^\renyiorder
},
\]
so clearly, $\generalFunctionTwo'(1) = \jointEntropy$, concluding the proof of the lemma.
\end{proof}

\bibliographystyle{IEEEtran}
\bibliography{quantum-resolvability-references}
\end{document}